%% file: final-arxiv.tex
\documentclass[10pt]{article}
\usepackage[lmargin=1in,rmargin=1in,tmargin=1in,bmargin=1in]{geometry}

\usepackage{amsmath,amsthm,amsfonts,amssymb,comment,slashed,mathtools}
\usepackage[utf8]{inputenc}
\usepackage[T1]{fontenc}  
\usepackage[english]{babel}
\usepackage{xcolor}

\usepackage{csquotes}
\usepackage{bm}
\usepackage[shortlabels]{enumitem}
\usepackage[affil-it]{authblk}

\usepackage[section]{placeins}
\usepackage{soul} 

\usepackage{bm}
\usepackage{mathrsfs, graphicx} 
\usepackage{stmaryrd}
\usepackage{epigraph}

\graphicspath{{figures/}}

\usepackage[backend=biber,style=alphabetic,giveninits=true,doi=false,isbn=false,url=false,sorting=nyt,maxbibnames=99,date=year]{biblatex}
\renewbibmacro{in:}{}
\usepackage{chngcntr}

\usepackage[%
bookmarks=true,
colorlinks=true,
linkcolor=black,
urlcolor=black,
citecolor=black,
plainpages=false,
pdfpagelabels,
final]{hyperref}
\usepackage[capitalise,nameinlink]{cleveref}

\numberwithin{equation}{section}

\theoremstyle{plain}
\newtheorem{thm}{Theorem}
\newtheorem{prop}{Proposition}[section]
\newtheorem{lem}[prop]{Lemma}
\newtheorem*{thm*}{Theorem}
\newtheorem{thmrough}{Theorem}

\newtheorem*{conj*}{Conjecture}
\newtheorem{conj}{Conjecture}
\theoremstyle{definition}
\newtheorem{defn}[prop]{Definition}

\theoremstyle{remark}
\newtheorem{rk}[prop]{Remark}
\crefname{thm}{Theorem}{Theorems} 
\crefname{thmrough}{Theorem}{Theorems} 
\crefname{lem}{Lemma}{Lemmas} 
\crefname{conj}{Conjecture}{Conjectures}
\crefname{prop}{Proposition}{Propositions}

\renewcommand{\Bbb}{\mathbb}
\newcommand{\ve}{\varepsilon}
\newcommand{\les}{\lesssim}

\newcommand{\id}{\operatorname{id}}

\newcommand{\Ric}{\mathrm{Ric}}
\newcommand{\out}{\mathrm{out}}
\newcommand{\ing}{\mathrm{in}}

\newcommand{\loc}{\mathrm{loc}}

\newcommand{\Int}{\operatorname{int}}

\newcommand{\R}{\Bbb R}

\newcommand{\stab}{\mathrm{stab}}

\usepackage{accents}

\makeatletter
\renewcommand{\paragraph}{%
  \@startsection{paragraph}{4}%
  {\z@}{1.25ex \@plus 1ex \@minus .2ex}{-1em}%
  {\normalfont\normalsize\bfseries}%
}
\makeatother

\begin{document}

\title{Nonlinear stability of extremal Reissner--Nordstr\"om black holes \\ in spherical symmetry}

\author[1]{Yannis Angelopoulos\thanks{yannis@bimsa.cn}}
\author[2]{Christoph~Kehle\thanks{kehle@mit.edu}}
\author[3]{Ryan Unger\thanks{runger@berkeley.edu}}
\affil[1]{\small Beijing Institute of Mathematical Sciences and Applications,

No.~544, Hefangkou Village, Huairou District, 101408 Beijing, China \vskip.1pc \ 
}
\affil[2]{\small  Massachusetts Institute of Technology, Department of Mathematics,

Building 2, 77 Massachusetts Avenue, Cambridge, MA 02139, United States of America \vskip.1pc \ 
}
 \affil[3]{\small  University of California, Berkeley,  Department of Mathematics,
	
970 Evans Hall, Berkeley, CA 94720, United States of America \vskip.1pc \  
	}

\date{January 14, 2026}

\maketitle

\begin{abstract}
In this paper, we prove the codimension-one nonlinear asymptotic stability of the extremal Reissner--Nordstr\"om family of black holes in the spherically symmetric Einstein--Maxwell-neutral scalar field model, up to and including the event horizon. 

 More precisely, we show that there exists a teleologically defined, codimension-one ``submanifold'' $\mathfrak M_\mathrm{stab}$ of the moduli space of spherically symmetric characteristic data for the Einstein--Maxwell-scalar field system lying close to the extremal Reissner--Nordstr\"om family, such that any data in $\mathfrak M_\mathrm{stab}$ evolve into a solution with the following properties as time goes to infinity: (i) the metric decays to a member of the extremal Reissner--Nordstr\"om family uniformly up to the event horizon, (ii) the scalar field decays to zero pointwise and in an appropriate energy norm, (iii) the first translation-invariant ingoing null derivative of the scalar field is approximately constant on the event horizon $\mathcal H^+$, (iv) for ``generic'' data, the second translation-invariant ingoing null derivative of the scalar field grows linearly along the event horizon. Due to the coupling of the scalar field to the geometry via the Einstein equations, suitable components of the Ricci tensor exhibit non-decay and growth phenomena along the event horizon.

 Points (i) and (ii) above reflect the ``stability'' of the extremal Reissner--Nordstr\"om family and points (iii) and (iv) verify the presence of the celebrated \emph{Aretakis instability} \cite{Aretakis-instability-2} for the linear wave equation on extremal Reissner--Nordstr\"om black holes in the full nonlinear Einstein--Maxwell-scalar field model.
\end{abstract}

\thispagestyle{empty}
\newpage 
\tableofcontents
\thispagestyle{empty}
\newpage

\section{Introduction}

Extremal black holes are special solutions of Einstein's equations of general relativity which have absolute zero temperature in the celebrated thermodynamic analogy of black hole mechanics. The simplest examples of extremal black holes are given by the \emph{extremal Reissner--Nordstr\"om} (ERN) metrics
\begin{equation}
    g_\mathrm{ERN}\doteq -\left(1-\frac{M}{r}\right)^2dt^2+\left(1-\frac{M}{r}\right)^{-2}dr^2+r^2(d\vartheta^2+\sin^2\vartheta\,d\varphi^2),\label{eq:ERN-1}
\end{equation}
where $M$ is a positive parameter known as the \emph{mass}. The metric \eqref{eq:ERN-1} solves the \emph{Einstein--Maxwell} equations,
\begin{equation}\label{eq:Einstein--Maxwell}
    R_{\mu\nu}-\tfrac 12 Rg_{\mu\nu} = 2F_{\mu\alpha}F_\nu{}^\alpha-\tfrac 12 g_{\mu\nu}F_{\alpha\beta}F^{\alpha\beta}\quad\text{and}\quad \nabla_\mu F^{\mu\nu}=0,
\end{equation}
and is spherically symmetric, asymptotically flat, and static, with time-translation Killing vector field $T\doteq \partial_t$. 

The extremal Reissner--Nordstr\"om metrics \eqref{eq:ERN-1} arise as an exceptional one-parameter subfamily of the full \emph{Reissner--Nordstr\"om family} \cite{reissner1916eigengravitation, nordstrom1918energy} of solutions to the Einstein--Maxwell equations,
\begin{equation}
  \label{eq:RN}  g_{M,e}\doteq -\left(1-\frac{2M}{r}+\frac{e^2}{r^2}\right)dt^2+\left(1-\frac{2M}{r}+\frac{e^2}{r^2}\right)^{-1}dr^2+r^2(d\vartheta^2+\sin^2\vartheta\,d\varphi^2),
\end{equation}
where $e$ is a real parameter representing the \emph{charge} of the electromagnetic field. The extremal case corresponds to the parameter values $|e|=M$. For the parameter range $|e|\le M$, the metric $g_{M,e}$ describes a black hole spacetime. When $|e|<M$, the solution is called \emph{subextremal}, and the $|e|=M$ case corresponds to the extremal Reissner--Nordstr\"om metric \eqref{eq:ERN-1} above. When $e=0$, $g_{M,e}$ reduces to the celebrated \emph{Schwarzschild solution} \cite{schw} of the Einstein vacuum equations. When $|e|>M$, the \emph{superextremal case}, the metric $g_{M,e}$ no longer describes a black hole. The role of superextremality will be discussed in \cref{sec:conjectures} below.

For any of the Reissner--Nordstr\"om black hole spacetimes, the Killing field satisfies 
\begin{equation}\label{eq:surface-gravity}
    \nabla_TT|_{\mathcal H^+}=\varkappa T|_{\mathcal H^+},
\end{equation} where $\mathcal H^+$ denotes the \emph{event horizon}, the boundary of the black hole region. The number $\varkappa=\varkappa(M,e)$, called the \emph{surface gravity} of $\mathcal H^+$, is given by 
\begin{equation*}
    \varkappa(M,e)\doteq \frac{\sqrt{M^2-e^2}}{(M+\sqrt{M^2-e^2})^2}
\end{equation*}
and quantifies the celebrated \emph{horizon redshift effect}: the null generators of $\mathcal H^+$ have exponentially decaying energy, with rate determined by $\varkappa$ (see for instance \cite{Sbierski-Gaussian}). 

The event horizon of subextremal Reissner--Nordstr\"om has $\varkappa>0$, while the event horizon of extremal Reissner--Nordstr\"om has $\varkappa=0$, a distinction which has fundamental repercussions for the behavior of perturbations of these spacetimes. In the seminal work \cite{Price-law}, Dafermos and Rodnianski proved the nonlinear asymptotic stability of the \emph{subextremal} Reissner--Nordstr\"om family as solutions of the Einstein--Maxwell equations coupled to a neutral scalar field in spherical symmetry. The horizon redshift effect is central to \cite{Price-law} and is a cornerstone of our understanding of linear waves on subextremal Reissner--Nordstr\"om black holes without symmetry assumptions \cite{dafermos2009red,dafermos2013lectures}. By the work of Dafermos--Holzegel--Rodnianski \cite{DHR19} and Blue \cite{Blue-maxwell} (see also \cite{Pasqualotto-Maxwell}) for $e=0$ and Giorgi \cite{Elena-small-Q,Elena-linear-stability} for $|e|<M$, subextremal Reissner--Nordstr\"om is now known to be linearly stable in Einstein--Maxwell theory outside of symmetry.

The theory of linear waves on extremal black holes---and by extension, nonlinear perturbations of extremal black holes---is very different. In a remarkable series of papers \cite{Aretakis-instability-1,Aretakis-instability-2,Aretakis-instability-3}, Aretakis showed that ingoing null derivatives of solutions to the linear wave equation on extremal Reissner--Nordstr\"om---even those arising from well-localized initial data---generically do not decay on the event horizon, and higher derivatives may even grow polynomially in time. This horizon instability for the linear wave equation, which has become known as the \emph{Aretakis instability}, has also been extended to gravitational perturbations \cite{LMR13,apetroaie} of extremal Reissner--Nordstr\"om and to axisymmetric linear waves on extremal Kerr \cite{Aretakis-Kerr,lucietti2012gravitational}. Gajic has recently shown that extremal Kerr is subject to additional stronger instabilities arising from higher azimuthal modes \cite{Gajic23} (see also the earlier heuristic analysis \cite{cgz-exkerr}), which will be discussed further in \cref{sec:EK} below. 

These horizon instabilities (along with the specter of superextremality which we will address in \cref{sec:conjectures} below) present a substantial obstacle to understanding the moduli space of solutions to the Einstein equations near extremal Reissner--Nordstr\"om, Kerr, and Kerr--Newman black holes. In a pioneering numerical study \cite{Reall-numerical}, Murata, Reall, and Tanahashi studied spherically symmetric perturbations of extremal Reissner--Nordstr\"om in the Einstein--Maxwell-neutral scalar field model and observed that the Aretakis instability for the scalar field is still activated on the dynamical background, but that the geometry is not completely disrupted in the process. These numerical results and rigorous work on nonlinear model problems by the first-named author, Aretakis, and Gajic \cite{A16,AAG-nonlinear-1,AAG20}, have given rise to the hope that extremal Reissner--Nordstr\"om could be \emph{stable} in spite of the Aretakis instability; see \cite[Conjecture~IV.2]{DHRT} and the recent essay by Dafermos \cite{Daf24}.

\subsection{Stability and instability of extremal Reissner--Nordstr\"om for the spherically symmetric Einstein--Maxwell-neutral scalar field system}

In this paper, we initiate the rigorous study of the (in)stability properties of extremal Reissner--Nordstr\"om black holes as solutions to the full nonlinear Einstein field equations. We work with the spherically symmetric Einstein--Maxwell-neutral scalar field model, which is the same model as in Murata--Reall--Tanahashi \cite{Reall-numerical} (and has been used in other influential works in recent years \cite{Dafermos-thesis,dafermos2005interior,Price-law,LukOhI}). This system consists of a spherically symmetric, charged spacetime $(\mathcal M^{3+1},g,F)$ together with a spherically symmetric massless scalar field $\phi:\mathcal M\to\Bbb R$ satisfying the linear wave equation \begin{equation}
    \Box_g\phi=0\label{eq:wave-1}
\end{equation} on the dynamical background, with total energy-momentum tensor given by
\begin{equation*}
    T_{\mu\nu}\doteq F_{\mu\alpha}F_\nu{}^\alpha-\tfrac 14 g_{\mu\nu}F_{\alpha\beta}F^{\alpha\beta}+\partial_\mu\phi\partial_\nu\phi -\tfrac 12g_{\mu\nu}\partial_\alpha\phi\partial^\alpha\phi.
\end{equation*}
 This is one of the simplest self-gravitating models in which one can entertain dynamical nonlinear perturbations of Reissner--Nordstr\"om, as a toy model for the electro-vacuum equations. See already \cref{sec:the-model} for the precise definitions and equations of the model.

We now state rough versions of our main theorems; the detailed statements and associated definitions will be presented in \cref{sec:setup-statement} below.

\begin{thmrough}[Codimension-one nonlinear stability of ERN, rough version]\label{thm:stability-rough} Let $\mathfrak M$ denote the moduli space of characteristic data for the spherically symmetric Einstein--Maxwell-neutral scalar field system posed on a bifurcate null hypersurface $C_\out\cup\underline C{}_\ing$, as in \cref{fig:stability-intro} below, which lie close to extremal Reissner--Nordstr\"om in an appropriate norm. There exists a ``codimension-one submanifold'' $\mathfrak M_\stab\subset \mathfrak M$ such that any data in $\mathfrak M_\stab$ evolve into a spacetime with the following properties:
\begin{enumerate}[(i)]
    \item Future null infinity $\mathcal I^+$ is complete and the causal past of future null infinity, $J^-(\mathcal I^+)$, is bounded by a regular event horizon $\mathcal H^+$, which itself bounds a nonempty black hole region $\mathcal{BH}\doteq \mathcal M\setminus J^-(\mathcal I^+)$.
    \item The metric remains close to the initial extremal Reissner--Nordstr\"om metric in the domain of outer communication and appropriately defined energy fluxes and pointwise $C^1$ norms of the scalar field $\phi$ are bounded in terms of their initial data on $C_\out\cup\underline C{}_\ing$.
    \item The metric decays polynomially in $C^0$ (as an appropriate notion of ``time'' tends to infinity) to a nearby member of the extremal Reissner--Nordstr\"om family, relative to a teleologically defined double null gauge uniformly in the entire domain of outer communication. The renormalized Hawking mass (see already \cref{sec:double-null-gauge}) converges uniformly to the final mass of the black hole. The scalar field decays polynomially to zero pointwise and in an appropriate energy norm. 
    \item The spacetime does not contain strictly trapped surfaces. Any marginally trapped surfaces lie on $\mathcal H^+$.  
\end{enumerate}
\end{thmrough}

 \begin{figure}[ht]
\centering{
\def\svgwidth{11pc}
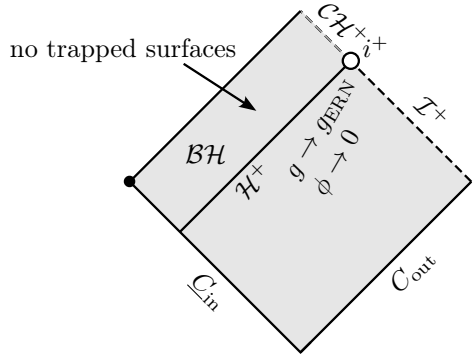}
\caption{A Penrose diagram showing the maximal development of a solution considered in \cref{thm:stability-rough}. The Cauchy data ends on the left in the solid point, where it is incomplete (but not singular).}
\label{fig:stability-intro}
\end{figure}

Some comments about this statement are in order:
\begin{enumerate}
    \item Because the extremal Reissner--Nordstr\"om family is already codimension-one within the full Reissner--Nordstr\"om family, any stability statement for it is necessarily a \emph{positive codimension} statement, with codimension one being sharp. This aspect of the problem is familiar from the proof of nonlinear stability of Schwarzschild outside of symmetry by Dafermos, Holzegel, Rodnianski, and Taylor \cite{DHRT}. For conjectures regarding the regularity of $\mathfrak M_\stab$ and what happens ``on either side'' of it, see already \cref{sec:conjectures}. 

        \item The theorem states that the metric decays in $C^0$ (and some Christoffel symbols) to that of extremal Reissner--Nordstr\"om and that the scalar field decays in $C^0$ to zero. We will show in \cref{thm:instability-rough} below that the metric does not necessarily decay to extremal Reissner--Nordstr\"om in $C^2$ (and may grow in $C^3$) and that the scalar field does not necessarily decay in $C^1$ (and may grow in $C^2$).
        
   \item The norms and decay rates for the scalar field $\phi$ are consistent with the norms and decay rates for spherically symmetric solutions of the linear wave equation on extremal Reissner--Nordstr\"om. We do not need to commute in order to close our bootstrap assumptions and hence do not prove sharp decay rates for $\phi$ everywhere. See already \cref{sec:proof-overview,sec:tails}.

    \item Since we show that there are no trapped surfaces behind the event horizon, we in fact prove that the maximal development of data in $\mathfrak M_\stab$ is the \emph{full double null rectangle} depicted in \cref{fig:stability-intro}. For the behavior of the scalar field and geometry in the black hole interior, see already \cref{sec:interior}.
\end{enumerate}

Our second main theorem shows that the Aretakis instability of the event horizon, i.e., non-decay of the first transverse derivative and linear growth of the second transverse derivative, persists in the Einstein--Maxwell-scalar field model. Moreover, in this coupled model, the instability affects the geometry as well.

Before stating the theorem, we recall briefly some features of the geometry of extremal Reissner--Nordstr\"om black holes. Let $Y$ denote the coordinate vector field $\partial_r$ in ingoing Eddington--Finkelstein coordinates $(v,r,\vartheta,\varphi)$. This null vector field is translation-invariant and transverse to the event horizon. Moreover, it is canonical in the sense that for \emph{any} spherically symmetric double null coordinates $(u,v,\vartheta,\varphi)$ on extremal Reissner--Nordstr\"om with $u$ ``ingoing,'' $Y$ can be written as $(\partial_ur)^{-1}\partial_u$. The $YY$-component of the Ricci tensor, $R_{YY}$, vanishes identically on any Reissner--Nordstr\"om solution.

\begin{thmrough}[Dynamical horizon instability, rough version]\label{thm:instability-rough} For any initial data lying in $\mathfrak M_\mathrm{stab}$, the following holds on the event horizon $\mathcal H^+$ of its maximal globally hyperbolic development:
    \begin{enumerate}[(i)]
        \item Let $Y\doteq (\partial_ur)^{-1}\partial_u$ denote the gauge-invariant null derivative transverse to $\mathcal H^+$, where $r$ is the area-radius of the spacetime, and $u$ is a retarded time coordinate, i.e., increasing towards the black hole. Then $R_{YY}$ and $Y(r\phi)$ are approximately constant on $\mathcal H^+$, i.e., do not necessarily decay.
        
        \item There exists a relatively open subset of $\mathfrak M_\stab$ for which the ``asymptotic Aretakis charge'' 
        \begin{equation}
         \label{eq:aretakis-intro-1}   H_0[\phi]\doteq \lim_{v\to\infty}Y(r\phi)|_{\mathcal H^+}(v)
        \end{equation}
        is nonvanishing. Here $v$ is an advanced time coordinate on the spacetime such that $v=\infty$ at future timelike infinity $i^+$. It holds that
        \begin{equation}
            \lim_{v\to\infty}R_{YY}|_{\mathcal H^+}(v)= 2M^{-2}\big(H_0[\phi]\big)^2.
        \end{equation}
        \item If $H_0[\phi]\ne 0$, then there exists a constant $c>0$ such that
        \begin{equation}
          \label{eq:aretakis-intro-2}\big|\nabla_YR_{YY}|_{\mathcal H^+}(v)\big|\ge cv,\qquad \big|Y^2(r\phi)|_{\mathcal H^+} (v)\big|\ge cv.
        \end{equation}
    \end{enumerate}
\end{thmrough}

As was mentioned before, the problem considered here was previously investigated numerically by Murata, Reall, and Tanahashi in \cite{Reall-numerical}. \cref{thm:stability-rough,thm:instability-rough} rigorously confirm all of their findings about the black hole exterior and event horizon for initial data in $\mathfrak M_\stab$ and make precise the nature of the ``fine tuning'' required to asymptote to extremality. Their findings about the black hole interior at extremality are also verified by combining \cref{thm:stability-rough} with the work of Gajic and Luk \cite{gajic-luk}; see already \cref{sec:interior}.

\begin{rk}
    The non-decay and growth of $R_{YY}$ and $\nabla_YR_{YY}$ along $\mathcal H^+$ found in the present paper is distinct from the non-decay and growth of $\underline\alpha$ and $\slashed\nabla_Y\underline \alpha$ found by Apetroaie in \cite{apetroaie} for the generalized Teukolsky system on extremal Reissner--Nordstr\"om, where $\underline\alpha{}_{AB}\doteq W(e_A,Y,e_B,Y)$, where $W$ is the Weyl tensor, and $\{e_A\}_{A=1,2}$ span the symmetry spheres. Indeed, $\underline\alpha$ vanishes identically for a spherically symmetric metric. 
\end{rk}

In this paper, we do not pursue the interesting question of which other geometric quantities exhibit instabilities at higher orders of differentiability.

\subsection{Overview of the proof}\label{sec:proof-overview}

The proof of \cref{thm:stability-rough} involves a bootstrap argument with a teleologically normalized double null gauge coupled to a discrete modulation argument performed on dyadic timescales. \cref{thm:instability-rough} is proved by combining the method of characteristics for the wave equation for $\phi$ with precise estimates on the dynamical degenerate redshift factor along the event horizon $\mathcal H^+$. We now describe the proofs in some detail, beginning with the relevant theory for the linear wave equation on extremal Reissner--Nordstr\"om. 

\subsubsection{Review of spherically symmetric linear waves on extremal Reissner--Nordstr\"om and the Aretakis instability} \label{sec:uncoupled}

Here we briefly review the theory of spherically symmetric solutions to the linear wave equation \eqref{eq:wave-1} on extremal Reissner--Nordstr\"om. Aretakis initiated the study of this problem in \cite{Aretakis-instability-1,Aretakis-instability-2} but we will make use of technical advances made by the first-named author, Aretakis, and Gajic in \cite{angelopoulos2018vector,angelopoulos2018late,angelopoulos2020late}. For a brief review of the geometry of extremal Reissner--Nordstr\"om, we refer the reader to \cref{sec:geometry-RN} of the present paper, \cite[Section 2]{Aretakis-instability-1}, and the appendix of \cite{stefanos-ern_full}. 

The general strategy to prove energy decay statements for waves on extremal Reissner--Nordstr\"om consists of, as in the subextremal case, deriving a hierarchy of weighted energy boundedness inequalities and time-integrated energy decay estimates. This hierarchy takes the form
\begin{multline}\label{eq:boundedness-intro-1} 
    \int_{C(\tau_2)}r^p(\partial_v\psi)^2\, dv+ \int_{\underline C(\tau_2)} (r-M)^{2-p} \frac{(\partial_u\psi)^2}{-\partial_ur}\,du\\ \les \int_{C(\tau_1)}r^p(\partial_v\psi)^2\, dv+ \int_{\underline C(\tau_1)} (r-M)^{2-p} \frac{(\partial_u\psi)^2}{-\partial_ur}\,du+\mathrm{l.o.t.},
\end{multline} \vspace{-5mm}
\begin{align}
\label{eq:rp-intro}    \int_{\tau_1}^{\tau_2}\int_{C(\tau)} r^{p-1}(\partial_v\psi)^2\, dvd\tau&\les \int_{C(\tau_1)}r^p(\partial_v\psi)^2\, dv + \text{l.o.t.},\\
\label{eq:r-Mp-intro}    \int_{\tau_1}^{\tau_2}\int_{\underline C(\tau)} (r-M)^{3-p} \frac{(\partial_u\psi)^2}{-\partial_ur}\,dud\tau&\les\int_{\underline C(\tau_1)}  (r-M)^{2-p} \frac{(\partial_u\psi)^2}{-\partial_ur}\,du + \text{l.o.t.},
\end{align}
where $(u,v)$ denote Eddington--Finkelstein double null coordinates on the domain of outer communication, $\tau$ is proper time along a timelike curve $\Gamma$ with constant area-radius, $\tau_1\le\tau_2$, $p\in[0,3)$ in \eqref{eq:boundedness-intro-1}, $p\in[1,3)$ in \eqref{eq:rp-intro} and \eqref{eq:r-Mp-intro}, ``l.o.t.'' denotes terms lower in the $p$-hierarchy, the foliations $C(\tau)$ and $\underline C(\tau)$ are defined pictorially in \cref{fig:ERN-intro} above, and $\psi$ denotes the \emph{radiation field}
\begin{equation*}
    \psi\doteq r\phi.
\end{equation*}

 \begin{figure}
\centering{
\def\svgwidth{11pc}
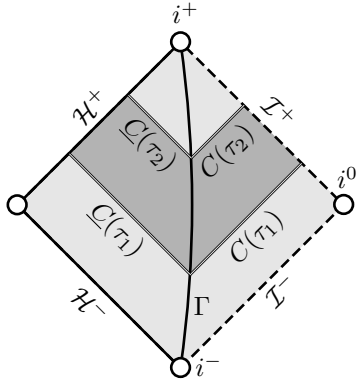}
\caption{A Penrose diagram of extremal Reissner--Nordstr\"om depicting the foliations $C(\tau)$ and $\underline C(\tau)$ used in the estimates \eqref{eq:boundedness-intro-1}--\eqref{eq:r-Mp-intro}. The region of integration in \eqref{eq:rp-intro} and \eqref{eq:r-Mp-intro} is shaded darker.}
\label{fig:ERN-intro}
\end{figure}

\begin{rk}
   The range of our horizon and infinity hierarchies is considered for $p \in [0,3-\delta]$. Specifically, at the horizon, this range is sharp, and it is necessary for us to go beyond $p \geq 2$ in order to simultaneously close the bootstrap argument for both the scalar field and the geometric quantities.
\end{rk}

\begin{rk}
    The expression $(\partial_u\psi)^2(-\partial_ur)^{-1}du$ is invariant under reparametrization of the double null gauge and represents a \emph{non-degenerate} one-form along $\underline C(\tau)$. Indeed, written in ingoing Eddington--Finkelstein coordinates $(v,r)$, it corresponds to $(\partial_r\psi)^2\,dr$ along $\underline C(\tau)$, which is manifestly nondegenerate. 
\end{rk}

The inequality \eqref{eq:rp-intro} is the celebrated \emph{$r^p$-weighted estimate} of Dafermos and Rodnianski \cite{dafermos2010new} and relies only on the asymptotic flatness of the extremal Reissner--Nordstr\"om metric. The estimate \eqref{eq:r-Mp-intro} is specific to extremal Reissner--Nordstr\"om and can be thought of as a ``horizon analogue'' of the $r^p$ estimates at infinity: The event horizon $\mathcal H^+$ of extremal Reissner--Nordstr\"om is located at $r=M$ and hence $r-M$ is a degenerate weight on $\mathcal H^+$. The estimate \eqref{eq:r-Mp-intro} states that the time integral of the $(p-1)$-weighted horizon flux is bounded by the initial value of the $p$-weighted horizon flux. This horizon hierarchy---some special cases of which were introduced in \cite{Aretakis-instability-1} and in full generality was presented in \cite{angelopoulos2020late}---replaces the fundamental redshift estimate from \cite{dafermos2009red,dafermos2013lectures} which holds for any stationary event horizon with positive surface gravity. Note that the $p=0$ horizon flux is equivalent to the $T$-energy, where $T$ is the time-translation Killing field. The duality between $r$ and $(r-M)^{-1}$ is related to the so-called \emph{Couch--Torrence conformal isometry} \cite{couch1984conformal} that exchanges $\mathcal H^+$ and $\mathcal I^+$ in extremal Reissner--Nordstr\"om. 

Using the pigeonhole principle as in \cite{dafermos2010new}, \eqref{eq:boundedness-intro-1}--\eqref{eq:r-Mp-intro} can be used to prove the energy decay estimate
\begin{equation}\label{eq:ED}
    \int_{C(\tau)} r^p(\partial_v\psi)^2\, dv+\int_{\underline C(\tau)} (r-M)^{2-p} \frac{(\partial_u\psi)^2}{-\partial_ur}\,du \le C_\star\tau^{-3+\delta+p},
\end{equation} for every $p\in[0,3-\delta]$ and $\tau\ge\tau_0$, where $C_\star$ is a constant depending on $\delta$ and the data at $\tau=\tau_0$. This estimate can then be used to prove pointwise decay of $\psi$ itself. By commuting the wave equation, higher order versions of \eqref{eq:boundedness-intro-1}--\eqref{eq:ED} can be proved, but we shall not need to do so in the present paper. 

\begin{rk}\label{rk:subextremal-decay}
    Observe that \eqref{eq:ED} for $p=2$ states that the nondegenerate energy of $\partial_u\psi$ decays only like $\tau^{-1+\delta}$, compared to $\tau^{-3+\delta}$ or $\tau^{-2+\delta}$ in the subextremal case (depending on the assumptions for the decay of the data at infinity---see for example \cite{angelopoulos2018vector}). Unfortunately, this is in fact sharp for generic data (see \cite{AAG-trapping}) and persists in the nonlinear theory, see already \cref{sec:tails,sec:sharpness}. This slow decay is responsible for many of the technical difficulties we face in the present paper and will be discussed again in \cref{sec:intro-geometry} below.
\end{rk}

While certain energies for $\phi$ do indeed decay by \eqref{eq:ED}, the Aretakis instability states that nondegenerate ingoing null derivatives of $\phi$ on $\mathcal H^+$ \emph{do not decay}, or can even \emph{grow polynomially}. We now briefly explain this mechanism. Let $Y\doteq \partial_r$ in ingoing Eddington--Finkelstein coordinates $(v,r)$.\footnote{The vector field $Y$ can be expressed as $(\partial_ur)^{-1}\partial_u$ in Eddington--Finkelstein double null coordinates $(u,v)$, or any reparametrization thereof.} An elementary calculation using the wave equation \eqref{eq:wave-1} shows that
\begin{equation}\label{eq:Y-psi-1}
    \partial_v\big(Y\psi|_{\mathcal H^+}\big)=0,
\end{equation}
i.e., $Y\psi$ is \emph{constant along $\mathcal H^+$}. Therefore, in sharp contrast to the subextremal case, $Y\psi$ does not decay along $\mathcal H^+$. This constant is written as $H_0[\phi]$ and is called the (zeroth) \emph{Aretakis charge} of $\phi$. 

By commuting the wave equation with $Y$, we can likewise derive an evolution equation for $Y^2\psi$ along $\mathcal H^+$. If $H_0[\phi]\ne 0$, we have that 
\begin{equation}\label{eq:Y-psi-2}
    \partial_v\big(Y^2\psi|_{\mathcal H^+}\big)= -2M^{-2}H_0[\phi]+\text{decaying terms},
\end{equation}
so upon integrating in $v$ we conclude that $|Y^2\psi|\gtrsim |H_0[\phi]|v$ on $\mathcal H^+$ for $v$ large. In fact, we have $|Y^k\psi|\gtrsim |H_0[\phi]|v^{k-1}$ on $\mathcal H^+$ for any $k\ge 1$ and $v$ large.

\subsubsection{Prescription of seed data and the modulation parameter \texorpdfstring{$\alpha$}{alpha}}\label{sec:intro-data}

Keeping in mind the ideas of \cref{sec:uncoupled}, we now turn to the outline of the proofs of \cref{thm:stability-rough,thm:instability-rough}. 

We recall the notion of \emph{renormalized Hawking mass}, which is the appropriate analogue of the usual Hawking mass $m\doteq \frac{r}{2}(1-g(\nabla r,\nabla r))$ for solutions of the spherically symmetric Einstein--Maxwell-neutral scalar field model:
\begin{equation*}
    \varpi\doteq m+\frac{e^2}{2r},
\end{equation*}
where $e$ is the constant charge of the solution. In Reissner--Nordstr\"om with parameters $M$ and $e$, $\varpi=M$.

We refer back to the Penrose diagram of the setup of our main results, \cref{fig:stability-intro}. Bifurcate characteristic seed data for the spherically symmetric Einstein--Maxwell-neutral scalar field model on $C_\out\cup\underline C{}_\ing$ consists of the restriction of $\phi$ to $C_\out\cup\underline C{}_\ing$, denoted by $\mathring\phi$, the area-radius $\Lambda$ of $C_\out\cap\underline C{}_\ing$, the renormalized Hawking mass $\varpi_0$ of $C_\out\cap\underline C{}_\ing$, and the electric charge $e$ of the Maxwell field $F$, which is defined in \eqref{eq:charge-definition} below. Since the scalar field is neutral, $e$ is in fact \emph{globally conserved} and is thus fixed once and for all by the initial data. The collection of such quadruples $(\mathring\phi,\Lambda,\varpi_0,e)$ is what we call the \emph{moduli space} of initial data $\mathfrak M$.

Fix extremal parameters $M_0=|e_0|$ and an area-radius $100M_0$ (which lies far outside the horizon located at $r=M_0$). We consider perturbations of the bifurcate cone in the extremal Reissner--Nordstr\"om solution with parameters $(M_0,e_0)$ with bifurcation sphere at $r=100M_0$. We consider the seed data norm 
\begin{equation}\label{eq:D-intro}
    \mathfrak D\approx |\Lambda-100M_0|+|\varpi_0-M_0|+|e-e_0|+ \sup_{\underline C_\ing}\big(|\phi|+|\partial_{\hat u}\phi|\big)+\sup_{ C_\out}\left(|\psi|+|r^2\partial_{\hat v}\psi|\right)
\end{equation}
and define a master smallness parameter $\ve\ge \mathfrak D$. Here $(\hat u,\hat v)$ are ``initial data normalized'' coordinates such that $\partial_{\hat u}r=-1$ on $\underline C{}_\ing$ and $\partial_{\hat v}r=1$ on $C_\out$. In particular, $\hat u$ will be regular across the event horizon. 

As in \cite{DHRT}, we in fact consider a \emph{family} of initial data which are indexed by a real parameter $\alpha$ such that $\varpi_0=M_0+\alpha$. Therefore, for $(\mathring \phi,\Lambda,e)$ fixed, the ``codimension-one'' aspect of \cref{thm:stability-rough} means that we find (at least one) $\alpha=\alpha_\star$ such that the solution generated by the seed data $(\mathring \phi,\Lambda,M_0+\alpha_\star,e)$ converges to extremal Reissner--Nordstr\"om with parameters $M=|e|$ and $e$. The critical parameter $\alpha_\star$ is determined in evolution and cannot be read off from the initial data. See already \cref{sec:intro-modulation}.

\begin{rk}
    Varying $\alpha$ in our construction corresponds exactly to varying the parameter $M_i$ in the numerical setup of \cite{Reall-numerical}, so we are indeed finding the ``same'' codimension-one family of asymptotically extremal Reissner--Nordstr\"om solutions. However, our modulation scheme is completely different from theirs.
\end{rk}

The setup for the initial data is given in \cref{sec:seed-data}.

\subsubsection{Setup of the bootstrap argument and the teleological gauge}\label{sec:intro-gauge}

 \begin{figure}
\centering{
\def\svgwidth{11pc}
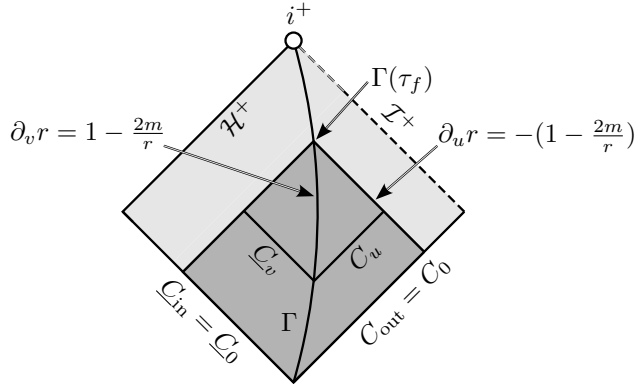}
\caption{A Penrose diagram of one of the bootstrap domains $\mathcal D_{\tau_f}\doteq J^-(\Gamma(\tau_f))$ used in the proof of \cref{thm:stability-rough}. Here $\Gamma\doteq\{r=\Lambda\}$ is the timelike curve which anchors the bootstrap domains and $(u,v)$ are double null coordinates teleologically normalized as depicted.}
\label{fig:bootstrap-intro}
\end{figure}

Our general setup for the bootstrap argument (depicted in \cref{fig:bootstrap-intro} below) is inspired by \cite{DHRT} and the work of Luk--Oh \cite{luk2019strong} on stability of subextremal Reissner--Nordstr\"om in spherical symmetry.

 In \cref{fig:bootstrap-intro}, the timelike curve $\Gamma$ has constant area-radius $r=\Lambda$ (from the seed data) and is parametrized by proper time $\tau\ge 1$. The parameter $\tau_f$ determines the bootstrap domain $\mathcal D_{\tau_f}$ and is sent to infinity in the course of the proof. The double null coordinates $(u,v)$ are teleologically normalized according to the conditions $\partial_ur=-(1-\frac{2m}{r})$ on the ingoing future boundary of $\mathcal D_{\tau_f}$, where $m$ is the Hawking mass of the spacetime, and $\partial_vr=1-\frac{2m}{r}$ on $\Gamma$. This latter gauge condition for the $v$ coordinate is nonstandard and marks a crucial difference with the subextremal case in \cite{luk2019strong,DHRT}. We extend $\tau$ to a function on $\mathcal D_{\tau_f}$ by setting it equal to proper time along $\Gamma$ and then declaring it to be constant along ingoing cones to the left of $\Gamma$ and constant along outgoing cones to the right of $\Gamma$.

Since the charge $e$ is conserved, we know a priori that we should aim to converge to extremal Reissner--Nordstr\"om with parameters $M\doteq|e|$ and $e$. We anchor a comparison extremal Reissner--Nordstr\"om solution with area-radius function $\bar r$ and parameters $(M,e)$ in Eddington--Finkelstein double null gauge to the bootstrap domain $\mathcal D_{\tau_f}$ by setting $\bar r(\Gamma(\tau_f))=\Lambda$. Relative to this background solution, we define energy norms motivated by the uncoupled case (recall \cref{sec:uncoupled}),
\begin{align*}
  \mathcal E_p(\tau)  &\doteq \int_{C_{\Gamma^u(\tau)}\cap\mathcal D_{\tau_f}}r^p(\partial_v\psi)^2\, dv+\cdots ,\\
   \underline{\mathcal E}{}_p(\tau) & \doteq\int_{\underline C{}_{\Gamma^u(\tau)}\cap\mathcal D_{\tau_f}} (\bar r-M)^{2-p}\frac{(\partial_u\psi)^2}{-\partial_u\bar r} \,du+\cdots,
\end{align*}
where we have only written the most important terms for now.

\begin{rk}
    In the definition of $\mathcal E_p$, we use the ``dynamical'' $r$ and not the background $\bar r$. This is because the $r^p$ estimates in spherical symmetry do not generate nonlinear errors and $r$ works just as well as $\bar r$ in the far region. The use of $\bar r-M$ in $\underline{\mathcal E}{}_p$ is crucial, however. 
\end{rk}

For the modulation parameter $\alpha$ lying in an appropriate range (to be explained in \cref{sec:intro-modulation} below), we make the bootstrap assumptions
\begin{equation}\label{eq:boot-geo-intro}
    \left|\frac{\partial_ur}{\partial_u\bar r}-1\right|\les \ve^{3/2}\tau^{-1+\delta},\quad |r-\bar r|\les \ve^{3/2}\tau^{-2+\delta},\quad |\varpi-M|\les \ve^{3/2}\tau^{-3+\delta}
\end{equation}
on $\mathcal D_{\tau_f}$ for the geometry and 
\begin{equation}\label{eq:boot-sf-intro}
    \mathcal E_p(\tau)+ \underline{\mathcal E}{}_p(\tau)\les \ve^2\tau^{-3+\delta+p}
\end{equation}
for $\tau\in[1,\tau_f]$ for the scalar field; compare with \eqref{eq:ED}.\footnote{We also have bootstrap assumptions for energy fluxes along outgoing cones in the near region and along ingoing cones in the far region, but suppress these at this level of discussion.} The main analytic content of \cref{thm:stability-rough} consists of recovering the bootstrap assumptions \eqref{eq:boot-geo-intro} and \eqref{eq:boot-sf-intro}.

The setup for the bootstrap argument is given in \cref{sec:setup}, the anchoring procedure is given in \cref{sec:anchor}, and the bootstrap assumptions are precisely stated in \cref{sec:bootstrap-definitions}.

\subsubsection{Estimates for the geometry: the role of the degenerate redshift}\label{sec:intro-geometry}

We utilize the well-known quantities 
\begin{equation*}
    \kappa\doteq \frac{\partial_vr}{1-\frac{2m}{r}},\quad \gamma\doteq \frac{\partial_ur}{1-\frac{2m}{r}},
\end{equation*}
which satisfy good evolution equations in $u$ and $v$, respectively. The background extremal Reissner--Nordstr\"om solution $\bar r $ has $\bar\kappa=-\bar\gamma=1$ globally. Therefore, the anchoring and gauge conditions ensure that $r(\Gamma(\tau_f))=\bar r(\Gamma(\tau_f))$, $\kappa=\bar \kappa=1$ along $\Gamma\cap\mathcal D_{\tau_f}$, and $\gamma=\bar\gamma=-1$ along the final ingoing cone in $\mathcal D_{\tau_f}$. 

The $v$-equation for $\gamma$ sees the flux $\mathcal E_0$ and hence gives us that $\gamma=-1 + O(\ve^2\tau^{-3+\delta})$, which is the best that can be expected. On the other hand, the $u$-equation for $\kappa$ reads $\partial_u\kappa = r\kappa (\partial_ur)^{-1}(\partial_u\phi)^2$, 
which sees $\underline{\mathcal E}{}_2$, and hence decays much slower.\footnote{In fact, if we only took $p$ up to 2 in our master hierarchy, we would not even have decay for the $p=2$ flux! Note again (recall \cref{rk:subextremal-decay}) that in the subextremal case, we would have $\tau^{-3+\delta}$ decay for $\kappa-1$ if we take $p$ up to $3-\delta$, and $\tau^{-2}$ decay (which is still integrable) if we only take $p$ up to $2$.} On the other hand, the quantity $(\bar r-M)^{2-p}(\kappa-\bar \kappa)$, for $p\in[0,2]$, obeys a $u$-evolution equation that sees the flux $\underline{\mathcal E}{}_p$ \emph{when integrating to the left of $\Gamma$}, and hence we show that
\begin{equation}\label{eq:kappa-intro}
    (\bar r-M)^{2-p}\kappa = (\bar r-M)^{2-p}+O(\ve^2\tau^{-3+\delta+p})
\end{equation}
to the left of $\Gamma$. This hierarchy of decay rates for $\kappa$ is a fundamental aspect of our geometric estimates. 

Because of this, it is important not to waste any powers of $\bar r-M$ in the system and our scheme is essentially sharp at the horizon in this regard. Of particular importance are the quantities $1-\frac{2m}{r}$ and $\partial_vr$, which we show admit ``Taylor expansions'' of the form
\begin{align}
  \label{eq:Taylor-1}1-\frac{2m}{r} &= \overline{1-\frac{2m}{r}} + \frac{2M}{\bar r^3}(\bar r-M)(r-\bar r)+O(\ve^{3/2}\tau^{-3+\delta}),\\
 \label{eq:Taylor-2}   \partial_vr &=\partial_v\bar r + \frac{2M\kappa}{\bar r^3}(\bar r-M)(r-\bar r)+ O(\ve^{3/2}\tau^{-3+\delta}).
\end{align}
In the background solution, the quantities $ \overline{1-\frac{2m}{r}} $ and $\partial_v\bar r$ vanish quadratically on $\mathcal H^+$, so these are to be viewed as expansions in powers of $\bar r-M$ with rapidly decaying error. 

The coefficients of the terms which are linear in $\bar r-M$ and the strong decay of the error terms in \eqref{eq:Taylor-1} and \eqref{eq:Taylor-2} are important. For example, the function $1-\frac{2m}{r}$ appears naturally in the $u$-equation for $\varpi$, which reads $\partial_u\varpi = \tfrac 12 (1-\frac{2m}{r})r^2(\partial_ur)^{-1}(\partial_u\phi)^2$, and proving $\tau^{-3+\delta}$ decay for this requires taking advantage of the linear term in \eqref{eq:Taylor-1}. In the subextremal case, the redshift effect implies one can prove $\tau^{-3+\delta}$ decay for $\varpi$ while completely ignoring the $1-\frac{2m}{r}$ weight, so no expansion of the form \eqref{eq:Taylor-1} is required (or helpful). Moreover, we use the linear term in \eqref{eq:Taylor-1} to control the error term in \eqref{eq:Taylor-2}.

The sign of the linear term in \eqref{eq:Taylor-2} is crucial: it is positive in the domain of outer communication, which reflects the \emph{global redshift effect} on extremal Reissner--Nordstr\"om. However, unlike the subextremal case, where the redshift leads to an exponentially decaying integrating factor for the analogue of \eqref{eq:Taylor-2} (compare \cite[Lemma 8.19]{luk2019strong}), on an asymptotically extremal black hole the effect degenerates on the horizon and the rapidly decaying error term becomes important, since integrating \eqref{eq:Taylor-2} in $v$ now loses a power of $\tau$. See also \cref{rk:redshift-tilde-r} below.

The geometric estimates are proved in \cref{sec:geometry}.

\subsubsection{Integrated local energy decay and the \texorpdfstring{$r^p$}{rp} and \texorpdfstring{$(\bar r-M)^{2-p}$}{(bar r-M){2-p}} energy hierarchies}\label{sec:intro-energy}

To recover the bootstrap assumption \eqref{eq:boot-sf-intro} for the scalar field, we follow the general strategy outlined in \cref{sec:uncoupled}. We prove the estimates 
\begin{equation}\label{eq:intro-energy-1}
\mathcal E_p(\tau_2)+ \underline{\mathcal E}{}_p(\tau_2) \les  \mathcal E_p(\tau_1)+ \underline{\mathcal E}{}_p(\tau_1)+ \text{decaying nonlinear error}
\end{equation}
for $1\le \tau_1\le \tau_2$ and $p\in[0,3-\delta]$, and 
\begin{equation}\label{eq:intro-energy-2}
  \int_{\tau_1}^{\tau_2}\big( \mathcal E_{p-1}(\tau)+ \underline{\mathcal E}{}_{p-1}(\tau)\big)\,d\tau \les  \mathcal E_p(\tau_1)+ \underline{\mathcal E}{}_p(\tau_1)+ \text{decaying nonlinear error}
\end{equation}
for $p\in[1,3-\delta]$. The choice of multiplier vector fields and lower order currents is inspired by work in the uncoupled case and \cite{luk2019strong}. We utilize a mix of ``dynamical'' (i.e., defined relative to the dynamical metric) and ``background'' (i.e., defined relative to the background extremal Reissner--Nordstr\"om metric $\bar r$) multipliers to minimize the number of error terms. The expansion \eqref{eq:Taylor-1} is used crucially to estimate nonlinear errors in the near-horizon region that are specific to the extremal case. Once \eqref{eq:intro-energy-1} and \eqref{eq:intro-energy-2} have been proved, decay is inferred by a standard application of the pigeonhole principle as in \cite{dafermos2010new}. 

The energy hierarchies are defined in \cref{sec:energy} and energy decay is proved in \cref{sec:decay-energy}.

\subsubsection{The codimension-one ``submanifold'' \texorpdfstring{$\mathfrak M_\stab$}{M stab} and modulation on dyadic timescales}\label{sec:intro-modulation}

While our bootstrap argument is performed continuously in time, the choice of allowed modulation parameters $\alpha$ is only decided when $\tau_f$ is dyadic, i.e., a power of 2. This approach is motivated by the purely dyadic approach of Dafermos--Holzegel--Rodnianski--Taylor in \cite{DHRT22} and turns out to be quite fortuitous compared to the continuous in time modulation theory employed in \cite{DHRT}.

In practice, we construct a sequence of nested compact intervals $\mathfrak A_0\supset\mathfrak A_1\supset\cdots$ such that $\mathfrak A_i$ consists of those $\alpha$'s which we consider for bootstrap time $\tau_f\in[2^i,2^{i+1})$. These sets are defined implicitly by the requirement that the renormalized Hawking mass $\varpi$ at the point $\Gamma(2^i)$ must lie within $\varepsilon^{3/2}2^{(-3+\delta)i}$ of $M$ for $\alpha\in\mathfrak A_i$.  By a shooting argument at each dyadic time $2^i$, we show that the implicitly defined set $\mathfrak A_i$ is nonempty. Therefore, as $\tau_f\to\infty$ (and, consequently, $i\to\infty$), we extract (at least one) critical parameter $\alpha_\star\in \bigcap_{i\ge 0}\mathfrak A_i$ for which the solution tends to extremality. The stable ``submanifold'' $\mathfrak M_\stab\subset\mathfrak M$ is the collection of all seed data sets of the form $(\mathring\phi,\Lambda,M_0+\alpha_\star,e)$ and is codimension-one in the sense that every ``line'' of constant $\mathring\phi$, $\Lambda$, and $e$ intersects it at least once.\footnote{In \cite{DHRT}, the codimension-three ``submanifold'' is characterized in a similar manner via three-planes in the moduli space of seed data sets.} Refer to \cref{fig:mod-space-lines} below.

Since we construct $\mathfrak A_i$ only after each dyadic time, we may always assume a $\tau^{-3+\delta}$ decay rate for the difference $\varpi-M$, which is sharp. This should be contrasted with \cite{DHRT}, where the assumption on angular momentum in the modulation set has a decay rate of $\tau_f^{-1}$, but the \emph{improvement} (i.e., the change in angular momentum along $\mathcal I^+$) is required to decay like $\tau_f^{-2}$, which is a significant technical issue. In our work, both the assumption and the improvement have the sharp decay rate $\tau_f^{-3+\delta}$ and the gain is purely in the power of $\ve$ ($\ve^{3/2}$ vs. $\ve^2$); see already \eqref{eq:algebra}. This strong decay rate for $\varpi-M$ is important, in particular because of the absence of redshift in the geometric estimates. 

We have depicted our modulation scheme in \cref{fig:mod-space-lines} above. The modulation sets $\mathfrak A_i$ are defined in \cref{sec:bootstrap-definitions} and the modulation argument is performed in \cref{sec:modulation}. Modulation on dyadic timescales, motivated by \cite{DHRT22}, has also been performed by K\'ad\'ar in \cite{Kadar,Kadar2}.

 \begin{figure}
\centering{
\def\svgwidth{22pc}
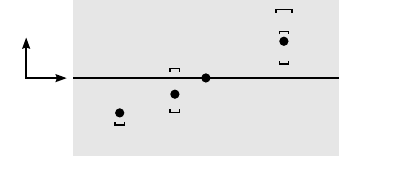}
\caption{A schematic depiction of our modulation scheme. Let $p\mapsto\mathring\phi(p)$ be a one-parameter family of characteristic initial data for the scalar field with $\mathring\phi(0)=0$. We can then consider the plane in $\mathfrak M$ parametrized by $(p,\alpha)$. Each $p$ generates a line segment $\mathcal L$ in $\mathfrak M$ which intersects the ``submanifold'' $\mathfrak M_\mathrm{stab}$ at least once. The horizontal line $\mathfrak M_0$ denotes the hyperplane in $\mathfrak M$ consisting of data sets with $\varpi_0=M_0$. On the three $\mathcal L$'s depicted here, we have also drawn three of the nested modulation sets $\mathfrak A_i$ which converge to $\mathcal L\cap \mathfrak M_\mathrm{stab}$. Note that we have drawn $\mathfrak M_\stab$ as a smooth, connected curve here, which is in line with our conjectures in \cref{sec:conjectures}, but we do not prove any such fine structure of it in this paper.}
\label{fig:mod-space-lines}
\end{figure}

\subsubsection{Construction of the eschatological gauge and background solution}

Recall from \cref{sec:intro-gauge} that the teleologically normalized coordinate $u$ and the background extremal Reissner--Nordstr\"om solution $\bar r$ depend on the bootstrap time $\tau_f$. Therefore, an important final step when taking $\tau_f\to\infty$ is to prove that we can define a unique ``final background solution'' which we converge to and an ``eschatological gauge'' (i.e., final teleological gauge) in which we converge to it. The regularity of the limiting gauge is actually related to the decay assumptions on initial data,\footnote{A similar phenomenon occurs in \cite{DHRT}.} and the best that can be hoped for given only the finiteness of \eqref{eq:D-intro} is that the final $u$ coordinate is a $C^2$ function of the initial data coordinate $\hat u$. These limiting arguments are carried out in \cref{sec:gauges,sec:limiting-argument}.

\subsubsection{Absence of trapped surfaces and rigidity of the apparent horizon}\label{sec:rigidity-intro}

After we have constructed the dynamical extremal spacetimes, it remains to prove part (iv) of \cref{thm:stability-rough}. This argument is originally due to Kommemi (in unpublished work) and a very similar argument appears in \cite[Appendix A]{LukOhI}. The first aspect of the argument is a proof that any point on the outermost apparent horizon $\mathcal A'$ must lie on $\mathcal H^+$ because of simple monotonicities to the right of $\mathcal A'$ and the fact that $r$ and $\varpi$ both asymptote to $M$ along $\mathcal H^+$. The second part involves Taylor expanding $\partial_vr$ in $u$ along $\mathcal H^+$ to show that there are no trapped or marginally trapped spheres behind $\mathcal H^+$. This second part of the argument is reminiscent of the proof of ``Israel’s observation'' (Proposition 1.1) in \cite{KU22}. We carry out these arguments in \cref{sec:putting-together}. 

\begin{rk}
These soft arguments, using monotonicity and the wave equation for $r$, do not directly carry over to the charged scalar field \cite{Kommemi13} or charged Vlasov \cite{KU24} models. 
\end{rk}

\subsubsection{The Aretakis instability on the dynamical geometry}\label{sec:aretakis-instability-intro}

Using the geometric estimates proved in \cref{sec:geometry}, the exact conservation law \eqref{eq:Y-psi-1} is replaced by the ``almost conservation law'' 
\begin{equation*}
    \partial_v\big(Y\psi|_{\mathcal H^+}\big) = O\big(\ve^{3}\tau^{-2+\delta}\big).
\end{equation*}
Note that the error is integrable in $\tau$ (again, thanks to going up to $p=3-\delta$ in the hierarchy!) and much better in $\ve$ than $Y\psi$. Therefore, integrating this immediately gives (i) of \cref{thm:instability-rough}. Part (ii) is completely soft and is essentially an immediate consequence of (i). To prove part (iii), we commute the wave equation on the dynamical background with $Y$ and show that the nonlinear error terms are decaying, which results in an estimate entirely analogous to \eqref{eq:Y-psi-2}. The instability of the geometry is obtained by directly inserting this behavior for $\phi$ into the Einstein equations. We do not prove growth of derivatives higher than second order because this would require proving higher order estimates for the geometry, which is not necessary for the proof of \cref{thm:stability-rough}. These arguments are carried out in \cref{sec:aretakis-proof}.

\subsection{The conjectural picture of the local moduli space}\label{sec:conjectures}

\cref{thm:stability-rough} makes no statement about the regularity (or even connectedness!) of the stable ``submanifold'' $\mathfrak M_\stab$, nor about behavior of solutions arising from data in $\mathfrak M\setminus\mathfrak M_\stab$. In this section, we propose conjectures addressing these issues, motivated by conjectures in \cite{DHRT}, \cite{Daf24}, and \cite{KU24} by the second- and third-named authors of the present paper.

\subsubsection{Asymptotically extremal black holes as a locally separating hypersurface between collapse and dispersion}\label{sec:separating}

In the following conjectures, we will write the symbol $\mathfrak M$ to mean a moduli space of initial data posed ``in the same way'' as in our \cref{thm:stability-rough}, but possibly topologized by a different norm than in \cref{sec:seed-data} below. In particular, it could be that the following conjectures are only true under stronger asymptotic decay assumptions. However, we will always suppose that $\mathfrak M$ carries a Banach space structure, so that we have access to the notion of $C^1$ submanifolds of $\mathfrak M$.

For $\mathfrak r\in [0,1]$, we can consider the subset $\mathfrak M^\mathfrak r_\stab\subset\mathfrak M$ consisting of initial data which form a black hole with asymptotic parameter ratio $|e|/M_\infty=\mathfrak r$, where $e$ is the conserved charge of the solution and $M_\infty\doteq\lim_{v\to\infty}\varpi|_{\mathcal H^+}$ is the teleologically determined final mass of the black hole. In this notation, the stable set from \cref{thm:stability-rough} is $\mathfrak M_\stab^1$ and $\bigcup_{\mathfrak r<1}\mathfrak M_\stab^\mathfrak r$ was studied by Dafermos and Rodnianski in \cite{Price-law}.

\begin{conj}[Regularity of $\mathfrak M_\mathrm{stab}^\mathfrak r$]\label{conj:regularity}
    For every $\mathfrak r\in[0,1]$, $\mathfrak M^\mathfrak r_\stab$ is a $C^1$ hypersurface in $\mathfrak M$. Moreover, if $\mathfrak M_\mathrm{black}$ denotes the subset of $\mathfrak M$ consisting of seed data which form a black hole in evolution, then $\mathfrak M_\mathrm{black}$ is foliated by the collection $\{\mathfrak M_\stab^\mathfrak r\}$. 
\end{conj}

In particular, we conjecture that the stable ``submanifold'' from \cref{thm:stability-rough} is in fact a connected, regular hypersurface in $\mathfrak M$. Our next conjecture addresses what happens on ``either side'' of $\mathfrak M_\stab^1$; see \cref{fig:local-moduli-space} above.

To motivate the following conjecture, it is helpful to consider the Reissner--Nordstr\"om family itself for a moment. If we consider characteristic data for the spherically symmetric Einstein--Maxwell system posed on a bifurcate null hypersurface as in \cref{fig:stability-intro} and then vary the renormalized Hawking mass $\varpi$ of the bifurcation sphere, the resulting developments sweep out a portion of the family of Reissner--Nordstr\"om solutions. In particular, we can observe a transition from \emph{dispersion within the domain of dependence of the data} (when the parameters are superextremal) to \emph{existence of a black hole} (when the parameters are extremal or subextremal). Extremal Reissner--Nordstr\"om is the critical solution in this phase transition. Note that none of these developments contain naked singularities. We conjecture that this critical behavior is preserved when we pass to the spherically symmetric Einstein--Maxwell-neutral scalar field model:

\begin{conj}[Extremality as a stable critical phenomenon] \label{conj:local-moduli-space}
Sufficiently close to extremal Reissner--Nordstr\"om, $\mathfrak M\setminus \mathfrak M_\stab^1$ has two connected components, which we denote by $\mathfrak M_\mathrm{sub}$ and $\mathfrak M_\mathrm{disp}$. These sets have the following properties:
\begin{enumerate}[(i)]
    \item The set $\mathfrak M_\mathrm{sub}$ consists of seed data leading to the formation of an asymptotically subextremal Reissner--Nordstr\"om black hole and is foliated by the hypersurfaces $\mathfrak M_\stab^\mathfrak r$ with $\mathfrak r<1$.
    \item Solutions evolving from seed data in $\mathfrak M_\mathrm{disp}$ do not form a black hole in the domain of dependence of the initial data hypersurface $C_\out\cup\underline C{}_\ing$.
\end{enumerate}
In particular, the ``extremal threshold'' $\mathfrak M_\stab^1$ delineates the boundary between black hole formation and dispersion in the domain of dependence of $C_\out\cup\underline C{}_\ing$.
\end{conj}

\begin{figure}
\centering{
\def\svgwidth{13pc}
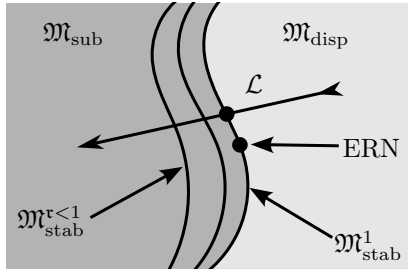}
\caption{A cartoon depiction of the conjectured structure of a neighborhood of extremal Reissner--Nordstr\"om in the moduli space $\mathfrak M$ of initial data posed as in \cref{fig:stability-intro}. We have suppressed infinitely many dimensions and emphasize the codimension-one property of the submanifolds $\mathfrak M_\stab^\mathfrak r$. We have drawn a distinguished point, which is extremal Reissner--Nordstr\"om. We have also drawn one of the lines $\mathcal L$ from \cref{fig:mod-space-lines} with the natural orientation given by increasing modulation parameter $\alpha$. The solid point on $\mathcal L$ corresponds to $\alpha=\alpha_\star$ (recall \cref{sec:intro-modulation}). See also \cref{fig:critical-behavior} below.}\label{fig:local-moduli-space} 
\end{figure}

This conjecture relates to \cref{fig:mod-space-lines} and our modulation scheme in the present paper as follows: We expect there to exist exactly one $\alpha_\star$ for which the modulation argument outlined in \cref{sec:intro-modulation} applies (as opposed to \emph{at least one}, which is what is shown in the present paper). For $\alpha>\alpha_\star$, we expect the seed data set $(\mathring\phi,\Lambda,M_0+\alpha,e)$ to lie in $\mathfrak M_\mathrm{sub}$ and if $\alpha<\alpha_\star$, we expect it to lie in $\mathfrak M_\mathrm{disp}$. We note that these conjectures are consistent with the numerical findings in \cite{Reall-numerical}. 

\begin{rk}
    By a standard Cauchy stability argument, generic solutions in $\mathfrak M_\mathrm{sub}$ which are very close to $\mathfrak M_\mathrm{stab}^1$ will experience a ``transient'' Aretakis instability, i.e., $Y\psi$ will remain constant and $Y^2\psi$ will grow for a long time before eventually decaying due to the redshift effect. This behavior was observed in \cite{Reall-numerical}.
\end{rk}

Since charge is conserved in the neutral scalar field model, it is not possible to have solutions with a nontrivial charge and a regular center. Therefore, \cref{conj:regularity,conj:local-moduli-space} are necessarily ``local near the event horizon'' because the ingoing cone $\underline C{}_\ing$ must terminate at a symmetry sphere with positive area-radius. Therefore, even in the ``dispersive'' case, the solutions are necessarily geodesically incomplete. The second- and third-named authors of the present paper have shown that extremal Reissner--Nordstr\"om black holes can form dynamically from initial data posed on $\Bbb R^3$ in the Einstein--Maxwell-charged scalar field and Einstein--Maxwell--charged Vlasov models \cite{KU22,KU24}. Therefore, \cref{conj:regularity,conj:local-moduli-space} can be formulated for these models, where data in $\mathfrak M_\mathrm{disp}$ now really ``exist globally'' and ``disperse'' in the literal sense. We refer to the introduction of \cite{KU24} for more details.

 \begin{figure}
\centering{
\def\svgwidth{35pc}
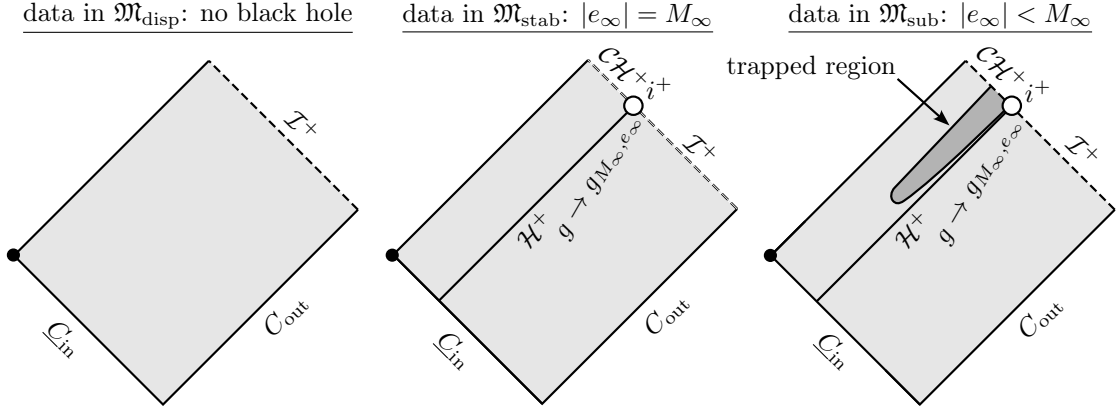}
\caption{Penrose diagrams depicting evolutions of seed data in the sets $\mathfrak M_\mathrm{disp}$, $\mathfrak M_\stab^1$, and $\mathfrak M_\mathrm{sub}$. One can think of these as arising from a one-parameter family of seed data crossing the extremal threshold in \cref{fig:local-moduli-space}. Spacetimes arising from $\mathfrak M_\mathrm{disp}$ have incomplete null infinity $\mathcal I^+$ because the ingoing cone $\underline C{}_\ing$ is incomplete and no black hole has formed. As proved in \cref{thm:stability-rough}, solutions on the critical threshold $\mathfrak M_\stab^1$ contain no trapped surfaces, but it follows from the work of Dafermos \cite{dafermos2005interior} that solutions arising from $\mathfrak M_\mathrm{sub}$ have a nonempty trapped region as depicted. The explicit Reissner--Nordstr\"om family itself already displays this transition behavior, in which case the trapped region in the third Penrose diagram would intersect the initial data, and the solid point would lie in the region $\{ r< M - \sqrt{M^2 - e^2} \}$ (in the left half of region III in \cite[Fig.~25]{HE73}) in the maximal analytic extension of subextremal Reissner--Nordstr\"om.}
\label{fig:critical-behavior}
\end{figure}

\subsubsection{The black hole interior at extremality}\label{sec:interior}

As explained in \cref{sec:rigidity-intro}, the black hole regions of the spacetimes evolving from data in $\mathfrak M_\stab$ are free of trapped surfaces, see \cref{fig:stability-intro}.  In particular, this shows that the interior is bounded to the future by a smooth outgoing null hypersurface (a trivial, smooth Cauchy horizon) emanating from the final sphere of $\underline C{}_\ing$ (the solid point in \cref{fig:stability-intro}) and a potentially singular ingoing Cauchy horizon $\mathcal{CH}^+$ emanating from $i^+$. Analogously to the surface gravity of $\mathcal H^+$ defined in \eqref{eq:surface-gravity}, one can define the surface gravity of $\mathcal{CH}^+$ in Reissner--Nordstr\"om for $0 < |e| \leq M$. In the extremal case $|e| = M$, this surface gravity also vanishes, resulting in the absence of Penrose's \emph{blueshift} instability \cite{Penrose:1968ar} and, somewhat ironically, renders the extremal Reissner--Nordstr\"om Cauchy horizon $\mathcal{CH}^+$ \emph{more} stable than its subextremal variant \cite{gajic1,gajic2}. Indeed, our decay estimates of \cref{thm:stability-rough} on the geometry and 
 the scalar field along $\mathcal H^+$ satisfy the assumptions in \cite[Theorem~5.1]{gajic-luk} and therefore the spacetimes evolving from $\mathfrak M_\stab$ can be extended in $C^{0,1/2}\cap H_{\mathrm{loc}}^1$ across $\mathcal{CH}^+$ and the Hawking mass remains uniformly bounded. This is in sharp contrast to the subextremal case, where generically the Hawking mass blows up identically along $\mathcal{CH}^+$ \cite{Dafermos-thesis,dafermos2005interior,LOSR23,Manohar} and extensions are believed to be less regular than $C^{0,1/2}\cap H^1_\mathrm{loc}$ \cite{LukOhI,luk2019strong}. In the extremal case, the existence of Cauchy data leading to any type of singularity at $\mathcal{CH}^+$ remains open---even for the linear wave equation! 
 
Note that according to \cref{conj:local-moduli-space}, asymptotically extremal Reissner--Nordstr\"om black holes only form from the (manifestly nongeneric) ``codimension~1 submanifold'' $\mathfrak M_\stab$, so their more regular Cauchy horizons do not endanger Penrose's strong cosmic censorship conjecture.

\subsection{Outlook}

We conclude the introduction with some thoughts and comments about future directions and related problems. We refer the reader also to the recent essay by Dafermos \cite{Daf24} about the stability problem for extremal black holes and the introduction of \cite{KU24} for connections with critical collapse.  

\subsubsection{Observational signatures and late-time tails}
\label{sec:tails}
In \cref{sec:sharp_asym}, we prove asymptotics for the scalar field in the near-horizon region, using the energy decay estimates from the proof of \cref{thm:stability-rough} (in particular the $(\bar{r}-M)^{2-p}$ hierarchy up to $p=3-\delta$). In \cref{sec:sharpness}, we then use this sharp decay estimate to show that the range of the $(\bar{r}-M)^{2-p}$ hierarchy is indeed sharp, as long as the asymptotic Aretakis charge is not too small compared to the master smallness parameter $\ve$ of the problem. Namely, we show that the non-degenerate integrated energy is \emph{unbounded} in any infinite area close to the horizon. This is in sharp contrast to the subextremal case, where the celebrated redshift estimate \cite{dafermos2009red} implies that this integrated energy is finite near the horizon. 

We note that compared to the uncoupled case, the leading order term in the decay for $\psi$ does not only include the asymptotic Aretakis charge $H_0[\phi]$ but also an error of order $\ve^3$ that comes from the geometric estimates. This marks a difference with the uncoupled case and is a manifestation of the failure of the Couch--Torrence conformal isometry in the dynamical case. This creates a marked difference in the proof of the failure of the redshift estimate compared to the uncoupled case, see already \cref{rk:failure}. 

In the present paper, we only obtain sharp leading asymptotics for the scalar field near the horizon. It would be an interesting problem to obtain asymptotics for the scalar field everywhere for our asymptotically extremal black holes, depending on the initial assumptions for the behavior of the scalar field at infinity. In the uncoupled case, assuming either compactly supported initial data or sufficiently fast decay so that the so-called \emph{Newman--Penrose constant} vanishes, one can extend the $r^p$ hierarchy to the range of $p\in(0,5)$. After some technical work (commutation and a ``time inversion'' process), one can show that the radiation field decays like $u^{-2}$ in a region close to $\mathcal I^+$, where the constant of the leading order decaying term depends not only on the Newman--Penrose constant but also on the Aretakis charge. This can be considered an \emph{observational signature} for extremal Reissner--Nordstr\"{o}m black holes and it would be interesting to extract this signature in our dynamical setting. For more details on observational signatures for extremal Reissner--Nordstr\"{o}m black holes, see \cite{angelopoulos2018late,angelopoulos2018late} and the numerical results of \cite{AKS-signature}.

\subsubsection{Outside of symmetry}\label{sec:EK}

The present work in symmetry is the first step towards addressing the stability problem for (and, more generally, the structure of moduli space near) extremal black holes in Einstein(--Maxwell) theory outside of symmetry. We wish to consider the \emph{extremal Kerr--Newman} family which consists of axisymmetric, rotating charged black holes with mass $M$, charge $e$, and specific angular momentum $a$ satisfying the relation $M^2=e^2+a^2$. When $a=0$, this reduces to extremal Reissner--Nordstr\"om and when $e=0$, to extremal Kerr. 

In the very slowly rotating case ($|a|\ll M$), we have the following:
\begin{conj}[Dafermos--Holzegel--Rodnianski--Taylor \cite{DHRT,Daf24}]\label{conj:DHRT}
    \cref{conj:regularity,conj:local-moduli-space} hold, suitably interpreted, in a neighborhood of very slowly rotating extremal Kerr--Newman, without symmetry. 
\end{conj} 

Significant progress on linear stability in the subextremal case has been made by Giorgi and Wan in \cite{Giorgi-KN-1,Giorgi-KN-2,Giorgi-KN-3,Giorgi-KN-4}). For extremal Reissner--Nordstr\"om, a complete understanding of the Teukolsky equation was recently obtained by Apetroaie \cite{apetroaie}, who showed both decay statements and that a version of the Aretakis instability is present (see also the earlier \cite{LMR13}).\footnote{One could also consider this conjecture further restricted to the codimension-three set of data with vanishing final angular momentum (so as to only consider asymptotically Reissner--Nordstr\"om solutions, as in \cite{DHRT}), in which case \cite{apetroaie} would become all the more relevant.} Unfortunately, an understanding of even the linear wave equation on very slowly rotating extremal Kerr--Newman (with $a\ne 0$) remains open. In addition to requiring a detailed understanding of linear theory around very slowly rotating extremal Kerr--Newman, we expect that resolving \cref{conj:DHRT} will also rely on a nontrivial understanding of the null structure of the Einstein equations in the near-horizon region, analogous to the nonlinear results of \cite{A16,AAG-nonlinear-1,AAG20} and the present work.

The case of extremal Kerr is considerably more speculative. While mode stability for the linear wave equation has been shown by Teixeira da Costa \cite{Rita-mode-stability}, axisymmetric scalar perturbations have been shown to exhibit a similar non-decay and growth hierarchy as general scalar perturbations of extremal Reissner–Nordstr\"om \cite{Aretakis-Kerr}. Moreover, even pointwise boundedness for linear waves remains open. Recently, Gajic has shown the existence of stronger \emph{azimuthal instabilities} (i.e., associated to $m\ne0$) on extremal Kerr \cite{Gajic23} (see also the earlier heuristic analysis \cite{cgz-exkerr}). In particular, growth is already triggered at the first order of differentiability. It remains to be seen if these instabilities are compatible with global existence and stability for nonlinear wave equations! For more discussion, we refer to \cite{Daf24}. 

\subsection*{Acknowledgments} The authors would like to thank Mihalis Dafermos for insightful conversations on teleology and dyadicism. They would also like to thank John Anderson, Jonathan Luk, Harvey Reall, Igor Rodnianski, Rita Teixeira da~Costa, and Claude Warnick for their interest, helpful discussions, and comments. R.U. acknowledges support from the NSF grant DMS-2401966 and a Miller Fellowship.

\section{Preliminaries}

\subsection{The Einstein--Maxwell-uncharged scalar field system in spherical symmetry}\label{sec:the-model}

\subsubsection{Double null gauge}\label{sec:double-null-gauge}

Let $(\mathcal M,g)$ be a smooth, connected, time-oriented, four-dimensional Lorentzian manifold. We say that $(\mathcal M,g)$ is \emph{spherically symmetric} if $\mathcal M $ splits diffeomorphically as $\mathring{\mathcal Q}\times S^2$ with metric 
\begin{align*}
    g = g_\mathcal{Q} + r^2g_{S^2},
\end{align*}
where $(\mathcal Q,g_\mathcal Q)$ is a (1+1)-dimensional Lorentzian spacetime with boundary (corresponding to the initial data hypersurface\footnote{In this paper, our spacetimes do not include a center of symmetry $\Gamma\subset\{r=0\}$.}), $g_{S^2}\doteq d\vartheta^2+\sin^2\vartheta\,d\varphi^2$ is the round metric on the unit sphere, and $r$ is a nonnegative function on $\mathcal Q$ which can be geometrically interpreted as the area-radius of the orbits of the isometric $\mathrm{SO}(3)$ action on  $(\mathcal M,g)$. We assume that $(\mathcal Q,g_\mathcal Q)$ admits a \emph{global double-null coordinate system} $(u,v)$ such that the metric $g$ takes the form
\begin{equation}
g = -\Omega^2\,dudv + r^2  g_{S^2}\label{eq:dn}
\end{equation}
for a positive function $\Omega^2\doteq -2g_\mathcal{Q}(\partial_u,\partial_v)$ on $\mathcal Q$ and such that $\partial_u$ and $\partial_v$ are future-directed. The constant $u$ and $v$ curves are null in $(\mathcal Q,g_\mathcal Q)$ and correspond to null hypersurfaces ``upstairs'' in the full spacetime $(\mathcal M,g)$. We will often refer interchangeably to  $(\mathcal M,g)$ and the reduced spacetime $(\mathcal Q,r,\Omega^2)$.

The double null coordinates $(u,v)$ above are not uniquely defined. For any strictly increasing smooth functions $U,V: \R \to \R$, $(\tilde u,\tilde v)=(U(u),V(v))$ defines a double null coordinate system on $\mathcal Q$ for which $g= -\tilde \Omega^2 \,d \tilde u d \tilde v + \tilde r^2g_{S^2}$, where $\tilde \Omega^2(\tilde u , \tilde v) = (U' V')^{-1} \Omega^2(U^{-1}(\tilde u), V^{-1}(\tilde v)) $ and $r (\tilde u, \tilde v) = r(U^{-1} (\tilde u), V^{-1} (\tilde v))$.

Recall the \emph{Hawking mass} $m:\mathcal M  \to\mathbb R$, which is defined by the relation  \begin{equation*}
   1-\frac{2m}{r} \doteq g(\nabla r, \nabla r)
\end{equation*}
and can be viewed as a function on $\mathcal Q$ according to 
\begin{equation}\label{eq:Hawking-mass}
    m = \frac r2 \left( 1+  \frac{4\partial_ur\partial_vr}{\Omega^2}\right).
\end{equation}
This function is clearly independent of the choice of double null gauge. 

We will consider spherically symmetric spacetimes equipped with spherically symmetric electromagnetic fields with constant Coulomb charge $e\in\Bbb R$. The field strength tensor takes the simple form 
\begin{equation}
    F = -\frac{\Omega^2e}{2r^2}du\wedge dv\label{eq:F-sph-sym}
\end{equation}
and the charge can be recovered from Gauss' formula 
\begin{equation}\label{eq:charge-definition}
    e=\frac{1}{4\pi}\int_{\{(u,v)\}\times S^2}\star F
\end{equation}
where $\star$ is the Hodge star operator of $(\mathcal M,g)$ defined relative to the orientation $du\wedge dv\wedge d\vartheta\wedge d\varphi$.

To see best the good structure of the spherically symmetric Einstein equations, it is very helpful to introduce some shorthand notation. We recall the \emph{renormalized Hawking mass} 
\begin{equation}
    \varpi\doteq m+\frac{e^2}{2r},\label{eq:varpi-defn}
\end{equation}
 the \emph{mass aspect function}
\begin{equation}
    \mu\doteq \frac{2m}{r}= \frac{2\varpi}{r}-\frac{e^2}{r^2},\label{eq:mu-defn}
\end{equation} and the traditional notation \begin{equation}
    \nu\doteq\partial_ur,\quad\lambda\doteq \partial_vr, \quad\kappa\doteq-\frac{\Omega^2}{4\partial_ur} = \frac{\lambda}{1-\mu},\quad \gamma\doteq - \frac{\Omega^2}{4\partial_vr} =\frac{\nu}{1-\mu}.\label{eq:greek-letters}
\end{equation}
As the degenerate redshift effect of the event horizon of extremal Reissner--Nordstr\"om will play an important role in this paper, we define the \emph{redshift factor}
\begin{equation}
 \label{eq:varkappa}   \varkappa\doteq \frac{1}{r^2}\left(\varpi-\frac{e^2}{r}\right).
\end{equation}
The connection of this \emph{function} $\varkappa$ with the \emph{constant} $\varkappa$ in \eqref{eq:surface-gravity} will be explained in \cref{rk:redshift-factor} below.

\subsubsection{The system of equations} 

\begin{defn}
    The \emph{Einstein--Maxwell-neutral scalar field system} consists of a spacetime $(\mathcal M,g)$ equipped with an electromagnetic field $F$ and a scalar field $\phi:\mathcal M\to\Bbb R$ satisfying the equations 
    \begin{equation}
    \Ric(g)-\tfrac 12 R(g)g = 2(T^\mathrm{EM}+T^\mathrm{SF}),\label{eq:EFE}
\end{equation}
\begin{equation}
    dF=0,\quad d\star F=0,\label{eq:Maxwell}
\end{equation}
\begin{equation}
    \Box_g \phi =0,\label{eq:phi-wave-general}
\end{equation}
where the energy-momentum tensors are defined by 
\begin{align}
 \nonumber   T^\mathrm{EM}_{\mu\nu}&\doteq F_{\mu\alpha}F_\nu{}^\alpha-\tfrac 14g_{\mu\nu}F_{\alpha\beta}F^{\alpha\beta},\\
 \label{eq:T-sf}   T^\mathrm{SF}_{\mu\nu}&\doteq \partial_\mu\phi\partial_\nu\phi-\tfrac 12 g_{\mu\nu}\partial_\alpha\phi\partial^\alpha\phi.
\end{align}
\end{defn}

We say that $(\mathcal M,g,F,\phi)$ is \emph{spherically symmetric} if $(\mathcal M,g)$ is spherically symmetric as defined in \cref{sec:double-null-gauge}, $F$ has the form \eqref{eq:F-sph-sym}, and $\phi$ is independent of the angular coordinates $\vartheta$ and $\varphi$. In this case, Einstein's equation \eqref{eq:EFE} reduces to the wave equations
\begin{align}
\label{eq:r-wave}    \partial_u\partial_vr&=-\frac{\Omega^2}{4r}-\frac{\partial_ur\partial_vr}{r}+\frac{\Omega^2e^2}{4r^3},\\
    \label{eq:Omega-wave}\partial_u\partial_v{\log\Omega^2}&=\frac{\Omega^2}{2r^2}+\frac{2\partial_ur\partial_vr}{r^2}-\frac{\Omega^2e^2}{r^4} - 2 \partial_u \phi \partial_v \phi ,
\end{align}
and Raychaudhuri's equations
\begin{align}
   \label{eq:Ray-u} \partial_u\left(\frac{\partial_ur}{\Omega^2}\right)&=-\frac{r}{\Omega^2}(\partial_u\phi)^2,\\
  \label{eq:Ray-v}\partial_v\left(\frac{\partial_vr}{\Omega^2}\right)&=-\frac{r}{\Omega^2}(\partial_v\phi)^2.
\end{align}
The Maxwell equations \eqref{eq:Maxwell} are automatically satisfied since $e$ is constant. Finally, the wave equation \eqref{eq:phi-wave-general} is equivalent to 
\begin{equation}
     \partial_u\partial_v\phi= - \frac{\partial_vr\partial_u\phi}{r}-\frac{\partial_ur\partial_v\phi}{r}.\label{eq:phi-wave-1}
\end{equation}

We will not actually work with the equations \eqref{eq:r-wave}--\eqref{eq:phi-wave-1} as written here and will instead use the quantities \eqref{eq:varpi-defn}--\eqref{eq:varkappa} which satisfy more helpful equations. First, the wave equation \eqref{eq:r-wave} can be written in the compact form 
\begin{equation}
    \label{eq:nu-v}\partial_u\lambda=\partial_v\nu=2\kappa\nu\varkappa,
\end{equation}
which allows us to write \eqref{eq:phi-wave-1} as an inhomogeneous wave equation for the \emph{radiation field} $\psi\doteq r\phi$,
\begin{equation}\label{eq:wave-equation-psi}
    \partial_u\partial_v\psi=(\partial_u\partial_vr)\phi=2\kappa\nu\varkappa\phi.
\end{equation}
Using \eqref{eq:r-wave}--\eqref{eq:Ray-v}, we derive the fundamental relations \begin{align}
  \label{eq:varpi-u}  \partial_u\varpi &= (1-\mu)\frac{r^2}{2\nu}(\partial_u\phi)^2,\\
 \label{eq:varpi-v}   \partial_v\varpi &= \frac{r^2}{2\kappa}(\partial_v\phi)^2,\\
\label{eq:kappa-u}    \partial_u\kappa &=\frac{r\kappa}{\nu}(\partial_u\phi)^2,\\
    \label{eq:gamma-v}    \partial_v\gamma&=\frac{r\gamma}{\lambda}(\partial_v\phi)^2.
\end{align}

Finally, we note that for any $\phi$ satisfying the wave equation \eqref{eq:phi-wave-1} and any $C^1$ vector field $X=X^u\partial_u+X^v\partial_v$ and any $C^2$ function $h=h(r)$, we have the multiplier identities \begin{equation}\label{eq:energy-identity}
\partial_v(r^2(\partial_u\phi)^2X^u)+\partial_u(r^2(\partial_v\phi)^2X^v)= r^2\partial_vX^u(\partial_u\phi)^2+r^2\partial_uX^v(\partial_v\phi)^2-2r\left(\nu X^u+\lambda X^v\right)\partial_u\phi\partial_v\phi
    \end{equation}
      and  \begin{equation}  \label{eq:lower-order-multiplier}
           \partial_u\partial_v ( h r\phi^2) - \partial_u( r  h' \lambda  \phi^2 ) - \partial_v ( r  h' \nu \phi^2 ) + (r\lambda \nu h'' + 2 r \kappa \nu \varkappa h' - 2 \kappa\varkappa \nu h 
        ) \phi^2 - 2h r\partial_u \phi \partial_v \phi = 0.
    \end{equation}

\subsection{The geometry of Reissner--Nordstr\"om}\label{sec:geometry-RN}

In this section, we briefly review some important geometric aspects of Reissner--Nordstr\"om black hole spacetimes. See \cref{fig:ERN} below for the Penrose diagram of extremal Reissner--Nordstr\"om. Let $M>0$ and $e\in\Bbb R$. Then the Reissner--Nordstr\"om metric is written in Schwarzschild coordinates $(t,r,\vartheta,\varphi)$ as
\begin{equation*}
     g_{M,e}= -D(r)dt^2+D(r)^{-1}dr^2+r^2g_{S^2},
\end{equation*}
 where $D(r)\doteq 1-\frac{2M}{r}+\frac{e^2}{r^2}$. When $|e|\le M$ and $g_{M,e}$ analytically extends to describe a black hole spacetime, these coordinates cover the domain of outer communication of the black hole for $(t,r) \in \mathbb R \times (r_+,\infty)$, where $r_+ \doteq  M+\sqrt{M^2 - e^2}$. In order to cover the event horizon $\mathcal H^+$ located at $r=r_+$, we introduce the \emph{tortoise coordinate} \begin{equation*}
    r_*(r)\doteq \int^r \frac{dr'}{D(r')},
\end{equation*} where we have left the integration constant ambiguous. In the extremal case $|e|=M$, we have
\begin{equation}
 r_*(r)=   r-M-\frac{M^2}{r-M}+2M\log(r-M) + \mathrm{constant},\label{eq:tortoise}
\end{equation} which will be used in \cref{sec:sharp_asym} below.
By defining the advanced time coordinate $v\doteq \frac 12(t+r_*)$, we bring $g_{M,e}$ into the \emph{ingoing Eddington--Finkelstein} form 
\begin{equation*}
    g_{M,e}= - 4Ddv^2+4dvdr+r^2g_{S^2},
\end{equation*} which is regular across $\mathcal H^+$. The vector field $T\doteq \frac 12\partial_v$ is the time-translation Killing vector field (equal to $\partial_t$ in Schwarzschild coordinates in the domain of outer communication).
In these coordinates, the vector field $Y\doteq \partial_r$ is past-directed null and is transverse to $\mathcal H^+$. It is also clearly translation-invariant in the sense that 
\begin{equation*}
    \mathcal L_TY=[\partial_v,\partial_r]=0.
\end{equation*}

In order to put the Reissner--Nordstr\"om metric in double null gauge, we define the retarded time coordinate $u\doteq \frac 12(t-r_*)$, so that the metric takes the \emph{Eddington--Finkelstein double null} form
\begin{equation*}
    g_{M,e}=-4D\,dudv+r^2g_{S^2}.
\end{equation*} The area-radius $r$ is now an implicit function of $u$ and $v$. These coordinates once again only cover the domain of outer communication if $|e|\le M$. The event horizon $\mathcal H^+$ formally corresponds to $u=+\infty$ and null infinity $\mathcal I^+$ formally corresponds to $v=+\infty$. From the identity $r_*=v-u$, we infer that in these coordinates, $\partial_ur=-D$ and $\partial_vr=D$. Since clearly $\Omega^2=4D$ and $D=1-\mu$, this implies that 
\begin{equation*}
    \gamma=-1\quad\text{and}\quad \kappa = 1
\end{equation*}
in these coordinates. Since $\partial_ur=-D$, we have that 
\begin{equation*}
    Y = (\partial_ur)^{-1}\partial_u.
\end{equation*} 

Finally, by simply changing the origin of the $(u,v)$ coordinates, we prove:

 \begin{figure}
\centering{
\def\svgwidth{11pc}
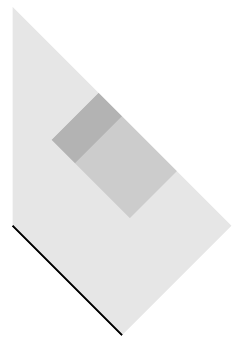}
\caption{A Penrose diagram of (one period of) the maximally extended extremal Reissner--Nordstr\"om solution. The union of the two darker shaded regions is the domain of dependence of the bifurcate null hypersurface $C_\out\cup\underline C{}_\ing$ and represents the solution we are perturbing around in \cref{thm:stability-rough}. We prove stability of the medium gray colored region.}
\label{fig:ERN}
\end{figure}

\begin{lem}\label{lem:existence-anchoring}
    Let $M>0$ and $|e|\le M$ be Reissner--Nordstr\"om black hole parameters. Then for any $R_0\in (r_+,\infty)$ and $(u_0,v_0)\in\Bbb R^2$, there exists a unique function $\bar r:\Bbb R^2\to (r_+,\infty)$ such that 
    \begin{equation}\label{eq:ERN-EF}
        \bar g\doteq -4D(\bar r)\,dudv+ \bar r^2(u,v)g_{S^2}
    \end{equation}
    on $\Bbb R^2\times S^2$ is isometric to the Reissner--Nordstr\"om black hole exterior with parameters $(M,e)$ and such that 
    \begin{equation*}\bar\varpi=M,\quad
      \partial_u\bar r=-(1-\bar\mu),\quad \partial_v\bar r=1-\bar\mu,\quad\text{and}\quad   \bar r(u_0,v_0)=R_0.
    \end{equation*}
\end{lem}

As mentioned above, the Eddington--Finkelstein double null coordinates defined in this section do not cover the event horizon, but it can be formally attached as the null hypersurface $u=+\infty$. With this understanding, we may extend geometric quantities associated to the metric \eqref{eq:ERN-EF} to the horizon by setting
\begin{equation*}
    \bar r(\infty,\cdot)=M,\quad  \bar \lambda(\infty,\cdot)=0, \quad \bar \mu(\infty,\cdot)=1,\quad \bar\varkappa(\infty,\cdot) =0.
\end{equation*}

\begin{rk}\label{rk:redshift-factor}
    In exact Reissner--Nordstr\"om, we may write the time-translation Killing vector field relative to a generic double null coordinate system as the \emph{Kodama vector field}
    \begin{equation}\label{eq:ERN-Kodama}
        T= \frac{2\lambda}{\Omega^2}\partial_u-\frac{2\nu}{\Omega^2}\partial_v,
    \end{equation} as long as $u$ and $v$ are chosen so that $T$ is future-directed timelike for $r>r_+$. 
 One can then compute
    \begin{equation*}
        \nabla_TT = -\varkappa T^u\partial_u+\varkappa T^v\partial_v,
    \end{equation*} where $\varkappa$ is the redshift factor defined in \eqref{eq:varkappa}.
    Therefore, on the event horizon $\mathcal H^+$, where $T^u=0$, 
    \begin{equation*}
        \nabla_TT|_{\mathcal H^+}=\varkappa T|_{\mathcal H^+}.
    \end{equation*}
Compare this identity with \eqref{eq:surface-gravity}.
\end{rk}

\subsection{The characteristic initial value problem}\label{sec:data-and-IVP}

\subsubsection{Characteristic data and existence in thin slabs}

Given $U_0<U_1$ and $V_0<V_1$, let
\begin{align*}
  \mathcal C(U_0,U_1,V_0,V_1)  &\doteq (\{U_0\}\times[V_0,V_1])\cup([U_0,U_1]\times\{V_0\}), \\
   \mathcal R(U_0,U_1,V_0,V_1) &\doteq [U_0,U_1]\times [V_0,V_1].
\end{align*} We will omit the decoration $(U_0,U_1,V_0,V_1)$ from $\mathcal C$ and $\mathcal R$ when the implied meaning is clear. A continuous function $f:\mathcal C\to \Bbb R$ is said to be \emph{smooth} if $f_\out\doteq f|_{\{U_0\}\times [V_0,V_1]}$ and $f_\ing\doteq f|_{[U_0,U_1]\times\{V_0\}}$ are $C^\infty$ functions.

\begin{defn}
  A smooth \emph{(bifurcate) characteristic initial data set} for the Einstein--Maxwell-scalar field system consists of a triple of smooth functions $\mathring r,\mathring\Omega^2,\mathring\phi:\mathcal C\to\Bbb R$ with $\mathring r$ and $\mathring\Omega^2$ positive, together with a real number $e$. The functions $\mathring r,\mathring\Omega^2$, and $\mathring\phi$ are assumed to satisfy \eqref{eq:Ray-u} on $[U_0,U_1]\times\{V_0\}$ and \eqref{eq:Ray-v} on $\{U_0\}\times[V_0,V_1]$.
\end{defn}

\begin{prop}\label{prop:slab-existence}
    For any $U_0<U_1$, $V_0<V_1$ and $B>0$, there exists a constant $\ve_\mathrm{slab}$ with the following property. Let $(\mathring r,\mathring\Omega^2,\mathring \phi,e)$ be a characteristic initial data set of the Einstein--Maxwell-scalar field system on $\mathcal C(U_0,U_1,V_0,V_1)$ satisfying
    \begin{equation*}
        \|{\log \mathring r}\|_{C^1(\mathcal C)}+\|{\log\mathring\Omega^2}\|_{C^1(\mathcal C)}+\|\mathring\phi\|_{C^1(\mathcal C)}+|e|\le B.
    \end{equation*}
    Then there exists a unique, smooth solution $(r,\Omega^2,\phi,e)$ to the spherically symmetric Einstein--Maxwell-scalar field system on the slab
    \begin{equation*}
      \mathcal D\doteq  \mathcal R(U_0,U_0+\min\{\ve_\mathrm{slab},U_1-U_0\},V_0,V_1)\cup\mathcal R(U_0,U_1,V_0,V_0+\min\{\ve_\mathrm{slab},V_1-V_0\})
    \end{equation*}
    which extends the initial data. Moreover, the norms $\|{\log r}\|_{C^k(\mathcal D)},$ $\|{\log\Omega^2}\|_{C^k(\mathcal D)}$, and $\|\phi\|_{C^k(\mathcal D)}$ are bounded in terms of appropriate higher order initial data norms. 
\end{prop}
\begin{proof}
  Local existence for the characteristic initial value problem in small rectangles is obtained by a standard iteration argument, see for instance \cite[Appendix A]{KU24}. In order to extend the region of existence to a thin slab, one exploits the null structure exhibited by the system \eqref{eq:r-wave}--\eqref{eq:phi-wave-1}: the equations can be viewed as linear ODEs in $u$ for the $v$-derivatives of the unknowns, and vice-versa. See \cite{Luk-local-existence} for an elaboration of this principle in a much more complicated setting. Alternatively, one can use the ``generalized extension principle'' for this model (see \cite{Kommemi13}) and argue as in \cite[Proposition 3.17]{KU24}.
\end{proof}

We also have a natural notion of maximal globally hyperbolic development \cite{CBG69,Zorn-slayed} in spherical symmetry, which can be directly realized as a subset of the domain of dependence of $\mathcal C$ viewed as a subset of $(1+1)$-dimensional Minkowski space.

\begin{prop}
 Let $(\mathring r,\mathring\Omega^2,\mathring \phi,e)$ be a characteristic initial data set of the Einstein--Maxwell-scalar field system on $\mathcal C(U_0,U_1,V_0,V_1)$, where $U_1$ and $V_1$ are allowed to take the value $+\infty$. Then there exists a set $\mathcal Q_\mathrm{max}\subset \mathcal R(U_0,U_1,V_0,V_1)$ with the following properties:
 \begin{enumerate}
 \item $\mathcal Q_\mathrm{max}$ is globally hyperbolic as a subset of $\Bbb R^{1+1}_{u,v}$ with Cauchy surface $\mathcal C$.
\item The solution $(r,\Omega^2,\phi,e)$ extends uniquely to $\mathcal Q_\mathrm{max}$.
\item $\mathcal Q_\mathrm{max}$ is maximal with respect to properties 1.~and 2.
\end{enumerate}
\end{prop}
For a characterization of the future boundary of $\mathcal Q_\mathrm{max}$, see \cite{Kommemi13}.

\subsubsection{Gauge-normalized seed data}

In defining the moduli space of initial data $\mathfrak M$ for our main theorem, it will be convenient to parametrize the space of bifurcate characteristic initial data as a linear space in such a way that certain gauge conditions are automatically satisfied on $\mathcal C$ in the maximal development. Note that this notion of seed data is distinct from the one used in \cite{KU22} because the gauge condition is different. 

\begin{defn}
    Let $\mathcal C$ be a spherically symmetric bifurcate null hypersurface. A \emph{seed data set} is a quadruple
    \begin{equation*}
        \mathcal S\doteq (\mathring \phi,r_0,\varpi_0,e),
    \end{equation*}
    where $\mathring\phi$ is a smooth function on $\mathcal C$, $r_0$ is a positive real number, and $\varpi_0,e$ are real numbers satisfying the condition
    \begin{equation*}
        \frac{2\varpi_0}{r_0}-\frac{e^2}{r_0^2}< 1. 
    \end{equation*}
\end{defn}

\begin{prop}\label{prop:seed-data-generation}
   Let $\mathcal S=(\mathring\phi,r_0,\varpi_0,e)$ be a seed data set on $\mathcal C(U_0,U_1,V_0,V_1)$ with $U_1-U_0<r_0$. Then there exists a unique characteristic initial data set $(\mathring r,\mathring \Omega^2,\mathring\phi,e)$ on $\mathcal C(U_0,U_1,V_0,V_1)$ such that the maximal development $(\mathcal Q_\mathrm{max},r,\Omega^2,\phi,e)$ of $(\mathring r,\mathring \Omega^2,\mathring\phi,e)$ has the following properties: 
   \begin{enumerate}
       \item $r(U_0,V_0)=r_0$,
       \item $\varpi(U_0,V_0)=\varpi_0$,
       \item $\nu=-1$ on $[U_0,U_1]\times\{V_0\}$, and
       \item $\lambda= 1$ on $\{U_0\}\times[V_0,V_1]$.
   \end{enumerate}
\end{prop}
We refer to characteristic data obtained from $\mathcal S$ in this manner as \emph{gauge-normalized} characteristic data determined by $\mathcal S$.
\begin{proof} For $(u,v)\in[U_0,U_1]\times[V_0,V_1]$, set
\begin{equation*}
    \mathring r_\ing(u)\doteq r_0-u,\quad \mathring r_\out(v)\doteq r_0+v,
\end{equation*}
and then 
\begin{equation*}
  \mathring\Omega^2_\ing(u)\doteq  \frac{4}{1-\frac{2\varpi_0}{r_0}+\frac{e^2}{r_0^2}} \exp\left(-\int_{U_0}^u \mathring r_\ing (\partial_u\mathring\phi_\ing)^2\,du'\right),\quad   \mathring\Omega^2_\out(v)\doteq  \frac{4}{1-\frac{2\varpi_0}{r_0}+\frac{e^2}{r_0^2}} \exp\left(\int_{V_0}^v \mathring r_\out (\partial_v\mathring\phi_\out)^2\,dv'\right).
\end{equation*}

Assembling these functions into a characteristic data set $(\mathring r,\mathring\Omega^2,\mathring\phi,e)$, we can immediately see that conditions 1.--4.~will be satisfied on the maximal development by virtue of the definitions.
\end{proof}

\section{Stability and instability of extremal Reissner--Nordstr\"om: setup and statements of the main theorems}\label{sec:setup-statement}

In this section, we give the detailed statements of our main results as well as the precise definitions needed to make the statements. In \cref{sec:seed-data}, we define the moduli space of seed data $\mathfrak M$ that features in \cref{thm:stability-rough}. In \cref{sec:setup}, we define the bootstrap domains $\mathcal D_{\tau_f}$ and the associated teleologically normalized gauges. In \cref{sec:anchor}, we define the anchored background extremal Reissner--Nordstr\"om spacetimes on $\mathcal D_{\tau_f}$ and the associated energy hierarchies. Finally, in \cref{sec:statements}, we give the precise statements of \cref{thm:stability-rough,thm:instability-rough}.

\subsection{Definition of the moduli space of seed data \texorpdfstring{$\mathfrak M$}{mathfrak M}}\label{sec:seed-data}

Fix a mass parameter $M_0>0$ and let $e_0$ be a charge parameter satisfying $|e_0|=M_0$. Let
\begin{equation*}
    U_*\doteq \frac{995}{10}M_0
\end{equation*}
and let
\begin{equation*}
    \hat{\mathcal C}\doteq \mathcal C(0,U_*,0,\infty) = \underline C{}_\ing \cup C_\out
\end{equation*}
denote the bifurcate null hypersurface on which we pose our data. We denote the null coordinates on $\hat{\mathcal C}$ by $\hat u$ and $\hat v$. Later, in the proof of the main theorem, $\hat u$ and $\hat v$ will be renormalized. 

Let $\mathcal S=(\mathring \phi,\Lambda,\varpi_0,e)$ denote a seed data set on $\hat{\mathcal C}$ with bifurcation sphere area-radius $r_0$ denoted by $\Lambda$. We define the \emph{seed data norm}
\begin{multline*}
       \mathfrak D[\mathcal S]\doteq |\Lambda-100M_0|+ |\varpi_0-M_0|+|e-e_0| +\sup_{\underline C{}_\ing}\left(|\mathring\phi_\ing|+|\partial_{\hat u}\mathring\phi_\ing|\right)\\+\sup_{C_\out} \left((1+\hat v)|\mathring\phi_\out|+(1+\hat v)^2|\partial_{\hat v}\mathring\phi_\out|+(1+\hat v)^2|\partial_{\hat v}(\hat v\mathring\phi_\out)|\right).
\end{multline*}

    The seed data norm $\mathfrak D$ and many of the constructions and smallness parameters in this paper will implicitly depend on the (fixed) choice of $M_0$.
 
\begin{rk}
 The unique seed data with $\mathfrak D=0$, namely $(0,100M_0,M_0,e_0)$, corresponds to extremal Reissner--Nordstr\"om with parameters $(M_0,e_0)$ and bifurcation sphere area-radius $100M_0$. Note that since $\Lambda$ is not fixed to be $100M_0$, there is a one-parameter family of seed data corresponding to the extremal Reissner--Nordstr\"om black hole with parameters $(M_0,e_0)$, but all except for the one with $\Lambda=100M_0$ will have $\mathfrak D>0$. The event horizon for the evolution of $(0,100M_0,M_0,e_0)$ is located at 
   \begin{equation}
        \hat u_{\mathcal H^+,0}\doteq 99M_0\label{eq:original-horizon}
    \end{equation}
    in the $(\hat u,\hat v)$ gauge determined by \cref{prop:seed-data-generation}.
\end{rk}

We now show that for $\mathfrak D[\mathcal S]$ sufficiently small, the associated characteristic data have ``no antitrapped spheres of symmetry.''

\begin{lem}\label{lem:epsilon-loc}
    There exists an $\ve_\mathrm{loc}>0$ depending only on $M_0$ such that if $\mathcal S$ is a seed data set for which
    \begin{equation}\label{eq:seed-smallness}
        \mathfrak D[\mathcal S]\le 3\ve_\mathrm{loc}, 
    \end{equation}
    then the maximal globally hyperbolic development $(\hat{\mathcal Q}_\mathrm{max},r,\Omega^2,\phi,e)$ of $\mathcal S$ has $\partial_{\hat u}r<0$ everywhere on $\hat{\mathcal Q}_\mathrm{max}$.
\end{lem}

\begin{rk}
    The role of the number $3$ in \eqref{eq:seed-smallness} is  that every seed data set in the moduli space $\mathfrak M$, defined below in \eqref{eq:M-defn}, will satisfy \eqref{eq:seed-smallness}.
\end{rk}

\begin{proof}[Proof of \cref{lem:epsilon-loc}] Let $\ve>0$ and suppose $\mathfrak D[\mathcal S]\le \ve$. We will show that $\hat\nu<0$ on $C_\out$ for $\ve$ sufficiently small, which is then propagated to the rest of $\hat{\mathcal Q}_\mathrm{max}$ by the inequality $\partial_{\hat u}(\hat \Omega^{-2}\hat\nu)\le 0$ obtained from Raychaudhuri's equation \eqref{eq:Ray-u}. We make the bootstrap assumption
\begin{equation}\label{eq:1-mu-boot}
    \tfrac 12\le 1-\mu\le 2
\end{equation}
on $\{0\}\times[0,\hat v_0]$, which is clearly satisfied for $\ve$ and $\hat v_0$ sufficiently small. Since $\hat\lambda=1$ on $C_\out$, $r(0,\hat v)=\Lambda+\hat v$ and
\begin{equation*}
    \left|\partial_{\hat v}\varpi(0,\hat v)\right| = \left|\tfrac 12 (1-\mu)(0,\hat v)(\Lambda+\hat v)^2(\partial_{\hat v}\mathring\phi_\out)^2 \right|\les \ve^2(1+\hat v)^{-2}
\end{equation*}
by \eqref{eq:varpi-v} and \eqref{eq:1-mu-boot}. Therefore, $|\varpi(0,\hat v)-M_0 |\les \ve^2$ for every $\hat v\in[0,\hat v_0]$, which easily allows us to propagate \eqref{eq:1-mu-boot} and therefore  \eqref{eq:1-mu-boot} holds on $C_\out$. It now follows easily from the relation $\hat\Omega^2 =-4\hat\lambda\hat\nu /(1-\mu)$
that $\hat\nu<0$ on $C_\out$.
\end{proof}

  We denote by $\mathcal S_0$ seed data sets on $\hat{\mathcal C}$ for which $\varpi_0=M_0$. These constitute a codimension-one affine subspace, denoted by $\mathfrak M_0$, of the vector space of seed data. Given such an $\mathcal S_0=(\mathring\phi,\Lambda,M_0,e)\in\mathfrak M_0$ and $\alpha\in\Bbb R$, we define the \emph{modulated} seed data set
  \begin{equation*}
      \mathcal S_0(\alpha)\doteq (\mathring\phi,\Lambda,M_0+\alpha,e).
  \end{equation*}
For $\ve>0$, we then define
\begin{equation*}
    \mathcal L(\mathcal S_{0},\ve)\doteq \{\mathcal S_{0}(\alpha):\alpha\in [-2\ve,2\ve]\},
\end{equation*}
which is a line segment in the vector space of seed data. 

\begin{defn}
Let $M_0>0$. The \emph{moduli space of seed data} centered on mass $M_0$ is the set
\begin{equation}\label{eq:M-defn}
    \mathfrak M\doteq \bigcup_{\mathcal S_{0}\in\mathfrak M_0:\mathfrak D[\mathcal S_0]\le \ve_\loc}\mathcal L(\mathcal S_0,\ve_\loc),
\end{equation} where $\ve_\loc$ is the small parameter from \cref{lem:epsilon-loc}. For $0<\ve\le\ve_\loc$, we define the moduli space with \emph{smallness parameter $\ve$} by
\begin{equation}\label{eq:M-eps-defn}
    \mathfrak M(\ve)\doteq \bigcup_{\mathcal S_{0}\in\mathfrak M_0:\mathfrak D[\mathcal S_0]\le \ve}\mathcal L(\mathcal S_0,\ve).
\end{equation}
We endow $\mathfrak M$ and $\mathfrak M(\ve)$ with the metric space topology associated to the norm $\mathfrak D$.
\end{defn}

\subsection{The geometric setup for the statement of stability}\label{sec:setup}

In this section, we explain several geometric constructions that will be required to state the precise version of our main theorem below and will be crucial to the proof. In order to explain these constructions, we will be required to make several assumptions, which we will later formalize as bootstrap assumptions in \cref{sec:bootstrap-definitions} below. Refer to \cref{fig:bootstrap-setup} for a diagram explaining our gauge choice and notation explained in this section and the following.

Let $\mathcal S\in \mathfrak M(\ve)$ with $0<\ve\le \ve_\loc$ 
and denote the maximal development of $\mathcal S$ by $(\hat{\mathcal Q}_\mathrm{max}, r,\Omega^2,\phi,e)$. Since we are aiming to converge to an extremal Reissner--Nordstr\"om black hole with charge $e$ (note that this parameter is conserved in the neutral scalar field model), we define the target mass parameter
\begin{equation*}
    M\doteq |e|. 
\end{equation*} We define the set
\begin{equation*}
    \Gamma\doteq\{r=\Lambda \}.
\end{equation*}
This is clearly a timelike curve near $\hat{\mathcal C}$  for $\ve$ sufficiently small and we will verify in the course of the proof of the main theorem that $\Gamma$ is an intextendible timelike curve in $\hat{\mathcal Q}_\mathrm{max}$. Assuming for the moment that this is the case, we may parametrize $\Gamma$ by its proper time $\tau$, which we normalize to start at $1$ at $\Gamma\cap\hat{\mathcal C}$. We write the components of $\Gamma$ as $\Gamma(\tau)=(\Gamma^{\hat u}(\tau),\Gamma^{\hat v}(\tau))$ in the ``initial data normalized'' $(\hat u,\hat v)$ coordinates on $\hat{\mathcal Q}_\mathrm{max}$.

For $\tau_f\in[1,\infty)$ such that $[1,\tau_f]$ lies in the domain of definition $\Gamma$, we define
\begin{equation*}
    \hat{\mathcal D}_{\tau_f}\doteq [0,\Gamma^{\hat u}(\tau_f)]\times[0,\Gamma^{\hat v}(\tau_f)].
\end{equation*} Assuming that $\hat\gamma<0$ on the final ingoing cone in $\hat{\mathcal D}_{\tau_f}$ and $\hat\kappa>0$ on $\Gamma\cap \hat{\mathcal D}_{\tau_f}$, we may define strictly increasing functions $
    \mathfrak u_{\tau_f}:[0,\Gamma^{\hat u}(\tau_f)]\to\Bbb R$ and $\mathfrak v:[0,\Gamma^{\hat v}(\tau_f)]\to \Bbb R$ by 
\begin{align}
 \label{eq:teleology-1}      \mathfrak u_{\tau_f}(\hat u)& \doteq  -\int_0^{\hat u} \hat\gamma(\hat u',\Gamma^{\hat v}(\tau_f)) \,d\hat u',\\
   \label{eq:teleology-2}          \mathfrak v(\hat v)&\doteq \int_0^{\hat v} \hat\kappa(\Gamma^{\hat u}((\Gamma^{\hat v})^{-1}(\hat v')),\hat v')\,d\hat v',
\end{align}
which then assemble into a map
\begin{align}
 \label{defn:Phi}   \Phi_{\tau_f}:\hat{\mathcal D}_{\tau_f}&\to \Bbb R^2\\
\nonumber(\hat u,\hat v) &\mapsto (\mathfrak u_{\tau_f}(\hat u),\mathfrak v(\hat v)),
\end{align}
which is a diffeomorphism onto its image. We denote the image of $\Phi_{\tau_f}$ by $\mathcal D_{\tau_f}$, which comes equipped with the double null coordinates $(u_{\tau_f},v)=\Phi_{\tau_f}(\hat u,\hat v)$. We call these coordinates the \emph{teleologically normalized coordinates}.  Let $\hat\Phi_{\tau_f}$ denote the inverse of $\Phi_{\tau_f}$. In the $(u_{\tau_f},v)$ coordinate system, the solution $(r,\hat\Omega^2,\phi,e)$ is given by $(r_{\tau_f},\Omega^2_{\tau_f},\phi_{\tau_f},e)$, where $r_{\tau_f}\doteq r\circ \hat\Phi_{\tau_f}$, $\phi_{\tau_f}\doteq \phi\circ\hat\Phi_{\tau_f}$, and 
\begin{equation*}
    \Omega^2_{\tau_f}\doteq \frac{1}{\mathfrak u_{\tau_f}'\circ\mathfrak u_{\tau_f}^{-1}}\frac{1}{\mathfrak v'\circ\mathfrak v^{-1}}\hat\Omega^2\circ\hat\Phi_{\tau_f}.
\end{equation*} In the $(u_{\tau_f},v)$ coordinate system, we write the coordinates of $\Gamma$ as $\Gamma(\tau)=(\Gamma^{u_{\tau_f}}(\tau),\Gamma^v(\tau))$.
We will frequently omit the subscript $\tau_f$ on $u_{\tau_f}$ and $(r_{\tau_f},\Omega^2_{\tau_f},\phi_{\tau_f},e)$ when it is clear that $\tau_f$ has been fixed. 

 \begin{figure}
\centering{
\def\svgwidth{13pc}
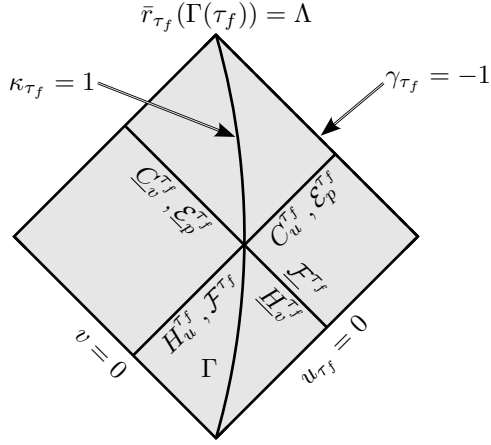}
\caption{A Penrose diagram showing the gauge conditions, null hypersurfaces, and energies in our bootstrap domain $\mathcal D_{\tau_f}$. The function $\tau$ measures advanced time to the left of $\Gamma$ and retarded time to the right of $\Gamma$.}
\label{fig:bootstrap-setup}
\end{figure}

By a slight abuse of notation, we define a continuous function $\tau$ on $\mathcal D_{\tau_f}$ implicitly by
\begin{equation}
 \label{def:tau}   \tau(u_{\tau_f},v)\doteq \begin{cases}
     \tau:\Gamma^{u_{\tau_f}}(\tau)=u_{\tau_f}   & \text{if } r(u_{\tau_f},v)\ge\Lambda  \\
       \tau: \Gamma^v(\tau)=v & \text{if } r(u_{\tau_f},v)< \Lambda 
    \end{cases}. 
\end{equation} The function $\tau$ measures (approximately) Bondi time in the region near null infinity and (approximately) Eddington--Finkelstein time near the event horizon. This will be made precise in \cref{lem:tau-properties} below.

We define four classes of null hypersurfaces in $\mathcal D_{\tau_f}$:
\begin{align*}
   C_u^{\tau_f}&\doteq (\{u\}\times[0,\Gamma^v(\tau)])\cap \{r\ge \Lambda \} ,&\underline C{}_v^{\tau_f}&\doteq([0,\Gamma^{u_{\tau_f}}(\tau_f)]\times\{v\})\cap\{r\le \Lambda \},\\
   H_u^{\tau_f}&\doteq (\{u\}\times[0,\Gamma^v(\tau)])\cap \{r\le \Lambda \},&\underline H{}_v^{\tau_f}&\doteq([0,\Gamma^{u_{\tau_f}}(\tau_f)]\times\{v\})\cap\{r\ge \Lambda \}.
\end{align*}
We shall often suppress the dependence of these hypersurfaces on $\tau_f$. It is also convenient to write
\begin{equation*}
    C^{\tau_f}(\tau)\doteq C_{\Gamma^{u_{\tau_f}}(\tau)}^{\tau_f}\quad\text{and}\quad \underline C^{\tau_f}(\tau)=\underline C{}_{\Gamma^v(\tau)}^{\tau_f}.
\end{equation*}

For the seed data sets we will ultimately consider, $\Gamma$ exists and remains timelike for all $\tau\in[1,\infty)$. We then define a number
\begin{equation}\label{eq:horizon-location}
    \hat u_{\mathcal H^+}\doteq \lim_{\tau\to\infty} \Gamma^{\hat u}(\tau).
\end{equation}
Since $\tau\mapsto \Gamma^{\hat u}(\tau)$ is monotone increasing, the existence of this limit is automatic and we will show the strict inequality 
\begin{equation*}
    \hat u_{\mathcal H^+}< U_*,
\end{equation*} see already \cref{lem:u-hat-H+}. We then set
\begin{equation*}
    \hat{\mathcal D}_\infty\doteq [0,\hat u_{\mathcal H^+})\times[0,\infty).
\end{equation*}
We will also show that there exists a surjective, strictly increasing $C^2$ function $u_\infty:[0,\hat u_{\mathcal H^+})\to [0,\infty)$, such that, if we use $(u_\infty,v)$ as coordinates on $\hat{\mathcal D}_\infty$, then $\partial_{u_\infty}r\to -1$ at null infinity $\mathcal I^+$. We denote $\hat{\mathcal D}_\infty$ by $\mathcal D_\infty$ under this change of coordinates. We refer to this ``final'' set of teleological double null coordinates as the \emph{eschatological gauge}. 

Finally, our geometric estimates will imply that for any $R\in(M,\infty)$, the set $\Gamma_R\doteq\{r=R\}$ is a timelike curve in $\mathcal D_{\tau_f}$ for any $\tau_f\in[1,\infty]$ (or possibly empty). For any given $v$, the intersection $(\Bbb R\times\{v\})\cap \Gamma_R$ is either empty or consists of a single point, which we then denote by $(u^R(v),v)$. Likewise, for any given $u$, the intersection $(\{u\}\times\Bbb R)\cap \Gamma_R$ is either empty or consists of a single point, which we then denote by $(u,v^R(u))$.

\subsection{Anchored extremal Reissner--Nordstr\"om solutions and definitions of the energies}\label{sec:anchor}

Let $(r,\Omega^2,\phi,e)$ be a spherically symmetric solution of the Einstein--Maxwell-scalar field system defined on a coordinate rectangle $\mathcal D_{\tau_f}$ with gauge conditions as explained in \cref{sec:setup}. 

First, assume $\tau_f<\infty$. We define the \emph{$\tau_f$-anchored background solution} to be the extremal Reissner--Nordstr\"om metric $(\bar r_{\tau_f},\bar\Omega^2_{\tau_f})$ with parameters $M=|e|$ in Eddington--Finkelstein double null form \eqref{eq:ERN-EF} which is uniquely determined by \cref{lem:existence-anchoring} according to  \begin{equation*}
        \bar r_{\tau_f}(\Gamma(\tau_f)) = \Lambda.
    \end{equation*} 
Given the anchored background solution we now adopt the following notation:
\begin{itemize}
    \item Barred quantities such as $\bar \lambda_{\tau_f},\bar\varpi_{\tau_f}=M$, or $\bar\kappa_{\tau_f}=1$ correspond to those of $(\bar r_{\tau_f},\bar\Omega^2_{\tau_f})$.
    \item Differences are denoted with a tilde, such as $\tilde r_{\tau_f}\doteq r_{\tau_f}-\bar r_{\tau_f}$, $\tilde\varpi_{\tau_f}\doteq \varpi_{\tau_f}-M$, or $\tilde\gamma_{\tau_f} \doteq \gamma_{\tau_f}+1$.
\end{itemize}

In the proof of \cref{thm:stability-rough}, we will send $\tau_f\to \infty$ and thus need to extend this definition to the case when $\tau_f=\infty$. Instead of trying to anchor directly ``at $\tau_f=\infty$,'' it is much easier to simply anchor the background solution at $\Gamma\cap\hat{\mathcal C}=\{(0,0)\}$. Therefore, we define the \emph{$\infty$-anchored background solution} to be the unique extremal Reissner--Nordstr\"om metric $(\bar r_{\infty},\bar\Omega^2_{\infty})$ with parameters $M=|e|$ in Eddington--Finkelstein double null form \eqref{eq:ERN-EF} with the property that  
\begin{equation*}
    \bar r_\infty(0,0)=\bar r_\star\doteq\lim_{\tau_f\to\infty}\bar r_{\tau_f}(0,0),
\end{equation*} see already \cref{fig:stability}.
We will show below in \cref{prop:r-convergence} that this limit actually exists. We will also show that this background solution is indeed anchored ``at $\tau_f=\infty$'' in the sense that
\begin{equation*}
    \lim_{\tau\to\infty}\bar r_\infty(\Gamma(\tau))=\Lambda,
\end{equation*}
see already \eqref{eq:r-infty-est}.

We may now define the fundamental weighted energy norms for the scalar field. Some of the norms will depend explicitly on the background solution $\bar r_{\tau_f}$ in a nontrivial manner. 

\begin{defn}\label{def-of-energies}
    Let $(r_{\tau_f},\Omega^2_{\tau_f},\phi_{\tau_f},e)$ be defined on $\mathcal D_{\tau_f}$ with teleologically normalized coordinates $(u_{\tau_f},v)$, where $\tau_f\in[1,\infty]$. Let $\bar r_{\tau_f}$ be the associated $\tau_f$-anchored background solution. Let 
    \begin{equation*}
        \psi_{\tau_f}\doteq r_{\tau_f}\phi_{\tau_f}
    \end{equation*}
    denote the \emph{radiation field} of $\phi_{\tau_f}$. For $\tau,\tau'\in[1,\tau_f]$, $p\in[0,3)$, and $(u,v)\in\mathcal D_{\tau_f}$, we define:
    \begin{enumerate}
        \item  The $(\bar r-M)^{2-p}$-weighted flux to the horizon:
\begin{equation*}
      \underline{\mathcal E}{}_p^{\tau_f}(\tau)\doteq\int_{\underline C^{\tau_f}(\tau)} \left((\bar r_{\tau_f}-M)^{2-p}\left[ \frac{(\partial_{u_{\tau_f}}\psi_{\tau_f})^2}{-\bar \nu_{\tau_f}}+ \frac{(\partial_{u_{\tau_f}}\phi_{\tau_f})^2}{-\bar \nu_{\tau_f}}\right]-c_{p_\star}(\bar r_{\tau_f}-M)^{-p_\star}\bar\nu_{\tau_f}\phi^2_{\tau_f} \right)du_{\tau_f},
\end{equation*}
where $p_\star= p$ if $p\in[0,1)$ and vanishes otherwise, and $c_{p_\star}\doteq(p_\star-1)^2$.

        \item  The $r^p$-weighted flux to null infinity: 
        \begin{equation}
             \mathcal E_p^{\tau_f}(\tau)\doteq \int_{C^{\tau_f}(\tau)}\big(r^p_{\tau_f}(\partial_v\psi_{\tau_f})^2+c_{p_\star}r^{p_\star+2}_{\tau_f}(\partial_v\phi_{\tau_f})^2+c_{p_\star}r^{p_\star}_{\tau_f}\phi^2_{\tau_f}\big)\,dv.
        \end{equation}
        \item The energy flux along outgoing cones in the near region:
        \begin{equation*}
           \mathcal F^{\tau_f}(u,\tau')\doteq \int_{H_u^{\tau_f}\cap\{\tau\ge\tau'\}}\big((\partial_v\phi_{\tau_f})^2+\bar\lambda_{\tau_f}\phi_{\tau_f}^2\big)\,dv.
        \end{equation*}
        \item The energy flux along ingoing cones in the far region:
        \begin{equation*}
            \underline{\mathcal F}^{\tau_f}(v,\tau')\doteq \int_{\underline H{}_v^{\tau_f}\cap\{\tau\ge\tau'\} }\big(r^2_{\tau_f}(\partial_{u_{\tau_f}}\phi_{\tau_f})^2+\phi_{\tau_f}^2\big)\,du_{\tau_f}.
        \end{equation*}
    \end{enumerate}
\end{defn}

\begin{rk}
   The reason for the dependence on $p_\star$ in $\underline{\mathcal E}{}_p^{\tau_f}$ and $\mathcal E_p^{\tau_f}$ is that when $p\ge 1$, $(\bar r_{\tau_f}-M)^{-p}(-\bar\nu_{\tau_f})\phi^2_{\tau_f}$ and $r^{p+2}_{\tau_f}(\partial_v\phi_{\tau_f})^2+r^p_{\tau_f}\phi^2_{\tau_f}$ are not integrable in $u_{\tau_f}$ and $v$ (uniformly in $\tau_f$), respectively, and hence cannot be controlled in the energy estimates. The reason for including the factor $c_{p_\star}$ is that the Hardy inequalities relating $(\bar r_{\tau_f}-M)^{-p}(-\bar\nu_{\tau_f})\phi^2_{\tau_f}$ to $(\bar r_{\tau_f}-M)^{2-p}(-\bar\nu_{\tau_f})^{-1}(\partial_{u_{\tau_f}}\psi_{\tau_f})^2$ and $r^{p+2}_{\tau_f}(\partial_v\phi_{\tau_f})^2+r^p_{\tau_f}\phi^2_{\tau_f}$ to $r^p_{\tau_f}(\partial_u\psi_{\tau_f})^2$ (which only hold for $p<1$, see already \cref{lem:horizon-Hardy-1,lem:null-infinity-Hardy-1}), degenerate as $p\nearrow 1$. Therefore, in order to have a $p$-independent estimate such as \eqref{eq:nrg-decay-main} below, we need to build this degeneration into the definition of the energy.
\end{rk}

\subsection{Detailed statements of the main theorems}\label{sec:statements}

We can now state our main theorems using the notation and definitions from \cref{sec:seed-data,sec:setup,sec:anchor}.

\subsubsection{Nonlinear stability}

\begin{thm}[Stability of extremal Reissner--Nordstr\"om in spherical symmetry]\label{thm:stability}

Let $M_0>0$, $e_0\in \mathbb R$ with $|e_0|=M_0$, and let $\delta$ be an arbitrary parameter satisfying \begin{equation}\label{defn:delta}
    0<\delta < \frac{1}{100}.
\end{equation} There exists a number $\ve_\stab(M_0,\delta)>0$, a set $\mathfrak M_\mathrm{stab}\subset \mathfrak M$, and a constant $C(M_0,\delta)$ (which is implicit in the notation $\les$ below) with the following properties:

\begin{enumerate}
    \item \ul{$\mathfrak M_{\stab}$ is ``codimension-one'' inside of $\mathfrak M(\ve)$}: For every $0<\ve\le \ve_\stab$ and $\mathcal S_0\in \mathfrak M_0$ with $\mathfrak D[\mathcal S_0]\le\ve$, it holds that
\begin{equation}
    \mathfrak M_\stab\cap \mathcal L(\mathcal S_0,\ve)\ne\emptyset.\label{eq:intersection}
\end{equation}
\item 
\underline{Existence of a black hole region}: Let $(\hat{\mathcal Q}_\mathrm{max},r,\Omega^2,\phi,e)$ be the maximal development of a seed data set in the intersection \eqref{eq:intersection}. Then $\hat{\mathcal Q}_\mathrm{max}=[0,U_*]\times[0,\infty)$. There exists a $\hat u_{\mathcal H^+}\in (0,U_*)$ such that $r(\hat u,\hat v)\to \infty$ as $\hat v\to \infty$  for every $\hat u\in[0,\hat u_{\mathcal H^+})$ and $r(\hat u_{\mathcal H^+},\hat v)\to |e|$ as $\hat v\to\infty$. Therefore, $[0,\hat u_{\mathcal H^+})\times\{\hat v=\infty\}$ may be regarded as future null infinity $\mathcal I^+$, there exists a nonempty black hole region
\begin{equation*}
    \mathcal{BH}\doteq \hat{\mathcal Q}_\mathrm{max}\setminus J^-(\mathcal I^+)=[\hat u_{\mathcal H^+},U_*]\times[0,\infty),
\end{equation*} and 
\begin{equation*}
    \mathcal H^+\doteq \partial J^-(\mathcal I^+)= \{\hat u_{\mathcal H^+}\}\times[0,\infty)
\end{equation*}
is the event horizon. Moreover, future null infinity is complete in the sense of Christodoulou \cite{christodoulou1999global}. There exist $C^2$ double null coordinates $(u_\infty,v)$ on the domain of outer communication $[0,\hat u_{\mathcal H^+})\times[0,\infty)$ such that $u_\infty$ is Bondi normalized, i.e., the event horizon $\mathcal H^+$ can be formally regarded as $\{u_\infty=\infty\}$ and $\partial_{u_\infty}r\to -1$ along any outgoing cone in the domain of outer communication.

 \begin{figure}
\centering{
\def\svgwidth{13pc}
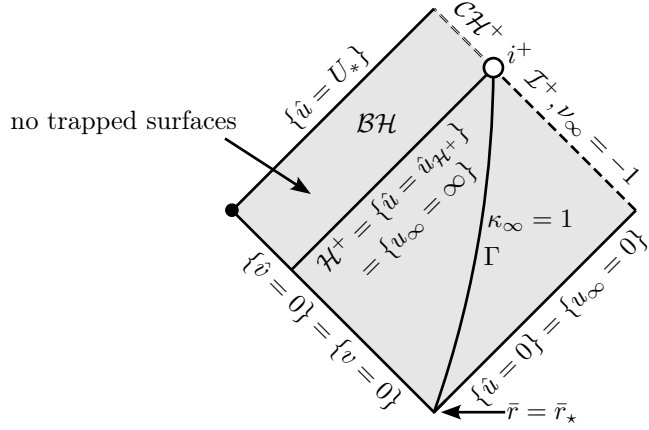}
\caption{A Penrose diagram depicting the maximal development of seed data lying in $\mathfrak M_\mathrm{stab}$.}
\label{fig:stability}
\end{figure}

\item \underline{Orbital stability}: There exists an $\infty$-anchored extremal Reissner--Nordstr\"om solution $\bar r_\infty$ in the $(u_\infty,v)$ coordinates whose parameters satisfy
\begin{equation}\label{eq:parameter-stability}
    |M-M_0|+|e-e_0|\les \ve.
\end{equation}
Relative to this background solution, the $p=3-\delta$ energy of the scalar field is bounded by its initial value,
\begin{equation}\label{eq:energy-orbital}
   \sup_{\tau\in[1,\infty)} \big(\mathcal E_{3-\delta}^\infty(\tau)+\underline{\mathcal E}{}_{3-\delta}^\infty(\tau)\big)\les \mathcal E_{3-\delta}^\infty(1) +\underline{\mathcal E}{}_{3-\delta}^\infty(1),
\end{equation}
and the scalar field is pointwise bounded in $C^1$ in the domain of outer communication in terms of its initial values,
\begin{multline}\label{eq:pointwise-orbital}
\sup_{J^-(\mathcal I^+)}\left(|r\phi|+  |r^2\partial_{\hat v}\psi|+
   |r^2\partial_{\hat v}\phi|+
   \left|\frac{\partial_{\hat u}\psi}{-\hat\nu}\right|+
  \left|r\frac{\partial_{\hat u}\phi}{-\hat\nu}\right|\right)\\\les \sup_{J^-(\mathcal I^+)\cap\hat{\mathcal C}}\left(|r\phi|+  |r^2\partial_{\hat v}\psi|+
   |r^2\partial_{\hat v}\phi|+
   \left|\frac{\partial_{\hat u}\psi}{-\hat\nu}\right|+
  \left|r\frac{\partial_{\hat u}\phi}{-\hat\nu}\right|\right).
\end{multline}
The right-hand side of \eqref{eq:energy-orbital} is $\les\ve^2$ and the right-hand side of $\eqref{eq:pointwise-orbital}$ is $\les \ve$.

\item \underline{Asymptotic stability}: The geometry decays towards the $\infty$-anchored extremal Reissner--Nordstr\"om solution in the following sense: 
\begin{equation}
    |\gamma_\infty+1|\les \ve^2r^{-1}\tau^{-3+\delta} \label{eq:geo-main-1}
\end{equation}
holds on $\mathcal D_{\infty} \cap\{r\ge\Lambda\}$,
\begin{equation}
    \left|\frac{\nu_\infty}{\bar\nu_\infty}-1\right|\les \ve^{2}\tau^{-1+\delta}
\end{equation}
holds on $\mathcal D_\infty$, and 
\begin{equation}\label{eq:geo-main-3}
     |r-\bar r_\infty|\les \ve^{2}\tau^{-2+\delta},\quad    |\lambda_\infty-\bar\lambda_\infty|\les \ve^{2}\tau^{-2+\delta},\quad  |\kappa_\infty-1|\les \ve^2\tau^{-1+\delta},\quad   |\varpi -M|\les \ve^2\tau^{-3+\delta}
\end{equation}
hold on $\mathcal D_\infty$, up to and including the event horizon $\mathcal H^+$. The scalar field decays to zero in energy norm,
\begin{equation}\label{eq:nrg-decay-main}
      \underline{\mathcal E}{}_p^\infty(\tau)+ \mathcal E_p^\infty(\tau)\les \ve^2\tau^{-3+\delta+p}
\end{equation}
for every $\tau\in[1,\infty)$ and $p\in[0,3-\delta]$ and its amplitude decays to zero pointwise,
\begin{equation}\label{eq:psi-decay-final}
    |(\bar r_\infty-M)^{1/2}\phi|\les \ve \tau^{-3/2+\delta/2},\quad |\psi|\les\ve \tau^{-1+\delta/2}
\end{equation}
on $\mathcal D_\infty$, up to and including the event horizon $\mathcal H^+$.

\item \ul{Absence of trapped surfaces}: On the spacetime $(\hat{\mathcal Q}_\mathrm{max},r,\hat\Omega^2)$, it holds that $\hat\nu<0$, i.e., there exist no antitrapped spheres of symmetry. Moreover: 
 \begin{enumerate} 
    \item $\hat\lambda\ge 0$ on $\hat{\mathcal Q}_\mathrm{max}$, i.e., there exist no (strictly) trapped symmetry spheres.  
    \item If $\hat\lambda(\hat u_0,\hat v_0)=0$ for some $(\hat u_0,\hat v_0)\in \hat{\mathcal Q}_\mathrm{max}$, then $\hat u_0=\hat u_{\mathcal H^+}$ and $\phi(\hat u_{\mathcal H^+},\hat v)=0$ for all $\hat v\ge \hat v_0$.
\end{enumerate}
\end{enumerate}
\end{thm}

\begin{rk}
    By \cite[Proposition B.1]{KU22}, $\hat{\mathcal Q}_\mathrm{max}\times S^2$ with the $(3+1)$-dimensional metric \eqref{eq:dn} does not contain \emph{any} trapped surfaces (i.e., without restriction to spherical symmetry).
\end{rk}

The proof of \cref{thm:stability} is given in \cref{sec:proof-main} and we will now briefly indicate to the reader where the various parts are shown. The set $\mathfrak M_\stab$ is defined in \cref{sec:M-defn}, where the codimension-one property \eqref{eq:intersection} is also proved. Part 2.~is proved in \cref{prop:existence-of-horizon}, while the eschatological gauge $(u_\infty,v)$ is constructed in \cref{sec:final-gauge}. The culmination of the proofs of Points 3.--5.~is given in \cref{sec:putting-together}. Points 3.~and 4.~rely in particular on estimates proved in \cref{sec:geometry,sec:energy,sec:decay-energy}. 

\subsubsection{The Aretakis instability}

 Let \begin{equation*}
     Y\doteq \hat\nu^{-1}\partial_{\hat u}
 \end{equation*}
 denote the gauge-invariant null derivative which is transverse to the event horizon $\mathcal H^+$, analogous to $\partial_r$ in ingoing Eddington--Finkelstein coordinates $(v,r)$ in Reissner--Nordstr\"om (recall \cref{sec:geometry-RN}). 

\begin{thm}[The Aretakis instability for dynamical extremal horizons] \label{thm:instability} Let $\mathfrak M_\stab$ denote the subset of the moduli space $\mathfrak M$ given by \cref{thm:stability} consisting of seed data asymptotically converging to extremal Reissner--Nordstr\"om in evolution. Then the following holds: 
    \begin{enumerate}[i)]
        \item For any solution $(\hat{\mathcal Q}_\mathrm{max},r,\hat\Omega^2,\phi,e)$ arising from $\mathcal S\in\mathfrak M_\stab$, the ``asymptotic Aretakis charge'' 
        \begin{equation*}
            H_0[\phi]\doteq\lim_{\hat v\to\infty} Y\psi|_{\mathcal H^+}
        \end{equation*}
        exists and it holds that 
\begin{align}
      |Y\psi|_{\mathcal H^+}(\hat v)- H_0[\phi]|&\les\ve^{3}(1+\hat v)^{-1+\delta},\label{eq:Aretakis-main}\\
      \big|R_{YY}|_{\mathcal H^+}(\hat v)-2M^{-2}\big(H_0[\phi]\big)^2\big|&\les \ve^2 (1+\hat v)^{-1+\delta/2},\label{eq:main-Ricci-1}
\end{align}
        where $\ve\ge \mathfrak D[\mathcal S]$. 
        
        \item The set
        \begin{equation*}
            \mathfrak M_\mathrm{stab}^{\ne 0}\doteq \{\mathcal S\in \mathfrak M_\mathrm{stab}:H_0[\phi]\ne 0\}
        \end{equation*}
        has nonempty interior as a subset of $\mathfrak M_\stab$. 
        
        \item For any solution arising from data lying in $\mathfrak M_\stab^{\ne 0}$, which satisfies the additional smallness assumption
\begin{equation}\label{eq:extra-ass}
    \sup_{\underline C{}_\ing}|Y^2\psi|\le C_2\ve,
\end{equation} it holds that 
        \begin{align}
       \big|Y^2(r\phi)|_{\mathcal H^+}(\hat v)\big|&\gtrsim |H_0[\phi]|\hat v \label{eq:aretakis-main-2},\\
              \big|\nabla_YR_{YY}|_{\mathcal H^+}(\hat v)\big|&\gtrsim \big(H_0[\phi]\big)^2\hat v\label{eq:main-Ricci-2}
        \end{align}
        for $\hat v\gtrsim 1+ |\ve H_0[\phi]^{-1}|^{1/(1-\delta)}$, where the implicit constants in these estimates also depend on $C_2$.
    \end{enumerate}
\end{thm}

The proof of this theorem is given in \cref{sec:aretakis-proof}.

\begin{rk} In this work, we do not prove pointwise decay of $\partial_v\phi$ as it is not required to close our energy estimates. Assuming the decay rate $\partial_v\phi|_{\mathcal H^+}=O(v^{-2})$ as in the uncoupled case, we would have the following:
    Relative to the null frame $\{e_1,e_2,e_3,e_4\}$ with $e_1=\partial_{\vartheta}$, $e_2=\partial_{\varphi}$, $e_3=Y$, and $e_4=\hat\kappa^{-1}\partial_{\hat v}$ (so that $e_3$ and $e_4$ are gauge invariant and $g(e_3,e_4)=-4$), the only components of the Riemann tensor that do not necessarily converge to those of extremal Reissner--Nordstr\"om along $\mathcal H^+$ are $R_{1313}$, $R_{2323}$, and appropriate permutations of indices (i.e., $\underline\alpha$ computed relative to the Riemann tensor). The only components of the covariant derivative of the Riemann tensor that are allowed to not decay or grow are $\nabla_{3}R_{1313}$, $\nabla_3R_{2323}$, and appropriate permutations of indices. At higher orders of differentiation, other components should also display non-decay and growth properties.
\end{rk}

\begin{rk}
    Using the Einstein equations, \eqref{eq:Aretakis-main}, and \eqref{eq:aretakis-main-2}, one can also prove
    \begin{align}\label{eq:kappa-aretakis-1}
        \big|Y\log\kappa|_{\mathcal H^+}(v)-M^{-2}\big(H_0[\phi]\big)^2\big|&\les \ve^2(1+ v)^{-1+\delta/2}\\
       \big|Y^2\log\kappa|_{\mathcal H^+}(v)\big| &\gtrsim \left(H_0[\phi]\right)^2 v,\label{eq:kappa-aretakis-2}
    \end{align} where we have interpreted $\kappa$ as being written in the ``mixed'' coordinate system $(\hat u,v)$ which still extends smoothly to $\mathcal H^+$. In this coordinate system, we have that $|{\log\kappa}|\les \ve^2(1+v)^{-1+\delta}$ up to $\mathcal H^+$ by \eqref{eq:geo-main-3}, so \eqref{eq:kappa-aretakis-1} and \eqref{eq:kappa-aretakis-2} represent yet another aspect of the extremal horizon instability.  
\end{rk}

\section{Setup for the proof of stability}\label{sec:bootstrap}

In this section, we set up the proof of \cref{thm:stability} and discuss the logic of the proof. In \cref{sec:bootstrap-definitions}, we state the bootstrap assumptions for the proof of \cref{thm:stability}, which also involves defining a sequence of sets used in our codimension-one modulation argument. In \cref{sec:B-nonempty}, we show that the bootstrap set is nonempty by virtue of the local existence theory. Finally, in \cref{sec:B-improving}, we state two propositions that encode the main analytic content of the proof of \cref{thm:stability}---improvability of the bootstrap assumptions.

\subsection{The bootstrap and modulation parameter sets}\label{sec:bootstrap-definitions}

We define in this section two sets of parameters: a bootstrap set $\mathfrak B$ containing the $\tau_f$'s for which we assume the solution exists on $\mathcal D_{\tau_f}$ and satisfies certain properties, and a sequence of compact intervals
\begin{equation*}
   \mathfrak A_0\supset \mathfrak A_1\supset \mathfrak A_2\supset\cdots
\end{equation*}
of $\alpha$ parameters which are used in the modulation argument to hit extremality. 

We first set 
\begin{equation}
    \label{modulation-definition-1}  \mathfrak A_0\doteq[|e|-|e_0|-\ve^{3/2},|e|-|e_0|+\ve^{3/2}],
\end{equation}
but the sets $\mathfrak A_i$ for $i\ge 1$ will only be properly defined in the course of the proof of \cref{thm:stability}; see already \cref{sec:modulation}.\footnote{This is because the sets $\mathfrak A_i$ for $i\ge 1$ are defined teleologically, i.e., they cannot be directly read off from the initial data. Their existence can only be inferred in the context of our continuity argument, in particular, the argument requires proving quantitative decay rates for $\tilde\varpi=\varpi-M$.} We now briefly indicate their construction to make the purpose of the bootstrap assumption~1.~below clear. We define continuous functions
\begin{align}
 \nonumber   \Pi_i:\mathfrak A_{i}&\to\Bbb R\\
  \label{eq:Pi-defn}  \alpha&\mapsto \varpi(\Gamma(L_i))-M=\tilde\varpi(\Gamma(L_i)),
\end{align}
where $L_i\doteq 2^i$, which will be assumed to satisfy the fundamental estimate
\begin{equation}
    \sup_{\mathfrak A_i}|\Pi_i|\le \ve^{3/2}L_i^{-3+\delta}.\label{modulation-definition-2}
\end{equation}
Of course, $\Pi_i$ is only defined if the solution corresponding to $\mathcal S_0(\alpha)$ for $\alpha\in\mathfrak A_i$ exists until the time $\tau=L_i$. Assuming that this is the case, we will then inductively define 
\begin{equation}
    \mathfrak A_{i+1}=[\alpha_{i+1}^-,\alpha_{i+1}^+],\label{modulation-definition-3}
\end{equation}
where $\alpha^\pm_{i+1}\in \Int(\mathfrak A_i)$ are chosen according to a simple algorithm ensuring that the improved estimate
\begin{equation*}
    \sup_{\mathfrak A_{i+1}}|\Pi_{i+1}|\le \ve^{3/2}L^{-3+\delta}_{i+1}
\end{equation*}
 holds; see already \cref{lem:inductive-mod}.

It is convenient to define the function $I(\tau_f)\doteq  \lfloor \log_2\tau_f \rfloor$, i.e., the largest integer such that $2^{I(\tau_f)}\le\tau_f$.

\begin{defn}\label{def:B} Let $A\ge 1$, $0<\ve\le \ve_\loc$, and $\mathcal S_0\in\mathfrak M_0$ with $\mathfrak D[\mathcal S_0]\le \ve$.   Then $\mathfrak B(\mathcal S_0,\ve,A)$ denotes the set of $\tau_f\in [1,\infty)$ such that:
   \begin{enumerate}
\item  For every $i\in\{0,1,\dotsc,I(\tau_f)\}$, there exist numbers $\alpha^\pm_i\in[|e|-|e_0|-\ve^{3/2},|e|-|e_0|+\ve^{3/2}]$, which may depend on $\mathcal S_0$ and $\ve$, but are independent of $A$ and $\tau_f$, with $\alpha_i^-<\alpha_i^+$ and $\alpha_0^\pm=|e|-|e_0|\pm\ve^{3/2}$, such that the nesting condition $\mathfrak A_{i+1}\subset \mathfrak A_i$
holds, where $\mathfrak A_i\doteq [\alpha_i^-,\alpha_i^+]$.
\end{enumerate}

Given $\alpha\in\mathfrak A_{I(\tau_f)}$, let $(\hat{\mathcal Q}_\mathrm{max},r,\hat\Omega^2,\phi,e)$ denote the maximal development of the modulated seed data $\mathcal S_0(\alpha)$ in the initial data gauge $(\hat u,\hat v)$ determined by \cref{prop:seed-data-generation} and \cref{lem:epsilon-loc}.

\begin{enumerate}
\item[2.]  For every $\alpha\in \mathfrak A_{I(\tau_f)}$, there exists a timelike curve $\Gamma:[1,\tau_f]\to \hat{\mathcal Q}_\mathrm{max}$, which is the unique smooth solution of the ODE
\begin{equation*}
\frac{d}{d\tau}(\Gamma^{\hat u},\Gamma^{\hat v})=\left.\left(\frac{\sqrt{1-\mu}}{-2\hat\nu},\frac{\sqrt{1-\mu}}{2\hat\lambda}\right)\right|_{\Gamma(\tau)},
\end{equation*}
with initial condition $\Gamma(1)=(0,0)$. 
\end{enumerate}

By global hyperbolicity of $\hat{\mathcal Q}_\mathrm{max}$, point 2.~implies that for every $\alpha\in \mathfrak A_{I(\tau_f)}$, \[\hat{\mathcal D}_{\tau_f}\doteq [0,\Gamma^{\hat u}(\tau_f)]\times[0,\Gamma^{\hat v}(\tau_f)]\subset\hat{\mathcal Q}_\mathrm{max}.\]

\begin{enumerate}
    \item[3.] For every $i\in\{0,1,\dotsc,I(\tau_f)\}$, the map $\Pi_i$ (recall \eqref{eq:Pi-defn}) is defined on $\mathfrak A_i$,
\begin{equation*}
    \Pi_i:\mathfrak A_i\to [-\ve^{3/2}L_i^{-3+\delta},\ve^{3/2}L_i^{-3+\delta}]
\end{equation*}
is surjective, and $\Pi_i(\alpha_i^\pm)=\pm \ve^{3/2}L_i^{-3+\delta}$.

       \item[4.] For every $\alpha\in \mathfrak A_{I(\tau_f)}$, $\hat\gamma$ is strictly negative on $[0,\Gamma^{\hat u}(\tau_f)]\times\{\Gamma^{\hat v}(\tau_f)\}$ and $\hat\kappa$ is strictly positive on $\Gamma$. Therefore, the teleologically normalized coordinates $(u_{\tau_f},v)$ are defined on $\hat{\mathcal D}_{\tau_f}$.
       \end{enumerate}

       On $\mathcal D_{\tau_f}$, let $(\bar r_{\tau_f},\bar\Omega^2_{\tau_f})$ be the $\tau_f$-anchored extremal Reissner--Nordstr\"om solution as defined in \cref{sec:anchor}. 

       \begin{enumerate}
        \item[5.]  For every $\alpha\in \mathfrak A_{I(\tau_f)}$, the following bootstrap assumptions for the geometry hold on $\mathcal D_{\tau_f}$:
        \begin{align}
\label{eq:boot-nu}   \left|\frac{\nu_{\tau_f}}{\bar\nu_{\tau_f}}-1\right| &\le A^3 \ve^{3/2}\tau^{-1+\delta},\\
 \label{eq:boot-r}    |\tilde r_{\tau_f}|   &\le A^2 \ve^{3/2} \tau^{-2+\delta},\\
  \label{eq:boot-pi}|\tilde\varpi_{\tau_f}|  &\le A\ve^{3/2} \tau^{-3+\delta}.
        \end{align} 
        \item[6.] For every $\alpha\in \mathfrak A_{I(\tau_f)}$, the following bootstrap assumptions for the scalar field hold:
        \begin{align}
   \label{eq:boot-rp}        \mathcal E_{p}^{\tau_f}(\tau) &\le A\ve^2 \tau^{-3+\delta+p},\\
  \label{eq:boot-horizon}         \underline{\mathcal E}{}_p^{\tau_f}(\tau) &\le A\ve^2 \tau^{-3+\delta+p},\\
        \label{eq:boot-flux-1}       \mathcal F^{\tau_f}(u,\tau)&\le A\ve^2 \tau^{-3+\delta},\\
      \label{eq:boot-flux-2}     \underline{\mathcal F}^{\tau_f}(v,\tau)&\le A\ve^2\tau^{-3+\delta}
        \end{align}
        for every $\tau\in [1,\tau_f]$, $(u,v)\in\mathcal D_{\tau_f}$, and $p\in [0,3-\delta]$.
   \end{enumerate}
\end{defn}

\subsection{Nonemptiness of the bootstrap set}\label{sec:B-nonempty}

We begin the continuity argument with the following simple consequence of the local existence theory:

\begin{prop}\label{prop:B-nonempty}
    For any $A\ge1$ and $\ve$ sufficiently small,
    \begin{equation}
        \{\mathcal S_0(\alpha):\alpha\in\mathfrak A_0\}\subset\mathcal L(\mathcal S_0,\ve)\label{eq:set-inclusion}
    \end{equation}
    and $\mathfrak B(\mathcal S_0,\ve,A)$ is nonempty. 
\end{prop}
\begin{proof} Observe that $|\alpha|\le |e-e_0|+\ve^{3/2}\le 2\ve$ for $\alpha\in\mathfrak A_0$ and $\ve\le 1$ by the definition of $\mathfrak D$, which verifies \eqref{eq:set-inclusion}. Points 1.~and 3.~of \cref{def:B} are trivial for $\tau_f\in[1,2)$. By local well-posedness and continuous dependence on initial data, every element of $\mathcal S_0(\alpha)$ extends to a neighborhood of the initial data hypersurface $\hat{\mathcal C}$ which is uniform near $(0,0)$. Since the solution is close to Reissner--Nordstr\"om in this small region by Cauchy stability, the bootstrap assumptions 2.~and 4.--6.~are automatically satisfied for $A$ sufficiently large, $\ve$ sufficiently small, and $|\tau_f-1|$ sufficiently small. Therefore, $\tau_f\in  \mathfrak B(\mathcal S_0,\ve,A)$  for $|\tau_f-1|$ sufficiently small.
\end{proof}

\subsection{Improving the bootstrap assumptions---the main estimates}\label{sec:B-improving}

The main analytic content of the proof of \cref{thm:stability} consists of the following two propositions, which will be used to show that the bootstrap set is \emph{open}.

\begin{prop}[Improving the bootstrap assumptions for the geometry] \label{prop:geometry-est-improve}
    There exist positive constants $A_0$ and $\ve_0$ depending only on $M_0$ and $\delta$, such that if $A= A_0$, $\ve\le \ve_0$, $\tau_f\in  \mathfrak B(\mathcal S_0,\ve,A)$, and $\alpha\in\mathfrak A_{I(\tau_f)}$, then the estimates 
            \begin{align}
   \left|\frac{\nu_{\tau_f}}{\bar\nu_{\tau_f}}-1\right| &\le \tfrac 12 A^3\ve^{3/2}\tau^{-1+\delta},\\
     |\tilde r_{\tau_f}|   &\le \tfrac 12 A^2 \ve^{3/2} \tau^{-2+\delta}, \\
  |\tilde\varpi_{\tau_f}|  &\le \tfrac 12 A\ve^{3/2} \tau^{-3+\delta} \label{varpi_dec}
        \end{align} hold on $\mathcal D_{\tau_f}$. 
\end{prop}

This proposition is proved in \cref{sec:geometry-improvement}.

\begin{prop}[Improving the bootstrap assumptions for the scalar field] \label{prop:energy-est-improve}
  There exist positive constants $A_0$ and $\ve_0$ depending only on $M_0$ and $\delta$, such that if $A= A_0$, $\ve\le \ve_0$, $\tau_f\in  \mathfrak B(\mathcal S_0,\ve,A)$, and $\alpha\in\mathfrak A_{I(\tau_f)}$, then the estimates 
        \begin{align}
           \mathcal E_{p}^{\tau_f}(\tau) &\le \tfrac 12A\ve^2 \tau^{-3+\delta+p}, \label{ep_decay}\\
           \underline{\mathcal{E}}{}_{p}^{\tau_f} (\tau) &\le \tfrac 12A\ve^2 \tau^{-3+\delta+p}, \label{uep_decay}\\
           \mathcal F^{\tau_f}(u,\tau)&\le \tfrac 12A\ve^2 \tau^{-3+\delta}, \label{f_decay}\\
           \underline{\mathcal F}^{\tau_f}(v,\tau)&\le \tfrac 12A\ve^2\tau^{-3+\delta} \label{uf_decay}
        \end{align}
       hold for every $\tau\in [1,\tau_f]$, $(u,v)\in\mathcal D_{\tau_f}$, and $p\in [0,3-\delta]$.
\end{prop}

This proposition is proved in \cref{sec:improving-energy} and relies on energy estimates proved in \cref{sec:energy}.

\begin{rk}
    The modulation argument for $\varpi$, i.e., the construction of the sets $\mathfrak A_i$, takes place in \cref{sec:modulation} when we show that  $\mathfrak B(\mathcal S_0,\ve,A)$  is \emph{closed}. It relies crucially on an estimate for $\tilde\varpi_{\tau_f}$ that will be proved in the course of the proof of \cref{prop:geometry-est-improve}. See already \cref{lem:varpi-est}.
\end{rk}

\section{Estimates for the geometry}\label{sec:geometry}

In this section, we prove estimates for the geometry in a bootstrap domain $\mathcal D_{\tau_f}$. In \cref{sec:initial-estimates}, we estimate $\tilde\kappa_{\tau_f}$, $\tilde \gamma_{\tau_f}$, $\tilde\kappa_{\tau_f}$, the behavior of $\Gamma$ and $\tau$, and provide the fundamental Taylor expansions for $\tilde\lambda_{\tau_f}$ and $\tilde\nu_{\tau_f}$. In \cref{sec:closing-geometry}, we estimate $\tilde r_{\tau_f}$, $\nu_{\tau_f}/\bar \nu_{\tau_f}$, and $\tilde\varpi_{\tau_f}$. We improve the bootstrap assumptions for the geometry (\cref{prop:geometry-est-improve}) in \cref{sec:geometry-improvement}.

\subsection{Conventions for \texorpdfstring{\cref{sec:geometry,sec:energy,sec:decay-energy}}{Sections 5 to 7}}\label{sec:notation} In this section, \cref{sec:energy}, and \cref{sec:decay-energy}, we will fix $\tau_f$ and essentially exclusively work in the teleological coordinate system $(u_{\tau_f},v)$. Therefore, for ease of reading, we will omit the subscript $\tau_f$ on $u_{\tau_f}$, the solution $(r_{\tau_f},\Omega^2_{\tau_f},\phi_{\tau_f},e)$, the background extremal Reissner--Nordstr\"om solution $\bar r_{\tau_f}$, and the energies. In order to keep track of the constant $A$ appearing in the bootstrap assumptions \eqref{eq:boot-nu}--\eqref{eq:boot-flux-2}, the notations $a\lesssim b$, $a\gtrsim b$, and $a\sim b$ mean that the estimate does not depend on $u$, $v$, $\alpha$, $A$, $\ve$, $\tau$, or $\tau_f$. We will omit the decoration $(\mathcal S_0,\ve,A)$ on the bootstrap set  $\mathfrak B(\mathcal S_0,\ve,A)$. We will also often absorb powers of the bootstrap constant $A$ by sacrificing a little bit of $\ve$ without comment, for example, $A\ve^{3/2}\les\ve$ for $\ve$ sufficiently small. Note that $A$ will be fixed at the end of \cref{sec:improving-energy} depending only on $M_0$ and $\delta$.

\subsection{Initial estimates}\label{sec:initial-estimates}

\subsubsection{Basic consequences of the bootstrap assumptions}

\begin{lem}\label{lem:basic-consequences} For any $A\ge 1$, $\ve$ sufficiently small, $\tau_f\in\mathfrak B$, and $\alpha\in\mathfrak A_{I(\tau_f)}$, it holds that 
    \begin{gather}\label{eq:rough-bound-varpi} 
     \tfrac 12M  \le \varpi\le 2M,\\
     \tfrac 12\le \frac{\nu}{\bar \nu}\le 2,\label{eq:nu-ratio}\\
     \lambda>0 \label{eq:lambda-positive}
    \end{gather}
    in $\mathcal D_{\tau_f}$ and 
    \begin{gather}
   \label{eq:mu-lower-bound}  1-\mu\ge \tfrac 34,\\
     \label{eq:lambda-lower-bound}   \lambda \geq \tfrac 34
     ,\\
   \label{eq:nu-1}     -2\le \nu\le-\tfrac 12,
    \end{gather}
    in $\mathcal D_{\tau_f}\cap\{r\ge \Lambda \}$.
\end{lem}

\begin{proof}
The bounds \eqref{eq:rough-bound-varpi} and \eqref{eq:nu-ratio} follow directly from \eqref{eq:boot-pi} and \eqref{eq:boot-nu}, respectively. For \eqref{eq:mu-lower-bound} we notice that  \begin{equation*} 1-\mu = 1 - \frac{2\varpi}{r} + \frac{e^2}{r^2}\geq 1 - \frac{4M}{\Lambda} \geq \frac 34. \end{equation*}
For \eqref{eq:lambda-lower-bound} we first observe  that the monotonicity of \eqref{eq:kappa-u} implies that $\kappa \geq 1$ in $\mathcal D_{\tau_f}\cap \{r \geq \Lambda\}$. Therefore, $\lambda = \kappa (1-\mu) \geq \frac 34$ implies \eqref{eq:lambda-lower-bound}. By monotonicity of Raychaudhuri's equation \eqref{eq:Ray-v}, $\lambda>0$ throughout $\mathcal D_{\tau_f}$. Finally, 
 to show \eqref{eq:nu-1} we note that $\bar \nu = - (1-\bar \mu) $ and thus $ -1 \leq \bar \nu \leq -\frac 34$ in  $\mathcal D_{\tau_f}\cap\{r\ge \Lambda \}$ using \eqref{eq:boot-r}. The estimate \eqref{eq:nu-1} follows then from \eqref{eq:boot-nu}. \end{proof}

\subsubsection{Estimates for \texorpdfstring{$\tilde \kappa$}{kappa} and \texorpdfstring{$\tilde \gamma$}{gamma}}

We now use the bootstrap assumptions for the energy decay of $\phi$ to derive bounds for $\tilde\kappa$ and $\tilde\gamma$. 

\begin{lem}\label{lem:kappa-bounds}
For any $A\ge 1$, $\ve$ sufficiently small, $\tau_f\in\mathfrak B$, and $\alpha\in\mathfrak A_{I(\tau_f)}$, it holds that
    \begin{align}
      \label{eq:kappa-1}   |\tilde\kappa|&\les A\ve^2 \tau^{-1+\delta},\\
      \label{eq:kappa-2}    (\bar r-M)|\tilde\kappa|&\les A\ve^2\tau^{-2+\delta},\\
       \label{eq:kappa-3}   (\bar r-M)^2|\tilde\kappa|&\les A\ve^2\tau^{-3+\delta}
    \end{align}
    in $\mathcal D_{\tau_f}\cap\{r\le\Lambda \}$ and
    \begin{equation}
     \label{eq:kappa-4}    |\tilde\kappa|\les A\ve^2\tau^{-3+\delta}
    \end{equation}
    in $\mathcal D_{\tau_f}\cap\{r\ge\Lambda \}$. 
\end{lem}
\begin{proof} \textsc{Estimate for $r(u,v)\le\Lambda $}: 
From \eqref{eq:kappa-u} and $\kappa = \tilde \kappa +1$ we have 
\begin{equation*}
    \partial_u \tilde \kappa = \frac{r (\tilde \kappa +1)}{\nu} (\partial_u\phi)^2.
\end{equation*}
Integrating to the future from $\Gamma$, where $\tilde\kappa=0$ by our gauge choice, and multiplying by $(\bar r- M)^{2-p}$ for $p\in \{0,1,2\}$ gives
\begin{equation*}
    (\bar r - M)^{2-p} |\tilde \kappa (u,v) | \leq (\bar r  - M)^{2-p}\int_{\underline C{}_v\cap\{u'\le u\}}  \frac{  |\tilde \kappa| +1}{-\nu} r(\partial_u\phi)^2 du \lesssim\int_{\underline C{}_v\cap\{u'\le u\}}   (\bar r - M)^{2-p} \frac{ |\tilde \kappa|+1 }{-\bar \nu} (\partial_u\phi)^2 du,
\end{equation*}
where we used that $\nu\sim\bar\nu $ by \eqref{eq:boot-nu} and that $\bar \nu  <0 $ on $\mathcal D_f$ in order to move the factor $(\bar r- M)^{2-p}$ inside the integral. Applying Grönwall's lemma then gives
\begin{equation*}
     (\bar r - M)^{2-p} |\tilde \kappa(u,v) | \lesssim \underline {\mathcal E}{}_p (\tau ) \exp\big({\underline {\mathcal E}{}_2 (\tau)}\big)\lesssim A \varepsilon^2 \tau^{-3 + \delta +p},
\end{equation*}
from which \eqref{eq:kappa-1}--\eqref{eq:kappa-3} follow.

    \textsc{Estimate for $r(u,v)\ge\Lambda $}: We integrate \eqref{eq:kappa-u} backwards from $\Gamma$ using an integrating factor and the gauge condition to obtain
    \begin{equation}
        \log\kappa(u,v)=\int_{\underline H{}_v\cap\{u'\ge u\}} r\left(\frac{\partial_u\phi}{-\bar \nu}\right)^2\left(\frac{\bar\nu}{\nu}\right)(-\bar\nu)\,du' \lesssim \underline{\mathcal F} (v, \tau(u,v)),
    \end{equation}
    using \eqref{eq:boot-nu}. The estimate \eqref{eq:kappa-4}  follows now from \eqref{eq:boot-flux-2}.
\end{proof}

Note that the estimates \eqref{eq:kappa-1} and \eqref{eq:kappa-4} imply that \begin{gather}
       \label{eq:kappa-bdd}   \tfrac 12 \leq \kappa \leq 2 ,\\
      \label{eq:lambda-bdd}    \lambda\le 2
    \end{gather}
        in $\mathcal D_{\tau_f}$.

\begin{lem}\label{lem:gamma} For any $A\ge 1$, $\ve$ sufficiently small, $\tau_f\in\mathfrak B$, and $\alpha\in\mathfrak A_{I(\tau_f)}$, it holds that
    \begin{align}
    \label{eq:gamma-1}    |\tilde \gamma|&\les A\ve^2  r^{-1}\tau^{-3+\delta},\\
    |\tilde \gamma|&\les A\ve^2  r^{-3/2}\tau^{-5/2+\delta}\label{eq:gamma-2}
    \end{align}    in $\mathcal D_{\tau_f}\cap\{r\ge\Lambda \}$.
\end{lem}
\begin{proof}
  From \eqref{eq:gamma-v} and $\gamma = \tilde \gamma -1$ we have  \begin{equation*}
        \partial_v \tilde \gamma = \frac{r (\tilde\gamma -1 )}{\lambda} (\partial_v \phi)^2.
    \end{equation*}
Integrating from $\underline C{}_{\Gamma^v(\tau_f)}$ and multiplying by $r^{1+p}$ for $p\in\{0,1/2\}$ gives for $(u,v) \in \mathcal D_{\tau_f} \cap \{ r \geq \Lambda\}$ that
\begin{equation*}
 |   r^{1+p} \tilde \gamma(u,v)|\leq r^{1+p} \int_{C_u \cap \{ v \leq v'\} } r \frac{  |\tilde \gamma| + 1 }{\lambda} (\partial_v\phi)^2 \,dv' \leq   \int_{C_u \cap \{ v \leq v'\} } r^{p+2} \frac{  |\tilde \gamma| + 1 }{\lambda}  (\partial_v\phi)^2\, dv',
\end{equation*}
where we used \eqref{eq:lambda-lower-bound} to move the $r$ weight into the integral. Applying Grönwall's inequality as before and noting that $\lambda \sim 1$ gives 
\begin{equation*}
    | r^{1+p} \tilde \gamma(u,v) |\lesssim \mathcal E_p(\tau) \exp \big(\mathcal E_0(\tau)  \big)
\end{equation*}
for $p\in \{0,1/2\}$. 
This shows \eqref{eq:gamma-1} and \eqref{eq:gamma-2} using bootstrap assumption \eqref{eq:boot-rp}.
\end{proof}

\begin{rk}
    The estimate \eqref{eq:gamma-2}, with its integrable $\tau$-weight \emph{and} integrable $r$-weight, is important in controlling the gauge $\Phi_{\tau_f}$ as $\tau_f\to\infty$. See already \cref{sec:convergence-to-id}.
\end{rk}

\subsubsection{Properties of \texorpdfstring{$\Gamma$}{Gamma} and \texorpdfstring{$\tau$}{tau}}

In this section, we derive some basic properties of the timelike curve $\Gamma=\{r=\Lambda \}$ and the associated function $\tau(u,v)$ defined by \eqref{def:tau}. Recall that $\tau$ was also used to denote the proper time parametrization of $\Gamma$, defined via the condition  
\begin{equation}
 \label{def:proper-time} g(\dot\Gamma,\dot\Gamma)=-1,
\end{equation}
where $\dot{}$ is used to denote $\frac{d}{d\tau}$. We write the components of $\Gamma$ with this parametrization in the teleological coordinate chart $(u,v)$ as $\Gamma^u(\tau)$ and $\Gamma^v(\tau)$.

\begin{lem}\label{lem:gamma-slope} For any $A\ge 1$, $\ve$ sufficiently small, $\tau_f\in\mathfrak B$, and $\alpha\in\mathfrak A_{I(\tau_f)}$, it holds that \begin{equation}
    \dot\Gamma^u(\tau)\sim\dot\Gamma^v(\tau)\sim 1\label{eq:Gamma-slope}
\end{equation}
for $\tau\in[1,\tau_f]$.
\end{lem}
\begin{proof} Since $r$ is constant along $\Gamma$, we have $\nu\dot\Gamma^u+\lambda\dot\Gamma^v=0$ on $\Gamma$. By the gauge condition $\lambda  = \kappa (1-\mu) |_{\Gamma} = 1-\mu|_{\Gamma} \sim 1$ and \eqref{eq:nu-1}, it follows that $\dot\Gamma^u\sim\dot\Gamma^v$. The proper time condition \eqref{def:proper-time} can be written as 
\begin{equation*}
    \frac{4\nu\lambda}{1-\mu}\dot\Gamma^u\dot\Gamma^v=-1,
\end{equation*}
which implies that $\dot\Gamma^u\dot\Gamma^v\sim 1$ by \eqref{eq:mu-lower-bound}. Now \eqref{eq:Gamma-slope} easily follows.
\end{proof} 

The following lemma relates the ``time function'' $\tau$ (which was defined in \eqref{def:tau}) to (approximate) Bondi time $u$ in the far region and (approximate) ingoing Eddington--Finkelstein time $v$ in the near-horizon region.
 
\begin{lem}\label{lem:tau-properties} For any $A\ge 1$, $\ve$ sufficiently small, $\tau_f\in\mathfrak B$, and $\alpha\in\mathfrak A_{I(\tau_f)}$, it holds that 
 \begin{equation}
        \tau(u_1,v)-\tau(u_2,v)\sim u_1-u_2
    \end{equation}
    for $(u_1,v),(u_2,v)\in\mathcal D_{\tau_f}\cap\{r\ge \Lambda  \}$ and 
\begin{equation}
        \tau(u,v_1)-\tau(u,v_2)\sim v_1-v_2\label{eq:tau-difference-v}
    \end{equation}
for $(u,v_1),(u,v_2)\in\mathcal D_{\tau_f}\cap\{r\le \Lambda \}$. Moreover, for any $\eta>1$ and $1\le \tau_1\le\tau_2\le \tau_f$, it holds that
\begin{align}
 \label{eq:Hbarv-integral}    \int_{\underline H{}_v\cap\{\tau_1\le\tau\le \tau_2\}}\tau^{-\eta}\,du &\les_\eta \tau_1^{-\eta+1},\\
    \label{eq:Hu-integral}      \int_{H_u\cap\{\tau_1\le\tau\le \tau_2\}}\tau^{-\eta}\,dv &\les_\eta \tau_1^{-\eta+1}.
\end{align}
\end{lem}
\begin{proof}
    This is immediate from \eqref{eq:Gamma-slope}, the fundamental theorem of calculus, and the change of variables formula.
\end{proof}

Next, we show that the $v$ and $\hat v$ coordinates are comparable. 

\begin{lem}\label{lem:v-est}
    For $\tau_f\in \mathfrak B$, $\alpha\in\mathfrak A_{I(\tau_f)}$, $\ve$ sufficiently small depending on $A$, and $\hat v_1,\hat v_2\in [0,\Gamma^{\hat v}(\tau_f)]$,
    it holds that 
    \begin{equation}
        \mathfrak v(\hat v_2)-\mathfrak v(\hat v_1)\sim \hat v_2-\hat v_1.\label{eq:v-est}
    \end{equation}
    Moreover,  on $\hat{\mathcal D}_{\tau_f}$ it holds that
    \begin{equation}
        \hat\kappa\sim 1.\label{eq:hat-kappa-est}
    \end{equation}
\end{lem}
\begin{proof}
Recall the definition of the map $\hat v\mapsto \mathfrak v(\hat v)$ in \eqref{eq:teleology-2}. The derivative of the inverse map $v\mapsto \mathfrak v^{-1}(v)$ is given by $\lambda(0,v)$, so  \eqref{eq:v-est} follows from \eqref{eq:lambda-lower-bound}, \eqref{eq:lambda-bdd}, and the fundamental theorem of calculus, and \eqref{eq:hat-kappa-est} follows from the identity
\begin{equation*}
    \hat\kappa(\hat u,\hat v)=\frac{\kappa(\mathfrak u(\hat u),\mathfrak v(\hat v))}{\lambda(0,\mathfrak v(\hat v))}
\end{equation*}
and \eqref{eq:kappa-bdd}.
\end{proof}

\subsubsection{Taylor expansions of  \texorpdfstring{$\widetilde{1-\mu}$, $\partial_u\tilde r$,}{1 minus mu} and \texorpdfstring{$\partial_v\tilde r$}{tilde lambda}}

We now derive Taylor expansions for $\widetilde{1-\mu}$, $\partial_u\tilde r$, and $\partial_v\tilde r$ in terms of tilde quantities. The precise form of the terms linear in $\tilde r$ will be crucial for arguments later in the paper. 

\begin{lem}\label{lem:mu}  For any $A\ge 1$, $\ve$ sufficiently small, $\tau_f\in\mathfrak B$, and $\alpha\in\mathfrak A_{I(\tau_f)}$, it holds that 
\begin{align}
  \label{eq:mu-1}   \left|\widetilde{1-\mu} - \frac{2M}{\bar r^3}(\bar r-M)\tilde r\right|&\les \frac{|\tilde r|^2}{\bar r^3}+\frac{|\tilde\varpi|}{\bar r},\\
     |\widetilde{1-\mu} |&\les A^2\ve^{3/2}\tau^{-2+\delta}\label{eq:mu-2}
\end{align}
in $\mathcal D_{\tau_f}$.
\end{lem}
\begin{proof}

Using \eqref{eq:boot-r}, we immediately derive the Taylor expansion 
    \begin{equation}
   \label{eq:r-Taylor}     \frac{1}{r^\eta}=\frac{1}{\bar r^\eta}-\frac{\eta\tilde r}{\bar r^{\eta+1}} +O_\eta\left(\frac{|\tilde r|^2}{\bar r^{\eta+2}}\right)
    \end{equation}
for any $\eta>0$. Using this, we then compute
\begin{align*}
    \widetilde{1-\mu} &= -\frac{2\varpi}{r}+\frac{2M}{\bar r}+\frac{e^2}{r^2}-\frac{e^2}{\bar r^2}= -\frac{2\tilde\varpi}{r}-2M\left(\frac{1}{r}-\frac{1}{\bar r}\right)+e^2\left(\frac{1}{r^2}-\frac{1}{\bar r^2}\right)\\
    &= -\frac{2\tilde\varpi}{r} + \frac{2M\tilde r}{\bar r^2}+O\left(\frac{|\tilde r|^2}{\bar r ^3}\right)- \frac{2e^2\tilde r}{\bar r^3}+ O\left(\frac{|\tilde r|^2}{\bar r^4}\right)= -\frac{2\tilde\varpi}{r} + \frac{2M}{\bar r^3}(\bar r-M)\tilde r+ O\left(\frac{|\tilde r|^2}{\bar r^3}\right),
\end{align*}
which gives \eqref{eq:mu-1}. We then immediately obtain \eqref{eq:mu-2} from the bootstrap assumptions.
\end{proof}

\begin{lem}
   For any $A\ge 1$, $\ve$ sufficiently small, $\tau_f\in\mathfrak B$, and $\alpha\in\mathfrak A_{I(\tau_f)}$, we have the expansions
    \begin{align}
    \label{eq:nu-expansion}    \partial_u\tilde r&= \frac{2M\gamma}{\bar r^3}(\bar r -M)\tilde r+E_u \quad\text{in }\mathcal D_{\tau_f}\cap\{r\ge\Lambda  \},\\
    \label{eq:lambda-expansion}    \partial_v\tilde r&= \frac{2M\kappa}{\bar r^3}(\bar r -M)\tilde r+E_v\quad\text{in }\mathcal D_{\tau_f},
    \end{align}
    where the error terms satisfy the estimates
    \begin{align}
      \label{eq:nu-expansion-error}  |E_u|&\les A\ve^{3/2}\bar r^{-1}\tau^{-3+\delta},\\
     \label{eq:lambda-expansion-error}   |E_v|&\les A\ve^{3/2}\tau^{-3+\delta}.
    \end{align}
\end{lem}
\begin{proof} \textsc{Expansion for $\partial_u\tilde r$}: We compute
\begin{equation*}
    \partial_u\tilde r = \tilde\nu = \gamma(1-\mu)-\bar\gamma(\overline{1-\mu})=\tilde\gamma(\overline{1-\mu}) + \gamma(\widetilde{1-\mu})
    =\tilde\gamma\left(1-\frac{M}{\bar r}\right)^2+ \gamma\left[\frac{2M}{\bar r^3}(\bar r-M)\tilde r+O\left(\frac{|\tilde r|^2}{\bar r^3}+\frac{|\tilde\varpi|}{\bar r}\right)\right],
\end{equation*}
which implies \eqref{eq:nu-expansion} with error term estimated by
\begin{equation*}
    |E_u|\les \frac{|\tilde r|^2}{\bar r^3}+\frac{|\tilde\varpi|}{\bar r}+|\tilde\gamma|\les A\ve^{3/2}\bar r^{-1}\tau^{-3+\delta}
\end{equation*}
by the bootstrap assumptions and \eqref{eq:gamma-1}.

  \textsc{Expansion for $\partial_v\tilde r$}: We compute
  \begin{equation}
      \partial_v\tilde r=\tilde\lambda = \kappa(1-\mu)-\bar\kappa(\overline{1-\mu})= \tilde\kappa(\overline{1-\mu})+\kappa(\widetilde{1-\mu})
      = \frac{(\bar r-M)^2}{\bar r^2}\tilde\kappa + \kappa \left[\frac{2M}{\bar r^3}(\bar r-M)\tilde r+O\left(\frac{|\tilde r|^2}{\bar r^3}+\frac{|\tilde\varpi|}{\bar r}\right)\right],\label{eq:lambda-tilde}
  \end{equation}
  which implies \eqref{eq:lambda-expansion} with error term estimated by
  \begin{equation*}
      |E_v|\les \frac{|\tilde r|^2}{\bar r^3}+\frac{|\tilde\varpi|}{\bar r}+\frac{(\bar r-M)^2}{\bar r^2}|\tilde\kappa|\les A\ve^{3/2}\tau^{-3+\delta}
  \end{equation*}
  by the bootstrap assumptions and \eqref{eq:kappa-3}.
\end{proof}

\begin{lem}\label{lem:varkappa}
  For any $A\ge 1$, $\ve$ sufficiently small, $\tau_f\in\mathfrak B$, and $\alpha\in\mathfrak A_{I(\tau_f)}$, it holds that
 \begin{equation}\label{eq:varkappa-decay}
     |\tilde\varkappa|\les A\ve^{3/2}  r^{-2} \tau^{-3+\delta} + A^2\ve^{3/2}  r^{-3} \tau^{-2+\delta}
 \end{equation}
 in $\mathcal D_{\tau_f}$.
 \begin{proof}
     This follows directly from the definition of $\varkappa$, \eqref{eq:varkappa}, and the bootstrap assumptions \eqref{eq:boot-r} and \eqref{eq:boot-pi}.
 \end{proof}
\end{lem}

\subsection{Improving the bootstrap assumptions for the geometry}\label{sec:closing-geometry}

\subsubsection{Estimates for \texorpdfstring{$\tilde r$}{r}}

In the following lemma, we estimate $\tilde r$ in three stages: Along $\Gamma$, we estimate $\tilde r$ by using the estimates \eqref{eq:kappa-4} and \eqref{eq:gamma-1} for $\tilde \kappa$ and $\tilde\gamma$, for $r\le \Lambda$ we use the equation \eqref{eq:lambda-expansion}, which is \emph{redshifted} backwards in $v$ (see already \cref{rk:redshift-tilde-r}), and for $r\ge \Lambda$ we use the equation \eqref{eq:nu-expansion}.

\begin{lem}\label{lem:r-est} For any $A\ge 1$, $\ve$ sufficiently small, $\tau_f\in\mathfrak B$, and $\alpha\in\mathfrak A_{I(\tau_f)}$, it holds that
    \begin{align}
\label{eq:tilde-r-aux-3}        |\tilde r|&\les A\ve^{3/2}\tau^{-2+\delta},\\
\label{eq:lambda-difference}
    |\tilde \lambda |&\lesssim A\varepsilon^{3/2} \tau^{-2 + \delta}
    \end{align}
    in $\mathcal D_{\tau_f}$.
\end{lem}
\begin{proof}
\textsc{Estimate along $\Gamma$}: Using the relation $\nu\dot\Gamma^u+\lambda\dot\Gamma^v=0$ along $\Gamma$, we compute 
\begin{equation*}
    \frac{d}{d\tau}\tilde r(\Gamma(\tau))= -\bar\nu\dot\Gamma^u-\bar\lambda\dot\Gamma^v = \left(\frac{\bar\lambda}{\lambda}\nu-\bar\nu\right)\dot\Gamma^u= (\overline{1-\mu})\left(\frac{\bar\kappa}{\kappa}-\frac{\bar\gamma}{\gamma}\right)\gamma\dot\Gamma^u.
\end{equation*}
Using \eqref{eq:kappa-4} and \eqref{eq:gamma-1}, we estimate
\begin{equation*}
    \left|\frac{\bar\kappa}{\kappa}-\frac{\bar\gamma}{\gamma}\right|\les A\ve^2\tau^{-3+\delta}
\end{equation*}
along $\Gamma$ and therefore conclude
  \begin{equation}
  \big|\tilde r|_\Gamma\big|\les A\ve^2\tau^{-2+\delta}\label{eq:tilde-r-on-Gamma}
\end{equation}
by the fundamental theorem of calculus and the anchoring condition $\tilde r(\Gamma(\tau_f))=0$.

\textsc{Estimate for $r\le\Lambda $}: We view \eqref{eq:lambda-expansion} as an ODE for $\tilde r$ in $v$. For $(u,v),(u,v_2)\in \mathcal D_{\tau_f}\cap\{r\le\Lambda \}$ we use an integrating factor to write
\begin{align}
   \nonumber  \tilde r(u,v)&=\exp\left(-\int_{v}^{v_2}\frac{2M\kappa}{\bar r^3}(\bar r-M)\,dv'\right)\left[\tilde r(u,v_2)-\int_{v}^{v_2}\exp\left(\int_{v'}^{v_2}\frac{2M\kappa}{\bar r^3}(\bar r-M)\,dv''\right)E_v\,dv'\right]\\
  \label{eq:tilde-r-near-formula}   &= \exp\left(-\int_{v}^{v_2}\frac{2M\kappa}{\bar r^3}(\bar r-M)\,dv'\right)\tilde r(u,v_2)- \int_{v}^{v_2}\exp\left(-\int_{v}^{v'}\frac{2M\kappa}{\bar r^3}(\bar r-M)\,dv''\right)E_v\,dv'.
\end{align}
For $(u,v)\in \mathcal D_{\tau_f}\cap\{r\le\Lambda \}$, let $v_2=v^\Lambda(u)$. Then \eqref{eq:tilde-r-near-formula}, \eqref{eq:tilde-r-on-Gamma}, \eqref{eq:lambda-expansion-error}, and \eqref{eq:Hu-integral} imply
\begin{equation}
    |\tilde r(u,v)|\les A\ve^2\tau^{-2+\delta}(u,v)+ A\ve^{3/2}\tau^{-2+\delta}(u,v)\label{eq:tilde-r-aux-1}
\end{equation}
after observing that the integrating factor is nonnegative and hence the exponentials are bounded by $1$. 

\textsc{Estimate for $r\ge \Lambda $}:  We view \eqref{eq:nu-expansion} as an ODE for $\tilde r$ in $u$. For $(u,v),(u_2,v)\in \mathcal D_{\tau_f}\cap\{r\geq\Lambda \}$ we use an integrating factor to write
\begin{equation}\label{eq:tilde-r-far-formula}
    \tilde r(u,v)= \exp\left(-\int_{u}^{u_2}\frac{2M\gamma}{\bar r^3}(\bar r-M)\,du'\right)\left[\tilde r(u_2,v) -\int_{u}^{u_2}\exp\left(\int_{u'}^{u_2}\frac{2M\gamma}{\bar r^3}(\bar r-M) \,du''\right)E_u\,du'\right].
\end{equation}
For $(u,v)\in \mathcal D_{\tau_f}\cap\{r\geq\Lambda\}$ let $u_2=u^\Lambda(v)$. By \eqref{eq:nu-1}, the integrating factor is bounded:
\begin{equation}\label{eq:integrating-factor-bounded}
    \int_{u}^{u_2}\frac{2M(-\gamma)}{\bar r^3}(\bar r-M)\,du'\les \int_{\underline H{}_v}\bar r^{-2}\,du'\les 1.
\end{equation}
It then follows from \eqref{eq:tilde-r-far-formula}, \eqref{eq:tilde-r-on-Gamma}, \eqref{eq:nu-expansion-error}, and \eqref{eq:Hbarv-integral} that
\begin{equation}
    \label{eq:tilde-r-aux-2} |\tilde r(u,v)|\les A\ve^2\tau^{-2+\delta}(u,v)+ A\ve^{3/2}\tau^{-2+\delta}(u,v).
\end{equation}
The estimates \eqref{eq:tilde-r-on-Gamma}, \eqref{eq:tilde-r-aux-1}, and \eqref{eq:tilde-r-aux-2} give \eqref{eq:tilde-r-aux-3} as desired. 

Finally, \eqref{eq:lambda-difference} now follows from \eqref{eq:lambda-expansion} and \eqref{eq:lambda-expansion-error}.
\end{proof}

\begin{rk}\label{rk:redshift-tilde-r}
    The good sign of the integrating factor in \eqref{eq:tilde-r-near-formula} is due to the \emph{global redshift effect} present on extremal Reissner--Nordstr\"om. To exploit this, it was crucial that we integrated the ODE \eqref{eq:lambda-expansion} \emph{backwards} in $v$. Indeed, integrating \eqref{eq:lambda-expansion} forwards in time results in a bad \emph{global blueshift}, despite the absence of the horizon redshift effect! Compare with \cite[Lemma 8.19]{luk2019strong}.
\end{rk}

\subsubsection{Estimates for \texorpdfstring{$\nu/\bar\nu$}{nu}}

\begin{lem}\label{lem:nu-estimate}
   For any $A\ge 1$, $\ve$ sufficiently small, $\tau_f\in\mathfrak B$, and $\alpha\in\mathfrak A_{I(\tau_f)}$, it holds that
    \begin{equation}
        \label{eq:nu-aux-1} \left|\frac{\nu}{\bar \nu}-1\right|\les A^2\ve^{3/2}\tau^{-1+\delta}
    \end{equation}
    on $\mathcal D_{\tau_f}$.
\end{lem}
\begin{proof}
    \textsc{Estimate for $r\ge\Lambda $}: We obtain directly from \eqref{eq:nu-expansion} and the bootstrap assumptions that
    \begin{equation}
        |\tilde\nu|\les \bar r^{-2}|\tilde r|+ A\ve^{3/2}\bar r^{-1}\tau^{-3+\delta}\les A^2\ve^{3/2}\bar r^{-1} \tau^{-2+\delta},\label{eq:nu-aux-2}
    \end{equation}
    which implies \eqref{eq:nu-aux-1} for $r\ge\Lambda $ after dividing by $|\bar\nu|\sim 1$. 

      \textsc{Estimate for $r\le\Lambda $}: 
      From the wave equation \eqref{eq:nu-v} we obtain
      \begin{equation}
      \label{eq:estimate-partialv-nu}
         \left|  \partial_v{\log\left(\frac{\nu}{\bar\nu}\right)} \right|= 2 |\widetilde{\kappa \varkappa}| = 2 |\kappa \tilde \varkappa| + 2 |\tilde \kappa \bar \varkappa| \lesssim A^2 \varepsilon^{3/2}\tau^{-2+\delta} +(\bar r - M)|\tilde \kappa| \lesssim A^2 \varepsilon^{ 3/2} \tau^{-2+\delta}
      \end{equation}
      by \eqref{eq:varkappa-decay} and \eqref{eq:kappa-2}. By \eqref{eq:nu-aux-2}, we have
      \begin{equation*}
         \left|\left. {\log\left(\frac{\nu}{\bar\nu}\right)}\right|_{\Gamma}\right|\les \big|\tilde\nu|_\Gamma\big|\les A^2\ve^{3/2}\tau^{-2+\delta}
      \end{equation*}
so integrating \eqref{eq:estimate-partialv-nu} backwards from $\Gamma$ shows \eqref{eq:nu-aux-1}. \end{proof}

\subsubsection{Estimates for \texorpdfstring{$\tilde\varpi$}{varpi}}

Recall the dyadic index $I(\tau_f)$ which was defined to be the largest integer such that $L_{I(\tau_f)}=2^{I(\tau_f)}\le\tau_f$. Using the energy decay bootstrap assumptions for $\phi$, we now show that $\varpi$ remains close to its value at $\Gamma(L_{I(\tau_f)})$. Note that the power of $\ve$ in \eqref{eq:varpi-aux-0} below is strictly better than in the bootstrap assumption \eqref{eq:boot-pi}. This is fundamental for our modulation argument, see already \cref{sec:modulation} below.

\begin{lem}\label{lem:varpi-est}
For any $A\ge 1$, $\ve$ sufficiently small, $\tau_f\in\mathfrak B$, and $\alpha\in\mathfrak A_{I(\tau_f)}$, it holds that
   \begin{equation}\label{eq:varpi-aux-0}
       |\tilde\varpi-\Pi_{I(\tau_f)}(\alpha)|\les A\ve^2\tau^{-3+\delta}
   \end{equation}
   on $\mathcal D_{\tau_f}$, where $\Pi_{I(\tau_f)}\doteq \tilde\varpi(\Gamma(L_{I(\tau_f)}))$.
\end{lem}
\begin{proof}
 \textsc{Estimate for} $u=\Gamma^u(L_{I(\tau_f)})$: We write $I=I(\tau_f)$ for short. By \eqref{eq:varpi-v} and \eqref{eq:kappa-bdd}, we have
    \begin{equation}
        |\partial_v(\tilde\varpi-\Pi_{I}(\alpha))|\les r^2(\partial_v\phi)^2.\label{eq:varpi-aux-1}
    \end{equation}
    Integrating this estimate forwards in $v$, we obtain from the bootstrap assumptions
\begin{equation}\label{eq:varpi-aux-2}
    |\tilde\varpi-\Pi_{I}(\alpha)|(\Gamma^u(L_{I}),v)\les \mathcal E_0(L_{I})\les A\ve^2L_{I}^{-3+\delta}\les A\ve^2\tau_f^{-3+\delta}
\end{equation}
for every $v\ge \Gamma^v(L_I)$. On the other hand, integrating \eqref{eq:varpi-aux-1} backwards in $v$, we obtain from the bootstrap assumptions and \eqref{eq:Hu-integral}
\begin{equation}\label{eq:varpi-aux-3}
    |\tilde\varpi-\Pi_{I}(\alpha)|(\Gamma^u(L_{I}),v)\les \mathcal F(\Gamma^u(L_{I}),\tau(\Gamma^u(L_{I},v))\les A\ve^2 \tau(\Gamma^u(L_{I}),v)^{-3+\delta}
\end{equation}
for every $v\le \Gamma^v(L_I)$.

    \textsc{Estimate for} $u<\Gamma^u(L_{I(\tau_f)})$:  By \eqref{eq:varpi-u}, the bootstrap assumptions, and \eqref{eq:mu-1}, we estimate
    \begin{align}
      \nonumber  |\partial_u(\tilde\varpi-\Pi_{I}(\alpha))|&\les r^2(\overline{1-\mu})\left(\frac{\partial_u\phi}{-\bar\nu}\right)^2(-\bar\nu)+ r^2|\widetilde{1-\mu}|\left(\frac{\partial_u\phi}{-\bar\nu}\right)^2(-\bar\nu)\\
      \nonumber &\les r^2(\overline{1-\mu})\left(\frac{\partial_u\phi}{-\bar\nu}\right)^2(-\bar\nu) +A^2\ve^{3/2}\tau^{-2+\delta}\left(\frac{\bar r-M}{\bar r^3}\right) \left(\frac{\partial_u\phi}{-\bar\nu}\right)^2(-\bar\nu) \\&\qquad +A\ve^{3/2}\tau^{-3+\delta}r^2\left(\frac{\partial_u\phi}{-\bar\nu}\right)^2(-\bar\nu). \label{eq:varpi-aux-4}
    \end{align}
    If $v\ge \Gamma^v(L_I)$, then $r(u,v)\ge\Lambda$ and by integrating \eqref{eq:varpi-aux-4} and using \eqref{eq:varpi-aux-2} and the bootstrap assumptions we have
    \begin{equation}
         |\tilde\varpi-\Pi_{I}(\alpha)|(u,v)\les A\ve^2\tau_f^{-3+\delta}+\underline{\mathcal F}(v,\tau(u,v))\les A\ve^2\tau(u,v)^{-3+\delta}.\label{eq:varpi-aux-5}
    \end{equation}

    On the other hand, if $v<\Gamma^v(L_I)$, we consider first the case $r(u,v)\le\Lambda$. Integrating \eqref{eq:varpi-aux-4} backwards from $(\Gamma^u(L_I),v)$ and using \eqref{eq:varpi-aux-2} and the bootstrap assumptions, we have 
      \begin{align}
     \nonumber    |\tilde\varpi-\Pi_{I}(\alpha)|(u,v)&\les A\ve^2\tau(\Gamma^u(L_I),v)^{-3+\delta}+ \underline{\mathcal E}{}_0(\tau(u,v))  \nonumber\\  \nonumber&\qquad+A^2\ve^{3/2}\tau(u,v)^{-2+\delta}\underline{\mathcal E}{}_1(\tau(u,v))+A\ve^{3/2}\tau(u,v)^{-3+\delta}\underline{\mathcal E}{}_2(\tau(u,v))\\
     \nonumber     &\les A\ve^2\tau(\Gamma^u(L_I),v)^{-3+\delta}+ A\ve^2\tau(u,v)^{-3+\delta}+ A^3\ve^{7/2}\tau(u,v)^{-4+2\delta}\\
         &\les A\ve^2\tau(u,v)^{-3+\delta},\label{eq:varpi-aux-6}
    \end{align}
    as desired, where we note that $\tau(\Gamma^u(L_I),v)=\tau(u,v)$. If instead $r(u,v)\ge \Lambda$, then the ingoing null cone emanating from $(u,v)$ will strike $\Gamma$ at the point $(u_\star,v)=(u^\Lambda(v),v)$. We first apply the estimate \eqref{eq:varpi-aux-6} at $(u_\star,v)$ to obtain \[|\tilde\varpi-\Pi_{I}(\alpha)|(u_\star,v)\les A\ve^2\tau(u_\star,v)^{-3+\delta}.\]
    Integrating \eqref{eq:varpi-aux-4} further backwards towards $(u,v)$, we finally obtain
       \begin{align*}
        |\tilde\varpi-\Pi_{I(\tau_f)}(\alpha)|(u,v)&\les A\ve^2\tau(u_\star,v)^{-3+\delta}+ \underline{\mathcal F}(v,\tau(u,v))\\&\les A\ve^2\tau(u_\star,v)^{-3+\delta}+A\ve^2 \tau(u,v)^{-3+\delta}\\&\les A\ve^2 \tau(u,v)^{-3+\delta}, 
    \end{align*}
    as desired. 
    
        \textsc{Estimate for} $u>\Gamma^u(L_{I(\tau_f)})$: This is similar to the case $u<\Gamma^u(L_{I(\tau_f)})$ and is left to the reader.
\end{proof}

\begin{rk}\label{rk:modulation-argument}
    It is clear from the proof that the estimate \eqref{eq:varpi-aux-0} also holds with $I(\tau_f)$ replaced by $I(\tau_f)-1$ if $I(\tau_f)\ge 1$. This observation will be used later in \cref{sec:modulation} when we perform the modulation argument for $\varpi$. 
\end{rk}

\subsubsection{The proof of \texorpdfstring{\cref{prop:geometry-est-improve}}{prop}}\label{sec:geometry-improvement}

We now put all of the estimates proved in this section together and improve the bootstrap assumptions for the geometry, \eqref{eq:boot-nu}--\eqref{eq:boot-pi}.

\begin{proof}[Proof of \cref{prop:geometry-est-improve}]
To summarize: by \cref{lem:r-est,lem:nu-estimate,lem:varpi-est}, we have the estimates
\begin{align*}
    |\tilde r|&\les A\ve^{3/2}\tau^{-2+\delta},\\
    \left|\frac{\nu}{\bar\nu}-1\right|&\les A^2\ve^{3/2}\tau^{-1+\delta},\\
    |\tilde\varpi-\Pi_{I(\tau_f)}(\alpha)|&\les A\ve^2\tau^{-3+\delta}
\end{align*}
in $\mathcal D_{\tau_f}$, where the implicit constant \ul{does not} depend on $A$. Therefore, by choosing $A$ sufficiently large, the first two estimates imply
\begin{align*}
|\tilde r|    &\le \tfrac 12 A^2\ve^{3/2}\tau^{-2+\delta},\\
   \left|\frac{\nu}{\bar\nu}-1\right|   &\le \tfrac 12  A^3\ve^{3/2}\tau^{-1+\delta}
\end{align*}
as desired. In order to improve the bootstrap assumption for $\tilde\varpi$, we use the definition of the modulation sets \eqref{modulation-definition-1} and \eqref{modulation-definition-2}. Namely, since $|\Pi_{I(\tau_f)}|\le \ve^{3/2}L_{I(\tau_f)}^{-3+\delta}$, we may estimate
\begin{equation*}
    |\tilde\varpi|\le |\tilde\varpi-\Pi_{I(\tau_f)}(\alpha)|+|\Pi_{I(\tau_f)}(\alpha)|\les A\ve^2\tau^{-3+\delta}+\ve^{3/2}\tau_f^{-3+\delta}\les (A\ve^{1/2}+1)\ve^{3/2}\tau^{-3+\delta}.
\end{equation*}
Therefore, by choosing $A$ sufficiently large and $\ve$ correspondingly sufficiently small, we have
\begin{equation*}
    |\tilde\varpi|\le \tfrac 12 A\ve^{3/2}\tau^{-3+\delta}
\end{equation*} in $\mathcal D_{\tau_f}$
as desired. 
\end{proof}

\section{Energy estimates for the scalar field}\label{sec:energy}

In this section, we derive the fundamental hierarchies of energy estimates for the scalar field which were used to improve the estimates for the geometry in \cref{sec:geometry}. In \cref{sec:Kodama}, we prove a basic degenerate energy boundedness statement for $\phi$ which is analogous to the usual $T$-energy estimate on extremal Reissner--Nordstr\"om. In \cref{sec:Morawetz}, we prove Morawetz estimates for $\phi$, i.e., weighted spacetime $L^2$ bounds for $\phi$ and its derivatives. In \cref{sec:horizon-hierarchy}, we prove $(\bar r-M)^{2-p}$-weighted energy estimates for $\partial_u\psi$ in the near-horizon region, where $\psi\doteq r\phi$. Finally, in  \cref{sec:rp} we prove $r^p$-weighted energy estimates for $\partial_v\psi$ in the far region.

\begin{center}
    \emph{In this section, we adopt the notational conventions outlined in \cref{sec:notation}. We will also often use the geometric bootstrap assumptions and their consequences, such as \cref{lem:basic-consequences}, without comment.}
\end{center}

The energy estimates in this section (except for \cref{lem:kodama-auxiliary}) will take place on a region $\mathcal R\subset\mathcal D_{\tau_f}$ defined as follows: Let $\tau_f\in\mathfrak B$, $(u_1,v_1),(u_2,v_2)\in\Gamma\cap\mathcal D_{\tau_f}$,  $u_2'>u_2$, and $v_2'>v_2$,  where $(u_2,v_2)$ is to the future of $(u_1,v_1)$ We then define
\begin{equation*}
    \mathcal R\doteq [u_1,u_2]\times[v_1,v_2']\cup [u_1,u_2']\times[v_1,v_2]
\end{equation*}
and 
\begin{equation*}
    \mathcal R_{\le\Lambda}\doteq\mathcal R\cap\{r\le\Lambda\},\quad  \mathcal R_{\ge\Lambda}\doteq\mathcal R\cap\{r\ge\Lambda\}.
\end{equation*}
The null segments constituting $\partial \mathcal R$ are numbered I--VI as depicted in \cref{fig:butterfly} below. We define $\tau_1$ to be the value of $\tau$ on $\mathrm{I}\cup\mathrm{II}$ and $\tau_2$ to be the value of $\tau$ on $\mathrm{IV}\cup\mathrm{V}$.

 \begin{figure}[h]
\centering{
\def\svgwidth{12pc}
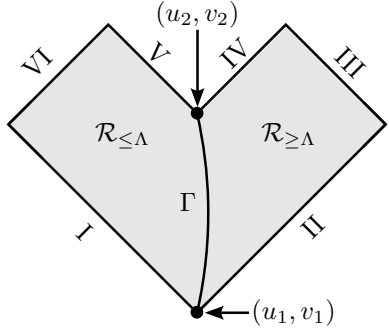}
\caption{A Penrose diagram depicting the region $\mathcal R$ and the hypersurfaces I--VI used in the energy estimates in this section.}
\label{fig:butterfly}
\end{figure}

\subsection{The degenerate energy estimate}\label{sec:Kodama}

We begin with an energy boundedness statement (with a small, decaying nonlinear error term) for an energy which degenerates quadratically in $\partial_u\phi$ at $\bar r=M$. 

\begin{prop}\label{prop:Kodama-estimate} For any $A\ge 1$, $\ve$ sufficiently small, $\tau_f\in\mathfrak B$, $\alpha\in\mathfrak A_{I(\tau_f)}$, $(u_1,v_1),(u_2,v_2)\in\Gamma\cap\mathcal D_{\tau_f}$,  $u_2'>u_2$, and $v_2'>v_2$,  where $(u_2,v_2)$ is to the future of $(u_1,v_1)$, it holds that \begin{multline}\label{eq:Kodama-energy-est}
    \int_\mathrm{III} \left(r^2(\partial_u\phi)^2+\phi^2\right) du+ \int_\mathrm{IV} \left(r^2(\partial_v\phi)^2+\phi^2\right)dv  +\int_\mathrm{V}\left((\bar r-M)^2\frac{(\partial_u\phi)^2}{-\bar\nu}-\bar\nu\phi^2\right)du\\+\int_\mathrm{VI}\left(r^2(\partial_v\phi)^2+\bar\lambda\phi^2\right)dv
    \les \int_\mathrm{I} \left((\bar r-M)^2\frac{(\partial_u\phi)^2}{-\bar\nu}-\bar\nu\phi^2\right)du + \int_\mathrm{II}\left(r^2(\partial_v\phi)^2+\phi^2\right)dv+\ve^3\tau_1^{-4+2\delta},
\end{multline}
where the hypersurfaces $\mathrm{I}$--$\mathrm{VI}$ are as depicted in \cref{fig:butterfly}.
\end{prop}

We prove this estimate by using the \emph{Kodama vector field} \cite{kodama1980conserved}
\begin{equation}\label{eq:Kodama}
    K\doteq \frac{1-\mu}{2}\left(\frac{1}{\lambda}\partial_v-\frac{1}{\nu}\partial_u\right)
\end{equation} 
to derive an energy identity on $\partial\mathcal R$. On Reissner--Nordstr\"om, $K$ is exactly the time-translation Killing vector field. Of course, $K$ is no longer Killing on a general spherically symmetric solution of the Einstein--Maxwell-scalar field system, but it has the remarkable property that the current $T_{\mu\nu}^\mathrm{SF}K^\nu$ is divergence-free, where $T^\mathrm{SF}$ is the energy-momentum tensor of the scalar field \eqref{eq:T-sf}. This property is expressed more plainly by the identity \eqref{eq:Kodama-no-bulk} below. In \cref{rk:removing-Kodama} below, we sketch a proof that avoids the use of $K$.  

We will augment the energy identity arising from $K$ with a ``null Lagrangian'' current \cite{action-principle} generated by a function $\sigma$ on $\mathcal D_{\tau_f}$ which allows us to simultaneously estimate zeroth order fluxes of $\phi$. This current is analogous to the two-form $\varpi$ used in \cite{DHRT22,Morawetz-note} and can be thought of as encoding a Hardy inequality on each face of $\partial\mathcal R$ \emph{without boundary terms}.

\begin{proof}[Proof of \cref{prop:Kodama-estimate}] We apply the general multiplier identity \eqref{eq:energy-identity} to the Kodama vector field \eqref{eq:Kodama} and add the trivial identity $\partial_u\partial_v(\sigma\phi^2)-\partial_v\partial_u(\sigma\phi^2)=0$, where $\sigma$ is an arbitrary smooth function on $\mathcal D_{\tau_f}$, to obtain 
\begin{equation}
\partial_u\left(\kappa^{-1}r^2(\partial_v\phi)^2+\partial_v(\sigma\phi^2)\right)+\partial_v\left(-\gamma^{-1}r^2(\partial_u\phi)^2-\partial_u(\sigma\phi^2)\right)=0.\label{eq:Kodama-no-bulk}
\end{equation}
Integrating this relation over the region $\mathcal R$ by parts, we obtain the energy identity 
\begin{multline}
 \int_\mathrm{III}  \left((1-\mu)r^2\frac{(\partial_u\phi)^2}{-\nu}-\partial_u(\sigma\phi^2)\right)du+ \int_\mathrm{IV} \left(\frac{r^2}{\kappa}(\partial_v\phi)^2+\partial_v(\sigma\phi^2)\right)dv\\+\int_\mathrm{V} \left((1-\mu)r^2\frac{(\partial_u\phi)^2}{-\nu}-\partial_u(\sigma\phi^2)\right)du+\int_\mathrm{VI}\left(\frac{r^2}{\kappa}(\partial_v\phi)^2+\partial_v(\sigma\phi^2)\right)dv\\
    =\int_\mathrm{I} \left((1-\mu)r^2\frac{(\partial_u\phi)^2}{-\nu}-\partial_u(\sigma\phi^2)\right)du + \int_\mathrm{II}\left(\frac{r^2}{\kappa}(\partial_v\phi)^2+\partial_v(\sigma\phi^2)\right)dv. \label{eq:Kodama-energy-id}
\end{multline}
For a small constant $\eta>0$ to be determined, we set $
    \sigma\doteq\eta(\bar r-M).$

If $I_u$ denotes the $du$ integrand in \eqref{eq:Kodama-energy-id}, we then compute \begin{equation*}
    I_u= r^2(1-\bar\mu)\frac{\bar\nu}{\nu} \frac{(\partial_u\phi)^2}{-\bar\nu} +r^2(\widetilde{1-\mu})\frac{\bar\nu}{\nu}\frac{(\partial_u\phi)^2}{-\bar\nu}-\eta\bar\nu\phi^2-2\eta (\bar r-M)\phi\partial_u\phi.
\end{equation*}
For $\eta$ sufficiently small, using \eqref{eq:mu-1} and Young's inequality, we estimate
\begin{align}
    \label{eq:Kod-aux-1}   I_u&\les r^2(1-\bar\mu)\frac{(\partial_u\phi)^2}{-\bar\nu}+A\ve^{3/2}  \tau^{-2+\delta}(\bar r-M)r\frac{(\partial_u\phi)^2}{-\bar\nu}+A\ve^{3/2} \tau^{-3+\delta}r^2\frac{(\partial_u\phi)^2}{-\bar\nu}-\bar\nu \phi^2,\\
    \label{eq:Kod-aux-2}       I_u&\gtrsim r^2(1-\bar\mu)\frac{(\partial_u\phi)^2}{-\bar\nu}-A\ve^{3/2}  \tau^{-2+\delta}(\bar r-M)r\frac{(\partial_u\phi)^2}{-\bar\nu}-A\ve^{3/2}  \tau^{-3+\delta}r^2\frac{(\partial_u\phi)^2}{-\bar\nu}-\bar\nu \phi^2
\end{align}
on $\mathcal D_{\tau_f}$. For the $dv$ integrand $I_v$, we readily estimate
\begin{equation}
     \label{eq:Kod-aux-3}   I_v\sim r^2(\partial_v\phi)^2+\bar\lambda\phi^2
\end{equation}
on $\mathcal D_{\tau_f}$ for $\eta$ sufficiently small.

Inserting the estimates \eqref{eq:Kod-aux-1}--\eqref{eq:Kod-aux-3} into \eqref{eq:Kodama-energy-id}, we obtain
\begin{multline}\label{eq:Kodama-energy-est-1}
    \int_\mathrm{III} \left(r^2(\partial_u\phi)^2+\phi^2\right) du+ \int_\mathrm{IV} \left(r^2(\partial_v\phi)^2+\phi^2\right)dv  +\int_\mathrm{V}\left((\bar r-M)^2\frac{(\partial_u\phi)^2}{-\bar\nu}-\bar\nu\phi^2\right)du\\+\int_\mathrm{VI}\left(r^2(\partial_v\phi)^2+\bar\lambda\phi^2\right)dv
    \les \int_\mathrm{I} \left((\bar r-M)^2\frac{(\partial_u\phi)^2}{-\bar\nu}-\bar\nu\phi^2\right)du + \int_\mathrm{II}\left(r^2(\partial_v\phi)^2+\phi^2\right)dv+A\ve^{3/2} E,
\end{multline}
where the error term $E$ is given by
\begin{equation*}
   E\doteq \int_\mathrm{I}\tau^{-2+\delta}(\bar r-M)\frac{(\partial_u\phi)^2}{-\bar\nu}\,du+\int_\mathrm{I}\tau^{-3+\delta}\frac{(\partial_u\phi)^2}{-\bar\nu}\,du+\int_\mathrm{V}\tau^{-2+\delta}(\bar r-M)\frac{(\partial_u\phi)^2}{-\bar\nu}\,du+\int_\mathrm{V}\tau^{-3+\delta}\frac{(\partial_u\phi)^2}{-\bar\nu}\,du.
\end{equation*}
Using the bootstrap assumption \eqref{eq:boot-horizon} with $p=1$ and $p=2$, we estimate
\begin{equation*}
    E\les \ve^3\tau_1^{-4+2\delta}+\ve^3\tau_2^{-4+2\delta},
\end{equation*}
which when inserted into \eqref{eq:Kodama-energy-est-1} completes the proof of \eqref{eq:Kodama-energy-est}. \end{proof}

With essentially the same proof, we also have:

\begin{lem}\label{lem:kodama-auxiliary}
    For any $A\ge 1$, $\ve$ sufficiently small, $\tau_f\in\mathfrak B$, $\alpha\in\mathfrak A_{I(\tau_f)}$, $(u_1,v_1),(u_2,v_2)\in\Gamma\cap\mathcal D_{\tau_f}$, where $(u_2,v_2)$ is to the future of $(u_1,v_1)$, it holds that 
    \begin{multline}\label{eq:Kodama-energy-est-square}
    \int_\mathrm{iii} \left(r^2(\partial_u\phi)^2+\phi^2\right) du+ \int_\mathrm{iv}\left(r^2(\partial_v\phi)^2+\bar\lambda\phi^2\right)dv
    \\\les \int_\mathrm{i} \left((\bar r-M)^2\frac{(\partial_u\phi)^2}{-\bar\nu}-\bar\nu\phi^2\right)du + \int_\mathrm{ii}\left(r^2(\partial_v\phi)^2+\phi^2\right)dv+\ve^3\tau_1^{-4+2\delta},
\end{multline}
where the hypersurfaces $\mathrm{i}$--$\mathrm{iv}$ are as depicted in \cref{fig:square}.
\end{lem}
 \begin{figure}[h]
\centering{
\def\svgwidth{7pc}
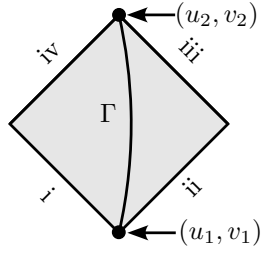}
\caption{A Penrose diagram depicting the hypersurfaces i--iv used in \cref{lem:kodama-auxiliary}.}
\label{fig:square}
\end{figure}

\begin{rk}\label{rk:removing-Kodama}
    The Kodama vector field can be avoided by instead using \begin{equation*}
    \bar T \doteq \frac{1}{2}(\partial_v+\partial_u)=\frac{1-\bar\mu}{2}\left(\frac{1}{\bar\lambda}\partial_v-\frac{1}{\bar\nu}\partial_u\right),
\end{equation*}
the time-translation Killing vector field of the comparison extremal Reissner--Nordstr\"om solution. In this remark, we sketch how to use $\bar T$ to prove the estimate \eqref{eq:Kodama-energy-est} (with $-4+2\delta$ replaced by $-4+3\delta$). 

For $\bar T$, the identity \eqref{eq:energy-identity} yields 
\begin{equation*}
      \partial_u\big(r^2(\partial_{ v}\phi)^2\big)+\partial_{ v}\big(r^2(\partial_u\phi)^2\big) = - 2r(\nu+\lambda)\partial_u\phi\partial_v\phi.
\end{equation*}
We now explain how to estimate the bulk error term
\begin{equation*}
    \iint_{\mathcal R} r|\lambda+\nu||\partial_u\phi||\partial_v\phi|\,dudv.
\end{equation*}

\textsc{Estimate in $\mathcal R_{\le\Lambda}$}:
Using the definitions, we have
\begin{equation*}
    \lambda+\nu =(1-\bar\mu)\left(\kappa-\frac{\nu}{\bar\nu}\right)+(\widetilde{1-\mu})\kappa,
\end{equation*}
which by \eqref{eq:boot-nu}, \eqref{eq:kappa-3}, and \eqref{eq:mu-2} implies 
\begin{equation}\label{eq:rk-aux-1}
    |\lambda+\nu|\les A^3\ve^{3/2}\tau^{-1+\delta}(1-\bar\mu) + A^2\ve^{3/2} \tau^{-2+\delta}.
\end{equation}
 Using this estimate and Young's inequality yields
\begin{multline*}
     \iint_{\mathcal R_{\le\Lambda}} r|\lambda+\nu||\partial_u\phi||\partial_v\phi|\,dudv\les A^3\ve^{3/2}  \iint_{\mathcal R_{\le\Lambda}} \left((\bar r-M)^4\frac{(\partial_u\phi)^2}{-\bar\nu}-\bar\nu(\partial_v\phi)^2\right)\,dudv \\+A^2\ve^{3/2} \iint_{\mathcal R_{\le\Lambda}} \tau^{-4+2\delta}\frac{(\partial_u\phi)^2}{-\bar\nu}\,dudv.
\end{multline*}
The first term can be absorbed by the Morawetz bulk (see already \cref{prop:Morawetz}) and the second term can be estimated using the bootstrap assumption \eqref{eq:boot-horizon},
\begin{equation*}
   A^2 \ve^{3/2} \iint_{\mathcal R_{\le\Lambda}} \tau^{-4+2\delta}\frac{(\partial_u\phi)^2}{-\bar\nu}\,dudv\les A^2\ve^{3/2}\int_{\tau_1}^{\tau_2}\tau^{-4+2\delta}\underline{\mathcal E}{}_2(\tau)\,d\tau \les \ve^3\int_{\tau_1}^{\tau_2}\tau^{-5+3\delta}\,d\tau\les \ve^3\tau_1^{-4+3\delta}.
\end{equation*}

\textsc{Estimate in $\mathcal R_{\ge\Lambda}$}: The estimate \eqref{eq:rk-aux-1} implies $|\lambda+\nu|\les A^3\ve^{3/2}\tau^{-1+\delta}$ in $\mathcal R_{\ge\Lambda}$. Therefore, we may estimate 
\begin{equation*}
      \iint_{\mathcal R_{\ge\Lambda}} r|\lambda+\nu||\partial_u\phi||\partial_v\phi|\,dudv\les  A^3\ve^{3/2} \iint_{\mathcal R_{\le\Lambda}} r^{1-\eta}(\partial_u\phi)^2 \,dudv +  A^3\ve^{3/2} \iint_{\mathcal R_{\le\Lambda}}\tau^{-1+\delta} r^{1+\eta}(\partial_v\phi)^2\,dudv,
\end{equation*}
where $\eta>0$ is the parameter in the Morawetz estimate in \cref{prop:Morawetz}. The first term can again be absorbed into the Morawetz bulk and the second term can be estimated by using the bootstrap assumption \eqref{eq:boot-rp},
\begin{equation*}
A^3\ve^{3/2}\iint_{\mathcal R_{\le\Lambda}} r^{1+\eta}(\partial_v\phi)^2\,dudv \les A^3\ve^{3/2}\int_{\tau_1}^{\tau_2} \tau^{-2+2\delta}\mathcal E_0(\tau)\,d\tau \les \ve^3\tau_1^{-4+3\delta},
\end{equation*}
which completes the proof.
\end{rk}

\subsection{Integrated local energy decay}\label{sec:Morawetz}

In this section, we prove the following \emph{Morawetz estimate} for the scalar field:

\begin{prop}\label{prop:Morawetz}
     Fix $\eta\in(0,1)$. Then for any $A\ge 1$, $\ve$ sufficiently small, $\tau_f\in\mathfrak B$, $\alpha\in\mathfrak A_{I(\tau_f)}$, $(u_1,v_1),(u_2,v_2)\in\Gamma\cap\mathcal D_{\tau_f}$,  $u_2'>u_2$, and $v_2'>v_2$, where $(u_2,v_2)$ is to the future of $(u_1,v_1)$, it holds that
     \begin{multline}\label{eq:Mor-main}
         \iint_{\mathcal R}\left(\left(1-\frac{M}{\bar r}\right)^4r^{1-\eta}\left(\frac{\partial_u\phi}{-\bar\nu}\right)^2+r^{1-\eta}(\partial_v\phi)^2+\left(1-\frac{M}{\bar r}\right)^2\frac{\phi^2}{r^{1+\eta}}\right)(-\bar\nu)dudv \\ \les \int_\mathrm{I} \left((\bar r-M)^2\frac{(\partial_u\phi)^2}{-\bar\nu}-\bar\nu\phi^2\right)du + \int_\mathrm{II}\left(r^2(\partial_v\phi)^2+\phi^2\right)dv+ \ve^3\tau_1^{-4+2\delta}.
     \end{multline}
\end{prop}

The proof of this estimate will be given in \cref{sec:Mor-proof} below after a series of lemmas. 

\subsubsection{Near-horizon Hardy inequalities I}\label{sec:Hardy-1}

In this section, we show how to control the zeroth order term on the left-hand side of \eqref{eq:Mor-main} by other quantities which will naturally appear later in \cref{sec:Mor-proof}. We begin with the following identity, which is proved by directly expanding the term in square brackets and integrating by parts (see also Lemma 8.30 in \cite{luk2019strong}).

\begin{lem} For any $[u_1,u_2]\times\{v\}\subset\mathcal D_{\tau_f}$, $f:\mathcal D_{\tau_f}\to\Bbb R$, and $\alpha\in\Bbb R$, it holds that 
\begin{multline}\label{eq:hardy_h}
    \frac{(\alpha +1 )^2}{4} \int_{u_1}^{u_2} ( \bar{r} - M )^{\alpha}f^2\,(-\bar\nu) du +  \int_{u_1}^{u_2} \frac{(\bar{r} - M )^{\alpha}}{-\bar{\nu}} \left[ (\bar{r} - M ) \partial_u f - \frac{\alpha+1}{2} (-\bar{\nu}) f \right]^2 \, du\\ =  \int_{u_1}^{u_2} (\bar{r} - M )^{\alpha+2} \left( \frac{\partial_u f}{-\bar{\nu}} \right)^2 (-\bar{\nu})  du+\frac{\alpha +1}{2} [ (\bar{r} - M )^{\alpha+1} f^2 ] (u_1,v) - \frac{\alpha+1}{2} [ (\bar{r} - M ) ^{\alpha+1} f^2 ] ( u_2 , v ).
\end{multline}
\end{lem}

We will use this identity several times later in \cref{sec:Hardy-2}, but for now use it to control the zeroth order bulk in \eqref{eq:Mor-main}:

\begin{lem}\label{lem:horizon-hardy-1} With hypotheses as in \cref{prop:Morawetz}, it holds that
   \begin{equation}
        \iint_{\mathcal R_{\le\Lambda}}(\bar r-M)^2\phi^2\,(-\bar\nu)dudv\les \iint_{\mathcal R_{\le\Lambda}} (\bar r-M)^4 \left(\frac{\partial_u\phi}{-\bar\nu}\right)^2(-\bar\nu)dudv+\int_\Gamma \phi^2\,ds,\label{eq:Hardy-bulk-1}
   \end{equation}
   where $ds$ is the arc length measure of $\Gamma$ with respect to the Euclidean metric $du^2+dv^2$.
\end{lem}
\begin{proof} Let $v\in [v_1,v_2]$. We apply the identity \eqref{eq:hardy_h} with $f=\phi$ and $\alpha=2$ along the segment \begin{equation}\label{eq:underline-S-v}\underline S{}_v\doteq[u^\Lambda(v),u_2']\times\{v\}\subset\mathcal R_{\le\Lambda}\end{equation}
to obtain the inequality 
\begin{equation}
    \int_{\underline S{}_v} (\bar r -M)^2\phi^2\,(-\bar\nu)du \les \int_{\underline S{}_v} (\bar r-M)^4 \left(\frac{\partial_u\phi}{-\bar\nu}\right)^2(-\bar\nu)du + \phi^2(u^\Lambda(v),v),
\end{equation} 
where we dropped the boundary term at $(u_2',v)$ due to its favorable sign. Integrating this inequality now over $v\in[v_1,v_2]$ and using the fact that $dv$ along $\Gamma$ is proportional to $ds$ by \eqref{eq:Gamma-slope}, we obtain \eqref{eq:Hardy-bulk-1}.
\end{proof}

\subsubsection{Proof of the Morawetz estimate}\label{sec:Mor-proof}

The proof of \cref{prop:Morawetz} involves exploiting the following identities and using our energy boundedness statement \cref{prop:Kodama-estimate} to control flux terms with unfavorable signs. 

\begin{lem}[Morawetz identities]
    Let $f:(0,\infty)\to\Bbb R$ be a $C^2$ function and let $\phi$ be a solution of \eqref{eq:phi-wave-1}. Then the following identities hold,     where $f=f(r)$ and $'=\frac{d}{dr}$:
\begin{multline}\label{eq:mora}
        r^2f'\left(\frac{\lambda}{-\gamma}(\partial_u\phi)^2-\frac{\nu}{\kappa}(\partial_v\phi)^2\right)+4rf(1-\mu)\partial_u\phi\partial_v\phi  +\frac{2fr^3}{\kappa\nu}  (\partial_u\phi)^2(\partial_v\phi)^2\\
        = \partial_v\left(\frac{r^2f}{-\gamma}(\partial_u\phi)^2\right) -\partial_u\left(\frac{r^2f}{\kappa}(\partial_v\phi)^2\right),
    \end{multline}
\begin{multline}\label{eq:mor_back2}
     - r^2f'\big(\lambda(\partial_u\phi)^2-\nu(\partial_v\phi)^2\big)-2fr(\lambda-\nu-2)\partial_u\phi\partial_v\phi+2\big(rf''\lambda+2(rf'-f)\kappa\varkappa\big)(-\nu)\phi^2\\=\partial_u\big(r^2f(\partial_v\phi)^2-2rf'\lambda\phi^2+\partial_v(rf\phi^2)\big)-\partial_v\big(r^2f(\partial_u\phi)^2+2rf'\nu\phi^2-\partial_u(rf\phi^2)\big).
\end{multline}
\end{lem}
\begin{proof} To prove \eqref{eq:mora}, we apply the general multiplier identity \eqref{eq:energy-identity} to the vector field
\begin{equation}\label{eq:mor-vf-1}
 -f(r)\left(\frac{1}{\gamma}\partial_u+\frac{1}{\kappa}\partial_v\right).
\end{equation} 

To prove \eqref{eq:mor_back2}, we first apply  \eqref{eq:energy-identity} to the vector field
\begin{equation}\label{eq:mor-vf-2}
 -f(r)\left(\partial_v-\partial_u\right)
\end{equation} 
to derive the identity 
     \begin{equation} \label{eq:mor_back1}
       r^2f'\big(\lambda(\partial_u\phi)^2-\nu(\partial_v\phi)^2\big)+2rf(\lambda-\nu)\partial_u\phi\partial_v\phi=\partial_v\big(r^2f(\partial_u\phi)^2\big)-\partial_u\big(r^2f(\partial_v\phi)^2\big).
    \end{equation}
 To obtain \eqref{eq:mor_back2} we now add the zeroth order current identity \eqref{eq:lower-order-multiplier} with $h= - 2 f$ to  \eqref{eq:mor_back1}.
\end{proof}

\begin{rk} 
    The vector field \eqref{eq:mor-vf-1} is the same as used in \cite[Lemma 8.24]{luk2019strong} and the ``modified'' identity \eqref{eq:mor_back2} is inspired by the identity (8.134) in \cite{luk2019strong}. 
\end{rk}

\begin{rk}
    The vector fields used in these identities are variants of the usual Morawetz vector field $f(r)\partial_{r_\star}$ used in \cite{dafermos2009red}, for instance. We use \eqref{eq:mor-vf-1} because for the choice $f=-r^{-3}$, the cross term in \eqref{eq:mora} (the middle term on the left-hand side) can be immediately absorbed into the good bulk using Young's inequality (see also \cite[Proposition 9.3.1]{Aretakis-instability-1}). One might worry about the quartic nonlinear error term, but it turns out to have a favorable sign with this choice of $f$! The point of the ``modified'' Morawetz identity \eqref{eq:mor_back2} is twofold: it introduces a (potentially) good zeroth order term and (potentially) improves the $r$-weight on the cross term. We will show that both of these good properties are true for a well-chosen $f$ below. However, this choice of $f$ would produce a bad quartic error term in \eqref{eq:mora}, so we use the ``background'' vector field \eqref{eq:mor-vf-2} to derive the modified identity. Note that we could also use $f=-r^{-3}$ in \eqref{eq:mor_back1}, but dealing with the cross term would take a bit more work. 
\end{rk}

The proof of \cref{prop:Morawetz} now proceeds in four steps:
\begin{enumerate}
    \item Use \eqref{eq:mor_back1} with $f=-r^{-3}$ in $\mathcal R$ to control $\partial\phi$ in spacetime $L^2$, but with a suboptimal $r$-weight. 
    \item Use \eqref{eq:mor_back1} with $f=-r^{-3}$ in $\mathcal R_{\ge \Lambda}$ to control $\partial\phi$ in $L^2$ along the timelike curve $\Gamma$. 
    \item Use \eqref{eq:mor_back2} with $f=-1+R^{-1}\chi r^{-\eta}$, where $R\ge\Lambda$ is a large constant and $\chi$ is an appropriately chosen cutoff, to improve the $r$-weights for $\partial\phi$ and to estimate $\phi$ in spacetime $L^2$ for $r\ge\Lambda$. 
    \item Combine these estimates with the Hardy inequality \eqref{eq:Hardy-bulk-1} to conclude \eqref{eq:Mor-main}.
\end{enumerate}

\begin{lem}\label{lem:mor-proof-1} With hypotheses as in \cref{prop:Morawetz}, it holds that \begin{multline}
     \iint_{\mathcal R}\frac{1}{r^2} \left(\left(1-\frac{M}{\bar r}\right)^4\left(\frac{\partial_u\phi}{-\bar\nu}\right)^2+(\partial_v\phi)^2\right)(-\bar\nu)dudv\\\les  \int_\mathrm{I} \left((\bar r-M)^2\left(\frac{\partial_u\phi}{-\bar\nu}\right)^2+\phi^2\right)(-\bar\nu)du + \int_\mathrm{II}\left(r^2(\partial_v\phi)^2+\phi^2\right)dv+ \ve^3\tau_1^{-4+2\delta}.\label{eq:Mor-compact}
 \end{multline}
\end{lem}
\begin{proof}
We integrate \eqref{eq:mora} over $\mathcal R$ with the choice $f(r)=-r^{-3}$.

\textsc{Estimating the fluxes:} For this choice of $f$, the integrands in the flux terms on the right-hand side of \eqref{eq:mora} are bounded by the integrands in the fluxes of the Kodama energy identity \eqref{eq:Kodama-energy-id}. Therefore, by the same reasoning as in the proof of \cref{prop:Kodama-estimate}, we may estimate  \begin{equation}\label{eq:Mor-aux-2}
   \iint_{\mathcal R} B\,(-\bar\nu)dudv \les  \int_\mathrm{I} \left((\bar r-M)^2\frac{(\partial_u\phi)^2}{-\bar\nu}-\bar\nu\phi^2\right)du + \int_\mathrm{II}\left(r^2(\partial_v\phi)^2+\phi^2\right)dv+ \ve^3\tau_1^{-4+2\delta}
\end{equation}
where
\begin{equation*}
  B\doteq 3r^{-2}\left(\lambda(1-\mu)\left(\frac{\partial_u\phi}{-\nu}\right)^2+\kappa^{-1}(\partial_v\phi)^2\right) 
      -4r^{-2}(1-\mu)\left(\frac{\partial_u\phi}{-\nu}\right)\partial_v\phi  +\frac{2}{\kappa \nu^2}  (\partial_u\phi)^2(\partial_v\phi)^2.
\end{equation*}

\textsc{Estimating the bulk:} By Young's inequality, the sign of the final term in $B$, \eqref{eq:nu-ratio}, and \eqref{eq:kappa-bdd}, we observe that 
\begin{equation*}
    B\ge  r^{-2}\lambda(1-\mu)\left(\frac{\partial_u\phi}{-\nu}\right)^2+r^{-2}\kappa^{-1}(\partial_v\phi)^2\gtrsim r^{-2}\lambda(1-\mu)\left(\frac{\partial_u\phi}{-\bar\nu}\right)^2+r^{-2}(\partial_v\phi)^2.
\end{equation*}
Using the expansions \eqref{eq:mu-1} and \eqref{eq:lambda-expansion}, we derive
\begin{equation*}
    \left|\lambda(1-\mu)-(1-\bar\mu)^2\right|\les A^2\ve^{3/2}\left((1-\bar\mu)\tau^{-2+\delta}+(1-\bar\mu)^{1/2}\tau^{-5+2\delta}+\tau^{-6+2\delta} \right)
\end{equation*}
on $\mathcal D_{\tau_f}$, and, therefore,  
\begin{equation}\label{eq:Mor-aux-1}
    \iint_{\mathcal R} \left[(\partial_v\phi)^2+(1-\bar\mu)^2\left(\frac{\partial_u\phi}{-\bar\nu}\right)^2\right] r^{-2} (-\bar\nu)dudv\les  \iint_{\mathcal R} B (-\bar\nu)dudv+A^2\ve^{3/2}E,
\end{equation}
with error term
\begin{equation*}
     E\doteq  \iint_{\mathcal R}r^{-2}\left((1-\bar\mu)\tau^{-2+\delta}+(1-\bar\mu)^{1/2}\tau^{-5+2\delta}+\tau^{-6+2\delta} \right)\left(\frac{\partial_u\phi}{-\bar\nu}\right)^2(-\bar\nu)dudv\doteq E_{\ge\Lambda}+E_{\le\Lambda},
\end{equation*}
where the notation $E_{\ge\Lambda}$ means the integral is restricted to $\mathcal R_{\ge\Lambda}$ and similarly for $E_{\le\Lambda}$. By \eqref{eq:mu-lower-bound} and \eqref{eq:lambda-lower-bound}, $E_{\ge\Lambda}$ can be absorbed into the integral of $B$ in \eqref{eq:Mor-aux-1}. To estimate $E_{\le\Lambda}$, we use the bootstap assumptions for $\phi$:
\begin{align}
 \label{eq:Mor-aux-3}   \iint_{\mathcal R\cap\{r\le\Lambda\}}(1-\bar\mu)\tau^{-2+\delta}\left(\frac{\partial_u\phi}{-\bar\nu}\right)^2(-\bar\nu)dudv&\les \int_{\tau_1}^{\tau_2}\tau^{-2+\delta}\underline{\mathcal E}{}_0(\tau)\,d\tau  \les  A\ve^2\tau_1^{-4+2\delta},\\
    \iint_{\mathcal R\cap\{r\le\Lambda\}}(1-\bar\mu)^{1/2}\tau^{-5+\delta}\left(\frac{\partial_u\phi}{-\bar\nu}\right)^2(-\bar\nu)dudv &\les\int_{\tau_1}^{\tau_2}\tau^{-5+\delta}\underline{\mathcal E}{}_1(\tau)\,d\tau \les A \ve^2 \tau_1^{-6+2\delta},\\
    \iint_{\mathcal R\cap\{r\le\Lambda\}}\tau^{-6+\delta}\left(\frac{\partial_u\phi}{-\bar\nu}\right)^2(-\bar\nu)dudv &\les\int_{\tau_1}^{\tau_2}\tau^{-6+\delta}\underline{\mathcal E}{}_2(\tau)\,d\tau \les A \ve^2 \tau_1^{-6+2\delta}. \label{eq:Mor-aux-4}
\end{align}
Combining \eqref{eq:Mor-aux-2}--\eqref{eq:Mor-aux-4} yields \eqref{eq:Mor-compact}, as desired.
\end{proof}

Next, we state a simple consequence of Stokes' theorem in the plane. 

\begin{lem}\label{lem:mor-proof-2}
    Let $j^u,j^v$, and $Q$ be smooth functions on $\mathcal R_{\ge \Lambda}$ satisfying $\partial_uj^u+\partial_vj^v=Q$. Then the divergence identity 
    \begin{equation*}
        \iint_{\mathcal R_{\ge \Lambda}} Q\,dudv= \int_\mathrm{IV}j^u\,dv-\int_\mathrm{II}j^u\,dv+\int_\mathrm{III}j^v\,du+\int_\Gamma (n_\Gamma^uj^u+n_\Gamma^vj^v)ds
    \end{equation*} holds,
    where 
    \begin{equation*}
         n_\Gamma \doteq \big((\dot\Gamma^u)^2+(\dot\Gamma^v)^2\big)^{-1/2}(\dot\Gamma^v\partial_u-\dot\Gamma^u\partial_v) 
    \end{equation*}
    is the outward-pointing normal to $\mathcal R_{\ge\Lambda}$ along $\Gamma$ with respect to the flat Euclidean metric $du^2+dv^2$ on $\mathcal R_{\ge \Lambda}$, $\dot{}$ is used to denote $\frac{d}{d\tau}$, and $ds$ is the Euclidean arc length measure of $\Gamma$.
\end{lem}

\begin{lem}\label{lem:mor-proof-3}
     With hypotheses as in \cref{prop:Morawetz}, it holds that \begin{equation}
    \int_\Gamma\big((\partial_u\phi)^2+(\partial_v\phi)^2\big)ds \les  \int_\mathrm{I} \left((\bar r-M)^2\frac{(\partial_u\phi)^2}{-\bar\nu}-\bar\nu\phi^2\right)du + \int_\mathrm{II}\left(r^2(\partial_v\phi)^2+\phi^2\right)dv+ \ve^3\tau_1^{-4+2\delta}.\label{eq:Mor-Gamma-1}
 \end{equation}
\end{lem}
\begin{proof}
We integrate \eqref{eq:mora} over $\mathcal R_{\ge \Lambda}$ with the choice $f=-r^{-3}$. As in the proof of \cref{lem:mor-proof-1}, the fluxes along the null hypersurfaces II, III, and IV can be estimated using \cref{prop:Kodama-estimate}. The bulk terms can be estimated using \eqref{eq:Mor-compact}. Observe that the contribution along $\Gamma$ is given by 
\begin{equation}
  I_{\Gamma}\doteq  \int_\Gamma \left(\frac{1}{\kappa r}(\partial_v\phi)^2n_\Gamma^u+\frac{1}{\gamma r}(\partial_u\phi^2)n^v_\Gamma\right) ds.
\end{equation}
Since $\Gamma$ is uniformly timelike (recall the estimate \eqref{eq:Gamma-slope}), we have $    n_\Gamma^u\sim - n_\Gamma^v\sim 1$ and therefore (using also that $\kappa\sim 1$ and $\gamma\sim -1$ on $\Gamma$)
\begin{equation}
    I_{\Gamma} \sim   \int_\Gamma\big((\partial_u\phi)^2+(\partial_v\phi)^2\big)ds.
\end{equation}
This completes the proof of \eqref{eq:Mor-Gamma-1}.
\end{proof}

\begin{lem}\label{lem:mor-proof-4}
     With hypotheses as in \cref{prop:Morawetz}, it holds that   \begin{multline}
         \iint_{\mathcal R_{\ge \Lambda}}\left(r^{1-\eta}(\partial_u\phi)^2+r^{1-\eta}(\partial_v\phi)^2+\frac{\phi^2}{r^{1+\eta}}\right)(-\bar\nu)dudv +\int_\Gamma  \phi^2\,ds \\ \les_\eta \int_\mathrm{I} \left((\bar r-M)^2\frac{(\partial_u\phi)^2}{-\bar\nu}-\bar\nu\phi^2\right)du + \int_\mathrm{II}\left(r^2(\partial_v\phi)^2+\phi^2\right)dv+ \ve^3\tau_1^{-4+2\delta}.\label{eq:mor-aux-11}
     \end{multline}
\end{lem}
\begin{proof}
We integrate \eqref{eq:mor_back2} over $\mathcal R\cap \{r\ge\Lambda\}$ with the choice $f(r)=-1+R^{-1}\chi(r) r^{-\eta}$, where $R\ge\Lambda$ is a large constant to be determined and $\chi$ is a cutoff function such that $\chi(r)=0$ for $r\le R$, $\chi(r)=1$ for $r\ge 2R$, and such that $|\chi'|\les R^{-1}$ and $|\chi''|\les R^{-2}$.

By \cref{lem:mor-proof-2}, we have
\begin{equation}\label{eq:mor-aux-5}
        \iint_{\mathcal R_{\ge \Lambda}} (F_1+F_2+Z\phi^2)\,(-\nu)dudv -\int_\Gamma (n_\Gamma^uj^u+n_\Gamma^vj^v)ds= \int_\mathrm{IV}j^u\,dv-\int_\mathrm{II}j^u\,dv+\int_\mathrm{III}j^v\,du,
    \end{equation} 
    where the bulk terms are given by 
    \begin{align*}
    F_1&\doteq R^{-1}r^{1-\eta}(\eta\chi -\chi' r)\left((\partial_v\phi)^2+ \lambda\frac{(\partial_u\phi)^2}{-\nu}\right),\\
          F_2&\doteq 2(1-R^{-1}\chi r^{-\eta})r(\lambda-\nu-2)\left(\frac{\partial_u\phi}{-\nu}\right)\partial_v\phi,\\
            Z&\doteq  R^{-1}\lambda\left(\chi'' r^{1-\eta}-2\eta \chi'r^{-\eta}+\eta(\eta+1)\chi r^{-1-\eta}\right)+2\left(1+R^{-1}\chi' r^{1-\eta}-(\eta+1)R^{-1}\chi r^{-\eta}\right)2\kappa\varkappa,
    \end{align*}
    and the flux terms are given by 
    \begin{equation*}
        j^u\doteq r^2f(\partial_v\phi)^2+rf'\lambda\phi^2+\lambda f\phi^2+2rf\phi\partial_v\phi,\quad  j^v\doteq -r^2f(\partial_u\phi)^2-rf'\nu\phi^2+\nu f\phi^2+2rf\phi\partial_u\phi.
    \end{equation*}

\textsc{Estimate for the zeroth order bulk}: We claim that 
\begin{equation}\label{eq:Mor-aux-6}
 Z=  2\kappa\varkappa+\eta(\eta+1)R^{-1}r^{-1-\eta}\chi -2R^{-1}(\eta+1)\kappa\varkappa r^{-\eta}\chi + 2R^{-1}(\kappa\varkappa r^{1-\eta}-\eta r^{-\eta})\chi'+R^{-1}r^{1-\eta}\chi'' \gtrsim R^{-1}r^{-1-\eta}
\end{equation}
on $\mathcal R_{\ge \Lambda}$ for $R$ sufficiently large. For $r\le R$, only the first term is nonzero and since $\varkappa\sim r^{-2}$ for $r\ge\Lambda$, we have $Z\gtrsim R^{-1}r^{-1}\ge R^{-1}r^{-1-\eta}$. For $R\le r\le 2R$, we estimate
\begin{equation*}
    Z\gtrsim r^{-2}+R^{-1}  r^{-1-\eta} \chi - R^{-2-\eta} \gtrsim r^{-2} \gtrsim R^{-1}r^{-1-\eta}
\end{equation*}
for $R$ sufficiently large. For $r\ge 2R$, only the first two terms are nonzero and \eqref{eq:Mor-aux-6} is evidently true. We therefore have
\begin{equation}
    \iint_{\mathcal R_{\ge \Lambda}} Z\phi^2\,(-\nu) dudv\gtrsim R^{-1}\iint_{\mathcal R_{\ge \Lambda}}r^{-1-\eta}\phi^2\,(-\bar\nu)dudv.\label{eq:Mor-aux-7}
\end{equation}

\textsc{Estimates for the first order bulks}: Let $
    F_0\doteq r^{-2}\left((\partial_u\phi)^2+(\partial_v\phi)^2\right). $
This quantity has the significance that its integral over $\mathcal R\cap \{r\ge\Lambda\}$ can be estimated by \cref{lem:mor-proof-1}. It is clear that $F_1\gtrsim R^{-1}r^{1-\eta}\left((\partial_u\phi)^2+(\partial_v\phi)^2\right)$ for $r\ge 2R$. For $F_2$, we estimate
\begin{equation*}
    |F_2|\les r|\lambda-\nu-2||\partial_u\phi||\partial_v\phi|\le r\left(|\bar\lambda-\bar\nu-2|+|\tilde\lambda|+|\tilde\nu|\right)|\partial_u\phi||\partial_v\phi|\les \left(1+A^2\ve^{3/2} r\tau^{-2+\delta}\right)|\partial_u\phi||\partial_v\phi|
\end{equation*}
using  \eqref{eq:lambda-difference} and \eqref{eq:nu-aux-2}. For $r\le R^2$, we have $ |\partial_u\phi||\partial_v\phi|\les R^4F_0$ and for $r\ge R^2$, we have 
\begin{equation*}
    |\partial_u\phi||\partial_v\phi|\le R^{-2+2\eta}r^{1-\eta}|\partial_u\phi||\partial_v\phi| \les R^{-1+\eta}F_1.
\end{equation*}
Therefore, for any $b>0$ we may choose $R$ sufficiently large that
\begin{equation*}
    |F_2|\les R^4 F_0 + b \textbf 1_{\{r\ge 2R\}} F_1 + A^2\ve^{3/2} r\tau^{-2+\delta}|\partial_u\phi||\partial_v\phi|.
\end{equation*}
Choosing $b$ sufficiently small to absorb the middle term, we therefore have that
\begin{multline}\label{eq:mor-aux-12}
    \iint_{\mathcal R_{\ge \Lambda}}(F_1+F_2)\,(-\nu)dudv\gtrsim R^{-3}\iint_{\mathcal R_{\ge \Lambda}}r^{1-\eta}\left((\partial_u\phi)^2+(\partial_v\phi)^2\right)dudv \\ -R^4\cdot(\text{LHS of \eqref{eq:Mor-compact}})-A^2\ve^{3/2}\iint_{\mathcal R_{\ge \Lambda}} r\tau^{-2+\delta}|\partial_u\phi||\partial_v\phi|\,dudv
\end{multline} by \cref{lem:mor-proof-1}
for $R$ sufficiently large. At this point, we fix $R$ and treat it as an implicit constant. For the final term on the right-hand side of this inequality, we use Young's inequality with a parameter $c>0$ and the bootstrap assumptions for $\phi$ to estimate
\begin{align*}
     \iint_{\mathcal R_{\ge \Lambda}} r\tau^{-2+\delta}|\partial_u\phi||\partial_v\phi|\,dudv &\les c\iint_{\mathcal R_{\ge \Lambda}} (\partial_u\phi)^2\,dudv + c^{-1}\iint_{\mathcal R_{\ge \Lambda}}\tau^{-4+2\delta}r^2(\partial_v\phi)^2\,dudv \\ & \les c\iint_{\mathcal R_{\ge \Lambda}} (\partial_u\phi)^2\,dudv +c^{-1}\int_{\tau_1}^{\tau_2} \tau^{-4+2\delta}\mathcal E_0(\tau)\,d\tau \\
     & \les c\iint_{\mathcal R_{\ge \Lambda}} (\partial_u\phi)^2\,dudv + c^{-1}A\ve^2 \tau_1^{-6+3\delta}.
\end{align*}
Choosing $c$ sufficiently small, we can absorb the first term on the right-hand side into the good bulk on the right-hand side of \eqref{eq:mor-aux-12} and conclude
   \begin{equation}\label{eq:mor-aux-13}
           \iint_{\mathcal R_{\ge \Lambda}}(F_1+F_2)\,(-\nu)dudv\gtrsim \iint_{\mathcal R_{\ge \Lambda}}r^{1-\eta}\left((\partial_u\phi)^2+(\partial_v\phi)^2\right)dudv -(\text{LHS of \eqref{eq:Mor-compact}})-\ve^3 \tau_1^{-6+3\delta}.
   \end{equation}

\textsc{Estimates for the null fluxes}: As in the proof of \cref{lem:mor-proof-1}, the terms on the right-hand side of \eqref{eq:mor-aux-5} are all estimated using \cref{prop:Kodama-estimate} by 
\begin{equation}\label{eq:Mor-aux-8}
   \text{RHS of \eqref{eq:mor-aux-5}}\les \int_\mathrm{I} \left((\bar r-M)^2\frac{(\partial_u\phi)^2}{-\bar\nu}-\bar\nu\phi^2\right)du + \int_\mathrm{II}\left(r^2(\partial_v\phi)^2+\phi^2\right)dv+ \ve^3\tau(u_1,v_1)^{-3+\delta}.
\end{equation}

\textsc{Estimate for the flux along $\Gamma$}: On $\Gamma$, we have by definition
\begin{equation*}
     j^v|_\Gamma= r^2(\partial_u\phi)^2-\nu \phi^2-2r\phi\partial_u\phi,\quad j^u|_\Gamma=-r^2(\partial_v\phi)^2-\lambda \phi^2-2r\phi\partial_v\phi.
\end{equation*}
Therefore, using again that $n_\Gamma^u\sim -n_\Gamma^v\sim 1$, we have 
\begin{equation}\label{eq:mor-aux-9}
    -\int_\Gamma (n_\Gamma^uj^u+n_\Gamma^vj^v)ds\gtrsim \int_\Gamma \phi^2\,ds - \int_\Gamma\big((\partial_u\phi)^2+(\partial_v\phi)^2\big)ds.
\end{equation}

Combining \eqref{eq:mor-aux-5}, \eqref{eq:Mor-aux-7}, \eqref{eq:mor-aux-13}, \eqref{eq:Mor-aux-8}, \eqref{eq:mor-aux-9}, and \eqref{eq:Mor-Gamma-1}, we conclude \eqref{eq:mor-aux-11}.
\end{proof}

\begin{proof}[Proof of \cref{prop:Morawetz}]
    Combine \eqref{eq:Hardy-bulk-1}, \eqref{eq:Mor-compact}, and \eqref{eq:mor-aux-11}.
\end{proof}

\subsection{The \texorpdfstring{$\mathcal H^+$}{H+}-localized hierarchy}\label{sec:horizon-hierarchy}

In this section, we prove the $(\bar r-M)^{2-p}$-hierarchy of weighted energy estimates in the near-horizon region. For the linear wave equation on extremal Reissner--Nordstr\"om, this hierarchy was introduced by Aretakis for $p=0,1$, and $2$ in \cite{Aretakis-instability-1} and for $p\in[0,3)$ (for the $\ell=0$ mode) by the first-named author, Aretakis, and Gajic in \cite{angelopoulos2020late}. For the following statement, recall the notations $p_\star$ and $c_{p_\star}$  from \cref{def-of-energies}.

\begin{prop}\label{prop:H-h-1}
    For any $p\in[\delta,3-\delta]$, $A\ge 1$, $\ve$ sufficiently small, $\tau_f\in\mathfrak B$, $\alpha\in\mathfrak A_{I(\tau_f)}$, $(u_1,v_1),(u_2,v_2)\in\Gamma\cap\mathcal D_{\tau_f}$,  $u_2'>u_2$, and $v_2'>v_2$,  where $(u_2,v_2)$ is to the future of $(u_1,v_1)$, it holds that 
    \begin{multline}\label{eq:horizon-1}
        \int_\mathrm{V}(\bar r-M)^{2-p} \left(\frac{(\partial_u\psi)^2}{-\bar\nu}+\frac{(\partial_u\phi)^2}{-\bar\nu}\right)du +\int_\mathrm{V}c_{p_\star}(\bar r-M)^{-p_\star}\phi^2\,(-\bar\nu)du\\  +\iint_{\mathcal R_{\le\Lambda}} (\bar r-M)^{3-p} \left(\frac{(\partial_u\psi)^2}{-\bar\nu}+\frac{(\partial_u\phi)^2}{-\bar\nu}\right)dudv+\iint_{\mathcal R_{\le\Lambda}}c_{(p-1)_\star}(\bar r-M)^{-(p-1)_\star}\phi^2\,(-\bar\nu)dudv\\
        \les  \int_\mathrm{I}(\bar r-M)^{2-p} \left(\frac{(\partial_u\psi)^2}{-\bar\nu}+\frac{(\partial_u\phi)^2}{-\bar\nu}\right)du +\int_\mathrm{I}\phi^2\,(-\bar\nu)du+\int_\mathrm{II}\left(r^2(\partial_v\phi)^2+\phi^2\right)dv+\ve^3\tau_1^{-4+2\delta+p}.
    \end{multline}
\end{prop}

The proof of this estimate will be given in \cref{sec:hh-proof} below after a series of lemmas. 

\subsubsection{Near-horizon Hardy inequalities II}\label{sec:Hardy-2}

We will again require several inequalities to handle lower order terms. 

\begin{lem}\label{lem:horizon-Hardy-1} Let $\alpha>-1$.  With hypotheses as in \cref{prop:H-h-1}, it holds that
\begin{equation}
       (\alpha+1)^2 \iint_{\mathcal R_{\le\Lambda}} (\bar r-M)^\alpha\psi^2\,(-\bar\nu)dudv\les \iint_{\mathcal R_{\le\Lambda}} (\bar r-M)^{\alpha+2}\frac{(\partial_u\psi)^2}{-\bar\nu}\,dudv+\int_\Gamma \phi^2\,ds\label{eq:Hardy-h-3},
\end{equation}
  \begin{equation}(\alpha+1)^2 \int_{\underline S{}_v}(\bar r-M)^\alpha\psi^2\,(-\bar\nu)du\les_\alpha \int_{\underline S{}_v}\left((\bar r-M)^{\alpha+2}\frac{(\partial_u\psi)^2}{-\bar\nu}+(\bar r-M)^2\frac{(\partial_u\phi)^2}{-\bar\nu}-\bar\nu\phi^2\right)du, \label{eq:Hardy-h-5}
   \end{equation}
   where $\underline S{}_v$ is defined in \eqref{eq:underline-S-v}.
\end{lem}
\begin{proof}
    To prove \eqref{eq:Hardy-h-3}, apply \eqref{eq:hardy_h} with $f=\psi$ and argue as in the proof of \cref{lem:horizon-hardy-1}. To prove \eqref{eq:Hardy-h-5}, let $\chi=\chi(r)$ be a cutoff function such that $\chi(r)=1$ for $r\le \frac 13\Lambda$ and $\chi(r)=0$ for $r\ge \frac 23\Lambda$ and apply \eqref{eq:hardy_h} to $f=\chi\psi$. The cutoff kills the boundary term along $\Gamma$ (which has an unfavorable sign) and by adding $       \int_{\underline S{}_v}\phi^2\,(-\bar\nu)du$    to both sides of the resulting inequality, we conclude \eqref{eq:Hardy-h-5}.
\end{proof}

\begin{lem} With hypotheses as in \cref{prop:H-h-1}, it holds that
    \begin{equation}
         \int_{\underline{S}{}_v} (\bar r-M)^{2-p}\frac{(\partial_u\phi)^2}{-\bar\nu}\,du\les    \int_{\underline{S}{}_v} (\bar r-M)^{2-p}\left(\frac{(\partial_u\phi)^2}{-\bar\nu}-\bar\nu\psi^2\right)du ,
         \label{eq:Hardy-h-4}
    \end{equation}
    \begin{equation}
         \iint_{\mathcal R_{\le\Lambda}} (\bar r-M)^{2-p}\frac{(\partial_u\phi)^2}{-\bar\nu}\,dudv\les    \iint_{\mathcal R_{\le\Lambda}} (\bar r-M)^{2-p}\left(\frac{(\partial_u\phi)^2}{-\bar\nu}-\bar\nu\psi^2\right)dudv.
         \label{eq:Hardy-h-6}
    \end{equation}
\end{lem}
\begin{proof} This follows immediately from the trivial estimate
$(\partial_u\phi)^2\les (\partial_u\psi)^2+\bar\nu^2\psi^2$
    in $\mathcal R_{\le\Lambda}$.
\end{proof}

\subsubsection{Proof of the \texorpdfstring{$(\bar r-M)^{2-p}$}{(bar r-M) 2-p} estimates}\label{sec:hh-proof}

We begin with a simple consequence of Stokes' theorem in the plane.

\begin{lem}\label{lem:hh-proof-1}
    Let $j^u,j^v$, and $Q$ be smooth functions on $\mathcal R_{\le \Lambda}$ satisfying $\partial_uj^u+\partial_vj^v=Q$. Then the divergence identity 
    \begin{equation*}
        \iint_{\mathcal R_{\le \Lambda}} Q\,dudv= \int_\mathrm{V}j^v\,du-\int_\mathrm{I}j^v\,du+\int_\mathrm{VI}j^u\,dv-\int_\Gamma (n_\Gamma^uj^u+n_\Gamma^vj^v)ds
    \end{equation*} holds,
    where 
    \begin{equation*}
         n_\Gamma \doteq \big((\dot\Gamma^u)^2+(\dot\Gamma^v)^2\big)^{-1/2}(\dot\Gamma^v\partial_u-\dot\Gamma^u\partial_v) 
    \end{equation*}
    is as in \cref{lem:mor-proof-2}.
\end{lem}

We now prove the main estimate of \cref{prop:H-h-1}.

\begin{lem}\label{lem:horizon-main}
    With hypotheses as in \cref{prop:H-h-1}, it holds that
    \begin{multline}\label{eq:horizon-aux-5}
   \int_\mathrm{V}(\bar r-M)^{2-p}\frac{(\partial_u\psi)^2}{-\bar\nu}\,du+ \iint_{\mathcal R_{\le\Lambda}} (\bar r-M)^{3-p}\frac{(\partial_u\psi)^2}{-\bar\nu}\,dudv \les  \int_\mathrm{I}(\bar r-M)^{2-p}\frac{(\partial_u\psi)^2}{-\bar\nu}\,du \\
        + \int_\mathrm{I} \left((\bar r-M)^2\frac{(\partial_u\phi)^2}{-\bar\nu}-\bar\nu\phi^2\right)du + \int_\mathrm{II}\left(r^2(\partial_v\phi)^2+\phi^2\right)dv+ \ve^3\tau_1^{-4+2\delta+p}
    \end{multline}
      
\end{lem}
\begin{proof}
    We evaluate the expression
    \begin{equation*}
        \iint_{\mathcal R_{\le\Lambda}} \partial_v\left((\bar r-M)^{2-p}\frac{(\partial_u\psi)^2}{-\bar\nu}\right)dudv
    \end{equation*}
first by integrating by parts (using \cref{lem:hh-proof-1}) and then by using the wave equation \eqref{eq:wave-equation-psi}. This leads to the identity
\begin{multline*}
    \int_\mathrm{V}(\bar r-M)^{2-p}\frac{(\partial_u\psi)^2}{-\bar\nu}\,du-\int_\mathrm{I}(\bar r-M)^{2-p}\frac{(\partial_u\psi)^2}{-\bar\nu}\,du+\int_\Gamma n^v_\Gamma(\bar r-M)^{2-p}\frac{(\partial_u\psi)^2}{-\bar\nu}\,ds\\ 
   = -\iint_{\mathcal R_{\le\Lambda}} \left(\frac{2M+(p-2)\bar r}{\bar r^3}\right)(\bar r-M)^{3-p}\frac{(\partial_u\psi)^2}{-\bar\nu}\,dudv- \iint_{\mathcal R_{\le\Lambda}}\frac{4\kappa\nu}{r\bar\nu}(\bar r-M)^{2-p}\varkappa \psi\partial_u\psi \,dudv
\end{multline*}
Since $p\in(0,3)$, $2M+(p-2)\bar r>0$ for $\bar r$ close to $M$, so we may add a large multiple of the Morawetz estimate \eqref{eq:Mor-main} and use additionally \eqref{eq:Mor-Gamma-1} and \eqref{eq:mor-aux-11} to estimate
\begin{multline}
    \int_\mathrm{V}(\bar r-M)^{2-p}\frac{(\partial_u\psi)^2}{-\bar\nu}\,du+ \iint_{\mathcal R_{\le\Lambda}} (\bar r-M)^{3-p}\frac{(\partial_u\psi)^2}{-\bar\nu}\,dudv \\ \les \int_\mathrm{I}(\bar r-M)^{2-p}\frac{(\partial_u\psi)^2}{-\bar\nu}\,du  + \iint_{\mathcal R_{\le\Lambda}}(\bar r-M)^{2-p}|\varkappa| |\psi||\partial_u\psi| \,dudv + \text{RHS of \eqref{eq:Mor-main}}\label{eq:horizon-aux-1}
\end{multline}
Using Young's inequality with a parameter $b>0$, the geometric estimate \eqref{eq:varkappa-decay}, and the observation that $0<\bar\varkappa\les \bar r-M$ in $\mathcal R_{\le\Lambda}$, we estimate the error term on the right-hand side by
\begin{multline}
   \iint_{\mathcal R_{\le\Lambda}}(\bar r-M)^{2-p}|\varkappa| |\psi||\partial_u\psi| \,dudv \les b\iint_{\mathcal R_{\le\Lambda}} (\bar r-M)^{3-p}\frac{(\partial_u\psi)^2}{-\bar\nu}\,dudv +b^{-1}\iint_{\mathcal R_{\le\Lambda}} (\bar r-M)^{3-p}\psi^2\,(-\bar\nu)dudv\\+A^2\ve^{3/2}\iint_{\mathcal R_{\le\Lambda}}\tau^{-2+\delta}(\bar r-M)^{2-p} \left(b\frac{(\partial_u\psi)^2}{-\bar\nu}+b^{-1}(-\bar\nu)\psi^2\right)dudv. \label{eq:horizon-aux-2}
\end{multline}
The first term can be absorbed into the bulk on the right-hand side of \eqref{eq:horizon-aux-1} if $b$ is chosen sufficiently small. The second term may be estimated by the Hardy inequality \eqref{eq:Hardy-h-3},
\begin{equation*}
    \iint_{\mathcal R_{\le\Lambda}} (\bar r-M)^{3-p}\psi^2\,(-\bar\nu)dudv\les \iint_{\mathcal R_{\le\Lambda}} (\bar r-M)^{5-p}\frac{(\partial_u\psi)^2}{-\bar\nu}\,dudv+\int_\Gamma\phi^2\,ds,
\end{equation*}
and since $5-p>3-p$, we may use \eqref{eq:Mor-Gamma-1} and add a large multiple of the Morawetz estimate \eqref{eq:Mor-main} to handle the error. The third term in \eqref{eq:horizon-aux-2} is estimated using the bootstrap assumption \eqref{eq:boot-horizon},
\begin{equation}
    A^2\ve^{3/2}\iint_{\mathcal R_{\le\Lambda}}\tau^{-2+\delta}(\bar r-M)^{2-p}\frac{(\partial_u\psi)^2}{-\bar\nu}\,dudv\les A^2\ve^{3/2}\int_{\tau_1}^{\tau_2}\tau^{-2+\delta}\underline{\mathcal E}{}_p(\tau)\,d\tau \les \ve^3\tau_1^{-4+2\delta+p}.
\end{equation}
The fourth term in \eqref{eq:horizon-aux-2} is estimated by first applying the Hardy inequality \eqref{eq:Hardy-h-5} and then the bootstrap assumption \eqref{eq:boot-horizon},
\begin{multline}
       A^2\ve^{3/2}\iint_{\mathcal R_{\le\Lambda}}\tau^{-2+\delta}(\bar r-M)^{2-p}\psi^2\,(-\bar\nu)dudv\\\les A^2\ve^{3/2} \int_{v_1}^{v_2}\tau^{-2+\delta}\int_{\underline S{}_v}\left((\bar r-M)^{4-p}\frac{(\partial_u\psi)^2}{-\bar\nu}+(\bar r-M)^2\frac{(\partial_u\phi)^2}{-\bar\nu}-\bar\nu\phi^2\right)dudv\\
       \les A^2\ve^{3/2} \int_{\tau_1}^{\tau_2}\tau^{-2+\delta}\underline{\mathcal E}{}_{\max\{p-2,0\}}(\tau)\,d\tau \les \ve^3\tau^{-4+2\delta+\max\{p-2,0\}}_1\les\ve^3\tau_1^{-4+2\delta+p}.\label{eq:horizon-aux-4}
\end{multline}
Combining \eqref{eq:horizon-aux-1}--\eqref{eq:horizon-aux-4} proves \eqref{eq:horizon-aux-5}.
\end{proof}

\begin{proof}[Proof of \cref{prop:H-h-1}]
    Combine \eqref{eq:horizon-aux-5} and \eqref{eq:Hardy-h-3}--\eqref{eq:Hardy-h-6}.
\end{proof}

\subsection{The \texorpdfstring{$\mathcal I^+$}{I+}-localized hierarchy}\label{sec:rp}

In this section, we prove Dafermos and Rodnianski's hierarchy of $r^p$-weighed energy estimates in the far region. Recall the notations $p_\star$ and $c_{p_\star}$ from \cref{def-of-energies}.

\begin{prop}\label{prop:I-h-1}
    For any $p\in[\delta,3-\delta]$, $A\ge 1$, $\ve$ sufficiently small, $\tau_f\in\mathfrak B$, $\alpha\in\mathfrak A_{I(\tau_f)}$, $(u_1,v_1),(u_2,v_2)\in\Gamma\cap\mathcal D_{\tau_f}$,  $u_2'>u_2$, and $v_2'>v_2$,  where $(u_2,v_2)$ is to the future of $(u_1,v_1)$, it holds that 
    \begin{multline}\label{eq:rp}
        \int_\mathrm{IV}\big(r^p(\partial_v\psi)^2+c_{p_\star}r^{p_\star+2}(\partial_v\phi)^2+c_{p_\star}r^{p_\star}\phi^2\big)\,dv \\+\iint_{\mathcal R_{\ge\Lambda}}\big(r^{p-1}(\partial_v\psi)^2+c_{(p-1)_\star}r^{(p-1)_\star+2}(\partial_v\phi)^2+c_{(p-1)_\star}r^{(p-1)_\star}\phi^2\big)\,dudv \\
        \les \int_\mathrm{II}\big(r^p(\partial_v\psi)^2+r^{2}(\partial_v\phi)^2+\phi^2\big)\,dv+\int_\mathrm{I}\left((\bar r-M)^2\frac{(\partial_u\phi)^2}{-\bar\nu}-\bar\nu\phi^2\right)du+ \ve^3\tau_1^{-4+2\delta}.
    \end{multline}
\end{prop}

The proof of this estimate will be given in \cref{sec:rp-proof} below after a series of lemmas.

\subsubsection{\texorpdfstring{$r^p$}{rp}-weighted Hardy inequalities}

We will again require several inequalities to handle lower order terms. The following identity is proved by expanding the term in square brackets and integrating by parts (see also \cite[Lemma 8.30]{luk2019strong}).

\begin{lem}  For any $\{u\}\times[v_1,v_2]\subset\mathcal D_{\tau_f}$, $f:\mathcal D_{\tau_f}\to\Bbb R$, and $\alpha\in\Bbb R$, it holds that
\begin{multline}\label{eq:hardy_v}
 \frac{(\alpha +1 )^2}{4} \int_{v_1}^{v_2}  r^{\alpha}f^2 \, \lambda dv +\int_{v_1}^{v_2} \frac{r^{\alpha}}{\lambda} \left[ r \partial_v f + \frac{\alpha+1}{2} \lambda f \right]^2 \, dv  \\ = \int_{v_1}^{v_2} \frac{r^{\alpha+2}}{\lambda} ( \partial_v f )^2 \, dv + \frac{\alpha +1}{2} ( r^{\alpha+1} f^2 ) (u,v_2) - \frac{\alpha+1}{2} ( r^{\alpha+1} f^2 ) ( u , v_1 ).
\end{multline}
\end{lem}

\begin{lem}\label{lem:null-infinity-Hardy-1} Let $\alpha<-1$.  With hypotheses as in \cref{prop:I-h-1}, it holds that
\begin{equation}
      (\alpha+1)^2  \iint_{\mathcal R_{\ge\Lambda}} r^\alpha\psi^2\,dudv\les\iint_{\mathcal R_{\ge\Lambda}} r^{\alpha+2}(\partial_v\psi)^2\,dudv+\int_\Gamma \phi^2\,ds\label{eq:Hardy-I-1},
\end{equation}
  \begin{equation} (\alpha+1)^2\int_{S_u}r^\alpha\psi^2\,dv\les \int_{S_u}
  \big(r^{\alpha+2}(\partial_v\psi)^2+\phi^2\big)\,dv, \label{eq:Hardy-I-2}
   \end{equation}
   where
   \begin{equation*}
      S_u\doteq  \{u\}\times[v^\Lambda(u),v_2'].
   \end{equation*}
\end{lem}
\begin{proof}
    To prove \eqref{eq:Hardy-I-1}, apply \eqref{eq:hardy_v} with $f=\psi$ to derive
    \begin{equation*}
        (\alpha+1)^2\int_{S_u}r^\alpha\psi^2\,dv\les \int_{S_u} r^{\alpha+2}(\partial_v\psi)^2\,dv + (r^{\alpha+1}\psi^2)(u,v^\Lambda(u)),
    \end{equation*}
    where we have used that the boundary term at $v_2'$ has a favorable sign. Now \eqref{eq:Hardy-I-1} is obtained by integrating this inequality in $u$. To prove \eqref{eq:Hardy-I-2}, $\chi=\chi(r)$ be a cutoff function such that $\chi(r)=1$ for $r\ge 3\Lambda$ and $\chi(r)=0$ for $r\le 2\Lambda$ and apply \eqref{eq:hardy_v} to $f=\chi\psi$. The cutoff kills the boundary term along $\Gamma$ and by adding $     \int_{S_u} \phi^2\, dv$ to both sides of the resulting inequality, we conclude \eqref{eq:Hardy-I-2}.
\end{proof}

\begin{lem} Let $\alpha<-1$. With hypotheses as in \cref{prop:I-h-1}, it holds that
    \begin{equation}
         \int_{S_u} r^{\alpha+2}(\partial_v\phi)^2\,dv\les    \int_{S_u} \big(r^\alpha (\partial_v\psi)^2+r^{\alpha-2}\psi^2\big)\,dv,
         \label{eq:Hardy-I-3}
    \end{equation}
    \begin{equation}
         \iint_{\mathcal R_{\ge\Lambda}}r^{\alpha+2}(\partial_v\phi)^2\,dudv \les    \iint_{\mathcal R_{\ge\Lambda}} \big(r^\alpha (\partial_v\psi)^2+r^{\alpha-2}\psi^2\big)\,dudv.
         \label{eq:Hardy-I-4}
    \end{equation}
\end{lem}
\begin{proof} This follows immediately from the trivial estimate
$(\partial_v\phi)^2\les r^{-2}(\partial_v\psi)^2+r^{-4}\psi^2$
    in $\mathcal R_{\le\Lambda}$.
\end{proof}

\subsubsection{Proof of the \texorpdfstring{$r^p$}{rp} estimates}\label{sec:rp-proof}

We now prove the main estimate of \cref{prop:I-h-1}.

\begin{lem}
    With hypotheses as in \cref{prop:I-h-1}, it holds that
    \begin{multline}\label{eq:rp-aux-3}
\int_\mathrm{IV}r^p(\partial_v\psi)^2\,dv+\iint_{\mathcal R_{\ge\Lambda}} r^{p-1}(\partial_v\psi)^2\,dudv  \les \int_\mathrm{II}r^p(\partial_v\psi)^2\,dv \\+ \int_\mathrm{I} \left((\bar r-M)^2\frac{(\partial_u\phi)^2}{-\bar\nu}-\bar\nu\phi^2\right)du + \int_\mathrm{II}\left(r^2(\partial_v\phi)^2+\phi^2\right)dv+ \ve^3\tau(u_1,v_1)^{-4+2\delta}.
    \end{multline}
\end{lem}
\begin{proof}
    We evaluate the expression
    \begin{equation*}
        \iint_{\mathcal R_{\ge\Lambda}} \partial_u\left(r^p(\partial_v\psi)^2\right)dudv
    \end{equation*}
   first by integrating by parts and then by using the wave equation \eqref{eq:wave-equation-psi}. This leads to the identity
    \begin{multline}
        \int_\mathrm{IV}r^p(\partial_v\psi)^2\,dv-\int_\mathrm{II}r^p(\partial_v\psi)^2\,dv+\int_\Gamma n^u_\Gamma r^p(\partial_v\psi)^2\,ds \\= -\iint_{\mathcal R_{\ge\Lambda}} pr^{p-1}(\partial_v\psi)^2\, (-\nu)dudv+\iint_{\mathcal R_{\ge\Lambda}} 4\kappa\nu\varkappa r^{p-1}\psi\partial_v\psi\,dudv.\label{eq:rp-aux-1}
    \end{multline}
    We estimate the second term by
    \begin{equation*}
      \left|\iint_{\mathcal R_{\ge\Lambda}} 4\kappa\nu\varkappa r^{p-1}\psi\partial_v\psi\,dudv\right|  \les\iint_{\mathcal R_{\ge\Lambda}} r^{p-3}\psi\partial_v\psi\les b\iint_{\mathcal R_{\ge\Lambda}}r^{p-1}(\partial_v\psi)^2\,dudv+b^{-1}\iint_{\mathcal R_{\ge\Lambda}} r^{p-5}\psi^2\,dudv 
    \end{equation*}
    where $b>0$ is an arbitrary constant. For $b$ sufficiently small, this can be absorbed by the good bulk on the right-hand side of \eqref{eq:rp-aux-1}. To estimate the lower order term, we apply Hardy's inequality \eqref{eq:Hardy-I-1} with $\alpha=p-5<-1$ to obtain 
    \begin{equation*}
        \iint_{\mathcal R_{\ge\Lambda}} r^{p-5}\psi^2\,dudv \les \iint_{\mathcal R_{\ge\Lambda}} r^{p-3}(\partial_v\psi)^2\,dv + \int_\Gamma \phi^2\,ds.
    \end{equation*}
    Since $p-3<p-1$, for any sufficiently small $c>0$ we have the estimate
    \begin{equation}\label{eq:rp-aux-2}
         \iint_{\mathcal R_{\ge\Lambda}} r^{p-3}(\partial_v\psi)^2\,dv\les c \iint_{\mathcal R_{\ge\Lambda}} r^{p-1}(\partial_v\psi)^2\,dv+c^{-1}\iint_{\mathcal R_{\ge\Lambda}}\left( r^{1-\eta}(\partial_v\phi)^2+r^{-1-\eta}\phi^2\right)dudv.
    \end{equation}
    Combining \eqref{eq:rp-aux-1} and \eqref{eq:rp-aux-2} with $c$ sufficiently small and applying the Morawetz estimates \eqref{eq:Mor-Gamma-1} and \eqref{eq:mor-aux-11} yields \eqref{eq:rp-aux-3} as desired.
\end{proof}

\begin{proof}[Proof of \cref{prop:I-h-1}]
    Combine \eqref{eq:rp-aux-3} and \eqref{eq:Hardy-I-1}--\eqref{eq:Hardy-I-4}.
\end{proof}

\section{Boundedness and decay of the scalar field}\label{sec:decay-energy}

In this section, we prove boundedness and decay estimates for the scalar field $\phi$. In \cref{sec:improving-energy}, we use a standard pigeonhole principle argument to prove decay of the energies $\mathcal E_p$, $\underline{\mathcal E}{}_p$, $\mathcal F$, and $\underline{\mathcal F}$, which proves \cref{prop:energy-est-improve}. In \cref{sec:pointwise}, we then use these energy estimates to prove pointwise decay and boundedness for $\phi$ and its derivatives. 

\begin{center}
    \emph{In this section, we adopt the notational conventions outlined in \cref{sec:notation}. We will also often use the geometric bootstrap assumptions and their consequences, such as \cref{lem:basic-consequences}, without comment.}
\end{center}

\subsection{Improving the bootstrap assumptions for the scalar field: the proof of \texorpdfstring{\cref{prop:energy-est-improve}}{fix me}}\label{sec:improving-energy}

Recall the definition of the energies $\mathcal E_p$, $\underline{\mathcal E}{}_p$, $\mathcal F$, and $\underline{\mathcal F}$ from \cref{sec:anchor}. With this compact notation, we can summarize the content of \cref{prop:Kodama-estimate,prop:H-h-1,prop:I-h-1} as:
\begin{prop}
     For any $A\ge 1$, $\ve$ sufficiently small, $\tau_f\in\mathfrak B$, $\alpha\in\mathfrak A_{I(\tau_f)}$, and $1\le\tau_1\le \tau_2\le\tau_f$, it holds that 
     \begin{equation}\label{eq:master-boundedness}
        \sup_{\tau\in[\tau_1,\tau_2]}\big(\mathcal E_p(\tau)+\underline{\mathcal E}{}_p(\tau)\big)\les \mathcal E_p(\tau_1)+\underline{\mathcal E}{}_p(\tau_1)+\ve^3\tau_1^{-4+2\delta+p}
     \end{equation}
     for every $p\in [0,3-\delta]$ and
     \begin{equation}\label{eq:master-integral}
         \int_{\tau_1}^{\tau_2}\big(\mathcal E_{p-1}(\tau)+\underline{\mathcal E}{}_{p-1}(\tau)\big)d\tau \les \mathcal E_p(\tau_1)+\underline{\mathcal E}{}_p(\tau_1)+\ve^3\tau_1^{-4+2\delta+p}
     \end{equation}
     for every $p\in[1,3-\delta]$.
\end{prop}

We shall require the following interpolation lemma:
\begin{lem}\label{lem:interpolation}
    Let $0\le p_1<p<p_2\le 3-\delta$ and $\tau\in[1,\tau_f]$. It then holds that \begin{align*}
      \mathcal E_p(\tau)&\les \big(\mathcal E_{p_1}(\tau))^\frac{p_2-p}{p_2-p_1}\big(\mathcal E_{p_2}(\tau))^\frac{p-p_1}{p_2-p_1},\\
     \underline{\mathcal E}{}_p(\tau) &\les \big(\underline{\mathcal E}{}_{p_1}(\tau))^\frac{p_2-p}{p_2-p_1}\big(\underline{\mathcal E}{}_{p_2}(\tau))^\frac{p-p_1}{p_2-p_1}.
  \end{align*}
\end{lem}
\begin{proof}
    This is a consequence of the following general inequality for a nonnegative function $w$ on a measure space $(X,\mu)$, which is  immediately obtained from H\"older's inequality:
    \begin{equation*}
        \int_X w^p\,d\mu\le\left(\int_X w^{p_1}\,d\mu\right)^\frac{p_2-p}{p_2-p_1}\left(\int_X w^{p_2}\,d\mu\right)^\frac{p-p_1}{p_2-p_1}.\qedhere
    \end{equation*}
\end{proof}

The last preparatory result we shall require is an estimate for the initial energy in terms of our seed data norm $\mathfrak D$.

\begin{lem}\label{lem:initial-energy-est} For any $A\ge 1$, $\ve$ sufficiently small, $\tau_f\in\mathfrak B$, and $\alpha\in\mathfrak A_{I(\tau_f)}$, it holds that 
    \begin{equation}
        \mathcal E_p(1)+\underline{\mathcal E}{}_p(1)\les \big( \mathfrak D[\mathcal S_0]\big)^2.\label{eq:data-estimate}
    \end{equation}
\end{lem}
\begin{proof} 
   Because $\lambda\sim 1$ on $C_0$  and $\nu\sim \bar\nu$ on $\underline C{}_0$, it follows from the definition of $\mathfrak D[\mathcal S_0]$ that
    \begin{equation}\label{eq:data-aux-1}
|r\phi|+r^2|\partial_v\phi|+r^2|\partial_v\psi|\les\mathfrak D[\mathcal S_0]
    \end{equation}
    on $C_0$ and
    \begin{equation}\label{eq:data-aux-2}
        |\phi|+\frac{|\partial_u\phi|}{-\bar\nu}+\frac{|\partial_u\psi|}{-\bar\nu}\les \mathfrak D[\mathcal S_0]
    \end{equation}
on $\underline C{}_0$. The bound \eqref{eq:data-estimate} for $\mathcal E_p(1)$ is now immediate. The bound for $\underline{\mathcal E}{}_p(1)$ follows once we observe that
    \begin{equation*}
        \int_{\underline C{}_0}(\bar r-M)^{2-p}(-\bar\nu)du \le \int_M^\Lambda (\bar r-M)^{2-p}\,d\bar r\les 1
    \end{equation*}
    for $p\in[0,3-\delta]$.
\end{proof}

\begin{proof}[Proof of \cref{prop:energy-est-improve}] First, we apply \eqref{eq:master-boundedness} to $[1,\tau_f]$ and use \eqref{eq:data-estimate} to infer
\begin{equation}\label{eq:decay-aux-1}
   \sup_{\tau\in[1,\tau_f]}\big(\mathcal E_{p}(\tau)+\underline{\mathcal E}{}_{p}(\tau)\big)\les \big(\mathfrak D[\mathcal S_0]\big)^2+\ve^3\les \ve^2 
\end{equation} for every $p\in[0,3-\delta]$ 
by definition of $\ve$. We assume $I(\tau_f)\ge 1$, else \eqref{eq:decay-aux-1} already suffices. For each $i\in \{0,\dotsc,I(\tau_f)-1\}$, we deduce from \eqref{eq:decay-aux-1} for $p=3-\delta$ and \eqref{eq:master-integral} for $p=3-\delta$ on the interval $[L_i,L_{i+1}]$, together with the pigeonhole principle, the existence of a $\tau'_i\in[L_i,L_{i+1}]$ such that 
\begin{equation}
    \mathcal E_{2-\delta}(\tau_i')+\underline{\mathcal E}{}_{2-\delta}(\tau_i')\les (L_{i+1}-L_i)^{-1}\left(\ve^2+\ve^3 L_i^{-1+\delta}\right)\les \ve^2 L_i^{-1}.\label{eq:decay-aux-2}
\end{equation}
Using now \eqref{eq:master-boundedness} for $p=2-\delta$ and the dyadicity of the sequence $L_i$, we upgrade \eqref{eq:decay-aux-2} to the statement that
\begin{equation}
     \mathcal E_{2-\delta}(\tau)+\underline{\mathcal E}{}_{2-\delta}(\tau)\les \ve^2 \tau^{-1}\label{eq:decay-aux-3}
\end{equation}
for every $\tau\in[1,\tau_f]$. Using \cref{lem:interpolation} to interpolate between \eqref{eq:decay-aux-1} for $p=3-\delta$ and \eqref{eq:decay-aux-3}, we obtain
\begin{equation}
     \mathcal E_{2}(\tau)+\underline{\mathcal E}{}_{2}(\tau) \les \ve^2\tau^{-1+\delta}\label{eq:decay-aux-4}
\end{equation}
for every $\tau\in[1,\tau_f]$. We now apply \eqref{eq:decay-aux-4} and \eqref{eq:master-integral} with $p=2$ to the intervals $[L_i,L_{i+1}]$ with $i\in\{0,\dotsc,I(\tau_f)-1\}$ to find $\tau_i''\in[L_i,L_{i+1}]$ such that
\begin{equation*}
     \mathcal E_{1}(\tau_i'')+\underline{\mathcal E}{}_{1}(\tau_i'')\les (L_{i+1}-L_i)^{-1}\big( \mathcal E_{2}(L_i)+\underline{\mathcal E}{}_{2}(L_i)+\ve^3L_i^{-2+2\delta}\big)\les \ve^2 L_i^{-2+\delta}. 
\end{equation*}
Using \eqref{eq:decay-aux-1}, this is again immediately upgraded to 
\begin{equation*}
      \mathcal E_{1}(\tau)+\underline{\mathcal E}{}_{1}(\tau)\les \ve^2\tau^{-2+\delta}
\end{equation*}
for every $\tau\in[1,\tau_f]$. Repeating this argument once more, we conclude 
\begin{equation*}
      \mathcal E_{0}(\tau)+\underline{\mathcal E}{}_{0}(\tau)\les \ve^2\tau^{-3+\delta}
\end{equation*}
for every $\tau\in[1,\tau_f]$. Interpolating between this estimate and \eqref{eq:decay-aux-1} for $p=3-\delta$, we finally have
\begin{equation}
      \mathcal E_{p}(\tau)+\underline{\mathcal E}{}_{p}(\tau)\les \ve^2\tau^{-3+\delta+p}\label{eq:decay-aux-5}
\end{equation}
for every $p\in[0,3-\delta]$ and $\tau\in[1,\tau_f]$.

Applying now \cref{lem:kodama-auxiliary}, we estimate
\begin{equation}\label{eq:decay-aux-7}
    \mathcal F(u,\tau)\les \ve^2\tau^{-3+\delta},\quad \underline{\mathcal F}(v,\tau)\les \ve^2\tau^{-3+\delta}
\end{equation}
for any $(u,v)\in\mathcal D_{\tau_f}$ and $\tau\in[1,\tau_f]$

For $A$ sufficiently large, \eqref{eq:decay-aux-5} implies \eqref{ep_decay} and \eqref{uep_decay} and \eqref{eq:decay-aux-7} implies \eqref{f_decay} and \eqref{uf_decay}.
\end{proof}

\subsection{Pointwise estimates}\label{sec:pointwise}

In this section we will prove several pointwise estimates for the scalar field $\phi$ and its derivatives, with the energy decay estimates of the previous section as a starting point. We shall also use the method of characteristics to bound derivatives of $\phi$. At this point, we fix $A=A_0$ sufficiently large that both \cref{prop:geometry-est-improve,prop:energy-est-improve} hold.

\begin{prop}
For any $\ve$ sufficiently small, $\tau_f\in\mathfrak B$, $\alpha\in\mathfrak A_{I(\tau_f)}$, and $1\le\tau_1\le \tau_2\le\tau_f$, we have the pointwise decay estimates
\begin{align}
 \label{est:phi_bound}    |(\bar r-M)^{1/2}\phi|&\les \ve \tau^{-3/2+\delta/2},\\
\label{est:psi_bound}    |\psi|&\les \ve\tau^{-1+\delta/2}
\end{align}
and the pointwise boundedness estimates
\begin{align}
  \label{est:vpsi_bound}   |r^2\partial_v\psi|&\les\ve,\\
 \label{est:vphi_bound}   |r^2\partial_v\phi|&\les\ve,\\
\label{est:upsi_bound}    \left|\frac{\partial_u\psi}{-\nu}\right|&\les\ve,\\
  \label{est:uphi_bound}  \left|r\frac{\partial_u\phi}{-\nu}\right|&\les\ve
\end{align}
on $\mathcal D_{\tau_f}$.
\end{prop}
\begin{proof} \textsc{Proof of \eqref{est:phi_bound} and \eqref{est:psi_bound} for $r\le\Lambda$}: Let $\chi_1=\chi_1(r)$ be a smooth cutoff satisfying $\chi_1(r)=1$ for $r\le \Lambda-M$ and $\chi_1(r)=0$ for $r\ge\Lambda$. By the same argument as in \cref{lem:gamma-slope}, the set $\{r=\Lambda-M\}$ is a timelike curve in $\mathcal D_{\tau_f}$. For $(u,v)\in\mathcal D_{\tau_f}$ with $r(u,v)\le\Lambda-M$, the segment $[u^\Lambda(v),u]\times\{v\}$ is entirely contained in $\mathcal D_{\tau_f}$. Therefore, for $\beta\in\{0,1\}$, we may write
\begin{equation*}
     (\bar r - M)^{1-\beta} \phi^2 (u,v)= \int_{u^\Lambda(v)}^u\big(\chi_1'\nu(\bar r - M)^{1-\beta} \phi^2+(1-\beta)\chi_1\nu\phi^2+2\chi_1(\bar r - M)^{1-\beta} \phi \partial_u\phi \big)du'.
\end{equation*}
The first and second terms can both be estimated by $\underline{\mathcal E}{}_0(\tau(u,v))$ and for the third we use Cauchy--Schwarz, 
\begin{align*}
    \int_{u^\Lambda(v)}^u\chi_1(\bar r - M)^{1-\beta} |\phi| |\partial_u\phi| \,du'&\les \left(\int_{u^\Lambda(v)}^u\phi^2\,(-\bar\nu)du'\right)^{1/2}\left(\int_{u^\Lambda(v)}^u(\bar r-M)^{2-2\beta}\frac{(\partial_u\phi)^2}{-\bar\nu}\,du'\right)^{1/2}\\
    &\les \big(\underline{\mathcal E}{}_0(\tau(u,v))\underline{\mathcal E}{}_{2\beta}(\tau(u,v))\big)^{1/2}\les \ve^2 \tau^{-3+\beta+\delta}(u,v).
\end{align*}
For $(u,v)\in \{ \Lambda-M\leq r \leq \Lambda\} \cap  \mathcal D_{\tau_f}$, define
\begin{equation*}
    v^*(u)\doteq\begin{cases}
        v^{\Lambda-M}(u)& \text{if }r(u,0)\ge\Lambda-M\\
        0 & \text{otherwise}
    \end{cases}.
\end{equation*}
Note that $\{u\}\times [v^\ast(u),v] \subset \mathcal D_{\tau_f}$, $0\le v-v^\ast(u)\lesssim 1$ and $\tau(u,v^\ast(u)) \sim \tau(u,v)$. Therefore,
\begin{equation*}
    |\phi(u,v)| \leq  |\phi(u,v^\ast(u))|+ \int_{v^\ast(u)}^v |\partial_v \phi| \,dv'\lesssim \ve \tau^{-3/2 + \delta/2} +\big(\mathcal F(u, \tau(u,v^\ast(u))) \big)^{1/2}   \lesssim \ve \tau^{-3/2 + \delta/2}(u,v) ,
\end{equation*}
where we used that $(u,v^\ast(u)) \in \{ r \leq \Lambda-M\}\cup\underline C{}_0$ and  the estimate \eqref{eq:decay-aux-7}. This proves \eqref{est:phi_bound} and  \eqref{est:psi_bound} for all $r\le\Lambda$.
 
\textsc{Proof of \eqref{est:phi_bound} and \eqref{est:psi_bound} for $r\ge\Lambda$}:  Let $\chi_2=\chi_2(r)$ be a smooth cutoff satisfying $\chi_2(r)=1$ for $r\ge \Lambda+M$ and $\chi_2(r)=0$ for $r\le\Lambda$. By the same argument as in \cref{lem:gamma-slope}, the set $\{r=\Lambda+M\}$ is a timelike curve in $\mathcal D_{\tau_f}$. For $(u,v)\in\mathcal D_{\tau_f}$ with $r(u,v)\ge\Lambda+M$, the segment $\{u\}\times[v^\Lambda(u),v]$ is entirely contained in $\mathcal D_{\tau_f}$. Therefore, for $\beta\in\{0,1\}$, we have
\begin{align*}
   \left| r^{-\beta} \psi^2 (u,v)\right|&=\left| \int_{v^\Lambda(u)}^v\big(\chi_2'\lambda r^{-\beta}\psi^2-\beta \chi_2r^{-\beta-1}\lambda\psi^2+2\chi_2r^{-\beta}\psi\partial_v\psi\big)\,dv'\right|\\&\les \mathcal E_0(\tau(u,v)) + \big( \mathcal E_0(\tau(u,v)) \mathcal E_{2-2\beta}(\tau(u,v))\big)^{1/2}\les\ve^2\tau^{-2-\beta+\delta}(u,v),
\end{align*}
which proves \eqref{est:phi_bound} and \eqref{est:psi_bound} for $r\geq \Lambda+M$. The region $\Lambda\le r\le\Lambda+M$ is handled completely analogously to $\Lambda-M\le r\le\Lambda$ and is omitted. This completes the proofs of \eqref{est:phi_bound} and \eqref{est:psi_bound}.

\textsc{Proof of \eqref{est:vpsi_bound}}: Using \eqref{eq:wave-equation-psi}, we compute
\begin{equation*}
    \partial_u(r^2\partial_v\psi) = \frac{2\nu}{r}(r^2\partial_v\psi) + 2r\kappa\nu\varkappa\psi
\end{equation*}
which can be solved for
\begin{equation}\label{eq:bounded-aux-1}
   ( r^2\partial_v\psi)(u,v) = \frac{r^2(u,v)}{r^2(0,v)}  ( r^2\partial_v\psi)(0,v)+r^2(u,v)\int_0^u \left(\frac{2\kappa\nu\varkappa}{r}\psi\right)(u',v) \,du'
\end{equation} using an integrating factor (note that $\exp(\int_{u_1}^{u_2}\frac{2\nu}{r}(u',v)\,du')=r^2(u_2,v)/r^2(u_1,v)$). Since $r(u_2,v)\le r(u_1,v)$ for $u_1\le u_2$ and 
\begin{equation*}
    r^2(u,v)\int_0^u r^{-3}(u',v)\,(-\nu)du'\les 1,
\end{equation*}
we readily obtain \eqref{est:vpsi_bound} from \eqref{eq:data-aux-1}, \eqref{est:psi_bound},  \eqref{eq:bounded-aux-1}, and the geometric estimates. 

\textsc{Proof of \eqref{est:vphi_bound}}: This follows immediately from the identity
\begin{equation*}
    r^2\partial_v\phi=r\partial_v\psi-\lambda\phi
\end{equation*}
and the previously proved estimates. 

\textsc{Proof of \eqref{est:upsi_bound}}: Using  \eqref{eq:nu-v} and \eqref{eq:wave-equation-psi}, we compute
\begin{equation}\label{eq:Ypsi-negative-1}
    \partial_v\left(\frac{\partial_u\psi}{-\nu}\right)= -2\kappa\varkappa \left(\frac{\partial_u\psi}{-\nu}\right) - 2\kappa\varkappa\phi,
\end{equation}
which can be solved for
\begin{equation}\label{eq:Ypsi-negative-2}
    \left(\frac{\partial_u\psi}{-\nu}\right)(u,v)=\exp\left(-\int_0^v2\kappa\varkappa\,dv'\right)\left(\frac{\partial_u\psi}{-\nu}\right)(u,0)-\int_0^v\exp\left(-\int_{v'}^v 2\kappa\varkappa\,dv''\right) 2\kappa\varkappa\phi\,dv'
\end{equation}
using an integrating factor, where the integrands are evaluated at constant $u$. Using the simple observation that $\kappa\bar\varkappa\ge 0$ on $\mathcal D_{\tau_f}$ and the decay estimate \eqref{eq:varkappa-decay}, we have
\begin{equation*}
    \exp\left(-\int_{v_1}^{v_2}2\kappa\varkappa\,dv'\right)\le \exp\left(-\int_{v_1}^{v_2}2\kappa\tilde\varkappa\,dv'\right)\les 1+\int_{v_1}^{v_2}|\tilde\varkappa|\,dv'\les 1+\ve^{3/2}v_1^{-1+\delta}\les 1
\end{equation*}
for $\ve$ sufficiently small. Next, we estimate 
\begin{align*}
    \int_0^v |\varkappa|| \phi|\,dv'& \les \int_0^vr^{-3}\big((\bar r-M)|\phi|+|\tilde\varkappa||\phi|\big)\,dv'\\
    &\les \ve\int_0^v \big(r^{-3} \tau^{-3/2+\delta/2}+r^{-2}\tau^{-3+\delta}+r^{-3}\tau^{-2+\delta}\big)\,dv'\les \ve
\end{align*}
using \eqref{est:phi_bound} and \eqref{eq:varkappa-decay} again. Therefore, \eqref{eq:Ypsi-negative-2} and \eqref{eq:data-aux-2} yield \eqref{est:upsi_bound} as desired. 

\textsc{Proof of \eqref{est:uphi_bound}}: This follows immediately from the identity 
\begin{equation*}
    r\frac{\partial_u\phi}{-\nu}= \frac{\partial_u\psi}{-\nu}+\phi
\end{equation*} and the previously proved estimates. \end{proof}

\section{The proof of nonlinear stability, \texorpdfstring{\cref{thm:stability}}{Theorem 2}}\label{sec:proof-main}

In this section, we complete the proof of \cref{thm:stability}. As the proof will involve repeatedly updating the gauge as $\tau_f\to\infty$, we will now reintroduce the $\tau_f$ sub- and superscripts on various relevant quantities. Also, as in \cref{sec:pointwise}, we fix $A=A_0$ so that \cref{prop:geometry-est-improve,prop:energy-est-improve} hold. 

We now briefly recall the notation for our gauges and refer the reader back to \cref{sec:setup,sec:anchor} for the precise definitions. The coordinates $(\hat u,\hat v)$ refer to the initial data normalized coordinates on the maximal development $\hat{\mathcal Q}_\mathrm{max}$, which are transformed into the teleologically normalized coordinates $(u_{\tau_f},v)$ by the diffeomoprhism $\Phi_{\tau_f}=(\mathfrak u_{\tau_f},\mathfrak v):\hat{\mathcal D}_{\tau_f}\to\mathcal D_{\tau_f}$. The inverse of $\Phi_{\tau_f}$ is denoted by $\hat\Phi_{\tau_f}$. Note that the background extremal Reissner--Nordstr\"om solutions $\bar r_{\tau_f}$ are defined on $\mathcal D_{\tau_f}$ in the coordinates $(u_{\tau_f},v)$. Therefore, part of the proof of \cref{thm:stability} will be to show that the pullbacks $\bar r_{\tau_f}\circ\Phi_{\tau_f}$ converge on (an appropriate subset of) $\hat{\mathcal Q}_\mathrm{max}$. 

In \cref{sec:open-and-closed}, we carry out the main continuity argument of the paper by showing that the bootstrap set $\mathfrak B$ is open and closed. In \cref{sec:gauges}, we prove estimates comparing the gauges $\Phi_{\tau_f}$ and background solutions $\bar r_{\tau_f}$ for different values of $\tau_f$. In \cref{sec:limiting-argument}, we extract the limiting comparison solution $\bar r_\infty$ as $\tau_f\to\infty$ and show that the energy hierarchies extend to the limit. Finally, we complete the proof of \cref{thm:stability} in \cref{sec:putting-together}. 

\subsection{The continuity argument}\label{sec:open-and-closed}

Recall the definition of the bootstrap set $\mathfrak B(\mathcal S_0,\ve,A_0)$ from \cref{sec:bootstrap-definitions}. We now have the following fundamental statement: 

\begin{prop}\label{prop:continuity} Let $M_0$ and $\delta$ be as in the statement of \cref{thm:stability}. There exists an $\ve_\stab(M_0,\delta)>0$ such that if $\ve\le\ve_\stab$ and $\mathcal S_0\in\mathfrak M_0$ with $\mathfrak D[\mathcal S_0]\le \ve $, then $\mathfrak B(\mathcal S_0,\ve,A_0)=[1,\infty)$, where $A_0$ is the constant for which \cref{prop:geometry-est-improve,prop:energy-est-improve} hold.
\end{prop}

\begin{proof}
    We infer this statement by proving that $\mathfrak B(\mathcal S_0,\ve,A_0)\subset[1,\infty)$ is nonempty, open, and closed for $\ve$ sufficiently small. Nonemptiness was proved in \cref{prop:B-nonempty} above. Openness will be proved in \cref{sec:open} below. Closedness will be proved in \cref{sec:modulation} below.
\end{proof}

\begin{rk}
    The number $\ve_\stab$ will be restricted one final time in \cref{prop:final-decay-geometry} below.
\end{rk}

\subsubsection{The proof of openness}  \label{sec:open}

We begin by showing that the bootstrap region always stops short of the future boundary of the initial data hypersurface $\hat{\mathcal C}$ which was defined in \cref{sec:seed-data}. See \cref{fig:open}.

\begin{lem}\label{lem:enough-data}  There exists a constant $\theta\in(0,1)$ such that for $\ve$ sufficiently small and $\tau_f\in \mathfrak B(\mathcal S_0,\ve,A_0)$, it holds that  $\Gamma^{\hat u}(\tau_f)\le \theta U_*$. 
\end{lem}
\begin{proof}
    Since $\partial_{\hat u}r=-1$ on $\underline C{}_\ing$, we see that $r(\Gamma^{\hat u}(\tau_f),0)=\Lambda-\Gamma^{\hat u}(\tau_f)$. Using \eqref{eq:boot-r}, we then estimate
    \begin{equation*}
        \Gamma^{\hat u}(\tau_f)=\Lambda-r(\Gamma^{\hat u}(\tau_f),0)=\Lambda-\bar r_{\tau_f}\circ\Phi_{\tau_f}(\Gamma^{\hat u}(\tau_f),0)-\tilde r\circ\Phi_{\tau_f}(\Gamma^{\hat u}(\tau_f),0)=99M_0+O(\ve),
    \end{equation*}
    which is quantitatively strictly smaller than $U_*=\frac{995}{10}M_0$ for $\ve$ sufficiently small.
\end{proof}

\begin{proof}[Proof that $\mathfrak B$ is open]
Let $\tau_f\in \mathfrak B(\mathcal S_0,\ve,A_0)$. We show that $\tau_f+\eta\in \mathfrak B(\mathcal S_0,\ve,A_0)$ for $\eta>0$ sufficiently small. 

Recall that at bootstrap time $\tau_f$, the solution $(r,\hat\Omega^2,\phi,e)$ is assumed to exist on the rectangle 
\begin{equation*}
    \hat{\mathcal D}_{\tau_f}\doteq[0,\Gamma^{\hat u}(\tau_f)]\times[0,\Gamma^{\hat v}(\tau_f)]
\end{equation*}
in the maximal development $\hat{\mathcal Q}_\mathrm{max}$ in the ``initial data coordinates'' $(\hat u,\hat v)$. By \cref{prop:slab-existence} and \cref{lem:enough-data}, there exists a small number $\sigma>0$ such that 
\begin{equation*}
     \hat{\mathcal D}_{\tau_f}^\mathrm{ext}\doteq[0,\Gamma^{\hat u}(\tau_f)+\sigma]\times[0,\Gamma^{\hat v}(\tau_f)+\sigma]\subset\hat{\mathcal Q}_\mathrm{max}.
\end{equation*} 
Since $\Gamma$ being a timelike curve is an open condition, $\Gamma|_{[0,\tau_f]}$ extends to an inextendible timelike curve in $\hat{\mathcal D}^\mathrm{ext}_{\tau_f}$ for $\sigma$ sufficiently small (refer also to the quantitative estimates \eqref{eq:Gamma-slope}). Again by continuity, \eqref{eq:mu-lower-bound}, \eqref{eq:nu-1}, and \eqref{eq:kappa-bdd} imply that $\hat\kappa>0$ and $\hat\gamma<0$ in $\hat{\mathcal D}^\mathrm{ext}_{\tau_f}\cap\{r\ge \Lambda\}$ for $\sigma$ sufficiently small. Therefore, $\hat{\mathcal D}_{\tau_f+\eta}\subset\hat{\mathcal Q}_\mathrm{max}$ for $\eta$ sufficiently small and the map $\Phi_{\tau_f+\eta}$ exists, which allows us to equip $\hat{\mathcal D}_{\tau_f+\eta}$ with teleologically normalized coordinates $(u_{\tau_f+\eta},v)=\Phi_{\tau_f+\eta}(\hat u,\hat v)$. By direct inspection of the definition \eqref{eq:teleology-1}, we see that $ \mathfrak  u_{\tau_f+\eta} \to \mathfrak u_{\tau_f}$ smoothly on $[0,\Gamma^{\hat u}(\tau_f)]$ as $\eta\to 0$.  Moreover, letting $\bar r_{\tau_f+\eta}$ and $\bar r_{\tau_f}$ denote the respective anchored extremal Reissner--Nordstr\"om backgrounds, we also have that $\bar r_{\tau_f+\eta}\circ\Phi_{\tau_f+\eta}\to\bar r_{\tau_f}\circ\Phi_{\tau_f} $ smoothly on $\hat{\mathcal D}_{\tau_f}$ as $\eta\to 0$. 

The step function $I(\tau)=\lfloor \log_2\tau\rfloor$ has the property that $I(\tau_f+\eta)=I(\tau_f)$ for $\eta$ sufficiently small. It trivially follows that $\Pi_i$ is defined on $\mathfrak A_i$ for every $i\in\{0,\dotsc,I(\tau_f+\eta)\}$. These soft arguments show that points 1.--4.~of the definition of $\mathfrak B(\mathcal S_0,\ve,A_0)$ are satisfied for $\tau_f+\eta$ if $\eta$ is chosen sufficiently small. 

We now invoke \cref{prop:geometry-est-improve,prop:energy-est-improve}: The bootstrap assumptions \eqref{eq:boot-nu}--\eqref{eq:boot-flux-2} hold on $\mathcal D_{\tau_f}$ with quantitatively strictly better constants ($1\mapsto \frac12$). Using now the smoothness of the limits $\Phi_{\tau_f+\eta}\to \Phi_{\tau_f}$ and $\bar r_{\tau_f+\eta}\circ\Phi_{\tau_f+\eta}\to \bar r_{\tau_f}\circ\Phi_{\tau_f}$, it is now immediate that the bootstrap assumptions \eqref{eq:boot-nu}--\eqref{eq:boot-flux-2} hold on $\mathcal D_{\tau_f}$ for $\eta$ sufficiently small. Therefore, $\tau_f+\eta\in\mathfrak B(\mathcal S_0,\ve,A_0)$.
\end{proof}

\begin{figure}
\centering{
\def\svgwidth{12pc}
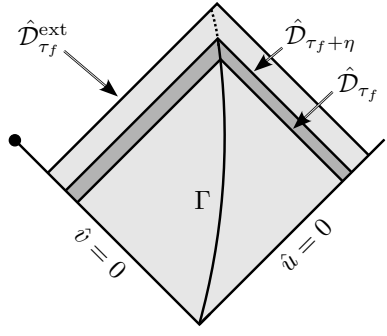}
\caption{A Penrose diagram of the openness argument in the proof of \cref{prop:continuity}. The extension $\hat{\mathcal D}_{\tau_f}^\mathrm{ext}$ avoids the solid black point by \cref{lem:enough-data}.}
\label{fig:open}
\end{figure}

\subsubsection{The proof of closedness: modulation of \texorpdfstring{$\varpi$}{varpi}}\label{sec:modulation}


\begin{proof}[Proof that $\mathfrak B$ is closed]
     Let $\tau_f^n\in \mathfrak B(\mathcal S_0,\ve,A_0)$ be a strictly increasing sequence of times with finite limit $\tau_f^\infty$ as $n\to\infty$. We aim to show that $\tau_f^\infty\in\mathfrak B(\mathcal S_0,\ve,A_0)$. 

\textsc{The case when $\tau_f^\infty$ is not dyadic}: We first argue that $\Gamma(\tau_f^n)$ has a limit point in the $(\hat u,\hat v)$-plane as $n\to\infty$. Since $\Gamma$ is timelike, the sequences $\Gamma^{\hat u}(\tau_f^n)$ and $\Gamma^{\hat v}(\tau_f^n)$ are strictly monotone increasing. By \cref{lem:enough-data}, $\Gamma^{\hat u}(\tau_f^n)$ converges to a number $\hat u_\star<U_*$. By \cref{lem:v-est} and \eqref{eq:tau-difference-v}, $\Gamma^{\hat v}(\tau_f^n)$ converges to a finite number $\hat v_\star$. Therefore, the set $\hat{\mathcal D}_\star\doteq [0,\hat u_\star)\times[0,\hat v_\star)$ is contained in $\hat{\mathcal Q}_\mathrm{max}$.

We now show that the closure of $\hat{\mathcal D}_\star$ is contained in $\hat{\mathcal Q}_\mathrm{max}$, i.e., that the solution $(r,\hat\Omega^2,\phi,e)$ extends to the closure of  $\hat{\mathcal D}_\star$ in $C^\infty$. By the bootstrap assumptions, \cref{lem:v-est}, \eqref{eq:lambda-lower-bound}, \eqref{eq:lambda-bdd}, \eqref{eq:varkappa-decay}, \eqref{est:psi_bound}, and \eqref{est:vphi_bound}, we immediately obtain the estimates $1\les r\les 1+ \tau_f^\infty$, $0<\hat\lambda\les 1$, $\hat\kappa\sim 1$, $|\phi|\les 1$, and $|\partial_{\hat v}\phi|\les 1$ on $\hat{\mathcal D}_\star$. Using these estimates and \eqref{eq:nu-v}, we have that $|{\log(-\hat\nu)}|\les 1+\tau_f^\infty$ on $\hat{\mathcal D}_\star$. We now use the gauge-invariant estimate \eqref{est:uphi_bound} to estimate $|\partial_{\hat u}\phi|\les 1+\tau_f^\infty$ on $\hat{\mathcal D}_\star$. Using the identity $\hat\Omega^2=-4\hat\kappa\hat\nu$, we also have $1\les\hat\Omega^2\les 1+\tau_f^\infty$ on $\hat{\mathcal D}_\star$. Now by integrating \eqref{eq:Omega-wave} in $\hat u$ and $\hat v$, we have that $|\partial_{\hat u}{\log\hat\Omega^2}|+|\partial_{\hat v}{\log\hat\Omega^2}|\les 1+(\tau_f^\infty)^2$ on $\hat{\mathcal D}_\star$. 

These arguments imply that $(r,\hat\Omega^2,\phi,e)$ extends to the closure of $\hat{\mathcal D}_\star$ in $C^1$. Standard propagation of regularity results now imply that this extension is actually $C^\infty$. Standard continuity arguments, such as those used in the proof of openness, imply that parts 2.~and 4.--6.~of \cref{def:B} hold on $\overline{\hat{\mathcal D}_\star}= \hat{\mathcal D}_{\tau_f^\infty}$. Since $\tau_f^\infty$ is not dyadic, i.e., not a power of 2, parts 1.~and 3.~of \cref{def:B} are automatically inherited from $\tau_f^n$ for $n$ sufficiently large.

\textsc{The case when $\tau_f^\infty$ is dyadic}: By the same arguments as in the previous case, the solution extends to $\hat{\mathcal D}_{\tau_f^\infty}$ and parts 2.~and 4.--6.~of \cref{def:B} are satisfied. Since $I\doteq I(\tau_f^\infty)>I(\tau_f^n)$ for any finite $n$, we must now construct the set $\mathfrak A_{I}$ and ensure that parts 1.~and 3. of \cref{def:B} hold. Note that $I\ge 1$. 

By assumption, there exist numbers $\alpha_i^-<\alpha_i^+$ for every $i\in\{0,\dotsc, I-1\}$ such that $\mathfrak A_i\doteq [\alpha_i^-,\alpha_i^+]$ are nested, i.e., $\mathfrak A_{I-1}\subset \mathfrak A_{I-2}\subset\cdots\subset\mathfrak A_0$, \begin{equation*}
    \Pi_i:\mathfrak A_i\to [-\ve^{3/2}L_i^{-3+\delta},\ve^{3/2}L_i^{-3+\delta}]
\end{equation*}
is surjective, and $\Pi_i(\alpha_i^\pm)=\pm \ve^{3/2}L_i^{-3+\delta}$. Clearly, the map $\Pi_I(\alpha)\doteq\tilde\varpi(\Gamma(L_I))$ is defined on $\mathfrak A_{I-1}$ since the solution exists on $\mathcal D_{\tau_f^\infty}$ for every $\alpha\in \mathfrak A_{I-1}$. We apply \cref{lem:inductive-mod} below with $f_1=\Pi_{I-1}$, $[x_1^-,x_2^+]=\mathfrak A_{I-1}$, $c_1=\ve^{3/2}L_{I-1}^{-3+\delta}$, $f_2=\Pi_I$, and $c_2=\ve^{3/2}L_I^{-3+\delta}$. In order to verify the assumption \eqref{eq:modulation-estimate}, we estimate, using \eqref{eq:varpi-aux-0} and \cref{rk:modulation-argument},
\begin{equation*}
    |\Pi_I(\alpha)-\Pi_{I-1}(\alpha)|\le C\ve^2 L_I^{-3+\delta}
\end{equation*}
for every $\alpha\in\mathfrak A_{I-1}$, where $C$ is a constant that does not depend on $\tau_f^\infty$. We now observe that
\begin{equation}\label{eq:algebra}
    \ve^{3/2} L_{I-1}^{-3+\delta}-\ve^{3/2}L_I^{-3+\delta} = (2^{3+\delta}-1)\ve^{3/2}L_I^{-3+\delta}>C\ve^2 L_I^{-3+\delta}
\end{equation}
for $\ve$ sufficiently small, which verifies \eqref{eq:modulation-estimate}. We may now take $\alpha_I^\pm = x_2^\pm$ and \cref{lem:inductive-mod} implies that
\begin{equation*}
   \Pi_I:\mathfrak A_I\to[-\ve^{3/2}L_I^{-3+\delta},\ve^{3/2}L_I^{-3+\delta}]
\end{equation*} is surjective,
where $\mathfrak A_I:=[\alpha_I^-,\alpha_I^+]$, as desired. \end{proof}

\begin{lem}\label{lem:inductive-mod} Let $0<c_2<c_1<0$ and let $f_1:[x_1^-,x_1^+]\to [-c_1,c_1]$ be a continuous surjective function satisfying $f_1(x_1^\pm)=\pm c_1$. Let $f_2:[x_1^-,x_1^+]\to \Bbb R$ be a continuous function satisfying the estimate
\begin{equation}
    \sup_{[x_1^-,x_1^+]}|f_2-f_1|<c_1-c_2.\label{eq:modulation-estimate}
\end{equation}
Then there exists an interval $[x_2^-,x_2^+]\subset (x_1^-,x_1^+)$ such that $f_2:[x_2^-,x_2^+]\to [-c_2,c_2]$ is surjective with $f_2(x_2^\pm)=\pm c_2$.
\end{lem}
\begin{proof} We observe that \begin{equation*}
    f_{2}(x_{1}^+)\ge f_1(x_1^+)-|f_{2}(x_1^+)-f_1(x_1^+)|> c_1-(c_1-c_2)=c_2
    \end{equation*}
    and similarly, $f_2(x_1^-) < -c_2$.
It follows from the intermediate value theorem that 
\begin{equation*}
    x_{2}^-\doteq \min f^{-1}_{2}(-c_{2})\quad\text{and}\quad x_{2}^+\doteq \max f^{-1}_{2}(c_{2})
\end{equation*}
exist and satisfy $x_1^-< x_2^- < x_2^+ < x_1^+ $. Surjectivity of $f_{2}$ on $[x_2^-,x_2^+]$ follows again from the intermediate value theorem.
\end{proof}

\subsubsection{The definition of the stable manifold \texorpdfstring{$\mathfrak M_\stab$}{M stab}}\label{sec:M-defn}


\begin{defn}\label{defn:stable-manifold}
    Let $\ve\le\ve_\stab$. For $\mathcal S_0\in \mathfrak M_0$ with $\mathfrak D[\mathcal S_0]\le \ve$, we have $\mathfrak B(\mathcal S_0,\ve,A_0)=[1,\infty)$ by \cref{prop:continuity}. Therefore, there exists a sequence of nested, compact, nonempty intervals $\{\mathfrak A_i\}_{i\ge 0}$ as in \cref{def:B}. Note that there may be multiple such $\{\mathfrak A_i\}_{i\ge 0}$ for which \cref{def:B} holds; we shall say that such an $\{\mathfrak A_i\}_{i\ge 0}$ is \emph{consistent} with this definition. See already \cref{rk:modulation}. We define
    \begin{equation}\label{def:M-stab-1}
         \mathfrak M_\stab(\mathcal S_0,\ve)\doteq \bigcup_{\{\mathfrak A_i\}_{i\ge 0}\text{ consistent}}\left\{\mathcal S_0(\alpha_\star)\in \mathcal L(\mathcal S_0,\ve):\alpha_\star\in\bigcap_{i\ge 0}\mathfrak A_i\right\}
    \end{equation}
    and 
    \begin{equation*}
    \mathfrak M_\stab\doteq \bigcup_{\ve\in[0,\ve_\stab]}\left(\bigcup_{\mathcal S_{0}\in\mathfrak M_0:\mathfrak D[\mathcal S_0]\le\ve} \mathfrak M_\stab(\mathcal S_0,\ve)\right).
\end{equation*}
\end{defn}

This set $\mathfrak M_\stab$ is the stable ``submanifold'' of seed data referred to in the statement of \cref{thm:stability}. We immediately infer the codimension-one property of $\mathfrak M_\stab$, recall \eqref{eq:intersection}. 

\begin{prop}\label{prop:codim-one} Let $\ve\le\ve_\stab$. For $\mathcal S_0\in \mathfrak M_0$ with $\mathfrak D[\mathcal S_0]\le \ve$, it holds that
    \begin{equation}\label{eq:stab}
    \mathfrak M_\stab\cap \mathcal L(\mathcal S_0,\ve)=  \mathfrak M_\stab(\mathcal S_0,\ve)\ne\emptyset.
\end{equation}
\end{prop}
\begin{proof}
The equality in \eqref{eq:stab} follows immediately from the definitions of $\mathfrak M_\stab$, $\mathcal L(\mathcal S_0,\ve)$, and $\mathfrak M_\stab(\mathcal S_0,\ve)$. The nonemptiness of $\mathfrak M_\stab(\mathcal S_0,\ve)$ follows from the existence of at least one valid sequence of modulation sets $\{\mathfrak A_i\}_{i\ge0}$ per \cref{def:B} and \cref{prop:continuity}.
\end{proof}

\begin{rk}\label{rk:modulation}
  We only prove that the intersection $\bigcap_{i\ge 0}\mathfrak A_i$ is nonempty, not that it only contains one element. Moreover, as was already mentioned, there might be multiple $\{\mathfrak A_i\}_{i\ge 0}$ consistent with \cref{def:B}. (The construction of \cref{lem:inductive-mod} gives one such choice, which is algorithmic but not necessarily unique in general.) In \eqref{def:M-stab-1}, we consider all possible descending chains of modulation sets which are consistent with \cref{def:B}. A priori, this could also result in the set $\mathfrak M_\stab(\mathcal S_0,\ve)$ containing more than one element.  One consequence of the conjectures in \cref{sec:separating} would be that $\mathfrak M_\stab(\mathcal S_0,\ve)$ contains a single seed data set and hence $\mathfrak M_\stab$ is \emph{exactly} codimension one.
\end{rk}

\subsection{Estimates for the gauge and background changes}\label{sec:gauges}

\subsubsection{Extension of the solution to \texorpdfstring{$\hat{\mathcal D}_\infty$}{D} and the eschatological gauge \texorpdfstring{$\Phi_\infty$}{Phi infinity}}\label{sec:final-gauge}

We now introduce the following convention:
\begin{center}
    \emph{For the remainder of this paper, we consider without further comment solutions arising from seed data in $\mathfrak M_\stab(\mathcal S_0,\ve)$. For such a  solution, all of the estimates proved in \cref{sec:geometry,sec:energy,sec:decay-energy} hold for every $\tau_f\ge 1$.}
\end{center}

Since $\mathfrak B(\mathcal S_0,\ve,A_0)=[1,\infty)$ by \cref{prop:continuity}, we have
    \begin{equation*}
      \hat{\mathcal D}_\infty\doteq  \bigcup_{\tau_f\ge 1}\hat{\mathcal D}_{\tau_f}\subset \hat{\mathcal Q}_\mathrm{max}.
    \end{equation*}
    In this section, we show that the \emph{eschatological gauge} $\Phi_\infty$
is well-defined and $C^2$ on $\hat{\mathcal D}_\infty$, which is a part of Part 2.~of the statement of \cref{thm:stability}.
\begin{prop}\label{prop:final-gauge}
        For $(\hat u,\hat v)\in\hat{\mathcal D}_\infty$, the limit
    \begin{equation*}
\Phi_\infty(\hat u,\hat v)\doteq \lim_{\tau_f\to\infty}\Phi_{\tau_f}(\hat u,\hat v)
    \end{equation*} exists and defines a $C^2$ diffeomorphism $\Phi_\infty:\hat{\mathcal D}_\infty\to [0,\infty)\times[0,\infty)$.
\end{prop}

First, we observe the following immediate consequence of \cref{lem:enough-data}:
\begin{lem}\label{lem:u-hat-H+}
  The limit $\hat u_{\mathcal H^+}\doteq\lim_{\tau\to\infty}\Gamma^{\hat u}(\tau)$ exists and satisfies $\hat u_{\mathcal H^+}\le \theta U_*$, where $\theta\in(0,1)$ is the constant from \cref{lem:enough-data}. Furthermore, $\hat{\mathcal D}_\infty=[0,\hat u_{\mathcal H^+})\times[0,\infty)$.
\end{lem}

Next, we show that the maps $\Phi_{\tau_f}$ are Cauchy in $C^2$.

\begin{lem}\label{lem:coord-change-1} For any $\hat u_0\in[0,\hat u_{\mathcal H^+})$, $\bar\tau_f$ sufficiently large, and $\tau_f\ge\bar\tau_f$, it holds that 
\begin{equation}
    \|\Phi_{\bar\tau_f}-\Phi_{\tau_f}\|_{C^2([0,\hat u_0]\times[0,\infty))}\les_{\hat u_0} \ve^2\bar\tau_f^{-1}.\label{eq:u-Cauchy}
\end{equation}
\end{lem}
\begin{proof} By definition and \eqref{eq:gamma-v} (applied twice), we have that
\begin{align}
      \nonumber  \mathfrak u_{\bar \tau_f}'(\hat u)-\mathfrak u_{\tau_f}'(\hat u) &= \hat\gamma(\hat u,\Gamma^{\hat v}(\tau_f))-\hat\gamma(\hat u,\Gamma^{\hat v}(\bar\tau_f))\\ \nonumber
    &=\hat\gamma(\hat u,\Gamma^{\hat v}(\bar\tau_f))\left[1-\exp\left(\int_{\Gamma^{\hat v}(\bar\tau_f)}^{\Gamma^{\hat v}(\tau_f)}\frac{r}{\hat\lambda}(\partial_{\hat v}\phi)^2\,d\hat v\right)\right]\\ \label{eq:diff-help-1}
        &=\frac{1}{1-\mu(\hat u,0)}\exp\left(\int_0^{\Gamma^{\hat v}(\bar \tau_f)}\frac{r}{\hat\lambda}(\partial_{\hat v}\phi)^2\,d\hat v\right) \left[1-\exp\left(\int_{\Gamma^{\hat v}(\bar\tau_f)}^{\Gamma^{\hat v}(\tau_f)}\frac{r}{\hat\lambda}(\partial_{\hat v}\phi)^2\,d\hat v\right)\right]
\end{align}
for $\hat u\in[0,\hat u_0]$. As the solution exists and is smooth on $\hat{\mathcal D}_{\bar\tau_f}$, we have that 
\begin{equation*}
    \left|\frac{1}{1-\mu(\hat u,0)}\exp\left(\int_0^{\Gamma^{\hat v}(\bar \tau_f)}\frac{r}{\hat\lambda}(\partial_{\hat v}\phi)^2\,d\hat v\right)\right|\le C(\hat u_0).
\end{equation*}
To see this, note that $1-\mu(\hat u,0)$ is nonvanishing for $\hat u\in[0,\hat u_{\mathcal H^+})$ by \eqref{eq:lambda-positive}. The integral over $r\le \Lambda$ is estimated softly by the compactness of the region $\{0\le \hat u\le\hat u_0\}\cap\{r\le\Lambda\}$ and for $r\ge \Lambda$ by \eqref{eq:boot-rp}. Using the pointwise estimate \eqref{est:vphi_bound}, we now estimate
\begin{equation*}
    \left|1-\exp\left(\int_{\Gamma^{\hat v}(\bar\tau_f)}^{\Gamma^{\hat v}(\tau_f)}\frac{r}{\hat\lambda}(\partial_{\hat v}\phi)^2\,d\hat v\right)\right|\les \int_{\Gamma^{\hat v}(\bar\tau_f)}^{\Gamma^{\hat v}(\tau_f)}\ve^2 r(\hat u,\hat v)^{-3}\,d\hat v \les \ve^2\bar 
    \tau_f^{-2}
\end{equation*} for $\bar \tau_f$ sufficiently large. 
Putting these estimates together yields $| \mathfrak u_{\bar \tau_f}'-\mathfrak u_{\tau_f}'|\les_{\hat u_0}\ve^2\bar\tau_f^{-2}$ and integrating once in $\hat u$ yields again $|\mathfrak u_{\bar\tau_f}-\mathfrak u_{\tau_f}|\les_{\hat u_0}\ve^2\bar \tau_f^{-2}$.

To estimate $\mathfrak u_{\bar \tau_f}''-\mathfrak u_{\tau_f}''$, we differentiate \eqref{eq:diff-help-1} to obtain \begin{multline}\label{eq:diff-helper-3}
     \mathfrak u_{\bar \tau_f}''(\hat u)-\mathfrak u_{\tau_f}''(\hat u) = \frac{\partial_{\hat u}\mu(\hat u,0)}{1-\mu(\hat u,0)}\big(\mathfrak u_{\bar \tau_f}'(\hat u)-\mathfrak u_{\tau_f}'(\hat u)\big) + \big(\mathfrak u_{\bar \tau_f}'(\hat u)-\mathfrak u_{\tau_f}'(\hat u)\big) \int_0^{\Gamma^{\hat v}(\bar \tau_f)}\partial_{\hat u}\left(\frac{r}{\hat\lambda}(\partial_{\hat v}\phi)^2\right)dv' \\ -\frac{1}{1-\mu(\hat u,0)}\exp\left(\int_0^{\Gamma^{\hat v}( \tau_f)}\frac{r}{\hat\lambda}(\partial_{\hat v}\phi)^2\,d\hat v\right) \int_{\Gamma^{\hat v}(\bar\tau_f)}^{\Gamma^{\hat v}(\tau_f)} \partial_{\hat u}\left(\frac{r}{\hat\lambda}(\partial_{\hat v}\phi)^2\right)dv'.
\end{multline}
Using the wave equations for $r$ and $\phi$, we compute
\begin{equation}\label{eq:diff-helper-2}
    \partial_{\hat u}\left(\frac{r}{\hat\lambda}(\partial_{\hat v}\phi)^2\right)= -\left(\frac{\hat\nu}{\hat\lambda}+ 2\frac{\hat\nu^2}{\hat \lambda^2}\hat\kappa\varkappa\right) (\partial_{\hat v}\phi)^2 - 2\partial_{\hat u}\phi \partial_{\hat v}\phi.
\end{equation}
Since $|\partial_{\hat v}{\log(-\hat\nu)}|\les r^{-2}$, it holds that $|\hat\nu|\les_{\hat u_0}1$ for $\hat u\le \hat u_0$. Therefore, using the pointwise estimates \eqref{est:vphi_bound} and \eqref{est:uphi_bound}, we find that the quantity in \eqref{eq:diff-helper-2} is, in magnitude, $\les_{\hat u_0} \ve^2 r^{-3}$. Therefore, applying the arguments used for $\mathfrak u_{\bar \tau_f}'-\mathfrak u_{\tau_f}'$ to \eqref{eq:diff-helper-3} proves that $| \mathfrak u_{\bar \tau_f}''-\mathfrak u_{\tau_f}''|\les_{\hat u_0}\ve^2\bar\tau_f^{-2}$ for $\bar\tau_f$ sufficiently large. This proves \eqref{eq:u-Cauchy}. 
\end{proof}

\begin{proof}[Proof of \cref{prop:final-gauge}]
    The existence and regularity of the map $\Phi_\infty=(\mathfrak u_\infty,\mathfrak v)$ is immediate from the estimate \eqref{eq:u-Cauchy}. We now show that $\mathfrak u_\infty:[0,\hat u_{\mathcal H^+})\to [0,\infty)$ is a diffeomorphism. First, $\mathfrak u_\infty$ is strictly increasing because
    \begin{equation*}
        \frac{d\mathfrak u_\infty}{d\hat u}=-\lim_{\hat v\to\infty}\hat\gamma(\hat u,\hat v)= \frac{1}{1-\mu(\hat u,0)}\exp\left(\int_0^\infty \frac{r}{\hat\lambda}(\partial_{\hat v}\phi)^2\,d\hat v\right)>0
    \end{equation*} and $\mathfrak u_\infty$ is surjective onto $[0,\infty)$ because $\mathfrak u_\infty(0)=0$ and $\mathfrak u_{\tau_f}(\Gamma^{\hat u}(\tau_f))\sim \tau_f$ as $\tau\to \infty$ by \cref{lem:tau-properties}. Finally, by \cref{lem:v-est} and its proof, $\mathfrak v:[0,\infty)\to[0,\infty)$ is a diffeomorphism. 
\end{proof}

We apply the coordinate transformation $(u_\infty,v)=\Phi_\infty(\hat u,\hat v)$ to $\hat{\mathcal D}_\infty$, after which we denote it by $\mathcal D_\infty$. However, since $\Phi_\infty$ is only $C^2$ in $\hat u$, the solution $(r,\hat\Omega^2,\phi,e)$ will only be $C^2$ in $u_\infty$ when expressed in the coordinates $(u_\infty,v)$. Moreover, the solution expressed in $(u_\infty,v)$ coordinates will satisfy the estimate
\begin{equation}\label{eq:gamma-final}
    |\gamma_\infty+1|\les \ve^2r^{-3/2}\tau^{-2+\delta}
\end{equation}
in $\mathcal D_\infty\cap\{r\ge\Lambda\}$ by simply passing to the limit in \eqref{eq:gamma-2}. 

\subsubsection{Uniform \texorpdfstring{$C^2$}{C2} convergence of \texorpdfstring{$\Phi_{\bar\tau_f}\circ\hat\Phi_{\tau_f}$}{Phi bar tau f circ hat Phi tau f} to the identity as \texorpdfstring{$\bar\tau_f\to\infty$}{infinity}}\label{sec:convergence-to-id}

As we saw in the proof of \cref{lem:coord-change-1}, the difference of the coordinate changes $\hat u\mapsto u_{\bar \tau_f}$ and $\hat u\mapsto  u_{\tau_f}$ does not satisfy a uniform estimate up to the horizon as $\bar\tau_f\to\infty$. In the following proposition, we prove that $u_{\bar\tau_f}\mapsto u_{\tau_f}$, however, is well-behaved in the entire domain of outer communication as $\bar\tau_f\to\infty$.  

\begin{prop} Let $1\le \bar\tau_f<\tau_f\le\infty$ and set $\Psi_{\bar\tau_f,\tau_f}\doteq \Phi_{\bar\tau_f}\circ\hat\Phi_{\tau_f}$. Then the following estimates hold:
\begin{align}
  \label{eq:Psi-1}    \sup_{J^-(\Gamma(\tau_f))\cap \mathcal D_{\tau_f}}  |\Psi_{\bar\tau_f,\tau_f}-\id|&\les \ve^2\bar\tau_f^{-1/2},\\
  \label{eq:Psi-2}          \sup_{J^-(\Gamma(\tau_f))\cap \mathcal D_{\tau_f}}  |d(\Psi_{\bar\tau_f,\tau_f}-\id)|&\les \ve^2\bar\tau_f^{-3/2}.
\end{align}
\end{prop}
\begin{proof} From the definitions, we have $\Psi_{\bar\tau_f,\tau_f}=\mathfrak g_{\bar\tau_f,\tau_f}\times\id$,   where \begin{equation}\label{eq:g-defn}
    \mathfrak g_{\bar\tau_f,\tau_f}(u_{\tau_f})\doteq -\int_0^{u_{\tau_f}} \gamma_{\tau_f}(u'_{\tau_f},\Gamma^v(\bar\tau_f))\,du'_{\tau_f}= u_{\tau_f} - \int_0^{u_{\tau_f}} \big(\gamma_{\tau_f}(u'_{\tau_f},\Gamma^v(\bar\tau_f))+1\big)\,du_{\tau_f}'.
\end{equation}
By applying the estimate \eqref{eq:gamma-2}, we have that
\begin{equation}
    | \mathfrak g_{\bar\tau_f,\tau_f}'-1|\le |\tilde\gamma_{\tau_f}| \les \ve^2 r^{-3/2}\bar\tau^{-2+\delta}_f,\label{eq:g-est-1}
\end{equation}
where the right-hand side is evaluated on $[0,\Gamma^{u_{\tau_f}}(\bar\tau_f)]\times\{\Gamma^v(\bar\tau_f)\}$. On $\{v=\Gamma^v(\bar\tau_f)\}\cap \{\tau\le \tfrac 12\bar\tau_f\}$, it is easy to see that $r\gtrsim \bar\tau_f$. By considering separately the two regions $\tau\ge \frac 12 \bar\tau_f$ and $\tau\le \frac 12 \bar\tau_f$ on $[0,\Gamma^{u_{\tau_f}}(\bar\tau_f)]\times\{\Gamma^v(\bar\tau_f)\}$, we use \eqref{eq:g-est-1} to estimate
\begin{equation}
     | \mathfrak g_{\bar\tau_f,\tau_f}'-1|\les \ve^2\bar\tau_f^{-3/2}.\label{eq:g-est-2}
\end{equation} 
This now readily implies \eqref{eq:Psi-1} and \eqref{eq:Psi-2}. \end{proof}

\subsubsection{Differences of the background solutions at different bootstrap times}

Next, we compare $\bar r_{\bar\tau_f}$ and $\bar r_{\tau_f}$ for two late bootstrap times. Since these are defined on different bootstrap domains with different coordinate systems, we need to pull them back onto the same domain. If we pull both back to $\hat{\mathcal D}_\infty$, we would have to contend with the poor estimate \eqref{eq:u-Cauchy}. Instead, we pull $\bar r_{\bar\tau_f}$ back to the ``later'' coordinate system $(u_{\tau_f},v)$, which takes advantage of the good estimate \eqref{eq:Psi-2}.

\begin{lem}\label{eq:background-r-difference-estimate} For any $\bar \tau\ge 1$ sufficiently large and $\bar\tau_f<\tau_f<\infty$, it holds that
    \begin{equation}\label{eq:varrho-estimate}
     \sup_{J^-(\Gamma(\tau_f))\cap\mathcal D_{\tau_f}}   |\bar r_{\tau_f}-\bar r_{\bar\tau_f}\circ\Psi_{\bar\tau_f,\tau_f}|\les \ve^{2}\bar\tau_f^{-1/2}.
    \end{equation}
\end{lem}
\begin{proof}
Let $\varrho\doteq \bar r_{\tau_f}-\bar r_{\bar\tau_f}\circ\Psi_{\bar\tau_f,\tau_f}$. We differentiate in $u_{\tau_f}$ to obtain
\begin{align*}
     \partial_{u_{\tau_f}}\varrho &= \bar\nu_{\tau_f}-(\bar\nu_{\bar\tau_f}\circ\Psi_{\bar\tau_f,\tau_f})\mathfrak g'_{\bar\tau_f,\tau_f}= -(1-\bar\mu_{\tau_f})+(1-\bar\mu_{\bar\tau_f}\circ\Psi_{\bar\tau_f,\tau_f})\mathfrak g'_{\bar\tau_f,\tau_f} \\
    & =\left(1-\frac{M}{\bar r_{\bar\tau_f}\circ\Psi_{\bar\tau_f,\tau_f}}\right)^2-\left(1-\frac{M}{\bar r_{\tau_f}}\right)^2+O(\ve^2\bar\tau_f^{-3/2})
\end{align*}
where we used \eqref{eq:g-est-2} and the Eddington--Finkelstein gauge condition for $\bar r_{\bar\tau_f}$ and $\bar r_{\tau_f}$. Under the bootstrap assumption that $|\varrho|\le \ve \bar \tau_f^{-1/3}$ with $\bar\tau_f$ large, we may Taylor expand to obtain
\begin{equation}
      \partial_{u_{\tau_f}}\varrho = -\frac{2M}{\bar r_{\tau_f}^3}(\bar r_{\tau_f}-M)\varrho +O(\bar r_{\tau_f}^{-3}\varrho^2)+O(\ve^2\bar\tau_f^{-3/2}).\label{eq:varrho-1}
\end{equation}
By \eqref{eq:tilde-r-on-Gamma}, we estimate $|\varrho(\Gamma(\bar\tau_f))|\les \ve^2 \bar\tau_f^{-3+\delta}$, which for $\bar \tau_f$ large shows that the bootstrap assumption is satisfied near $\Gamma(\bar\tau_f)$. Integrating \eqref{eq:varrho-1} backwards from $\Gamma(\bar\tau_f)$ (with the help of an integrating factor, see also \eqref{eq:integrating-factor-bounded}), we estimate 
\begin{equation}\label{eq:varrho-3}
|\varrho(u_{\tau_f},\Gamma^v(\bar\tau_f))|\les \ve^2\bar\tau_f^{-3+\delta}+\int^{\Gamma^{u_{\tau_f}}(\bar\tau_f)}_{u_{\tau_f}} \big(\bar r_{\tau_f}^{-3}\varrho^2+\ve^2\bar\tau_f^{-3/2}\big)\,du_{\tau_f}'\les \ve^2\bar\tau_f^{-1/2}+\ve^2 \bar \tau_f^{-2/3}\les \ve^2\bar \tau^{-1/2}_f.
\end{equation}
For $\bar \tau_f$ sufficiently large, this improves the bootstrap assumption on $\varrho$ and therefore \eqref{eq:varrho-3} holds for $u_{\tau_f}\in [0,\Gamma^{u_{\tau_f}}(\bar\tau_f)]$.

Next, similarly to \eqref{eq:varrho-1}, we derive the equation
\begin{equation*}
    \partial_v\varrho = \frac{2M}{\bar r^3_{\tau_f}}(\bar r_{\tau_f}-M)\varrho + O(\bar r^{-3}_{\tau_f}\varrho^2).
\end{equation*}
This can be integrated backwards from $[0,\Gamma^{u_{\tau_f}}(\bar\tau_f)]\times\{\Gamma^v(\bar\tau_f)\}$ and by using \eqref{eq:varrho-3} and the good sign of the coefficient of the linear term on the right-hand side, we obtain \eqref{eq:varrho-estimate} as desired. 
\end{proof}

\subsection{Taking \texorpdfstring{$\tau_f\to\infty$}{infinity} and completing the proof of \texorpdfstring{\cref{thm:stability}}{Theorem 1}}\label{sec:limiting-argument}

With estimates for the gauges and background solutions in hand, we can now justify passing to the limit $\tau_f\to\infty$ in the constructions and estimates established in \cref{sec:bootstrap,sec:geometry,sec:energy,sec:decay-energy}. Once the ``final'' estimates on $\mathcal D_\infty$ have been established, we can infer the existence of the event horizon $\mathcal H^+$ and complete the proof of \cref{thm:stability}.

\subsubsection{Extracting the final anchored background solution \texorpdfstring{$\bar r_\infty$}{bar r infinity}}

\begin{prop}\label{prop:r-convergence} The limit
\begin{equation}\label{eq:r-star-limit}
    \bar r_\star \doteq \lim_{\tau_f\to\infty} \bar r_{\tau_f}(0,0)
\end{equation} exists and satisfies
\begin{equation}
    |\bar r_\star-\Lambda|\les \ve^{2}.\label{eq:r-star-estimate}
\end{equation}
Let $\bar r_\infty$ denote the $\infty$-anchored background solution on $\mathcal D_\infty$ with bifurcation sphere area-radius $\bar r_\infty(0,0)=\bar r_\star$, as defined in \cref{sec:anchor}. Then it holds that
 \begin{equation}\label{eq:r-infty-est}
        \sup_{J^-(\Gamma(\tau_f))\cap\mathcal D_\infty}|\bar r_\infty-\bar r_{\tau_f}\circ\Psi_{\tau_f,\infty}|\les \ve^2 \tau_f^{-1/2}.
    \end{equation}
\end{prop}
\begin{proof}
The existence of the limit \eqref{eq:r-star-limit} follows immediately from \eqref{eq:varrho-estimate} evaluated at $(0,0)$ and then the estimate \eqref{eq:r-star-estimate} follows from \eqref{eq:tilde-r-on-Gamma}. To prove \eqref{eq:r-infty-est}, first observe that by \eqref{eq:varrho-estimate}, the limit
\begin{equation*}
    \mathfrak r(\hat u,\hat v)\doteq \lim_{\tau_f\to \infty}\bar r_{\tau_f}\circ\Phi_{\tau_f}
\end{equation*}
exists on $\hat{\mathcal D}_\infty$ and satisfies
\begin{equation}\label{eq:mathfrak-r-est}
    \sup_{J^-(\Gamma(\tau_f))\cap\hat{\mathcal D}_\infty}|\mathfrak r-r_{\tau_f}\circ\Phi_{\tau_f}|\les \ve^{3/2}\tau_f^{-1/2}.
\end{equation}
We claim that $\bar r_\infty=\mathfrak r\circ\hat\Phi_\infty$, whence \eqref{eq:r-infty-est} follows from \eqref{eq:mathfrak-r-est}.  As $\bar r_\infty(0,0)=\bar r_\star=\mathfrak r\circ\hat\Phi_\infty(0,0)$, it suffices to show that $\mathfrak r\circ\hat\Phi_\infty$ satisfies the Eddington--Finkelstein gauge conditions in the coordinates $(u_\infty,v)$. Indeed, 
\begin{equation}
    \partial_{u_\infty}(\bar r_{\tau_f}\circ\Psi_{\tau_f,\infty}) = -(1-\bar\mu_{\tau_f}\circ\Psi_{\tau_f,\infty})\mathfrak g'_{\tau_f,\infty}=-\left(1-\frac{M}{\bar r_{\tau_f}\circ\Psi_{\tau_f,\infty}}\right)^2\mathfrak g'_{\tau_f,\infty}
\end{equation}
by the definitions. Using \eqref{eq:g-est-2}, we see that this converges uniformly to $-(1-M/\mathfrak r\circ\hat\Phi_\infty)^2$ as $\tau_f\to\infty$, which verifies the Eddington--Finkelstein gauge condition for $\partial_{u_\infty} (\mathfrak r\circ\hat\Phi_\infty)$. The corresponding argument for $\partial_v (r\circ\hat\Phi_\infty)$ is trivial, which completes the proof.
\end{proof}

\begin{rk}
   Note that $\bar r_\infty$ is smooth on $\mathcal D_\infty$ in the $(u_\infty,v)$ coordinate system, but $\bar r_\infty\circ\Phi_\infty$ is only $C^2$ in $\hat u$ on $\hat{\mathcal D}_\infty$.
\end{rk}

\subsubsection{Convergence of the energy hierarchies}

\begin{prop}\label{prop:energy-convergence} For any $\tau\in[1,\infty)$ and $p\in[0,3-\delta]$, it holds that 
\begin{align*}
      \lim_{\tau_f\to\infty}\mathcal E_p^{\tau_f}(\tau)&=\mathcal E_p^{\infty}(\tau) ,\\    \lim_{\tau_f\to\infty}\underline{\mathcal E}{}_p^{\tau_f}(\tau)&=\underline{\mathcal E}{}_p^{\infty}(\tau) . 
\end{align*}
\end{prop}

\begin{proof}
We give the proof for the $\partial_u\phi$ term in $\underline{\mathcal E}{}_p^{\tau_f}(\tau)$ as every other term in any of the energies is either similarly or less difficult. We use the change of variables formula to write the $du_{\tau_f}$ integral as a $du_\infty$ integral: \begin{multline}   \int_{\underline C^{\tau_f}(\tau)} (\bar r_{\tau_f}-M)^{2-p}\frac{(\partial_{u_{\tau_f}}\phi_{\tau_f})^2}{-\bar\nu_{\tau_f}}\,du_{\tau_f} =  \int_{\underline C^{\tau_f}(\tau)}  (\bar r_{\tau_f}-M)^{2-p}\left(\frac{\partial_{u_{\tau_f}}\phi_{\tau_f}}{-\bar\nu_{\tau_f}}\right)^2(-\bar\nu_{\tau_f})du_{\tau_f} \\
     =\int_{[\Gamma^{u_\infty}(\tau),\Gamma^{u_\infty}(\tau_f)]\times\{v\}}\frac{(\bar r_{\tau_f}\circ\Psi_{\tau_f,\infty}-M)^{4-p} }{(\bar r_{\tau_f}\circ\Psi_{\tau_f,\infty})^{2}}\left(\frac{\partial_{u_\infty}\phi_\infty}{(1-\bar\mu_{\tau_f})\circ\Psi_{\tau_f,\infty}}\right)^2\,du_\infty,\label{eq:energy-conv-aux-1}
\end{multline}
where $v=(\Gamma^v)^{-1}(\tau)$. We shall use the dominated convergence theorem to prove that this integral converges to
\begin{equation}
    \int_{[\Gamma^{u_\infty}(\tau),\infty)\times\{v\}}\frac{(\bar r_\infty-M)^{4-p} }{\bar r_\infty^{2}}\left(\frac{\partial_{u_\infty}\phi_\infty}{1-\bar\mu_{\infty}}\right)^2\,du_\infty =  \int_{\underline C^{\infty}(\tau)} (\bar r_{\infty}-M)^{2-p}\frac{(\partial_{u_{\infty}}\phi_{\infty})^2}{-\bar\nu_{\infty}}\,du_{\infty}\label{eq:energy-conv-aux-2}
\end{equation}
as $\tau_f\to \infty$. 

It is clear from \cref{prop:r-convergence} that the integrand in \eqref{eq:energy-conv-aux-1} converges to the integrand in \eqref{eq:energy-conv-aux-2} pointwise, so we just have to prove a uniform bound by an $L^1$ function of $u_\infty$. For the scalar field, we estimate
\begin{equation}
\left|\frac{\partial_{u_\infty}\phi_\infty}{1-\bar\mu_{\infty}} \right|= \lim_{\tau_f\to\infty}\left|  \frac{\partial_{u_\infty}\phi_\infty}{(1-\bar\mu_{\tau_f})\circ\Psi_{\tau_f,\infty}} \right|=\lim_{\tau_f\to\infty} \left|\mathfrak g'_{\tau_f,\infty}\left(\frac{\partial_{u_{\tau_f}}\phi_{\tau_f}}{-\bar\nu_{\tau_f}}\right)\circ\Psi_{\tau_f,\infty}\right|\les \ve\label{eq:energy-conv-aux-3}
\end{equation}
by \eqref{est:uphi_bound} and \eqref{eq:g-est-2}. To estimate the degenerate factor on the horizon, we note that
\begin{equation*}
    \partial_{u_{\tau_f}}\big((\bar r_{\tau_f}-M)^{-1}\big) = \bar r^{-2}_{\tau_f},
\end{equation*}
from which it readily follows that
\begin{equation}
    \bar r_{\tau_f}(u_{\tau_f},v)-M \les \big(1+u_{\tau_f}-u_{\tau_f}^\Lambda(v)\big)^{-1}\label{eq:energy-conv-aux-4}
\end{equation}
for $(u_{\tau_f},v)\in\mathcal D_{\tau_f}\cap\{r\le\Lambda\}$. Therefore, combining \eqref{eq:energy-conv-aux-3} and \eqref{eq:energy-conv-aux-4}, we have
\begin{equation*}
  \left|  \frac{(\bar r_{\tau_f}\circ\Psi_{\tau_f,\infty}-M)^{4-p} }{(\bar r_{\tau_f}\circ\Psi_{\tau_f,\infty})^{2}}\left(\frac{\partial_{u_\infty}\phi_\infty}{(1-\bar\mu_{\tau_f})\circ\Psi_{\tau_f,\infty}}\right)^2\right|\les \ve^2 \big(1+\mathfrak g_{\tau_f,\infty}(u_\infty)-u_{\tau_f}^\Lambda(v)\big)^{p-4},
\end{equation*}
which is integrable as $p-4<-1$ (using also \eqref{eq:g-est-2} to say that $\mathfrak g_{\tau_f,\infty}(u_\infty)-u_{\tau_f}^\Lambda(v)\sim u_\infty$ for $u_\infty$ large).
\end{proof}

By \cref{prop:continuity,prop:energy-convergence}, the estimates \eqref{eq:boot-rp}--\eqref{eq:boot-flux-2} hold for $\tau_f=\tau_\infty$ and every $\tau\in[1,\infty)$. This proves \eqref{eq:nrg-decay-main} of \cref{thm:stability}.

\subsubsection{Global structure of \texorpdfstring{$\hat{\mathcal Q}_\mathrm{max}$}{Q max}}

We now prove the existence of a black hole region and a regular event horizon in the maximal development $(\hat{\mathcal Q}_\mathrm{max},r,\hat\Omega^2,\phi,e)$ of seed data lying in $\mathfrak M_\mathrm{stab}$.

\begin{prop}\label{prop:existence-of-horizon}
The maximal development of any seed data lying in $\mathfrak M_\mathrm{stab}$ has $\hat{\mathcal Q}_\mathrm{max}=[0,U_*]\times[0,\infty)$ and there exists a $\hat u_{\mathcal H^+}\in (0,U_*)$ satisfying
  \begin{equation}\label{eq:horizon-location-estimate}
        |\hat u_{\mathcal H^+}-\hat u_{\mathcal H^+,0}|\les \ve,
    \end{equation} where $\hat u_{\mathcal H^+,0}$ was defined in \eqref{eq:original-horizon}, such that 
  $  [0,\hat u_{\mathcal H^+}]\times[0,\infty)\subset\hat{\mathcal Q}_\mathrm{max},$
\begin{equation}\label{eq:null-infinity-2}
    \lim_{\hat v\to\infty}r(\hat u,\hat v)=\infty
\end{equation} for every $\hat u\in[0,\hat u_{\mathcal H^+})$, and 
\begin{equation}\label{eq:horizon-2}
    \lim_{\hat v\to\infty} r(\hat u_{\mathcal H^+},\hat v)=    \lim_{\hat v\to\infty}\varpi(\hat u_{\mathcal H^+},\hat v)=M=|e|.
\end{equation}
Therefore, $[0,\hat u_{\mathcal H^+})\times\{\hat v=\infty\}$ may be regarded as future null infinity $\mathcal I^+$ and 
\begin{equation*}
    \mathcal H^+\doteq J^-(\mathcal I^+)= \{\hat u_{\mathcal H^+}\}\times[0,\infty)
\end{equation*}
is the event horizon. The black hole region is
\begin{equation*}
    \mathcal{BH}\doteq \hat{\mathcal Q}_\mathrm{max}\setminus J^-(\mathcal I^+)=[\hat u_{\mathcal H^+},U_*]\times[0,\infty)\ne\emptyset
\end{equation*}
and future null infinity is complete in the sense of Christodoulou \cite{christodoulou1999global}.
\end{prop}
\begin{proof} 
By the geometric estimates, $r$ is bounded from below on $\hat{\mathcal D}_\infty$ and $\hat\lambda>0$ by \eqref{eq:lambda-positive}. It follows from the extension principle\footnote{Note that since we actually already have a pointwise estimate for $\phi$ on $\hat{\mathcal D}_\infty$ by \eqref{est:phi_bound}, one could give a much quicker argument for the extension.}  \cite{dafermos2005naked} that $\overline{\hat{\mathcal D}_\infty}\subset\hat{\mathcal{Q}}_\mathrm{max}$, where the closure is taken in the $(\hat  u,\hat v)$-plane. By \cref{lem:u-hat-H+}, $\overline{\hat{\mathcal D}_\infty}\setminus\hat{\mathcal D}_\infty=\{\hat u_{\mathcal H^+}\}\times[0,\infty)$. 

By \eqref{eq:boot-r} and \eqref{eq:r-infty-est}, it holds that 
\begin{equation*}
    |r-\bar r_\infty\circ\Phi_\infty|\les \ve^{3/2}\tau^{-2+\delta}
\end{equation*}
on $\hat{\mathcal D}_\infty$ and since
\begin{equation*}
    \lim_{\hat u\to \hat u_{\mathcal H^+}}\bar r_\infty\circ\Phi_\infty(\hat u,\hat v) = \lim_{u_\infty\to\infty} \bar r_\infty(u_\infty,\mathfrak v(\hat v))=M
\end{equation*}
by the geometry of extremal Reissner--Nordstr\"om, we conclude that
\begin{equation}\label{eq:horizon-location-estimate-2}
    |r(\hat u_{\mathcal H^+},\hat v)-M|\les \ve^{3/2}\tau^{-2+\delta}(\hat u_{\mathcal H^+},\hat v)
\end{equation}
for every $\hat v\in [0,\infty)$. Furthermore, by passing to the limit in \eqref{eq:boot-pi}, we have
\begin{equation}\label{eq:horizon-location-estimate-3}
    |\varpi-M|\les \ve^{3/2}\tau^{-3+\delta}
\end{equation}
on $\overline{\hat{\mathcal D}}_\infty$. Now \eqref{eq:horizon-location-estimate-2} and \eqref{eq:horizon-location-estimate-3} imply \eqref{eq:horizon-2}. 

Evaluating \eqref{eq:horizon-location-estimate-2} at $\hat v=0$ and using the gauge condition $\hat\nu(\hat u,0)=-1$, we find that 
\begin{equation*}
\hat u_{\mathcal H^+}-\hat u_{\mathcal H^+,0} = \hat u_{\mathcal H^+} - (M-\Lambda)+O(\ve) = r(\hat u_{\mathcal H^+},0) - M + O(\ve) = O(\ve),
\end{equation*}
which verifies \eqref{eq:horizon-location-estimate}.

To show that $\hat{\mathcal Q}_\mathrm{max}$ contains the rectangle $(\hat u_{\mathcal H^+},U_*]\times[0,\infty)$, we use the logically independent fact that $\hat\lambda\ge 0$ everywhere on $\hat{\mathcal Q}_\mathrm{max}$, which will be shown in \cref{sec:putting-together} below. Since $r$ is bounded below on $[0,U_*]\times\{0\}$, is it bounded below on $\hat{\mathcal Q}_\mathrm{max}$. The extension principle \cite{dafermos2005naked} therefore implies that $\hat{\mathcal Q}_\mathrm{max}=[0,U_*]\times[0,\infty)$.

Finally, since $r$ is bounded on $\mathcal H^+$, completeness of $\mathcal I^+$ follows from the work of Dafermos \cite{dafermos-trapped-surface} (see also \cite{Kommemi13}).\end{proof}

This completes the proof of Part 2.~of \cref{thm:stability}.

\subsubsection{The final geometric estimates}\label{sec:geometric-limits}

In this section, we record the final geometric estimates in the eschatological gauge $(u_\infty,v)$, with sharp improvements in $\ve$. This will prove \eqref{eq:geo-main-1}--\eqref{eq:geo-main-3} of \cref{thm:stability}.

While the $(u_\infty,v)$ coordinates do not cover the event horizon $\mathcal H^+$, we can formally attach it as the $u_\infty=\infty$ limiting curve (as is done in \cite{Price-law}, for instance). We write this ``extended'' manifold as $\mathcal D_\infty\cup\mathcal H^+$. On $\mathcal D_\infty\cup\mathcal H^+$, we can extend certain background quantities to $\mathcal H^+$ as explained in \cref{sec:geometry-RN}.

\begin{prop}\label{prop:final-decay-geometry} If $\ve_\stab$ is sufficiently small, then in the eschatological gauge $(u_\infty,v)$ with background solution $\bar r_\infty$, it holds that
\begin{equation}\label{eq:gamma-final-2}
    |\gamma_\infty+1|\les \ve^2r^{-1}\tau^{-3+\delta}
\end{equation}
on $\mathcal D_{\infty} \cap\{r\ge\Lambda\}$,
\begin{equation}
    \left|\frac{\nu_\infty}{\bar\nu_\infty}-1\right|\les \ve^{2}\tau^{-1+\delta} \label{eq:nu-final}
\end{equation}
on $\mathcal D_\infty$, and
\begin{align}
 |r-\bar r_\infty|&\les \ve^{2}\tau^{-2+\delta}, \label{eq:r-final}\\
    |\lambda_\infty-\bar\lambda_\infty|&\les \ve^{2}\tau^{-2+\delta},\label{eq:lambda-final}\\
        |\kappa_\infty-1|&\les \ve^2\tau^{-1+\delta},\label{eq:kappa-final}\\
    |\varpi -M|&\les \ve^2\tau^{-3+\delta},\label{eq:varpi-final-2}\\
    |(1-\mu)-(1-\bar\mu_\infty)|&\les \ve^{2}\tau^{-2+\delta},\label{eq:1-mu-final}\\
    | \varkappa - \bar{\varkappa}_{\infty} | & \lesssim \varepsilon^{2} r^{-2} \tau^{-3+\delta} + \varepsilon^{2} r^{-3} \tau^{-2+\delta}  \label{eq:varkappa-final}
\end{align}
on $\mathcal D_\infty\cup\mathcal H^+$.
\end{prop}

\begin{rk}
    Note that compared to the estimates proved in \cref{sec:geometry}, the estimates \eqref{eq:nu-final}, \eqref{eq:r-final}, \eqref{eq:lambda-final}, \eqref{eq:varpi-final-2}, \eqref{eq:1-mu-final}, and \eqref{eq:varkappa-final} have factors of $\ve^2$ instead of $\ve^{3/2}$. This improvement is the main content of the following proof. 

\end{rk}

\begin{proof}[Proof of \cref{prop:final-decay-geometry}]
The estimates \eqref{eq:gamma-final-2} and \eqref{eq:kappa-final} are obtained by simply passing to the limit in \cref{lem:kappa-bounds,lem:gamma}. To prove \eqref{eq:varpi-final-2}, we revisit the estimate \eqref{eq:varpi-aux-0} and apply it to $\alpha=\alpha_\star$. Since $\alpha_\star\in\bigcap_{i\ge 0}\mathfrak A_i$, it holds that $|\Pi_{I(L_i)}(\alpha_\star)|\le \ve^{3/2}L_i^{-3+\delta}$ for every $i\ge 0$. Therefore, by \eqref{eq:varpi-aux-0} we have $    |\varpi-M|\les \ve^{3/2}L_i^{-3+\delta} + \ve^2\tau^{-3+\delta}$
on $\mathcal D_{L_i}$. Sending $i\to \infty$ proves \eqref{eq:varpi-final-2} on $\mathcal D_\infty$.

We now prove the remaining estimates via a bootstrap argument in the $(u_\infty,v)$ gauge on $\mathcal D_\infty$. Let $\check{\mathfrak B}$ denote the set of $\tau_f\in[1,\infty)$ for which 
\begin{align}
 \label{eq:r-boot-2}   |r-\bar r_\infty|&\le \check A \ve^2\tau^{-2+\delta},\\
  \label{eq:nu-boot-2}  \left|\frac{\nu_\infty}{\bar\nu_\infty}-1\right|& \le \check A^2\ve^2\tau^{-1+\delta},
\end{align}
where $\check A\ge1$ is a constant to be determined. 

Let $\tau_f\in\check{\mathfrak B}$. Using these bootstrap assumptions and \eqref{eq:varpi-final-2}, we improve the Taylor expansions \eqref{eq:nu-expansion} and \eqref{eq:lambda-expansion} to
\begin{equation*}
    \partial_u(r-\bar r_\infty) = \frac{2M\gamma_\infty}{\bar r^3_\infty}(\bar r_\infty-M)(r-\bar r_\infty)+ O\big(\ve^2\bar r_\infty^{-1}\tau^{-3+\delta}\big)
\end{equation*}
in $\mathcal D_\infty\cap\{\tau\le\tau_f\}\cap\{r\ge\Lambda\}$ and 
\begin{equation*}
    \partial_v(r-\bar r_\infty)= \frac{2M\kappa_\infty}{\bar r_\infty^3}(\bar r_\infty-M)(r-\bar r_\infty) +O\big(\ve^2\tau^{-3+\delta}\big)
\end{equation*}
in $\mathcal D_\infty\cap\{\tau\le\tau_f\}$. By passing to the limit in \eqref{eq:tilde-r-on-Gamma}, we have
\begin{equation*}
    \big|(r-\bar r_\infty)|_\Gamma\big|\les \ve^2\tau^{-2+\delta}.
\end{equation*}
By now repeating the arguments of \cref{lem:r-est,lem:nu-estimate}, we improve the bootstrap assumptions \eqref{eq:r-boot-2} and \eqref{eq:nu-boot-2} for $\check A$ chosen sufficiently large. Therefore, $\check{\mathfrak B}=[1,\infty)$ and the estimates \eqref{eq:nu-final} and \eqref{eq:r-final} hold. 

Finally, \eqref{eq:1-mu-final} and \eqref{eq:varkappa-final} are proved by repeating the arguments of \cref{lem:mu,lem:varkappa} with our improved estimates at hand.
\end{proof}

Using \eqref{eq:varpi-final-2}, we can prove a sharp-in-$\ve$ estimate for the final modulation parameter $\alpha_\star$.     The following result can be compared with \cite[Remark 6.3.5]{DHRT}.

\begin{prop}
   Let $\mathcal S_0(\alpha_\star)=(\mathring\phi,r_0,M_0+\alpha_\star,e)\in \mathfrak M_\stab(\mathcal S_0,\ve)$. Then
    \begin{equation}
        \alpha_\star=|e|-|e_0|+O(\ve^2)\label{eq:alpha-star-estimate}
    \end{equation}
\end{prop}
\begin{proof} By \eqref{eq:varpi-final-2}, we have
\begin{equation*}
    \alpha_\star = \varpi(0,0)-M_0 = \varpi(0,0) - M -(M_0-M) = |e|-|e_0| +O(\ve^2).\qedhere
\end{equation*}
\end{proof}

\subsubsection{Putting everything together}\label{sec:putting-together}

Everything is now in place to complete the proof of nonlinear stability of Reissner--Nordstr\"om, \cref{thm:stability}. 

\begin{proof}[Proof of \cref{thm:stability}]
Part 1.~of the theorem is a direct restatement of \cref{prop:codim-one} above. Part 2.~of the theorem is a combination of \cref{prop:final-gauge,prop:existence-of-horizon}, and \eqref{eq:gamma-final-2} implies that $\nu_\infty\to -1$ at $\mathcal I^+$. Orbital stability of the parameters, \eqref{eq:parameter-stability}, follows immediately from the definition of $\mathfrak D$ and the conservation of charge $e$ in the neutral scalar field model. 

\textsc{Proof of orbital stability for the $p=3-\delta$ energy}: Unfortunately, this does not immediately follow from \cref{prop:energy-convergence} applied to $p=3-\delta$ since the smallness parameter $\ve\ge \mathfrak D[\mathcal S_0]$ only bounds the energy $\mathcal E^\infty_{3-\delta}(1)+\underline{\mathcal E}{}_{3-\delta}^\infty(1)$ from \emph{above} (recall \cref{lem:initial-energy-est}). Furthermore, the estimates \eqref{eq:horizon-1} and \eqref{eq:rp} for $p=3-\delta$ have nonlinear errors that depend on $\ve$ and are not bounded by $\mathcal E^\infty_{3-\delta}(1)+\underline{\mathcal E}{}_{3-\delta}^\infty(1)$ as the estimates are currently written. We now sketch how to remove this deficiency. Working on $\mathcal D_\infty$ in the eschatological gauge $(u_\infty,v)$, we define a bootstrap set $\breve{\mathfrak B}$ consisting of $\tau_f\in[1,\infty)$ such that the bootstrap assumptions \eqref{eq:boot-nu}--\eqref{eq:boot-flux-2} hold with $\ve$ replaced by $\breve\ve\doteq \big(\mathcal E^\infty_{3-\delta}(1)+\underline{\mathcal E}{}_{3-\delta}^\infty(1)\big)^{1/2}$ and $A$ replaced by a possibly larger constant $\breve A$. We now repeat \emph{all} of the arguments in \cref{sec:geometry}, \cref{sec:energy}, and \cref{sec:improving-energy} with $\tau_f$ replaced by $\infty$ and $\ve$ by $\breve \ve$ to show that $\breve{\mathfrak B}=[1,\infty)$. Note that the only instance (besides the bootstrap assumptions) where $\ve$ enters into the proofs of the estimates in \cref{sec:geometry}, \cref{sec:energy}, and \cref{sec:improving-energy}, as they are currently written, is in \eqref{eq:data-estimate}, which is now replaced by the definition of $\breve \ve$. We therefore have
\begin{equation*}
    \mathcal E^\infty_{p}(\tau)+\underline{\mathcal E}{}_{p}^\infty(\tau)\les \breve\ve^2\tau^{-3+\delta+p}= \big(\mathcal E^\infty_{3-\delta}(1)+\underline{\mathcal E}{}_{3-\delta}^\infty(1)\big)\tau^{-3+\delta+p}
\end{equation*}
for every $\tau\in[1,\infty)$ and $p\in[0,3-\delta]$. For $p=3-\delta$, this estimate gives \eqref{eq:energy-orbital}, as desired.

\textsc{Proof of orbital stability for the pointwise $C^1$ norm}: Revisiting the proofs of \eqref{est:psi_bound}--\eqref{est:uphi_bound}, we observe that the estimates only depend on energies and pointwise norms coming from initial data of the solution restricted to $J^-(\mathcal I^+)\cap\hat{\mathcal C}$. Therefore, the $\ve$ on the right-hand side of the estimates \eqref{est:psi_bound}--\eqref{est:uphi_bound} can be replaced by the right-hand side of \eqref{eq:pointwise-orbital}.

Part 4.~of the theorem follows from \cref{prop:energy-convergence}, \cref{prop:final-decay-geometry}, \eqref{est:phi_bound}, and \eqref{est:psi_bound}.

\textsc{Proof of the absence of trapped surfaces}: We first recall the notions of \emph{apparent horizon}
    \begin{equation*}
          \mathcal A\doteq\{(\hat u,\hat v)\in\hat{\mathcal Q}_\mathrm{max}:\hat\lambda(\hat u,\hat v)=0\},
    \end{equation*}
   \emph{outermost apparent horizon}
    \begin{equation*}
        \mathcal A'\doteq\{(\hat u,\hat v)\in\hat{\mathcal Q}_\mathrm{max}:\hat\lambda(\hat u,\hat v)=0\text{ and }\hat\lambda(\hat u',\hat v)>0\text{ for every }\hat u'<\hat u\}
    \end{equation*}
    and the set\footnote{Sometimes called the ``regular region.''}
    \begin{equation*}
        \mathcal R\doteq\{(\hat u,\hat v)\in\hat{\mathcal Q}_\mathrm{max}:\hat\lambda(\hat u,\hat v)>0\}.
    \end{equation*}
    We observe from \eqref{eq:varpi-u} and \eqref{eq:varpi-v} that if $(\hat u,\hat v)\in\mathcal R\cup\mathcal A$, then we have the monotonicities
    \begin{equation}
        \partial_{\hat u}\varpi(\hat u,\hat v)\le 0\quad\text{and}\quad \partial_{\hat v}\varpi(\hat u,\hat v)\ge 0,\label{eq:monotone}
    \end{equation}
    
    \textsc{Proof that $\mathcal A'\subset\mathcal H^+$}: Let $(\hat u_0,\hat v_0)\in\mathcal A'$.\footnote{This hypothesis could be empty. In this case, the hypothesis of part 5.~of the theorem is itself empty.} By the monotonicities \eqref{eq:monotone} and the outermost property of $(\hat u_0,\hat v_0)$, it follows that $M\ge\varpi(\hat u_0,\hat v_0)$ and $M\ge r(\hat u_0,\hat v_0)$.
    We also have $1-\mu(\hat u_0,\hat v_0)=0$, which can be solved for $r(\hat u_0,\hat v_0)=\varpi(\hat u_0,\hat v_0)\pm \sqrt{\varpi(\hat u_0,\hat v_0)^2-M^2}$, which implies that $\varpi(\hat u_0,\hat v_0)\ge M$. Hence, $\varpi(\hat u_0,\hat v_0)=M$ and $r(\hat u_0,\hat v_0)=M$. Since $\hat\nu<0$ on $\hat{\mathcal Q}_\mathrm{max}$ (recall \cref{lem:epsilon-loc}), this in turn implies that $(\hat u_0,\hat v_0)\in \mathcal H^+$ and that $\{\hat u_{\mathcal H^+}\}\times[\hat v_0,\infty)\subset\mathcal A'$. By Raychaudhuri's equation \eqref{eq:Ray-v}, it follows that $\partial_{\hat v}\phi$ vanishes identically on $\{\hat u_{\mathcal H^+}\}\times[\hat v_0,\infty)$. Since $\phi$ decays along $\mathcal H^+$ by \eqref{eq:psi-decay-final}, $\phi$ itself vanishes on $\{\hat u_{\mathcal H^+}\}\times[\hat v_0,\infty)$.

    \textsc{Proof that $\hat\lambda>0$ behind $\mathcal H^+$}: Suppose $\hat\lambda(\hat u_0,\hat v_0)\le 0$. By the argument in the previous paragraph, $(\hat u_{\mathcal H^+},\hat v_0)\in \mathcal A'$ and it holds that $r(\hat u_{\mathcal H^+},\hat v_0)=\varpi(\hat u_{\mathcal H^+},\hat v_0)=M$. Therefore, by \eqref{eq:nu-v}, $\partial_{\hat u}\hat\lambda(\hat u_{\mathcal H^+},\hat v_0)=0$. Because $\partial_{\hat u}\varpi(\hat u_{\mathcal H^+},\hat v_0)=0$, we may compute
    \begin{equation*}
        \partial_{\hat u}^2\hat\lambda(\hat u_{\mathcal H^+},\hat v_0)= \frac{2e^2\hat\kappa\hat\nu^2}{r^4}(\hat u_{\mathcal H^+},\hat v_0)>0.
    \end{equation*}
       If $\hat u_0>\hat u_{\mathcal H^+}$, it follows from Taylor's theorem that there exists a $\hat u_*\in (\hat u_{\mathcal H^+},\hat u_0]$ such that $\hat\lambda(\hat u_*,\hat v_0)=0$ and $\hat\lambda\ge 0$ on $[\hat u_{\mathcal H^+},\hat u_*]\times\{\hat v_0\}$. We may then argue again using monotonicity that $M\ge\varpi(\hat u_*,\hat v_0)$ and $M\ge r(\hat u_*,\hat v_0)$. The argument in the previous paragraph now implies that $\varpi(\hat u_*,\hat v_0)=r(\hat u_*,\hat v_0)=M$, which contradicts the assumption that $\hat u_0>\hat u_{\mathcal H^+}$.

       This completes the proof of \cref{thm:stability}.
\end{proof}

\section{Fine properties of the scalar field on and near the event horizon}

In this section, we use the decay estimates for the geometry and the scalar field derived in the previous section to further analyze the behavior of the scalar field on and near the event horizon $\mathcal H^+$ of the solutions given by \cref{thm:stability}. In \cref{sec:aretakis-proof}, we investigate the Aretakis instability on the dynamical geometry, which proves \cref{thm:instability}. In particular, we prove that each of the solutions given by \cref{thm:stability} has an ``asymptotic Aretakis charge'' $H_0[\phi]$, which while not actually constant, is ``almost'' constant in a quantitative sense. In \cref{sec:sharp_asym}, we derive sharp pointwise asymptotics for $\psi$ near $\mathcal H^+$. Finally, in \cref{sec:sharpness}, we show that if the final Aretakis charge $H_0[\phi]$ is ``quantitatively nonvanishing,'' then no nondegenerate integrated energy decay statement is true near $\mathcal H^+$, which proves the sharpness of the horizon hierarchy proved in \cref{sec:energy}. For a discussion of these results, we refer the reader back to \cref{sec:tails}. 

\begin{center}
    \emph{In this section, we will always work with one of the solutions $(\hat{\mathcal Q}_\mathrm{max},r,\hat\Omega^2,\phi,e)$ arising from seed data lying in $\mathfrak M_\stab$. In particular, the conclusions of \cref{thm:stability} hold for this solution.}
\end{center}

\subsection{Proof of the Aretakis instability, \texorpdfstring{\cref{thm:instability}}{Theorem 2}}\label{sec:aretakis-proof}

Recall the gauge-invariant vector field 
\begin{equation*}
    Y\doteq \nu^{-1}\partial_u,
\end{equation*}
which equals $\partial_r$ in standard ingoing Eddington--Finkelstein coordinates $(v,r)$ in Reissner--Nordstr\"om. For a general spherically symmetric spacetime, we derive at once the commutation formula
    \begin{equation*}
        [\partial_v,Y]= -\frac{\partial_u\partial_vr}{\nu}Y,
    \end{equation*}
    which will be useful in deriving equations for $\partial_v(Y\psi)$ and $\partial_v(Y^2\psi)$, where, as before, $\psi\doteq r\phi$.

\begin{lem} Let $(\mathcal Q,r,\Omega^2,\phi,e)$ be a solution of the spherically symmetric Einstein--Maxwell-scalar field system. Then
    \begin{equation}
        \partial_v(Y\psi)+(2\kappa\varkappa) Y\psi =   2\kappa\varkappa\,\phi.\label{eq:dvYpsi}
    \end{equation}
\end{lem}
\begin{proof}
    Using the commutation formula and \eqref{eq:phi-wave-1}, we compute
    \begin{equation*}
        \partial_v(Y\psi)=[\partial_v,Y]\psi +Y(\partial_v\psi)=-\frac{\partial_u\partial_vr}{\nu}Y\psi + \frac{\partial_u\partial_vr}{\nu}\phi.
    \end{equation*}
   By \eqref{eq:r-wave}, we conclude \eqref{eq:dvYpsi}.
\end{proof}

\begin{lem} Let $(\mathcal Q,r,\Omega^2,\phi,e)$ be a solution of the spherically symmetric Einstein--Maxwell-scalar field system. Then
\begin{equation}
     \partial_v(Y^2\psi)+(4 \kappa\varkappa) Y^2\psi = -\frac{2\kappa e^2}{r^4}Y\psi+E,\label{eq:dv2Ypsi}
\end{equation}
where
\begin{equation}
    E\doteq \frac{2\kappa e^2}{r^3}\phi+(1-\mu)r\kappa (Y\phi)^3+2\kappa\varkappa\left(r^2(Y\phi)^3-Y\phi\right).
    \label{eq:dv2Ypsi-error}
\end{equation}
\end{lem}
\begin{proof}
We rewrite \eqref{eq:dvYpsi} as
\begin{equation*}
    \partial_v(Y\psi)= \frac{\partial_u\partial_vr}{\nu}(\phi-Y\psi)
\end{equation*}
and apply again the commutation formula to obtain
\begin{align*}
    \partial_v(Y^2\psi)&= [\partial_v,Y]Y\psi + Y\partial_v(Y\psi)\\
    &= -\frac{\partial_u\partial_vr}{\nu}Y^2\psi+ Y\left[\frac{\partial_u\partial_vr}{\nu}(\phi-Y\psi)\right]\\
    &= -\frac{\partial_u\partial_vr}{\nu}Y^2\psi+ Y\left(\frac{\partial_u\partial_vr}{\nu}\right)(\phi-Y\psi)+ \frac{\partial_u\partial_vr}{\nu}Y(\phi-Y\psi),\\
    &= -2\frac{\partial_u\partial_vr}{\nu}Y^2\psi +\frac{\partial_u\partial_vr}{\nu}Y\phi+Y\left(\frac{\partial_u\partial_vr}{\nu}\right)(\phi-Y\psi)
\end{align*}
Using \eqref{eq:nu-v}, \eqref{eq:varpi-u}, and \eqref{eq:kappa-u}, we compute
\begin{equation*}
     Y\left(\frac{\partial_u\partial_vr}{\nu}\right)=2r\kappa\varkappa(Y\phi)^2-\frac{4\kappa\varkappa}{r}+(1-\mu)\kappa (Y\phi)^2+\frac{2\kappa e^2}{r^4}.
\end{equation*}
By grouping terms appropriately, we arrive at \eqref{eq:dv2Ypsi} and \eqref{eq:dv2Ypsi-error}.
\end{proof}

\begin{proof}[Proof of \cref{thm:instability}]
We work in one of the spacetimes $(\hat{\mathcal Q}_\mathrm{max},r,\hat \Omega^2,\phi,e)$ given by \cref{thm:stability} with the initial data normalized coordinates $(\hat u,\hat v)$. Recall that the event horizon $\mathcal H^+$ is located at $\hat u=\hat u_{\mathcal H^+}$. 

    \textsc{Almost-conservation of $Y\psi|_{\mathcal H^+}$}: We use an integrating factor to solve \eqref{eq:dvYpsi} on $\mathcal H^+$:
    \begin{equation*}
            Y\psi(\hat u_{\mathcal H^+},\hat v_2)=\exp\left(-\int_{\hat v_1}^{\hat v_2}2\hat\kappa\varkappa\,d\hat v'\right)Y\psi(\hat u_{\mathcal H^+},\hat v_1)
     +\int_{\hat v_1}^{\hat v_2} \exp\left(-\int_{\hat v'}^{\hat v_2}2\hat\kappa\varkappa\, d\hat v''\right)2\hat\kappa\varkappa\,\phi\,d\hat v'
    \end{equation*}
  for every $0\le \hat v_1\le \hat v_2$, where every quantity on the right-hand side is evaluated at $\hat u=\hat u_{\mathcal H^+}$. By integrating \eqref{eq:varkappa-final} along $\mathcal H^+$ and using \eqref{eq:hat-kappa-est}, we find
  \begin{equation*} 
    \int_{\hat v_1}^{\hat v_2}|\hat\kappa\varkappa|_{\mathcal H^+}|\,d\hat v'\les \ve^{2}(1+\hat v_1)^{-1+\delta},\quad  \exp\left(-\int_{\hat v_1}^{\hat v_2}2\hat\kappa\varkappa|_{\mathcal H^+}\,d\hat v'\right) = 1+O\big(\ve^{2}(1+\hat v_1)^{-1+\delta}\big).
  \end{equation*}
  Using these, \eqref{eq:psi-decay-final}, and \eqref{eq:pointwise-orbital}, we estimate
\begin{equation}\label{est:aretakis_charge}
    |Y\psi(\hat u_{\mathcal H^+},\hat v_2)-Y\psi(\hat u_{\mathcal H^+},\hat v_1)|\les \ve^{2}(1+\hat v_1)^{-1+\delta}|Y\psi(u_{\mathcal H^+},\hat v_1)|+ \int_{\hat v_1}^{\hat v_2}|\hat\kappa\varkappa||\phi|\,d\hat v'\les \ve^{3}(1+\hat v_1)^{-1+\delta}.
\end{equation}
By an elementary Cauchy sequence argument, the limit 
\begin{equation*}
    H_0[\phi]\doteq \lim_{\hat v\to\infty} Y\psi(\hat u_{\mathcal H^+},\hat v)
\end{equation*}
exists and the estimate \eqref{eq:Aretakis-main} holds.

  \textsc{Proof that $\mathfrak M_\stab^{\ne 0}$ has nonempty interior}: Let $\mathring \varphi$ be a compactly supported smooth function on $\hat{\mathcal C}$ such that 
  \begin{equation}
     \left. \frac{d}{d\hat u}\big((100M_0-\hat u)\mathring\varphi_\ing\big)\right|_{\hat u=\hat u_{\mathcal H^+,0}}\ne 0.\label{eq:nonempty-interior-aux-1}
  \end{equation}
   Let $0<\ve\le\ve_\stab$ and set $\mathcal S_0'\doteq ( c\ve\mathring\varphi, 100M_0,M_0,e_0 )$, where $c>0$ is a constant chosen so that $\mathfrak D[\mathcal S_0']\le \frac 12\ve$. In particular, $c$ is independent of $\ve$. For $\eta>0$ small, we now set
\begin{equation*}
      \mathfrak O_\eta\doteq\{\mathcal L(\mathcal S_0,\ve):\mathcal S_0\in\mathfrak M_0,\mathfrak D[\mathcal S_0-\mathcal S_0']<\eta\}.
\end{equation*}
Clearly, $\mathfrak O_\eta$ is open in $\mathfrak M$ and we claim that for $\ve$ and $\eta$ sufficiently small, $\emptyset\ne\mathfrak O_\eta\cap \mathfrak M_\stab\subset\mathfrak M_\stab^{\ne 0}$.
  
For $\eta$ small and $\mathcal S_0\in\mathfrak O_\eta$, \cref{thm:stability} applies to $\mathcal L(\mathcal S_0,\ve)$. Let $\mathcal S\in\mathfrak M_\stab(\mathcal S_0,\ve)$, let $H_0[\phi]$ be the associated asymptotic Aretakis charge. By \eqref{eq:horizon-location-estimate} and \eqref{eq:nonempty-interior-aux-1}, there exists a constant $c'>0$ such that \[|Y\psi(\hat u_{\mathcal H^+},0)|\gtrsim \ve - c'\eta^{1/2}\gtrsim \ve\] for $\ve$ and $\eta$ sufficiently small. By \eqref{est:aretakis_charge}, $H_0[\phi]$ differs from $Y\psi(\hat u_{\mathcal H^+},0)$ by an $O(\ve^{3})$ quantity, and is therefore nonzero for $\ve$ sufficiently small. As this applies for every $\mathcal S_0\in\mathfrak O_\eta$, we have proved that $\emptyset\ne\mathfrak O_\eta\cap \mathfrak M_\stab\subset\mathfrak M_\stab^{\ne 0}$, which shows that $\mathfrak M_\stab^{\ne 0}$ has nonempty interior in the subspace topology.
  
\textsc{Linear growth of $Y^2\psi|_{\mathcal H^+}$}: We use an integrating factor to solve \eqref{eq:dv2Ypsi} on $\mathcal H^+$:
\begin{equation*}
     Y^2\psi(\hat u_{\mathcal H^+},\hat v)=\exp\left(-\int_{0}^{\hat v}4\hat\kappa\varkappa\,d\hat v'\right)Y^2\psi(\hat u_{\mathcal H^+},0)
     +\int_{0}^{\hat v} \exp\left(-\int_{\hat v'}^{\hat v_2}4\hat\kappa\varkappa\,d\hat v''\right)\left(-\frac{2\hat\kappa e^2}{r^4}Y\psi+E\right)d\hat v'.
\end{equation*}
The first term is $O(\ve)$ by our assumption on the initial data (recall \eqref{eq:extra-ass}) and by \eqref{eq:1-mu-final}, \eqref{eq:varkappa-final}, \eqref{eq:psi-decay-final}, and \eqref{eq:pointwise-orbital}, we have 
\begin{equation*}
    \big|E|_{\mathcal H^+}\big|\les \ve\tau^{-1+\delta/2}.
\end{equation*} Note that the worst decaying term in $E$ (by far) comes from the zeroth order term, for which we only have the nonintegrable decay estimate \eqref{eq:psi-decay-final}.
Using again the geometric estimates on the horizon and \eqref{est:aretakis_charge}, we have 
\begin{equation*}
    \int_0^{\hat v}\frac{2\hat\kappa e^2}{r^4}Y\psi\,d\hat v'=\int_0^{\mathfrak v(\hat v)}\frac{2\kappa e^2}{r^4}Y\psi\,dv'= \frac{2}{M^2}H_0[\phi]\hat v + O\big(\ve^{3}(1+\hat v)^\delta\big).
\end{equation*}
Putting these estimates together, we have 
\begin{equation*}
    \left|Y^2\psi(\hat u_{\mathcal H^+},\hat v)+\frac{2}{M^2}H_0[\phi]\hat v\right| \les \ve(1+\hat v)^{\delta/2}+\ve^{3}(1+\hat v)^\delta,
\end{equation*}
which implies \eqref{eq:aretakis-main-2}.

\textsc{Behavior of the Ricci tensor along $\mathcal H^+$}: By trace-reversing the Einstein equation, we find
\begin{equation*}
    R_{\mu\nu}= 2T_{\mu\nu}^\mathrm{EM}+2\partial_\mu\phi\partial_\nu\phi
\end{equation*}
and using \eqref{eq:F-sph-sym} we compute
\begin{equation*}
   T^\mathrm{EM}= \frac{\hat\Omega^2e^2}{4r^4}(d\hat u\otimes d\hat v+d\hat v\otimes d\hat u)+\frac{e^2}{2r^2}g_{S^2}. 
\end{equation*}
Therefore, 
\begin{equation*}
    R_{\mu\nu}Y^\mu Y^\nu = 2(Y\phi)^2 = 2r^{-2}(Y\psi)^2-4r^{-3}\psi Y\psi +2r^{-4}\psi^2,
\end{equation*}
from which \eqref{eq:main-Ricci-1} follows readily.

Next, we compute
\begin{equation*}
    \nabla_\rho R_{\mu\nu} = 2\nabla_\rho T_{\mu\nu}^\mathrm{EM}+2\nabla_\rho\nabla_\mu\phi\nabla_\nu\phi+2\nabla_\mu\phi\nabla_\rho\nabla_\nu\phi.
\end{equation*}
Using again the form of $T^\mathrm{EM}$, we have $\nabla_\rho T_{\mu\nu}^\mathrm{EM}Y^\rho Y^\mu Y^\nu=0$. Since \[\Gamma^{\hat u}_{\hat u\hat u}=\partial_{\hat u}{\log\hat\Omega^2}=\partial_{\hat u}{\log(-4\hat\kappa\hat\nu)}=\hat\kappa^{-1}\partial_{\hat u}\hat\kappa+\hat\nu^{-1}\partial_{\hat u}\hat\nu,\] it holds that
\begin{align*}
    \nabla_\rho\nabla_\mu\phi \,Y^\rho Y^\mu& = \hat \nu^{-2} \partial_{\hat u}^2\phi -\hat\nu^{-2}\big(\hat\kappa^{-1}\partial_{\hat u}\hat\kappa+\hat\nu^{-1}\partial_{\hat u}\hat\nu\big)\partial_{\hat u}\phi\\
    &=\hat\nu^{-2}\partial_{\hat u}(\hat\nu Y\phi) - \hat\nu^{-1}\big(r\hat\nu (Y\phi)^2+\hat\nu^{-1}\partial_{\hat u}\hat\nu\big) Y\phi\\
    &= Y^2\phi -r(Y\phi)^3
\end{align*}
and hence 
\begin{equation*}
     \nabla_\rho R_{\mu\nu} Y^\rho Y^\mu Y^\nu = 4 Y\phi Y^2\phi - 4r(Y\phi)^4,
\end{equation*}
from which \eqref{eq:main-Ricci-2} follows readily.
\end{proof}

\subsection{Sharp decay for the scalar field near the horizon}\label{sec:sharp_asym}

In this section, we work in the eschatological gauge $(u_\infty,v)$ on the domain $\mathcal D_\infty$ with the final anchored extremal Reissner--Nordstr\"om solution $\bar r_\infty$. 

We require the following lemma (see also \eqref{eq:tortoise}).

\begin{lem}
    On $\mathcal D_\infty\cap\{\bar r_\infty\le \tfrac 12\Lambda\}$, it holds that
    \begin{equation}\label{eq:r_u_rel2}
      \frac{1}{\bar r_\infty-M}\sim u_\infty-v.
    \end{equation}
\end{lem}
\begin{proof}
   We have the identity
    \begin{equation}\label{eq:tortoise-2}
        v-u_\infty  = \int_{\bar r_\star}^{\bar r_\infty}\frac{d\bar r'}{D(\bar r')}=  -\frac{M^2}{\bar r_\infty-M}+\frac{M^2}{\bar r_\star-M}+2M\log\left(\frac{\bar r_\infty-M}{\bar r_\star-M}\right)+\bar r_\infty-\bar r_\star
    \end{equation} on $\mathcal D_\infty$, where $\bar r_\star$ was defined in \eqref{eq:r-star-limit} and $D(r)\doteq(1-M/r)^2$. Indeed, this is trivially true at the bifurcation sphere $(u_\infty,v)=(0,0)$ by definition of $\bar r_\star$, and holds everywhere because the $\partial_{u_\infty}$ and $\partial_v$ derivatives of $v-u_\infty$ and the integral are equal. Writing \eqref{eq:tortoise-2} as
    \begin{equation}\label{eq:tortoise-3}
        u_\infty-v = \frac{M^2}{\bar r_\infty-M} - 2M\log(\bar r_\infty-M) +O(1),
    \end{equation}
    we infer \eqref{eq:r_u_rel2} by inspection. 
    \end{proof}

Let us consider the curve 
\begin{equation}\label{eq:curve_hor}
\mathcal{\sigma}_{\beta} \doteq \{(u_\infty,v)\in\mathcal D_\infty: u_\infty -v =  u_{\infty}^{\beta} + C_\beta \}, 
\end{equation}
where the constant $C_\beta$ is chosen so that $\sigma_\beta\subset \{\bar r_\infty\le \frac 12\Lambda\}$. We also define the spacetime region
\begin{equation}\label{eq:reg_asym}
\mathfrak{t}_{\beta} \doteq  \{(u_\infty,v)\in\mathcal D_\infty: v \leq u_{\infty}-u_{\infty}^{\beta} - C_\beta \}.
\end{equation}
 Using \eqref{eq:r_u_rel2}, we observe that in the region $\mathfrak{t}_{\beta}$,
\begin{equation}\label{est:aux_sigma}
 \bar{r}_{\infty} - M \lesssim (1+u_{\infty})^{-\beta} \lesssim (1+v)^{-\beta} .
\end{equation}

With these definitions, we have the following result:
\begin{prop}\label{prop:asym_partial_u} Under the assumption \eqref{eq:extra-ass}, for every $\beta\in(\frac 23,1)$, there exists a $\beta'>0$ such that
\begin{equation}\label{est:asym_partial_u}
| u_{\infty}^2 \partial_{u_{\infty}} \psi + M^2 H_0 [ \phi ] | \lesssim \ve^{3} +\ve u_{\infty}^{-\beta'}
\end{equation}
in $\mathfrak t_\beta$, where $H_0[\phi]$ is the asymptotic Aretakis charge of the wave.
\end{prop}

We first verify the proposition along $\underline C{}_\ing$.

    \begin{lem}\label{lem:aux-Taylor}
    Let $\eta>0$. Under the assumption \eqref{eq:extra-ass}, it holds that
    \begin{equation}
    |u^2_\infty\partial_{u_\infty}\psi(u_\infty,0)+M^2H_0[\phi]|\les \ve^3+\ve u_\infty^{-1+\eta}\label{eq:expansion-data}
    \end{equation}
    for all $u_\infty>0$. 
\end{lem}
\begin{proof} Using the chain rule and \eqref{eq:nu-final}, we estimate
\begin{equation*}
    \partial_{\hat u}\big((\bar r_\infty-M)\circ\Phi_\infty(\hat u,0)\big) = \bar\nu_\infty\big(\Phi_\infty(\hat u,0)\big) \partial_{\hat u}\Phi(\hat u,0)\sim \nu_\infty\big(\Phi_\infty(\hat u,0)\big) \partial_{\hat u}\Phi(\hat u,0)=\hat\nu(\hat u,0)=-1.
\end{equation*}
Therefore, integrating backwards from $\hat u_{\mathcal H^+}$, we find that
\begin{equation*}
    (\bar r_\infty-M)\circ\Phi_\infty(\hat u,0)\sim \hat u_{\mathcal H^+}-\hat u
\end{equation*}
for $\hat u\in[0,\hat u_{\mathcal H^+}]$. By Taylor's theorem, the estimate \eqref{est:aretakis_charge}, and the assumption \eqref{eq:extra-ass}, we have that 
\begin{equation*}
    Y\psi(\hat u,0)= H_0[\phi] + O\big(\ve(\bar r_\infty-M)\big)+O(\ve^3)
\end{equation*}
for $\hat u\in[0,\hat u_{\mathcal H^+}]$. Of course, $H_0[\phi]=O(\ve)$ as well. 

By \eqref{eq:tortoise-3}, the identity $\bar \nu_\infty = -\bar r_\infty^{-2}(\bar r_\infty-M)^2$, and the Taylor expansion $\bar r^{-2}_\infty = M^{-2}+O(\bar r_\infty-M)$, we have
    \begin{equation*}
        u^2_\infty \bar \nu_\infty = -M^2+O\big((\bar r_\infty-M)[1+\log(\bar r_\infty-M)]\big)
    \end{equation*}
    at $v=0$, and we may therefore compute
\begin{equation*}
    u^2_\infty\partial_{u_\infty}\psi = u^2_\infty \nu_\infty Y\psi =u_\infty^2\bar\nu_\infty\big(1+O(\ve^3)\big)Y\psi 
    =-M^2H_0[\phi] +O\big(\ve(\bar r_\infty-M)[1+\log(\bar r_\infty-M)]\big)+O(\ve^3)
\end{equation*}
at $v=0$. Using the estimate $x\log x\les x^{1-\eta}$, we arrive at \eqref{eq:expansion-data}. \end{proof}

\begin{proof}[Proof of \cref{prop:asym_partial_u}]
Integrating the wave equation \eqref{eq:wave-equation-psi} along the segment $\{u_\infty\}\times[0,v]\subset\mathfrak t_\beta$, we have
$$ u_{\infty}^2 \partial_{u_{\infty}} \psi ( u_{\infty} , v ) = u_{\infty}^2 \partial_{u_{\infty}} \psi ( u_{\infty} , 0 ) + u_{\infty}^2 \int_0^v \frac{2\kappa_\infty\nu_\infty\varkappa}{r} \psi \, dv' . $$
By \eqref{eq:kappa-1}, \eqref{eq:nu-aux-1}, \eqref{eq:varkappa-decay}, and \eqref{eq:kappa-final}, \eqref{eq:nu-final}, \eqref{eq:varkappa-final}, the second term on the right-hand side satisfies
\begin{multline}
 \left|  u_{\infty}^2 \int_0^v \frac{2\kappa_\infty\nu_\infty\varkappa}{r} \psi \, dv'\right| \les u_{\infty}^2 \int_0^v ( \bar{r}_{\infty} - M)^3 | \psi | \, dv' + u_{\infty}^2 \int_0^v \ve^2 (1+v')^{-2+\delta}( \bar{r}_{\infty} - M )^2 | \psi | \, dv' \\
     \les u_{\infty}^{-\beta'}  \int_0^v u_{\infty}^{2+\beta'} ( \bar{r}_{\infty} - M)^3 | \psi | \, dv' + \int_0^v \ve^{2} u_{\infty}^2 (1+v')^{-2+\delta}( \bar{r}_{\infty} - M )^2| \psi | \, dv'.\label{eq:asy-aux-1}
\end{multline}
For the first term  of the last expression we note that by \eqref{eq:psi-decay-final} and \eqref{est:aux_sigma}, we have
\begin{equation*}
    u_{\infty}^{-\beta'}  \int_0^v u_{\infty}^{2+\beta'} ( \bar{r}_{\infty} - M)^3 | \psi | \, dv' \les \ve u_{\infty}^{-\beta'} \int_0^v (1+v' )^{-1+\delta/2 -3\beta + 2 + \beta' } dv' \lesssim \ve u_{\infty}^{-\beta'} ,
\end{equation*}
as $-1+\delta/2 - 3 \beta + 2 + \beta' < -1$ for $\beta \in (2/3 , 1)$ and $\beta' > 0$ sufficiently small.

To handle the second term on the right-hand side of \eqref{eq:asy-aux-1}, we break the region $\mathfrak{t}_{\beta}$ into two pieces. We split $\mathfrak t_\beta$ along the curve \[\rho_\beta\doteq \{(u_\infty,v)\in\mathcal D_\infty:v= \tfrac 12 u_\infty-C_\beta'\},\]
where $C_\beta'$ is chosen so that $\rho_\beta\cap\sigma_\beta=\emptyset$. Denote the region between $\rho_\beta$ and $\mathcal H^+$ by $\mathfrak t^1_\beta$ and the region between $\rho_\beta$ and $\sigma_\beta$ by $\mathfrak t^2_\beta$, so that $\mathfrak t_\beta=\mathfrak t^1_\beta\cup\mathfrak t^2_\beta$. In $\mathfrak{t}^1_{\beta}$, it holds that $\bar{r}_{\infty}-M \lesssim (1+u_{\infty})^{-1}$, which implies that 
\begin{equation*}
    \int_0^v \ve^{2} u_{\infty}^2 (1+v')^{-2+\delta}( \bar{r}_{\infty} - M )^2| \psi |\mathbf{1}_{\mathfrak t_\beta^1} \, dv' \les \ve^3.
\end{equation*}
On the other hand, in $\mathfrak{t}^2_{\beta}$, it holds that $(1+v)^{-1}\les (1+ u_{\infty})^{-1}$, so using \eqref{est:aux_sigma} we estimate
\begin{equation*}
     \int_0^v \ve^{2} u_{\infty}^2 (1+v')^{-2+\delta}( \bar{r}_{\infty} - M )^2| \psi |\mathbf{1}_{\mathfrak t_\beta^2} \, dv'\les \ve^2u_\infty^{2-k-2\beta}\int_0^v(1+v')^{-3+k+3\delta/2} \,dv'\les \ve^3u_\infty^{-\beta'}
\end{equation*}
where $k$ is chosen so that $-2+k+2\beta>\beta'$ and $-3+k+3\delta/2<-1$.

By applying \cref{lem:aux-Taylor} with $\eta$ chosen such that $-1+\eta<\beta'$, we conclude \eqref{est:asym_partial_u}. \end{proof}

We now use the previous proposition to obtain improved decay for $\psi$ in $\mathfrak{t}_{\beta}$.

\begin{prop}\label{prop:asym_psi}
Under the assumptions of \cref{prop:asym_partial_u} and $\beta\in(\frac 23,1-\delta-2\beta')$, it holds that
\begin{equation}\label{est:psi_dec_sharp}
    \left| \psi ( u_{\infty} , v )   + M^2 H_0 [ \phi ] \left(v^{-1} - u_\infty^{-1} \right) \right| \les \ve^3|v^{-1} - u_\infty^{-1} |+\ve (1+v)^{-1-\beta'}
\end{equation}
in $\mathfrak t_\beta$.
\end{prop}
\begin{proof} Let $u_\beta=u_\beta(v)$ be defined by $(u_\beta,v)\in\sigma_\beta$. Then
\begin{equation}
    |\psi(u_\beta,v)|=(\bar r_\infty-M)^{-1/2}(u_\beta,v)|(\bar r_\infty-M)^{1/2}\psi|(u_\beta,v)\les \ve (1+v)^{-3/2+\delta/2+\beta/2}\les \ve(1+v)^{-1-\beta'}
\end{equation}
by \eqref{eq:psi-decay-final}, the fact that $(\bar r_\infty-M)^{-1}\sim 1+v$ along $\sigma_\beta$, and the assumption on $\beta$. Using \eqref{est:asym_partial_u}, we find for $(u_\infty,v)\in\mathfrak t_\beta$,
\begin{equation*}
    \psi (u_{\infty} , v)-\psi (u_\beta , v) =\int_{u_\beta}^{u_\infty}\partial_{u_\infty}\psi\,du_\infty'=-M^2H_0[\phi](u_\beta^{-1}-u_\infty^{-1})+O\big(\ve^3(u_\beta^{-1}-u_\infty^{-1})+\ve(u_\beta^{-1-\beta'}-u_\infty^{-1-\beta'})\big).
\end{equation*}
Now we note that 
\begin{equation*}
    u_\beta^{-1}-u_\infty^{-1}=(u_\beta^{-1}-v^{-1})+(v^{-1}-u_\infty^{-1})= (v^{-1}-u_\infty^{-1})+O(v^{-2+\beta})
\end{equation*}
and $u_\beta+1\sim v+1$, and \eqref{est:psi_dec_sharp} readily follows.
\end{proof}

\subsection{Sharpness of the horizon hierarchy}\label{sec:sharpness}
We will show that the range of $p$ in the horizon hierarchy of \cref{sec:horizon-hierarchy} is sharp, resulting in failure of boundedness of the integrated non-degenerate energy near the horizon.

\begin{prop}\label{prop:sharp} For any $\eta>0$ and $C_2>0$ there exists a constant $0<\ve_{\eta,C_2}\le \ve_\stab$ such that the following holds. Let $\mathcal S_0(\alpha_\star)\in\mathfrak M_\stab$ with $0< \mathfrak D[\mathcal S_0]\le\ve\le \ve_{\eta,C_2}$. Suppose also that $\mathring\phi$ satisfies the second order smallness condition \eqref{eq:extra-ass} and that the asymptotic Aretakis charge is quantitatively nonvanishing in the sense that 
\begin{equation}\label{eq:aretakis-nonzero}
    |H_0[\phi]|\ge\ve^{2-\eta}.
\end{equation} 
Then for any $R> M$, it holds that 
\begin{equation}\label{eq:blowup}
    \iint_{\mathcal D_\infty\cap\{\bar r_\infty\le R\} }\frac{(\partial_{u_{\infty}} \psi)^2}{-\bar{\nu}_{\infty}} \,du_\infty dv = \infty.
\end{equation}
\end{prop}
\begin{rk}\label{rk:finite}
Note that by the Morawetz estimate \cref{prop:Morawetz}, the integral in \eqref{eq:blowup}, taken instead over $\mathcal D_\infty\cap\{R'\le\bar r_\infty\le R\}$ where $M<R'<R$, is finite.
\end{rk}

\begin{proof}[Proof of \cref{prop:sharp}] We repeat the main calculation for the horizon hierarchy, \cref{lem:horizon-main}, with $p=3$. Let
\begin{equation*}
    \mathcal A_{u_f}\doteq \mathcal D_\infty\cap \{r\le\Lambda\}\cap\{u_\infty\le u_f\},
\end{equation*}
where we will let $u_f\to\infty$ and $v_0$ will be chosen later.

Integrating the expression \begin{equation*}
    \iint_{ \mathcal A_{u_f}} \partial_v\left((\bar r-M)^{-1}\frac{(\partial_{u_\infty}\psi)^2}{-\bar\nu_\infty}\right) du_\infty dv
\end{equation*} by parts and using the wave equation \eqref{eq:wave-equation-psi}, we find that
\begin{multline}\label{eq:p=3}
   - \int_{\underline C{}_{0}\cap \mathcal A_{u_f}} (\bar r_\infty-M)^{-1} \frac{(\partial_{u_\infty}\psi)^2}{-\bar\nu_\infty}\,du_\infty +\int_{\Gamma \cap \mathcal A_{u_f}}n^v_\Gamma (\bar r_\infty-M)^{-1}\frac{(\partial_{u_\infty}\psi)^2}{-\bar\nu_\infty}\,ds\\ =-\iint_{\mathcal A_{u_f}}\left(\frac{2M+\bar r_\infty}{\bar r_\infty^3}\right)\frac{(\partial_{u_\infty}\psi)^2}{-\bar\nu_\infty}\,du_\infty dv -\iint_{\mathcal A_{u_f}}\frac{4\kappa_\infty \nu_\infty}{r\bar \nu_\infty}(\bar r_\infty-M)^{-1}\varkappa\psi\partial_{u_\infty}\psi\,du_\infty dv.
\end{multline}
The term along $\Gamma$ is $\les\ve^2$ using \eqref{eq:Mor-Gamma-1} and \eqref{eq:mor-aux-11}. For $\ve$ sufficiently small using  \eqref{eq:expansion-data} with $\eta=1-\beta'$ and the condition \eqref{eq:aretakis-nonzero}, we have
\begin{equation*}
    \big(u^2_\infty \partial_{u_\infty}\psi(u_\infty,0)\big)^2\gtrsim |H_0[\phi]|^2 -O(\ve^2u_\infty^{-\beta'})
\end{equation*}
Therefore, by \eqref{eq:r_u_rel2}, it holds that
\begin{multline}\label{eq:p=3-1}
       \int_{\underline C{}_{0}\cap \mathcal A_{u_f}} (\bar r_\infty-M)^{-1} \frac{(\partial_{u_\infty}\psi)^2}{-\bar\nu_\infty}\,du_\infty \gtrsim  \int_{\underline C{}_{0}\cap \mathcal A_{u_f}} (\bar r_\infty-M)^{-3} (\partial_{u_\infty}\psi)^2\,du_\infty \gtrsim    \int_{\underline C{}_{0}\cap \mathcal A_{u_f}}u_\infty^3 (\partial_{u_\infty}\psi)^2\,du_\infty\\\gtrsim\int_{\underline C{}_{0}\cap \mathcal A_{u_f}}u_\infty^{-1}\big(|H_0[\phi]|^2-O(\ve^2 u_\infty^{-\beta'})\big)\,du_\infty\gtrsim |H_0[\phi]|^2\log(u_f)-\varepsilon^2,
\end{multline}
where we have used that $u_\infty^{-1-\beta'}$ is integrable.

We now estimate the mixed term on the right-hand side of \eqref{eq:p=3}. We use the geometric estimate \eqref{eq:varkappa-final} and break up the region of integration using $\mathfrak t_\beta$ to estimate \begin{multline}\label{eq:p=3-2}
   \left| \iint_{\mathcal A_{u_f}}\frac{4\kappa_\infty \nu_\infty}{r\bar \nu_\infty}(\bar r_\infty-M)^{-1}\varkappa\psi\partial_{u_\infty}\psi\,du_\infty dv\right|\les \iint_{\mathcal A_{u_f}}|\psi||\partial_{u_\infty}\psi|\,du_\infty dv \\
   +\left(\iint_{\mathcal A_{u_f}\setminus \mathfrak t_\beta}+\iint_{\mathcal A_{u_f}\cap \mathfrak t_\beta}\right)\ve^2(\bar r_\infty-M)^{-1}\tau^{-2+\delta}|\psi||\partial_{u_\infty}\psi|\,du_\infty dv  \doteq \mathrm{I}+\mathrm{II}+\mathrm{III}.
\end{multline}
By \eqref{eq:r_u_rel2}, we have $(\bar r_\infty-M)^{-1}\les u-v\le u$ in $\mathcal A_{u_f}\cap\{\bar r_\infty\le \tfrac 12\Lambda\}$ and by definition of $\sigma_\beta$, $(\bar r_\infty-M)^{-1}\les 1+v$ in $\mathcal A_{u_f}\setminus \mathfrak t_\beta$. It follows that $\mathrm{II}\les \ve^2\mathrm{I}$. Using the bulk terms in the energy estimate \cref{prop:H-h-1} with $p=5/2$ (after passing to $\tau_f\to\infty$, refer also to the proof of \cref{thm:stability} in \cref{sec:putting-together}), we have $\mathrm{I}\les \underline{\mathcal E}^\infty_{5/2}(1)\les\ve^2$. By \cref{prop:asym_psi}, we have $|\psi| \les \ve u^{-1}_\infty +\ve v^{-1}$ in $\mathfrak t_\beta$ and by \cref{prop:asym_partial_u}, we have $|u_\infty^2\partial_{u_\infty}\psi|\les \ve$ in $\mathfrak t_\beta$. Therefore, we estimate
\begin{multline}\label{eq:p=3-3}
    \mathrm{III}\les \iint_{\mathcal A_{u_f}\cap \mathfrak t^1_\beta}\ve^2u_\infty^{-1}v^{-2+\delta}|\psi||u_\infty^2\partial_{u_\infty}\psi|\,du_\infty dv\\\les \ve^4 + \ve^4\iint_{\{1\le u_\infty\le u_f\}\cap\{v\ge 1\}} \big(u_\infty^{-2}v^{-2+\delta}+u_\infty^{-1}v^{-3+\delta}\big) du_\infty dv\les \ve^4(1+\log(u_f)).
\end{multline}

Combining \eqref{eq:p=3}, \eqref{eq:p=3-1}, \eqref{eq:p=3-2}, and \eqref{eq:p=3-3}, we find 
\begin{equation*}
    \iint_{\mathcal A_{u_f}}\frac{(\partial_{u_\infty}\psi)^2}{-\bar\nu_\infty}\,du_\infty dv\gtrsim \big(|H_0[\phi]|^2-\ve^4\big)\log(u_f)-\ve^2\gtrsim |H_0[\phi]|^2\log(u_f)-\ve^2
\end{equation*}
for $\ve$ sufficiently small. Letting $u_f\to\infty$ completes the proof.
\end{proof}

\begin{rk}\label{rk:failure} There is an interesting difference between the proof of \cref{prop:sharp} and the corresponding fact in the uncoupled case. In our case, the mixed bulk term on the right-hand side of \eqref{eq:p=3} is in general \emph{not} bounded as $u_f\to\infty$, while it \emph{is} bounded in the uncoupled case. This is because the dynamical redshift factor $\varkappa$ does not in general exactly cancel out the singular factor of $(\bar r_\infty-M)^{-1}$. Though the growth rate is logarithmic, just as with the first term on the left-hand side of \eqref{eq:p=3}, it comes with an additional smallness factor which allows the proof to go through. The cost is that we can only prove \cref{prop:sharp} under the ``quantitative nonvanishing'' assumption \eqref{eq:aretakis-nonzero} on the asymptotic Aretakis charge $H_0[\phi]$. In the uncoupled case, where $H_0[\phi]$ is constant, it suffices to assume that $H_0[\phi]\ne 0$.

On the other hand, in the uncoupled case on a \textit{subextremal} Reissner--Nordstr\"{o}m background, the same bulk term is also infinite. It should be noted though, in the uncoupled subextremal case, that it is infinite in a way that cancels out the other unbounded term in \eqref{eq:p=3} (the first term on the left-hand side). As a result, in the uncoupled subextremal case, there is no contradiction with the fact that the spacetime integral \eqref{eq:blowup} is finite due to the redshift effect \cite{dafermos2009red}. See \cite[Section 4.6]{AAG-trapping} for more details.
\end{rk}

\printbibliography[heading=bibintoc] 
\end{document}

%% file: figures/stability-intro.pdf_tex
\begingroup%
  \makeatletter%
  \providecommand\color[2][]{%
    \errmessage{(Inkscape) Color is used for the text in Inkscape, but the package 'color.sty' is not loaded}%
    \renewcommand\color[2][]{}%
  }%
  \providecommand\transparent[1]{%
    \errmessage{(Inkscape) Transparency is used (non-zero) for the text in Inkscape, but the package 'transparent.sty' is not loaded}%
    \renewcommand\transparent[1]{}%
  }%
  \providecommand\rotatebox[2]{#2}%
  \newcommand*\fsize{\dimexpr\f@size pt\relax}%
  \newcommand*\lineheight[1]{\fontsize{\fsize}{#1\fsize}\selectfont}%
  \ifx\svgwidth\undefined%
    \setlength{\unitlength}{130.2406275bp}%
    \ifx\svgscale\undefined%
      \relax%
    \else%
      \setlength{\unitlength}{\unitlength * \real{\svgscale}}%
    \fi%
  \else%
    \setlength{\unitlength}{\svgwidth}%
  \fi%
  \global\let\svgwidth\undefined%
  \global\let\svgscale\undefined%
  \makeatother%
  \begin{picture}(1,1.03249467)%
    \lineheight{1}%
    \setlength\tabcolsep{0pt}%
    \put(0,0){\includegraphics[width=\unitlength,page=1]{stability-intro.pdf}}%
    \put(0.34958889,0.44207093){\color[rgb]{0,0,0}\rotatebox{45}{\makebox(0,0)[lt]{\lineheight{1.25}\smash{\begin{tabular}[t]{l}$\mathcal H^+$\end{tabular}}}}}%
    \put(0.82966204,0.71275232){\color[rgb]{0,0,0}\rotatebox{-45}{\makebox(0,0)[lt]{\lineheight{1.25}\smash{\begin{tabular}[t]{l}$\mathcal I^+$\end{tabular}}}}}%
    \put(0.68441128,0.8565502){\color[rgb]{0,0,0}\makebox(0,0)[lt]{\lineheight{1.25}\smash{\begin{tabular}[t]{l}$i^+$\end{tabular}}}}%
    \put(0.1777498,0.5528981){\color[rgb]{0,0,0}\makebox(0,0)[lt]{\lineheight{1.25}\smash{\begin{tabular}[t]{l}$\mathcal{BH}$\end{tabular}}}}%
    \put(0.19046056,0.21175655){\color[rgb]{0,0,0}\rotatebox{-45}{\makebox(0,0)[lt]{\lineheight{1.25}\smash{\begin{tabular}[t]{l}$\underline C{}_\ing$\end{tabular}}}}}%
    \put(0.79089638,0.1785535){\color[rgb]{0,0,0}\rotatebox{45}{\makebox(0,0)[lt]{\lineheight{1.25}\smash{\begin{tabular}[t]{l}$C_\out$\end{tabular}}}}}%
    \put(0.4962972,0.50436896){\color[rgb]{0,0,0}\rotatebox{60}{\makebox(0,0)[lt]{\lineheight{1.25}\smash{\begin{tabular}[t]{l}$g\to g_\mathrm{ERN}$\\$\phi\to 0$\end{tabular}}}}}%
    \put(0,0){\includegraphics[width=\unitlength,page=2]{stability-intro.pdf}}%
    \put(-0.32530936,0.85867915){\color[rgb]{0,0,0}\makebox(0,0)[lt]{\lineheight{1.25}\smash{\begin{tabular}[t]{l}no trapped surfaces\end{tabular}}}}%
    \put(0,0){\includegraphics[width=\unitlength,page=3]{stability-intro.pdf}}%
    \put(0.55331292,0.98707402){\color[rgb]{0,0,0}\rotatebox{-45}{\makebox(0,0)[lt]{\lineheight{1.25}\smash{\begin{tabular}[t]{l}$\mathcal{CH}^+$\end{tabular}}}}}%
    \put(0,0){\includegraphics[width=\unitlength,page=4]{stability-intro.pdf}}%
  \end{picture}%
\endgroup%

%% file: figures/ERN-intro.pdf_tex
\begingroup%
  \makeatletter%
  \providecommand\color[2][]{%
    \errmessage{(Inkscape) Color is used for the text in Inkscape, but the package 'color.sty' is not loaded}%
    \renewcommand\color[2][]{}%
  }%
  \providecommand\transparent[1]{%
    \errmessage{(Inkscape) Transparency is used (non-zero) for the text in Inkscape, but the package 'transparent.sty' is not loaded}%
    \renewcommand\transparent[1]{}%
  }%
  \providecommand\rotatebox[2]{#2}%
  \newcommand*\fsize{\dimexpr\f@size pt\relax}%
  \newcommand*\lineheight[1]{\fontsize{\fsize}{#1\fsize}\selectfont}%
  \ifx\svgwidth\undefined%
    \setlength{\unitlength}{112.32566773bp}%
    \ifx\svgscale\undefined%
      \relax%
    \else%
      \setlength{\unitlength}{\unitlength * \real{\svgscale}}%
    \fi%
  \else%
    \setlength{\unitlength}{\svgwidth}%
  \fi%
  \global\let\svgwidth\undefined%
  \global\let\svgscale\undefined%
  \makeatother%
  \begin{picture}(1,1.08580229)%
    \lineheight{1}%
    \setlength\tabcolsep{0pt}%
    \put(0,0){\includegraphics[width=\unitlength,page=1]{ERN-intro.pdf}}%
    \put(0.72677897,0.78340836){\color[rgb]{0,0,0}\rotatebox{-45}{\makebox(0,0)[lt]{\lineheight{1.25}\smash{\begin{tabular}[t]{l}$\mathcal I^+$\end{tabular}}}}}%
    \put(0.47505493,1.02414453){\color[rgb]{0,0,0}\makebox(0,0)[lt]{\lineheight{1.25}\smash{\begin{tabular}[t]{l}$i^+$\end{tabular}}}}%
    \put(0.54139065,0.01309571){\color[rgb]{0,0,0}\makebox(0,0)[lt]{\lineheight{1.25}\smash{\begin{tabular}[t]{l}$i^-$\end{tabular}}}}%
    \put(0.94724456,0.55144719){\color[rgb]{0,0,0}\makebox(0,0)[lt]{\lineheight{1.25}\smash{\begin{tabular}[t]{l}$i^0$\end{tabular}}}}%
    \put(0.53665354,0.18846795){\color[rgb]{0,0,0}\makebox(0,0)[lt]{\lineheight{1.25}\smash{\begin{tabular}[t]{l}$\Gamma$\end{tabular}}}}%
    \put(0.77316146,0.20413792){\color[rgb]{0,0,0}\rotatebox{45}{\makebox(0,0)[lt]{\lineheight{1.25}\smash{\begin{tabular}[t]{l}$\mathcal I^-$\end{tabular}}}}}%
    \put(0.21882601,0.72549469){\color[rgb]{0,0,0}\rotatebox{45}{\makebox(0,0)[lt]{\lineheight{1.25}\smash{\begin{tabular}[t]{l}$\mathcal H^+$\end{tabular}}}}}%
    \put(0.17853476,0.23768495){\color[rgb]{0,0,0}\rotatebox{-45}{\makebox(0,0)[lt]{\lineheight{1.25}\smash{\begin{tabular}[t]{l}$\mathcal H^-$\end{tabular}}}}}%
    \put(0,0){\includegraphics[width=\unitlength,page=2]{ERN-intro.pdf}}%
    \put(0.66444272,0.33122183){\color[rgb]{0,0,0}\rotatebox{45}{\makebox(0,0)[lt]{\lineheight{1.25}\smash{\begin{tabular}[t]{l}$C(\tau_1)$\end{tabular}}}}}%
    \put(0.58050146,0.57550254){\color[rgb]{0,0,0}\rotatebox{45}{\makebox(0,0)[lt]{\lineheight{1.25}\smash{\begin{tabular}[t]{l}$C(\tau_2)$\end{tabular}}}}}%
    \put(0.23645891,0.48414155){\color[rgb]{0,0,0}\rotatebox{-45}{\makebox(0,0)[lt]{\lineheight{1.25}\smash{\begin{tabular}[t]{l}$\underline C(\tau_1)$\end{tabular}}}}}%
    \put(0.33026187,0.7276198){\color[rgb]{0,0,0}\rotatebox{-45}{\makebox(0,0)[lt]{\lineheight{1.25}\smash{\begin{tabular}[t]{l}$\underline C(\tau_2)$\end{tabular}}}}}%
  \end{picture}%
\endgroup%

%% file: figures/bootstrap-setup-intro.pdf_tex
\begingroup%
  \makeatletter%
  \providecommand\color[2][]{%
    \errmessage{(Inkscape) Color is used for the text in Inkscape, but the package 'color.sty' is not loaded}%
    \renewcommand\color[2][]{}%
  }%
  \providecommand\transparent[1]{%
    \errmessage{(Inkscape) Transparency is used (non-zero) for the text in Inkscape, but the package 'transparent.sty' is not loaded}%
    \renewcommand\transparent[1]{}%
  }%
  \providecommand\rotatebox[2]{#2}%
  \newcommand*\fsize{\dimexpr\f@size pt\relax}%
  \newcommand*\lineheight[1]{\fontsize{\fsize}{#1\fsize}\selectfont}%
  \ifx\svgwidth\undefined%
    \setlength{\unitlength}{130.2406275bp}%
    \ifx\svgscale\undefined%
      \relax%
    \else%
      \setlength{\unitlength}{\unitlength * \real{\svgscale}}%
    \fi%
  \else%
    \setlength{\unitlength}{\svgwidth}%
  \fi%
  \global\let\svgwidth\undefined%
  \global\let\svgscale\undefined%
  \makeatother%
  \begin{picture}(1,1.03249467)%
    \lineheight{1}%
    \setlength\tabcolsep{0pt}%
    \put(0,0){\includegraphics[width=\unitlength,page=1]{bootstrap-setup-intro.pdf}}%
    \put(0.32574715,0.7090499){\color[rgb]{0,0,0}\rotatebox{45}{\makebox(0,0)[lt]{\lineheight{1.25}\smash{\begin{tabular}[t]{l}$\mathcal H^+$\end{tabular}}}}}%
    \put(0.75875339,0.78058104){\color[rgb]{0,0,0}\rotatebox{-45}{\makebox(0,0)[lt]{\lineheight{1.25}\smash{\begin{tabular}[t]{l}$\mathcal I^+$\end{tabular}}}}}%
    \put(0.48656942,1.03774326){\color[rgb]{0,0,0}\makebox(0,0)[lt]{\lineheight{1.25}\smash{\begin{tabular}[t]{l}$i^+$\end{tabular}}}}%
    \put(0.12134934,0.28086278){\color[rgb]{0,0,0}\rotatebox{-45}{\makebox(0,0)[lt]{\lineheight{1.25}\smash{\begin{tabular}[t]{l}$\underline C{}_\ing=\underline C{}_0$\end{tabular}}}}}%
    \put(0.72179376,0.10945088){\color[rgb]{0,0,0}\rotatebox{45}{\makebox(0,0)[lt]{\lineheight{1.25}\smash{\begin{tabular}[t]{l}$C_\out=C_0$\end{tabular}}}}}%
    \put(0,0){\includegraphics[width=\unitlength,page=2]{bootstrap-setup-intro.pdf}}%
    \put(0.4714531,0.15859642){\color[rgb]{0,0,0}\makebox(0,0)[lt]{\lineheight{1.25}\smash{\begin{tabular}[t]{l}$\Gamma$\end{tabular}}}}%
    \put(0,0){\includegraphics[width=\unitlength,page=3]{bootstrap-setup-intro.pdf}}%
    \put(-0.30450545,0.70521732){\color[rgb]{0,0,0}\makebox(0,0)[lt]{\lineheight{1.25}\smash{\begin{tabular}[t]{l}$\partial_vr=1-\frac{2m}{r}$\end{tabular}}}}%
    \put(0.92187731,0.69828153){\color[rgb]{0,0,0}\makebox(0,0)[lt]{\lineheight{1.25}\smash{\begin{tabular}[t]{l}$\partial_ur =-(1-\frac{2m}{r})$\end{tabular}}}}%
    \put(0.69203772,0.3280767){\color[rgb]{0,0,0}\rotatebox{45}{\makebox(0,0)[lt]{\lineheight{1.25}\smash{\begin{tabular}[t]{l}$C_u$\end{tabular}}}}}%
    \put(0.3860161,0.38223202){\color[rgb]{0,0,0}\rotatebox{-45}{\makebox(0,0)[lt]{\lineheight{1.25}\smash{\begin{tabular}[t]{l}$\underline C{}_v$\end{tabular}}}}}%
    \put(0,0){\includegraphics[width=\unitlength,page=4]{bootstrap-setup-intro.pdf}}%
    \put(0.73053503,0.86104969){\color[rgb]{0,0,0}\makebox(0,0)[lt]{\lineheight{1.25}\smash{\begin{tabular}[t]{l}$\Gamma(\tau_f)$\end{tabular}}}}%
  \end{picture}%
\endgroup%

%% file: figures/mod-space-with-lines.pdf_tex
\begingroup%
  \makeatletter%
  \providecommand\color[2][]{%
    \errmessage{(Inkscape) Color is used for the text in Inkscape, but the package 'color.sty' is not loaded}%
    \renewcommand\color[2][]{}%
  }%
  \providecommand\transparent[1]{%
    \errmessage{(Inkscape) Transparency is used (non-zero) for the text in Inkscape, but the package 'transparent.sty' is not loaded}%
    \renewcommand\transparent[1]{}%
  }%
  \providecommand\rotatebox[2]{#2}%
  \newcommand*\fsize{\dimexpr\f@size pt\relax}%
  \newcommand*\lineheight[1]{\fontsize{\fsize}{#1\fsize}\selectfont}%
  \ifx\svgwidth\undefined%
    \setlength{\unitlength}{197.08441066bp}%
    \ifx\svgscale\undefined%
      \relax%
    \else%
      \setlength{\unitlength}{\unitlength * \real{\svgscale}}%
    \fi%
  \else%
    \setlength{\unitlength}{\svgwidth}%
  \fi%
  \global\let\svgwidth\undefined%
  \global\let\svgscale\undefined%
  \makeatother%
  \begin{picture}(1,0.46744847)%
    \lineheight{1}%
    \setlength\tabcolsep{0pt}%
    \put(0,0){\includegraphics[width=\unitlength,page=1]{mod-space-with-lines.pdf}}%
    \put(0.09294654,0.23505551){\color[rgb]{0,0,0}\makebox(0,0)[lt]{\lineheight{1.25}\smash{\begin{tabular}[t]{l}$p$\end{tabular}}}}%
    \put(0.0158228,0.30538865){\color[rgb]{0,0,0}\makebox(0,0)[lt]{\lineheight{1.25}\smash{\begin{tabular}[t]{l}$\alpha$\end{tabular}}}}%
    \put(0,0){\includegraphics[width=\unitlength,page=2]{mod-space-with-lines.pdf}}%
    \put(0.58151244,0.23557674){\color[rgb]{0,0,0}\makebox(0,0)[lt]{\lineheight{1.25}\smash{\begin{tabular}[t]{l}$\mathfrak M_0$\end{tabular}}}}%
    \put(0.50572019,0.01219153){\color[rgb]{0,0,0}\makebox(0,0)[lt]{\lineheight{1.25}\smash{\begin{tabular}[t]{l}$\mathcal L$\end{tabular}}}}%
    \put(0.72176394,0.40254775){\color[rgb]{0,0,0}\makebox(0,0)[lt]{\lineheight{1.25}\smash{\begin{tabular}[t]{l}$\mathfrak M_\stab$\end{tabular}}}}%
    \put(0.5238294,0.39972322){\color[rgb]{0,0,0}\makebox(0,0)[lt]{\lineheight{1.25}\smash{\begin{tabular}[t]{l}ERN\end{tabular}}}}%
    \put(0,0){\includegraphics[width=\unitlength,page=3]{mod-space-with-lines.pdf}}%
  \end{picture}%
\endgroup%

%% file: figures/local-moduli-space.pdf_tex
\begingroup%
  \makeatletter%
  \providecommand\color[2][]{%
    \errmessage{(Inkscape) Color is used for the text in Inkscape, but the package 'color.sty' is not loaded}%
    \renewcommand\color[2][]{}%
  }%
  \providecommand\transparent[1]{%
    \errmessage{(Inkscape) Transparency is used (non-zero) for the text in Inkscape, but the package 'transparent.sty' is not loaded}%
    \renewcommand\transparent[1]{}%
  }%
  \providecommand\rotatebox[2]{#2}%
  \newcommand*\fsize{\dimexpr\f@size pt\relax}%
  \newcommand*\lineheight[1]{\fontsize{\fsize}{#1\fsize}\selectfont}%
  \ifx\svgwidth\undefined%
    \setlength{\unitlength}{116.24999279bp}%
    \ifx\svgscale\undefined%
      \relax%
    \else%
      \setlength{\unitlength}{\unitlength * \real{\svgscale}}%
    \fi%
  \else%
    \setlength{\unitlength}{\svgwidth}%
  \fi%
  \global\let\svgwidth\undefined%
  \global\let\svgscale\undefined%
  \makeatother%
  \begin{picture}(1,0.64516125)%
    \lineheight{1}%
    \setlength\tabcolsep{0pt}%
    \put(0,0){\includegraphics[width=\unitlength,page=1]{local-moduli-space.pdf}}%
    \put(0.09091466,0.55332653){\color[rgb]{0,0,0}\makebox(0,0)[lt]{\lineheight{1.25}\smash{\begin{tabular}[t]{l}$\mathfrak M_\mathrm{sub}$\end{tabular}}}}%
    \put(0.79727988,0.03984704){\color[rgb]{0,0,0}\makebox(0,0)[lt]{\lineheight{1.25}\smash{\begin{tabular}[t]{l}$\mathfrak M^1_\mathrm{stab}$\end{tabular}}}}%
    \put(0,0){\includegraphics[width=\unitlength,page=2]{local-moduli-space.pdf}}%
    \put(0.67155986,0.55332653){\color[rgb]{0,0,0}\makebox(0,0)[lt]{\lineheight{1.25}\smash{\begin{tabular}[t]{l}$\mathfrak M_\mathrm{disp}$\end{tabular}}}}%
    \put(0,0){\includegraphics[width=\unitlength,page=3]{local-moduli-space.pdf}}%
    \put(0.57986778,0.42612826){\color[rgb]{0,0,0}\makebox(0,0)[lt]{\lineheight{1.25}\smash{\begin{tabular}[t]{l}$\mathcal L$\\\end{tabular}}}}%
    \put(0.02823296,0.097732){\color[rgb]{0,0,0}\makebox(0,0)[lt]{\lineheight{1.25}\smash{\begin{tabular}[t]{l}$\mathfrak M_\stab^{\mathfrak r<1}$\end{tabular}}}}%
    \put(0,0){\includegraphics[width=\unitlength,page=4]{local-moduli-space.pdf}}%
    \put(0.81728131,0.27030843){\color[rgb]{0,0,0}\makebox(0,0)[lt]{\lineheight{1.25}\smash{\begin{tabular}[t]{l}ERN\end{tabular}}}}%
    \put(0,0){\includegraphics[width=\unitlength,page=5]{local-moduli-space.pdf}}%
  \end{picture}%
\endgroup%

%% file: figures/critical-behavior.pdf_tex
\begingroup%
  \makeatletter%
  \providecommand\color[2][]{%
    \errmessage{(Inkscape) Color is used for the text in Inkscape, but the package 'color.sty' is not loaded}%
    \renewcommand\color[2][]{}%
  }%
  \providecommand\transparent[1]{%
    \errmessage{(Inkscape) Transparency is used (non-zero) for the text in Inkscape, but the package 'transparent.sty' is not loaded}%
    \renewcommand\transparent[1]{}%
  }%
  \providecommand\rotatebox[2]{#2}%
  \newcommand*\fsize{\dimexpr\f@size pt\relax}%
  \newcommand*\lineheight[1]{\fontsize{\fsize}{#1\fsize}\selectfont}%
  \ifx\svgwidth\undefined%
    \setlength{\unitlength}{363.11812754bp}%
    \ifx\svgscale\undefined%
      \relax%
    \else%
      \setlength{\unitlength}{\unitlength * \real{\svgscale}}%
    \fi%
  \else%
    \setlength{\unitlength}{\svgwidth}%
  \fi%
  \global\let\svgwidth\undefined%
  \global\let\svgscale\undefined%
  \makeatother%
  \begin{picture}(1,0.37047617)%
    \lineheight{1}%
    \setlength\tabcolsep{0pt}%
    \put(0,0){\includegraphics[width=\unitlength,page=1]{critical-behavior.pdf}}%
    \put(0.35627857,0.35093646){\color[rgb]{0,0,0}\makebox(0,0)[lt]{\lineheight{1.25}\smash{\begin{tabular}[t]{l}$\underline{\text{data in $\mathfrak M_\mathrm{stab}$: $|e_\infty|=M_\infty$}}$\end{tabular}}}}%
    \put(0.70397359,0.35090144){\color[rgb]{0,0,0}\makebox(0,0)[lt]{\lineheight{1.25}\smash{\begin{tabular}[t]{l}$\underline{\text{data in $\mathfrak M_\mathrm{sub}$: $|e_\infty|<M_\infty$}}$\end{tabular}}}}%
    \put(0.01467861,0.35140309){\color[rgb]{0,0,0}\makebox(0,0)[lt]{\lineheight{1.25}\smash{\begin{tabular}[t]{l}$\underline{\text{data in $\mathfrak M_\mathrm{disp}$: no black hole}}$\end{tabular}}}}%
    \put(0,0){\includegraphics[width=\unitlength,page=2]{critical-behavior.pdf}}%
    \put(0.46903904,0.14413892){\color[rgb]{0,0,0}\rotatebox{45}{\makebox(0,0)[lt]{\lineheight{1.25}\smash{\begin{tabular}[t]{l}$\mathcal{H}^+$\end{tabular}}}}}%
    \put(0.80924488,0.14407369){\color[rgb]{0,0,0}\rotatebox{45}{\makebox(0,0)[lt]{\lineheight{1.25}\smash{\begin{tabular}[t]{l}$\mathcal{H}^+$\end{tabular}}}}}%
    \put(0.5758598,0.28152821){\color[rgb]{0,0,0}\makebox(0,0)[lt]{\lineheight{1.25}\smash{\begin{tabular}[t]{l}$i^+$\end{tabular}}}}%
    \put(0.91689539,0.27855679){\color[rgb]{0,0,0}\makebox(0,0)[lt]{\lineheight{1.25}\smash{\begin{tabular}[t]{l}$i^+$\end{tabular}}}}%
    \put(0.60865711,0.24004035){\color[rgb]{0,0,0}\rotatebox{-45}{\makebox(0,0)[lt]{\lineheight{1.25}\smash{\begin{tabular}[t]{l}$\mathcal I^+$\end{tabular}}}}}%
    \put(0.95275236,0.23867795){\color[rgb]{0,0,0}\rotatebox{-45}{\makebox(0,0)[lt]{\lineheight{1.25}\smash{\begin{tabular}[t]{l}$\mathcal I^+$\end{tabular}}}}}%
    \put(0.24942253,0.2612303){\color[rgb]{0,0,0}\rotatebox{-45}{\makebox(0,0)[lt]{\lineheight{1.25}\smash{\begin{tabular}[t]{l}$\mathcal I^+$\end{tabular}}}}}%
    \put(0.53153656,0.31954595){\color[rgb]{0,0,0}\rotatebox{-45}{\makebox(0,0)[lt]{\lineheight{1.25}\smash{\begin{tabular}[t]{l}$\mathcal{CH}^+$ \end{tabular}}}}}%
    \put(0.87208171,0.32041619){\color[rgb]{0,0,0}\rotatebox{-45}{\makebox(0,0)[lt]{\lineheight{1.25}\smash{\begin{tabular}[t]{l}$\mathcal{CH}^+$ \end{tabular}}}}}%
    \put(0,0){\includegraphics[width=\unitlength,page=3]{critical-behavior.pdf}}%
    \put(0.50384991,0.15367345){\color[rgb]{0,0,0}\rotatebox{60}{\makebox(0,0)[lt]{\lineheight{1.25}\smash{\begin{tabular}[t]{l}$g\to g_{M_\infty,e_\infty}$\end{tabular}}}}}%
    \put(0.84663534,0.14989681){\color[rgb]{0,0,0}\rotatebox{60}{\makebox(0,0)[lt]{\lineheight{1.25}\smash{\begin{tabular}[t]{l}$g\to g_{M_\infty,e_\infty}$\end{tabular}}}}}%
    \put(0.38625447,0.06316811){\color[rgb]{0,0,0}\rotatebox{-45}{\makebox(0,0)[lt]{\lineheight{1.25}\smash{\begin{tabular}[t]{l}$\underline C{}_\ing$\end{tabular}}}}}%
    \put(0.72851005,0.06235608){\color[rgb]{0,0,0}\rotatebox{-45}{\makebox(0,0)[lt]{\lineheight{1.25}\smash{\begin{tabular}[t]{l}$\underline C{}_\ing$\end{tabular}}}}}%
    \put(0.03462726,0.07336437){\color[rgb]{0,0,0}\rotatebox{-45}{\makebox(0,0)[lt]{\lineheight{1.25}\smash{\begin{tabular}[t]{l}$\underline C{}_\ing$\end{tabular}}}}}%
    \put(0.58248668,0.06684257){\color[rgb]{0,0,0}\rotatebox{45}{\makebox(0,0)[lt]{\lineheight{1.25}\smash{\begin{tabular}[t]{l}$C_\out$\end{tabular}}}}}%
    \put(0.23946958,0.06584848){\color[rgb]{0,0,0}\rotatebox{45}{\makebox(0,0)[lt]{\lineheight{1.25}\smash{\begin{tabular}[t]{l}$C_\out$\end{tabular}}}}}%
    \put(0.91911465,0.06395481){\color[rgb]{0,0,0}\rotatebox{45}{\makebox(0,0)[lt]{\lineheight{1.25}\smash{\begin{tabular}[t]{l}$C_\out$\end{tabular}}}}}%
    \put(0,0){\includegraphics[width=\unitlength,page=4]{critical-behavior.pdf}}%
    \put(0.64809366,0.30362254){\color[rgb]{0,0,0}\makebox(0,0)[lt]{\lineheight{1.25}\smash{\begin{tabular}[t]{l}trapped region\end{tabular}}}}%
    \put(0,0){\includegraphics[width=\unitlength,page=5]{critical-behavior.pdf}}%
  \end{picture}%
\endgroup%

%% file: figures/ERN.pdf_tex
\begingroup%
  \makeatletter%
  \providecommand\color[2][]{%
    \errmessage{(Inkscape) Color is used for the text in Inkscape, but the package 'color.sty' is not loaded}%
    \renewcommand\color[2][]{}%
  }%
  \providecommand\transparent[1]{%
    \errmessage{(Inkscape) Transparency is used (non-zero) for the text in Inkscape, but the package 'transparent.sty' is not loaded}%
    \renewcommand\transparent[1]{}%
  }%
  \providecommand\rotatebox[2]{#2}%
  \newcommand*\fsize{\dimexpr\f@size pt\relax}%
  \newcommand*\lineheight[1]{\fontsize{\fsize}{#1\fsize}\selectfont}%
  \ifx\svgwidth\undefined%
    \setlength{\unitlength}{115.60038289bp}%
    \ifx\svgscale\undefined%
      \relax%
    \else%
      \setlength{\unitlength}{\unitlength * \real{\svgscale}}%
    \fi%
  \else%
    \setlength{\unitlength}{\svgwidth}%
  \fi%
  \global\let\svgwidth\undefined%
  \global\let\svgscale\undefined%
  \makeatother%
  \begin{picture}(1,1.42171134)%
    \lineheight{1}%
    \setlength\tabcolsep{0pt}%
    \put(0,0){\includegraphics[width=\unitlength,page=1]{ERN.pdf}}%
    \put(0.72843557,0.75625535){\color[rgb]{0,0,0}\rotatebox{-45}{\makebox(0,0)[lt]{\lineheight{1.25}\smash{\begin{tabular}[t]{l}$\mathcal I^+$\end{tabular}}}}}%
    \put(0.53895523,0.97074267){\color[rgb]{0,0,0}\makebox(0,0)[lt]{\lineheight{1.25}\smash{\begin{tabular}[t]{l}$i^+$\end{tabular}}}}%
    \put(0.54829805,0.00776047){\color[rgb]{0,0,0}\makebox(0,0)[lt]{\lineheight{1.25}\smash{\begin{tabular}[t]{l}$i^-$\end{tabular}}}}%
    \put(0.94265685,0.53086408){\color[rgb]{0,0,0}\makebox(0,0)[lt]{\lineheight{1.25}\smash{\begin{tabular}[t]{l}$i^0$\end{tabular}}}}%
    \put(0.10574607,0.96950344){\color[rgb]{0,0,0}\makebox(0,0)[lt]{\lineheight{1.25}\smash{\begin{tabular}[t]{l}$\mathcal{BH}$\end{tabular}}}}%
    \put(0,0){\includegraphics[width=\unitlength,page=2]{ERN.pdf}}%
    \put(0.26755398,1.2234589){\color[rgb]{0,0,0}\rotatebox{-45}{\makebox(0,0)[lt]{\lineheight{1.25}\smash{\begin{tabular}[t]{l}$\mathcal{CH}^+$\end{tabular}}}}}%
    \put(0.5648809,0.30704566){\color[rgb]{0,0,0}\makebox(0,0)[lt]{\lineheight{1.25}\smash{\begin{tabular}[t]{l}$\Gamma$\end{tabular}}}}%
    \put(0.77350436,0.19339173){\color[rgb]{0,0,0}\rotatebox{45}{\makebox(0,0)[lt]{\lineheight{1.25}\smash{\begin{tabular}[t]{l}$\mathcal I^-$\end{tabular}}}}}%
    \put(0.20937056,0.53361517){\color[rgb]{0,0,0}\rotatebox{45}{\makebox(0,0)[lt]{\lineheight{1.25}\smash{\begin{tabular}[t]{l}$\mathcal H^+$\end{tabular}}}}}%
    \put(0.19571942,0.22598861){\color[rgb]{0,0,0}\rotatebox{-45}{\makebox(0,0)[lt]{\lineheight{1.25}\smash{\begin{tabular}[t]{l}$\mathcal H^-$\end{tabular}}}}}%
    \put(0.10712271,0.57212832){\color[rgb]{0,0,0}\rotatebox{45}{\makebox(0,0)[lt]{\lineheight{1.25}\smash{\begin{tabular}[t]{l}$r=M$\end{tabular}}}}}%
    \put(0.0149583,0.84298265){\color[rgb]{0,0,0}\rotatebox{90}{\makebox(0,0)[lt]{\lineheight{1.25}\smash{\begin{tabular}[t]{l}$r=0$\end{tabular}}}}}%
    \put(0.17464844,1.16283491){\color[rgb]{0,0,0}\rotatebox{-45}{\makebox(0,0)[lt]{\lineheight{1.25}\smash{\begin{tabular}[t]{l}$r=M$\end{tabular}}}}}%
    \put(0.64188639,0.52082121){\color[rgb]{0,0,0}\rotatebox{45}{\makebox(0,0)[lt]{\lineheight{1.25}\smash{\begin{tabular}[t]{l}$C_\out$\end{tabular}}}}}%
    \put(0.30930522,0.64534923){\color[rgb]{0,0,0}\rotatebox{-45}{\makebox(0,0)[lt]{\lineheight{1.25}\smash{\begin{tabular}[t]{l}$\underline C{}_\ing$\end{tabular}}}}}%
    \put(0,0){\includegraphics[width=\unitlength,page=3]{ERN.pdf}}%
  \end{picture}%
\endgroup%

%% file: figures/bootstrap-setup.pdf_tex
\begingroup%
  \makeatletter%
  \providecommand\color[2][]{%
    \errmessage{(Inkscape) Color is used for the text in Inkscape, but the package 'color.sty' is not loaded}%
    \renewcommand\color[2][]{}%
  }%
  \providecommand\transparent[1]{%
    \errmessage{(Inkscape) Transparency is used (non-zero) for the text in Inkscape, but the package 'transparent.sty' is not loaded}%
    \renewcommand\transparent[1]{}%
  }%
  \providecommand\rotatebox[2]{#2}%
  \newcommand*\fsize{\dimexpr\f@size pt\relax}%
  \newcommand*\lineheight[1]{\fontsize{\fsize}{#1\fsize}\selectfont}%
  \ifx\svgwidth\undefined%
    \setlength{\unitlength}{130.2406275bp}%
    \ifx\svgscale\undefined%
      \relax%
    \else%
      \setlength{\unitlength}{\unitlength * \real{\svgscale}}%
    \fi%
  \else%
    \setlength{\unitlength}{\svgwidth}%
  \fi%
  \global\let\svgwidth\undefined%
  \global\let\svgscale\undefined%
  \makeatother%
  \begin{picture}(1,1.03249467)%
    \lineheight{1}%
    \setlength\tabcolsep{0pt}%
    \put(0,0){\includegraphics[width=\unitlength,page=1]{bootstrap-setup.pdf}}%
    \put(0.16065444,0.25700745){\color[rgb]{0,0,0}\rotatebox{-45}{\makebox(0,0)[lt]{\lineheight{1.25}\smash{\begin{tabular}[t]{l}$v=0$\end{tabular}}}}}%
    \put(0.73075049,0.13623418){\color[rgb]{0,0,0}\rotatebox{45}{\makebox(0,0)[lt]{\lineheight{1.25}\smash{\begin{tabular}[t]{l}$u_{\tau_f}=0$\end{tabular}}}}}%
    \put(0,0){\includegraphics[width=\unitlength,page=2]{bootstrap-setup.pdf}}%
    \put(0.4714531,0.15859642){\color[rgb]{0,0,0}\makebox(0,0)[lt]{\lineheight{1.25}\smash{\begin{tabular}[t]{l}$\Gamma$\end{tabular}}}}%
    \put(0.00492191,0.85632703){\color[rgb]{0,0,0}\makebox(0,0)[lt]{\lineheight{1.25}\smash{\begin{tabular}[t]{l}$\kappa_{\tau_f}=1$\end{tabular}}}}%
    \put(0.91810596,0.87246224){\color[rgb]{0,0,0}\makebox(0,0)[lt]{\lineheight{1.25}\smash{\begin{tabular}[t]{l}$\gamma_{\tau_f}=-1$\end{tabular}}}}%
    \put(0.65941632,0.47004635){\color[rgb]{0,0,0}\rotatebox{45}{\makebox(0,0)[lt]{\lineheight{1.25}\smash{\begin{tabular}[t]{l}$C_u^{\tau_f},\mathcal E_p^{\tau_f}$\end{tabular}}}}}%
    \put(0.3071373,0.64913215){\color[rgb]{0,0,0}\rotatebox{-45}{\makebox(0,0)[lt]{\lineheight{1.25}\smash{\begin{tabular}[t]{l}$\underline C{}_v^{\tau_f}, \underline{\mathcal E}{}_p^{\tau_f}$\end{tabular}}}}}%
    \put(0,0){\includegraphics[width=\unitlength,page=3]{bootstrap-setup.pdf}}%
    \put(0.32486564,1.01915963){\color[rgb]{0,0,0}\makebox(0,0)[lt]{\lineheight{1.25}\smash{\begin{tabular}[t]{l}$\bar r_{\tau_f}(\Gamma(\tau_f))=\Lambda$\end{tabular}}}}%
    \put(0,0){\includegraphics[width=\unitlength,page=4]{bootstrap-setup.pdf}}%
    \put(0.39738241,0.211735){\color[rgb]{0,0,0}\rotatebox{45}{\makebox(0,0)[lt]{\lineheight{1.25}\smash{\begin{tabular}[t]{l}$H_u^{\tau_f}, \mathcal F^{\tau_f}$\end{tabular}}}}}%
    \put(0.6093631,0.35454309){\color[rgb]{0,0,0}\rotatebox{-45}{\makebox(0,0)[lt]{\lineheight{1.25}\smash{\begin{tabular}[t]{l}$\underline H{}_v^{\tau_f}$\end{tabular}}}}}%
    \put(0.67737158,0.41504352){\color[rgb]{0,0,0}\rotatebox{-45}{\makebox(0,0)[lt]{\lineheight{1.25}\smash{\begin{tabular}[t]{l}$\underline{\mathcal F}^{\tau_f}$\\\end{tabular}}}}}%
  \end{picture}%
\endgroup%

%% file: figures/stability.pdf_tex
\begingroup%
  \makeatletter%
  \providecommand\color[2][]{%
    \errmessage{(Inkscape) Color is used for the text in Inkscape, but the package 'color.sty' is not loaded}%
    \renewcommand\color[2][]{}%
  }%
  \providecommand\transparent[1]{%
    \errmessage{(Inkscape) Transparency is used (non-zero) for the text in Inkscape, but the package 'transparent.sty' is not loaded}%
    \renewcommand\transparent[1]{}%
  }%
  \providecommand\rotatebox[2]{#2}%
  \newcommand*\fsize{\dimexpr\f@size pt\relax}%
  \newcommand*\lineheight[1]{\fontsize{\fsize}{#1\fsize}\selectfont}%
  \ifx\svgwidth\undefined%
    \setlength{\unitlength}{130.2406275bp}%
    \ifx\svgscale\undefined%
      \relax%
    \else%
      \setlength{\unitlength}{\unitlength * \real{\svgscale}}%
    \fi%
  \else%
    \setlength{\unitlength}{\svgwidth}%
  \fi%
  \global\let\svgwidth\undefined%
  \global\let\svgscale\undefined%
  \makeatother%
  \begin{picture}(1,1.04470467)%
    \lineheight{1}%
    \setlength\tabcolsep{0pt}%
    \put(0,0){\includegraphics[width=\unitlength,page=1]{stability.pdf}}%
    \put(0.24715221,0.35678888){\color[rgb]{0,0,0}\rotatebox{45}{\makebox(0,0)[lt]{\lineheight{1.25}\smash{\begin{tabular}[t]{l}$\mathcal H^+=\{\hat u=\hat u_{\mathcal H^+}\}$\\$\quad=\{u_\infty=\infty\}$\end{tabular}}}}}%
    \put(0.72600811,0.82861625){\color[rgb]{0,0,0}\rotatebox{-45}{\makebox(0,0)[lt]{\lineheight{1.25}\smash{\begin{tabular}[t]{l}$\mathcal I^+, \nu_\infty=-1$\end{tabular}}}}}%
    \put(0.68441128,0.86876021){\color[rgb]{0,0,0}\makebox(0,0)[lt]{\lineheight{1.25}\smash{\begin{tabular}[t]{l}$i^+$\end{tabular}}}}%
    \put(0.31960812,0.70266752){\color[rgb]{0,0,0}\makebox(0,0)[lt]{\lineheight{1.25}\smash{\begin{tabular}[t]{l}$\mathcal{BH}$\end{tabular}}}}%
    \put(0.04260651,0.39039153){\color[rgb]{0,0,0}\rotatebox{-45}{\makebox(0,0)[lt]{\lineheight{1.25}\smash{\begin{tabular}[t]{l}$\{\hat v=0\}=\{v=0\}$\end{tabular}}}}}%
    \put(0.61930172,0.03609366){\color[rgb]{0,0,0}\rotatebox{45}{\makebox(0,0)[lt]{\lineheight{1.25}\smash{\begin{tabular}[t]{l}$\{\hat u=0\}=\{u_\infty=0\}$\end{tabular}}}}}%
    \put(0,0){\includegraphics[width=\unitlength,page=2]{stability.pdf}}%
    \put(0.55331292,0.99928403){\color[rgb]{0,0,0}\rotatebox{-45}{\makebox(0,0)[lt]{\lineheight{1.25}\smash{\begin{tabular}[t]{l}$\mathcal{CH}^+$\end{tabular}}}}}%
    \put(0,0){\includegraphics[width=\unitlength,page=3]{stability.pdf}}%
    \put(0.16727098,0.70727495){\color[rgb]{0,0,0}\rotatebox{45}{\makebox(0,0)[lt]{\lineheight{1.25}\smash{\begin{tabular}[t]{l}$\{\hat u=U_*\}$\end{tabular}}}}}%
    \put(0,0){\includegraphics[width=\unitlength,page=4]{stability.pdf}}%
    \put(0.62831582,0.37839035){\color[rgb]{0,0,0}\makebox(0,0)[lt]{\lineheight{1.25}\smash{\begin{tabular}[t]{l}$\Gamma$\\\end{tabular}}}}%
    \put(0.68280789,0.00614311){\color[rgb]{0,0,0}\makebox(0,0)[lt]{\lineheight{1.25}\smash{\begin{tabular}[t]{l}$\bar r=\bar r_\star$\end{tabular}}}}%
    \put(0,0){\includegraphics[width=\unitlength,page=5]{stability.pdf}}%
    \put(-0.51921481,0.70940323){\color[rgb]{0,0,0}\makebox(0,0)[lt]{\lineheight{1.25}\smash{\begin{tabular}[t]{l}no trapped surfaces\end{tabular}}}}%
    \put(0,0){\includegraphics[width=\unitlength,page=6]{stability.pdf}}%
    \put(0.63632215,0.47270984){\color[rgb]{0,0,0}\makebox(0,0)[lt]{\lineheight{1.25}\smash{\begin{tabular}[t]{l}$\kappa_\infty=1$\end{tabular}}}}%
  \end{picture}%
\endgroup%

%% file: figures/butterfly.pdf_tex
\begingroup%
  \makeatletter%
  \providecommand\color[2][]{%
    \errmessage{(Inkscape) Color is used for the text in Inkscape, but the package 'color.sty' is not loaded}%
    \renewcommand\color[2][]{}%
  }%
  \providecommand\transparent[1]{%
    \errmessage{(Inkscape) Transparency is used (non-zero) for the text in Inkscape, but the package 'transparent.sty' is not loaded}%
    \renewcommand\transparent[1]{}%
  }%
  \providecommand\rotatebox[2]{#2}%
  \newcommand*\fsize{\dimexpr\f@size pt\relax}%
  \newcommand*\lineheight[1]{\fontsize{\fsize}{#1\fsize}\selectfont}%
  \ifx\svgwidth\undefined%
    \setlength{\unitlength}{128.68440763bp}%
    \ifx\svgscale\undefined%
      \relax%
    \else%
      \setlength{\unitlength}{\unitlength * \real{\svgscale}}%
    \fi%
  \else%
    \setlength{\unitlength}{\svgwidth}%
  \fi%
  \global\let\svgwidth\undefined%
  \global\let\svgscale\undefined%
  \makeatother%
  \begin{picture}(1,0.78199001)%
    \lineheight{1}%
    \setlength\tabcolsep{0pt}%
    \put(0,0){\includegraphics[width=\unitlength,page=1]{butterfly.pdf}}%
    \put(0.45508353,0.28638132){\color[rgb]{0,0,0}\makebox(0,0)[lt]{\lineheight{1.25}\smash{\begin{tabular}[t]{l}$\Gamma$\end{tabular}}}}%
    \put(0,0){\includegraphics[width=\unitlength,page=2]{butterfly.pdf}}%
    \put(0.66962418,0.47121108){\color[rgb]{0,0,0}\makebox(0,0)[lt]{\lineheight{1.25}\smash{\begin{tabular}[t]{l}$\mathcal R_{\ge\Lambda}$\end{tabular}}}}%
    \put(0.23250829,0.47015214){\color[rgb]{0,0,0}\makebox(0,0)[lt]{\lineheight{1.25}\smash{\begin{tabular}[t]{l}$\mathcal R_{\le\Lambda}$\end{tabular}}}}%
    \put(0.17951927,0.22363701){\color[rgb]{0,0,0}\rotatebox{-45}{\makebox(0,0)[lt]{\lineheight{1.25}\smash{\begin{tabular}[t]{l}I\end{tabular}}}}}%
    \put(0.80242934,0.21150807){\color[rgb]{0,0,0}\rotatebox{45}{\makebox(0,0)[lt]{\lineheight{1.25}\smash{\begin{tabular}[t]{l}II\end{tabular}}}}}%
    \put(0.85740322,0.68840907){\color[rgb]{0,0,0}\rotatebox{-45}{\makebox(0,0)[lt]{\lineheight{1.25}\smash{\begin{tabular}[t]{l}III\end{tabular}}}}}%
    \put(0.5910003,0.65822687){\color[rgb]{0,0,0}\rotatebox{45}{\makebox(0,0)[lt]{\lineheight{1.25}\smash{\begin{tabular}[t]{l}IV\end{tabular}}}}}%
    \put(0.3618237,0.70274093){\color[rgb]{0,0,0}\rotatebox{-45}{\makebox(0,0)[lt]{\lineheight{1.25}\smash{\begin{tabular}[t]{l}V\end{tabular}}}}}%
    \put(0.06992374,0.63176608){\color[rgb]{0,0,0}\rotatebox{45}{\makebox(0,0)[lt]{\lineheight{1.25}\smash{\begin{tabular}[t]{l}VI\end{tabular}}}}}%
    \put(0,0){\includegraphics[width=\unitlength,page=3]{butterfly.pdf}}%
    \put(0.39056685,0.78742316){\color[rgb]{0,0,0}\makebox(0,0)[lt]{\lineheight{1.25}\smash{\begin{tabular}[t]{l}$(u_2,v_2)$\end{tabular}}}}%
    \put(0.64355413,0.00448479){\color[rgb]{0,0,0}\makebox(0,0)[lt]{\lineheight{1.25}\smash{\begin{tabular}[t]{l}$(u_1,v_1)$\end{tabular}}}}%
    \put(0,0){\includegraphics[width=\unitlength,page=4]{butterfly.pdf}}%
  \end{picture}%
\endgroup%

%% file: figures/square.pdf_tex
\begingroup%
  \makeatletter%
  \providecommand\color[2][]{%
    \errmessage{(Inkscape) Color is used for the text in Inkscape, but the package 'color.sty' is not loaded}%
    \renewcommand\color[2][]{}%
  }%
  \providecommand\transparent[1]{%
    \errmessage{(Inkscape) Transparency is used (non-zero) for the text in Inkscape, but the package 'transparent.sty' is not loaded}%
    \renewcommand\transparent[1]{}%
  }%
  \providecommand\rotatebox[2]{#2}%
  \newcommand*\fsize{\dimexpr\f@size pt\relax}%
  \newcommand*\lineheight[1]{\fontsize{\fsize}{#1\fsize}\selectfont}%
  \ifx\svgwidth\undefined%
    \setlength{\unitlength}{68.74015748bp}%
    \ifx\svgscale\undefined%
      \relax%
    \else%
      \setlength{\unitlength}{\unitlength * \real{\svgscale}}%
    \fi%
  \else%
    \setlength{\unitlength}{\svgwidth}%
  \fi%
  \global\let\svgwidth\undefined%
  \global\let\svgscale\undefined%
  \makeatother%
  \begin{picture}(1,1.08911337)%
    \lineheight{1}%
    \setlength\tabcolsep{0pt}%
    \put(0,0){\includegraphics[width=\unitlength,page=1]{square.pdf}}%
    \put(0.41550955,0.53756346){\color[rgb]{0,0,0}\makebox(0,0)[lt]{\lineheight{1.25}\smash{\begin{tabular}[t]{l}$\Gamma$\end{tabular}}}}%
    \put(0,0){\includegraphics[width=\unitlength,page=2]{square.pdf}}%
    \put(0.14170376,0.19660043){\color[rgb]{0,0,0}\rotatebox{-45}{\makebox(0,0)[lt]{\lineheight{1.25}\smash{\begin{tabular}[t]{l}i\end{tabular}}}}}%
    \put(0.83234451,0.19398779){\color[rgb]{0,0,0}\rotatebox{45}{\makebox(0,0)[lt]{\lineheight{1.25}\smash{\begin{tabular}[t]{l}ii\end{tabular}}}}}%
    \put(0,0){\includegraphics[width=\unitlength,page=3]{square.pdf}}%
    \put(0.75114125,0.98836031){\color[rgb]{0,0,0}\makebox(0,0)[lt]{\lineheight{1.25}\smash{\begin{tabular}[t]{l}$(u_2,v_2)$\end{tabular}}}}%
    \put(0.76833428,0.00984154){\color[rgb]{0,0,0}\makebox(0,0)[lt]{\lineheight{1.25}\smash{\begin{tabular}[t]{l}$(u_1,v_1)$\end{tabular}}}}%
    \put(0,0){\includegraphics[width=\unitlength,page=4]{square.pdf}}%
    \put(0.7645823,0.87028928){\color[rgb]{0,0,0}\rotatebox{-45}{\makebox(0,0)[lt]{\lineheight{1.25}\smash{\begin{tabular}[t]{l}iii\end{tabular}}}}}%
    \put(0.17770535,0.79407377){\color[rgb]{0,0,0}\rotatebox{45}{\makebox(0,0)[lt]{\lineheight{1.25}\smash{\begin{tabular}[t]{l}iv\end{tabular}}}}}%
  \end{picture}%
\endgroup%

%% file: figures/open.pdf_tex
\begingroup%
  \makeatletter%
  \providecommand\color[2][]{%
    \errmessage{(Inkscape) Color is used for the text in Inkscape, but the package 'color.sty' is not loaded}%
    \renewcommand\color[2][]{}%
  }%
  \providecommand\transparent[1]{%
    \errmessage{(Inkscape) Transparency is used (non-zero) for the text in Inkscape, but the package 'transparent.sty' is not loaded}%
    \renewcommand\transparent[1]{}%
  }%
  \providecommand\rotatebox[2]{#2}%
  \newcommand*\fsize{\dimexpr\f@size pt\relax}%
  \newcommand*\lineheight[1]{\fontsize{\fsize}{#1\fsize}\selectfont}%
  \ifx\svgwidth\undefined%
    \setlength{\unitlength}{131.18456503bp}%
    \ifx\svgscale\undefined%
      \relax%
    \else%
      \setlength{\unitlength}{\unitlength * \real{\svgscale}}%
    \fi%
  \else%
    \setlength{\unitlength}{\svgwidth}%
  \fi%
  \global\let\svgwidth\undefined%
  \global\let\svgscale\undefined%
  \makeatother%
  \begin{picture}(1,0.85238369)%
    \lineheight{1}%
    \setlength\tabcolsep{0pt}%
    \put(0,0){\includegraphics[width=\unitlength,page=1]{open.pdf}}%
    \put(0.16689118,0.23638418){\color[rgb]{0,0,0}\rotatebox{-45}{\makebox(0,0)[lt]{\lineheight{1.25}\smash{\begin{tabular}[t]{l}$\hat v=0$\end{tabular}}}}}%
    \put(0.73567162,0.1389791){\color[rgb]{0,0,0}\rotatebox{45}{\makebox(0,0)[lt]{\lineheight{1.25}\smash{\begin{tabular}[t]{l}$\hat u=0$\end{tabular}}}}}%
    \put(0,0){\includegraphics[width=\unitlength,page=2]{open.pdf}}%
    \put(0.48848102,0.30789768){\color[rgb]{0,0,0}\makebox(0,0)[lt]{\lineheight{1.25}\smash{\begin{tabular}[t]{l}$\Gamma$\end{tabular}}}}%
    \put(0,0){\includegraphics[width=\unitlength,page=3]{open.pdf}}%
    \put(0.86352036,0.61586959){\color[rgb]{0,0,0}\makebox(0,0)[lt]{\lineheight{1.25}\smash{\begin{tabular}[t]{l}$\hat{\mathcal D}_{\tau_f}$\end{tabular}}}}%
    \put(0.72653614,0.74274564){\color[rgb]{0,0,0}\makebox(0,0)[lt]{\lineheight{1.25}\smash{\begin{tabular}[t]{l}$\hat{\mathcal D}_{\tau_f+\eta}$\end{tabular}}}}%
    \put(0.02414452,0.73896993){\color[rgb]{0,0,0}\makebox(0,0)[lt]{\lineheight{1.25}\smash{\begin{tabular}[t]{l}$\hat{\mathcal D}_{\tau_f}^\mathrm{ext}$\end{tabular}}}}%
    \put(0,0){\includegraphics[width=\unitlength,page=4]{open.pdf}}%
  \end{picture}%
\endgroup%